\documentclass[a4paper,11pt]{report}
\usepackage{amssymb}
\usepackage{amsmath,amsfonts}
\usepackage{amscd}
\usepackage{amsthm}
\usepackage{latexsym}
\usepackage{fancyhdr}
\usepackage{hyperref}
\hypersetup{
    colorlinks,
    citecolor=blue,
    filecolor=black,
    linkcolor=blue,
    urlcolor=black
}
\input xy
\xyoption{all}
\usepackage{xcolor}
\colorlet{RED}{black}
\colorlet{BLUE}{black}
\usepackage[titles]{tocloft}

\makeatletter

 \newtheorem{thm}{Theorem}[section]
 \newtheorem{cor}[thm]{Corollary}
 \newtheorem{lem}[thm]{Lemma}
 \newtheorem{prop}[thm]{Proposition}
 \theoremstyle{definition}
 \newtheorem{defn}[thm]{Definition}
 \theoremstyle{remark}
 \newtheorem{rem}[thm]{Remark}
 \numberwithin{equation}{chapter}

\title{\bf{\textsf{A mathematical perspective on \\ the phenomenology
of \\ non-perturbative Quantum Field Theory}}}

\vspace{2in}

\author{\bf{Ali Shojaei-Fard}}

\date{\textsc{\today}}

\begin{document}

\maketitle \pagenumbering{arabic}
\pagestyle{fancy} \lhead{{\small \textsf{Ali Shojaei-Fard}}}
\chead{} \rhead{{\tiny \textsf{\textbf{The phenomenology of non-perturbative QFT: \\
Mathematical Perspective}}}} \lfoot{{\small \textsf{}}} \cfoot{\thepage} \rfoot{}

\thispagestyle{empty}

\newpage

\verb"Monograph" \\
\verb"A mathematical perspective on the phenomenology of"\\
\verb"non-perturbative Quantum Field Theory" \\

\vspace{0.3in}

\verb"Ali Shojaei-Fard" \\
$\bullet$ \verb"Postdoctoral Research Fellow, Max-Planck-Institut fur Mathematik," \\
\verb"Vivatsgasse 7, 53111 Bonn, Germany." \\
$\bullet$ \verb"Ph.D. Mathematician and Independent Scholar,"\\
\verb"1461863596 Marzdaran Blvd., Tehran, Iran." \\
\verb"Email: shojaeifa@yahoo.com" \\
\verb"Email: shojaei-fard@mpim-bonn.mpg.de"\\

\vspace{0.3in}

\textbf{$\bullet$ Keywords:} \\
{\footnotesize
\verb"Quantum Field Theory and Feynman Diagrams;" \\
\verb"Dyson--Schwinger Equations and Connes--Kreimer Renormalization Hopf Algebra;" \\
\verb"Non-Perturbative Renormalization Theory via Combinatorics and Infinite Graph Theory;" \\
\verb"Evolution in Quantum Motions via Functional Analysis and Differential Calculus;" \\
\verb"Multi-Scale Non-Perturbative Renormalization Group and Theory of Computation;" \\
\verb"Noncommutative (Differential) Geometry Methods;" \\
\verb"QFT Entanglement via Connes--Marcolli Differential Galois Theory and Tannakian Formalism;" \\
\verb"QFT Logic and Foundations via Category Theory."\\
}

\vspace{0.3in}

\textbf{$\bullet$ Mathematics Subject Classification 2010:}
\begin{verbatim}
05E15; 06A11; 16T05; 60B15; 81P10; 81Q30; 81T16; 81T17; 81T18; 81T75
\end{verbatim}

\vspace{0.3in}

\textbf{$\bullet$ Physics and Astronomy Classification Scheme:}
\begin{verbatim}
02.10.De; 02.10.Ox; 02.30.Sa; 02.40.Pc; 03.65.Ud; 03.70.+k; 11.15.-q
\end{verbatim}

\vspace{0.3in}

\textbf{{\large{$\bullet$ Acknowledgment.}}}
\begin{verbatim}
The researcher is grateful to Max Planck Institute for Mathematics
for the support and hospitality.
\end{verbatim}

\newpage

\thispagestyle{empty}
\begin{abstract}
This monograph aims to build some new mathematical structures originated from Dyson--Schwinger equations for the description of non-perturbative aspects of gauge field theories whenever bare or running coupling constants are strong enough.
\end{abstract}



\newpage

\thispagestyle{empty}

{\small \setcounter{tocdepth}{3} \tableofcontents}


\addcontentsline{toc}{chapter}{PREFACE}

\newpage

\chapter*{PREFACE}

{\footnotesize{\it The strength of Mathematics is its ability to create models \\
which are absolutely vital for producing physical parameters. \\
A mathematician is like a surrealist painter who can design \\
the purest portraits from known and unknown universes.}}

\vspace{0.2in}

Recent discoveries in Science and Technology from the smallest to the largest scales have approved clearly the importance of advanced research
activities in Basic Science which can bring a new package of fundamental knowledge for the analysis of complicated systems in natural phenomena. To obtain a comprehensive description of those complexities requires to build
any possible
interrelations among different fields in Mathematics (as the purest
mental production of human beings). The resulting connections can lead us
to achieve some new theoretical methodologies which are essential tools for scientists
to build advanced practical models in dealing with complexities of the nature. The designed models together with some
computational algorithms will lead scientists to solution
procedures.

This research work has a multidisciplinary foundation in the context
of Mathematics, High Energy Theoretical Physics and Theoretical
Computer Science. It plans to discover some new knowledge about the
most complicated or unknown parts of Quantum Field Theories whenever the coupling constants are strong enough in terms of building some new advanced mathematical structures. The
outstanding consequence of this research work is to provide a new
mathematical interpretation of the phenomenology of
Quantum Field Theory with strong coupling constants under discrete,
analytic and logical settings. If we study simultaneously these different
but related settings, then our mathematical outputs will be useful for the better understanding of the
behavior of quantum physical systems in non-perturbative situations. We study the phenomenology of non-perturbative aspects of gauge field theories in terms of some new combinatorial, geometric and categorical tools. We apply graph limits to formulate a new analytic generalization for solutions of Dyson--Schwinger equations which is useful for a complete theory of renormalization for these non-perturbative equations. We also show the use of Tutte polynomials for the combinatorial representation of solutions of Dyson--Schwinger equations. Then we propose a new concept of complexity for the study of Dyson--Schwinger equations under different running coupling constants. The analytic behavior of these equations are also discussed via some new noncommutative geometric tools. We finally apply graphon models of Dyson--Schwinger equations to build a new theory of ordered algebraic sub-structures for the analysis of entanglement in strongly coupled gauge field theories and a new topos model for the logical description of topological regions of Feynman diagrams which encode non-perturbative parameters. All these observations enable us to describe non-perturbative phenomenology and intrinsic foundations of strongly coupled gauge field theories under discrete, analytic and logical platforms.

The Lagrangian approach to Quantum Field Theory, which is on the
basis of the Feynman path integral formalism, has made extraordinary
theoretical and experimental progress for the study of elementary
particles and their interactions at the highest level of energies
and the smallest scales under a perturbative setting. This approach
encodes physical information of a quantum system with infinite degrees of freedom in terms of Green's functions as infinite formal expansions of ill-defined iterated integrals and powers of coupling constants.

Quantum Electrodynamics (QED) concerns interactions among matter
(electron, positron) and light (photons). Quantum Flavourdynamics
(QFD) concerns weak interactions inside the nucleus of an atom which
change the flavour or type of quarks to describe $\beta^{-}$ decay
and $\beta^{+}$ decay under $W,Z$ bosons. Quantum Chromodynamics
(QCD) concerns strong interactions of quarks and gluons inside the
nucleus of an atom to build composite hadrons such as protons and
neutrons. Standard Model, as the most successful experienced model,
has provided a practical platform to collect quantum field theories
corresponding to electromagnetic, weak and strong interactions into
a united Quantum Field Theory model. The modified versions of the
Standard Model in the context of Noncommutative Geometry have also
provided a new updated (theoretical) model which is (minimally)
coupled to gravity as the weakest fundamental force in the nature.
The constructions of gauge field theories in Theoretical and
Experimental High Energy Physics, as updated Quantum Field Theory
models, are on the basis of the modified Standard Model of
elementary particles. These gauge theories can analyze electroweak and strong
interactions of elementary particles in the scale of distances down
to the order of $10^{-16}$ centimeters while neutrino masses have
also been accounted. In addition, under a more theoretical setting,
String Theory as the other class of Quantum Field Theory models, which
does not have ultraviolet divergencies, has been introduced and
developed to deal with gravity in terms of Quantum Field Theories with matrix fields and higher generalizations of matrix models. The classical one-loop Feynman diagram should be
replaced with its stringy counterpart which is a torus. Other more general Feynman diagrams should be
replaced with Riemann surfaces and world sheets. This
mathematical theory is capable of describing Quantum Gravity in
Space-Time.

The first fundamental challenge in perturbative setting is the
appearance of so complicated nested (sub-)divergencies which live in
each term of Green's functions. These ill-defined terms, known as
Feynman integrals, can be theoretically reduced to some finite
values as the result of the renormalization machinery and many loop
techniques. However some extra parameters (i.e. counterterms) should be
added to the original Lagrangian of the physical theory during the extraction of finite values. The
discovery of a comultiplication structure hidden inside of the
(Bogoliubov)--Zimmermann's forest formula has led us to understand
the Bogoliubov--Parasiuk--Hepp--Zimmermann perturbative
renormalizaton in the language of the Connes--Kreimer Hopf algebra
of Feynman diagrams and the Riemann--Hilbert problem. Thanks to this
setting, a geometric interpretation of Dimensional Regularization on
the basis of flat equi-singular connections had been formulated by
Connes and Marcolli. This study was lifted immediately onto a universal
categorical setting where it is possible to associate a category of Lie group
representations to each renormalizable Quantum Field Theory. Thanks to this platform, nowadays there exists
a diverse spectrum of advanced mathematical techniques and
tools to deal with ill-defined iterated Feynman integrals in
physical theories to generate finite values from
infinities.

The second fundamental challenge in perturbative setting is dealing
with complicated infinities originated from Green's functions which
encode quantum motions in physical theories with strong couplings.
The lack of a rigorous mathematical methodology for the study of
aspects beyond perturbation boundary has made so many difficulties
to understand completely Quantum Field Theory. In physical theories
with strong (running or bare) couplings, it is already impossible to
study the full behavior of quantum systems under perturbation series
and in this situation, we need to concern non-perturbative methods
such as numerical methods, Borel summation method, theory of
instantons and lattice model. In addition, the self-similar nature
of Green's functions makes an alternative way for us to study
non-perturbative aspects in the context of fixed point equations of
Green's functions. The resulting equations, which are known as
Dyson--Schwinger equations, contain an infinite collection of coupled integral
equations depended on the coupling constants. In QCD
with higher energies and short distances, the strength of couplings is small enough where we can expect the asymptotic freedom behavior of the physical system. In this situation, we can consider Dyson--Schwinger equations via
some perturbative tools such as many loop computational techniques. However in QCD with relatively lower
energies and long distances, the strength of couplings is more than
or equal to 1 where the physical system these equations behave non-perturbatively. In this situation, we need to consider non-linear Dyson--Schwinger equations via non-perturbative methods of computations. Work on the phenomenology of
strong running couplings in QCD has been considered under a physical
perspective to provide some computational methods in dealing with
non-perturbative parameters. Thanks to the applications of the
Connes--Kreimer renormalization Hopf algebra of Feynman diagrams to
Quantum Field Theory, we already have a combinatorial reformulation
for Dyson--Schwinger equations in the language of Hochschild
cohomology theory. The unique solution of each equation DSE
determines a free commutative connected graded Hopf subalgebra of
the renormalization Hopf algebra. This mathematical approach to
Dyson--Schwinger equations has already provided some new
combinatorial and geometric tools for the computation of some
non-perturbative parameters. The foundations of a differential
Galois theory and a Tannakian formalism for the study of
non-perturbative aspects of Quantum Field Theories have been
designed and developed by the author on the basis of the Connes--Marcolli
universal category of flat equi-singular vector bundles. The author applied these platforms to clarify a new motivic method for the study of (systems of ) Dyson--Schwinger equations in terms of sub-categories of mixed Tate motives. In addition, the author has shown a new method for identifying the amount of non-computability of non-perturbative parameters in the context of the renormalization Hopf algebra of the Halting problem in the theory of computation.

This research work proposes some new applications of mathematical
tools originated from Combinatorics, Functional Analysis,
Noncommutative Geometry, Category Theory and Logic to deal with
infinite graphs generated by solutions of Dyson--Schwinger
equations. These mathematical platforms can provide some new
techniques for the computation of non-perturbative parameters. In
addition, they suggest a new methodology for the description of the
intrinsic foundations of strongly coupled gauge field theories (such as quantum
entanglement and quantum logic) under a non-perturbative setting.
These investigations will help us to understand the indeterministic
nature of non-perturbative Quantum Field Theory models.

Generally speaking, the achievements of this research work can
improve our knowledge about the phenomenology of non-perturbative
Quantum Field Theory in terms of studying an individual Dyson--Schwinger equations or studying these equations with respect to each other.

In the first level, we focus on the mathematical foundations of
Dyson--Schwinger equations to bring some new computational tools in
dealing with non-perturbative parameters generated by large Feynman
diagrams. At this level, we consider each Dyson--Schwinger equation
as an individual object in the vector space $\mathcal{S}^{\Phi,g}$ generated by all Dyson--Schwinger
equations derived from Green's functions of a given Quantum Field
Theory $\Phi$ under different scales $\lambda g$ of the bare
coupling constant $g$ where $0 < \lambda \le 1$. We equip this
infinite dimensional vector space with a topological structure
defined via the graphon representations of Feynman diagrams. Under a
combinatorial setting, we discuss the structure of a new model for
large Feynman diagrams in the language of combinatorial polynomials
and random graphs. Furthermore, we discuss the complexity of
non-perturbative parameters generated by Dyson--Schwinger equations
in the context of theory of computation. In this direction we try to show the importance of a new multi-scale non-perturbative Renormalization Group for the description of the Kolmogorov complexity in dealing with Dyson--Schwinger equations. Under a geometric setting,
we explain the dynamics of non-perturbative aspects in a physical
theory with respect to the mathematical structures originated from
Dyson--Schwinger equations. We build a Noncommutative Geometry model
for each Dyson--Schwinger equation which leads us to interpret
quantum motions in the context of theory of spectral triples and
noncommutative differential forms. Under a Functional Analysis
setting, we discuss the evolution of fixed point equations of
Green's functions by defining a new generalized version of the
Fourier transformation on the Banach algebra $L^{1}(\mathcal{S}^{\Phi,g},\mu_{{\rm Haar}})$ with respect to a new Haar measure $\mu_{{\rm Haar}}$ integration theory on the space of solutions of Dyson--Schwinger equations.

In the second level, we focus on the mathematical foundations of
non-perturbative Quantum Field Theory where we must deal with all
possible Dyson--Schwinger equations under different running couplings. We explain
the construction of a new Hopf algebra structure
$\mathcal{S}^{\Phi}_{{\rm graphon}}$ on the topological vector space of graphons which contribute to representations of Feynman diagrams and
their finite or infinite formal expansions. The resulting
topological Hopf algebra is capable to encode large Feynman diagrams
generated by solutions of Dyson--Schwinger equations in different
rescalings of the bare coupling constant $g$. Therefore we can embed
the space $\mathcal{S}^{\Phi,g}$ into $\mathcal{S}^{\Phi}_{{\rm
graphon}}$. This Hopf algebra leads us to formulate a new topological Hopf algebraic
renormalization theory for large Feynman diagrams via a topological generalization of the Connes--Kreimer BPHZ renormalization program. We then define a new multi-scale Renormalization Group on
the collection $\mathcal{S}^{\Phi,g}$ to control the simultaneously rescaling of the momentum parameter and the bare coupling constant in Dyson--Schwinger equations of a given gauge field theory $\Phi$. This new Renormalization Group is useful to study the computational complexity of a strongly coupled equation ${\rm DSE}(g)$ in terms of the computational complexities of a sequence of Dyson--Schwinger equations under weaker running couplings. Furthermore, we explain the foundations of a differential calculus
theory on $\mathcal{S}^{\Phi,g}$ (as a separable Banach space with respect to the cut-distance topology) where thanks to the Feynman graphon models of Dyson--Schwinger equations and the theory of
G\^{a}teaux derivative, we formulate a new theory of differentiation and
integration on $\mathcal{S}^{\Phi,g}$. Under an ordered algebraic setting, we organize Feynman graphon models of solutions of Dyson--Schwinger equations into some lattices of substructures to provide a new interpretation for the quantum entanglement in interacting gauge field theories. We explain mathematically the information flow among elementary (virtual) particles in QFT models via a new class of
topological Hopf algebras generated by Dyson--Schwinger equations. We then lift this lattice models onto the level of Tannakian sub-categories of the Connes--Marcolli category of flat equi-singular vector
bundles to show the importance of this universal category for the study of the geometry of quantum entanglement. Under a categorical setting, we organize topological Hopf subalgebras derived from solutions of Dyson--Schwinger equations into a new small category to build a new topos model which can encode the logical evaluations of non-perturbative situations. The structure of this topos model is depended on the strength of running couplings to show the importance of the phenomenology of strong running couplings in the determination of truth values of propositions about Dyson--Schwinger equations.

Thanks to these two levels of observations, we expect to provide a
new insight into the complicated problems of non-perturbative
situations where the strength of the coupling constants do really
change the mathematics and the logics of quantum theory models.


\chapter{\textsf{Introduction}}

\vspace{1in}

$\bullet$ \textbf{\emph{Physical backgrounds}} \\
$\bullet$ \textbf{\emph{Mathematical backgrounds}} \\
$\bullet$ \textbf{\emph{Recent progress and objectives}}

\newpage

\section{\textsl{Physical backgrounds}}

Modern Theoretical and Experimental High Energy Physics have been
established on the basis of Quantum Field Theory models under the
Lagrangian setting which is (minimally) coupled to gravity via the
incorporation of massive neutrinos. The foundations of Quantum Field
Theory were initiated in terms of the interpretation of the
quantized version of Electrodynamics in the language of the Feynman
path integral formalism under a perturbation setting. The appearance
of gauge field theories which include Quantum Electrodynamics (QED),
Electroweak theory, Quantum Chromodynamics (QCD) and Quantum Gravity
have developed rigorously our theoretical knowledge about the
fundamental properties of elementary particles before we could reach
to appropriate empirical information. Thanks to these backgrounds,
mathematicians and theoretical physicists have already made
outstanding achievements for the description of interactions of
elementary particles under different settings in the context of
advanced mathematical models. For example, mathematical tools in
Noncommutative Geometry, Algebraic Geometry, Combinatorics and
Category Theory have been applied to formulate Standard Model and
other extended theories which include supersymmetry, gravitational
interactions or extended objects such as strings and brane theory.
We can also address tensor models as higher dimensional
generalizations of matrix models which aim to achieve a theory of
random geometries in dimensions higher than two. This class of
theories helps us for the construction of discrete approaches
to quantizing gravity. \cite{brouder-frabetti-1, brouder-frabetti-2,
breitenlohner-maison-1, calaque-strobl-1, connes-marcolli-1,
duboisviolette-4, krajewski-toriumi-1, krajewski-toriumi-2,
kreimer-6, kreimer-sars-suijlekom-1, malyshev-1, marino-3,
mirkasimov-1, nair-1, roberts-1, roberts-schmidt-1, tanasa-2,
weinzierl-2, yeats-1}

The Lagrangian formalism enables us to understand Quantum Field
Theory by working on Green's functions as infinite formal expansions
of Feynman integrals or their corresponding diagrams where the
amount of some fundamental parameters such as the strength of the
bare coupling constants or the domain of momenta make the
resulting series divergent or asymptotic free. In perturbative
physical theories we expect to have some convergent series.

For example, in $\phi^{4}$ model the partition function is given by
\begin{equation}
Z[B]:= \int \mathcal{D} \phi  \  e^{-L(\phi)+\int B\phi}
\end{equation}
such that
\begin{equation}
L(\phi)=\int d^{4}x  \  \big(\frac{1}{2} (\nabla\phi)^{2} +
\frac{1}{2}r_{0}\phi^{2}+\frac{1}{4!}u_{0}\phi^{4} \big) ,
\end{equation}
and $B$ is an external field. If we set
\begin{equation}
Z_{0}:= \int \mathcal{D} \phi \ e^{-L_{0}(\phi)+\int B\phi}, \ \ \ \
\ L_{0}(\phi):= \int d^{4}x \ \big(\frac{1}{2} (\nabla\phi)^{2} +
\frac{1}{2}r_{0}\phi^{2} \big),
\end{equation}
then we can develop $Z$ as a series in $u_{0}$ around $Z_{0}$ to
achieve
\begin{equation}
Z=\int \mathcal{D} \phi \ \big(1-\frac{u_{0}}{4!}\int_{x_{1}}
\phi^{4}(x_{1}) + \frac{1}{2}(\frac{u_{0}}{4!})^{2}
\int_{x_{1},x_{2}} \phi^{4}(x_{1}) \phi^{4}(x_{2}) + ...\big)
e^{-L_{0}(\phi) + \int B \phi}.
\end{equation}
This expansion can be represented in the language of Feynman
diagrams which leads us to a combinatorial formulation for Green's
functions. For $n$ elementary particles we have
\begin{equation}
G_{n}(x_{1},...,x_{n}) = <\phi(x_{1}) ... \phi(x_{n})> = \frac{\int e^{-S[\phi]} \phi(x_{1}) ... \phi(x_{n}) \mathcal{D} \phi}{\int e^{-S[\phi]} \mathcal{D} \phi}
\end{equation}
such that $S[\phi] = S_{0}[\phi] + g S_{{\rm int}}[\phi]$ where $g$ is the (bare) coupling constant.
The fluctuations generated by the $\phi^{4}$ term around
the Gaussian integral $Z_{0}$ are large where they determine
iterated integrals with the general form
\begin{equation}
\int^{\Lambda} d^{d}q_{1}...d^{d}q_{L} \prod_{i} ({\rm
propagator(q_{i})})
\end{equation}
such that the ultra-violet regulator $\Lambda$ provides a cut-off at
the upper bound type of integral. The dependency of these integrals
to the parameter $\Lambda$ makes a rigorous challenge for the
computation of universal quantities. Theory of perturbative
renormalization provides the machinery to reparametrize the
perturbative expansion in such a way that the sensitive dependence
on $\Lambda$ has been eliminated. In this situation, the
renormalization group enables us to partially resum the perturbative
expansions to achieve some universal computational results.
\cite{breitenlohner-maison-1, calaque-strobl-1, nair-1,
wilson-kogut-1}

In general speaking, it is possible to investigate the situations beyond perturbation theory in terms of some expressions such as
\begin{equation} \label{coupling-poly-1}
P(g) = X_{0}+X_{1}g+X_{2}g^{2}+...+X_{n}g^{n}+...
\end{equation}
such that $g$ is the coupling constant and each term $X_{i}$
represents the class of Feynman diagrams which contribute to the
$i$-order of perturbative expansion. It is obvious that physical
theories with very small $g$ could be encoded by only some beginning
finite number of terms from the above expansion while physical
theories with strong coupling $g$ produce infinite number of terms.
These non-perturbative aspects have been concerned in Theoretical
Physics via Dyson--Schwinger equations as a quantized version of the
Euler--Lagrange equations of motion originated from the principal of
the least action. These equations, which can be determined by
fixed point equations of Green's functions, have been studied under
analytic and numerical methods in Theoretical and Mathematical
Physics such that we can address standard techniques such as Borel
summation, theory of instantons, lattice models, etc in dealing with
these equations to generate some estimations for non-perturbative
parameters. In physical theories with strong couplings such as
Hadron Physics, we should deal with hadrons such as protons and
neutrons as the composite particles build up from quarks and gluons
(as elementary particles under strong interaction). In general, QCD,
as a nonabelian gauge theory with the symmetry group ${\rm SU}(3)$,
has provided a modern understanding of the complicated nature of
hadrons and nuclei where we study the strong interactions of quarks
and gluons under confinement and chiral symmetry breaking. The
appearance of nuclear weak force enables us to describe any change in
quark's flavour via $W$ bosons. In fact, the weak force is not only
responsible for interactions between particles, but it also allows
heavy particles to decay by emitting or absorbing some of the force
carriers. We can describe QCD as a matrix-valued modification of
electromagnetic theory in terms of replacing photons by gluons and
electrons by quarks while quantum fluctuations of the fields could
determine the force law. The quantization of Chromodynamics involves
the regularization and renormalization of ultraviolet divergencies
which generate a mass-scale where mass-dimensionless quantities
become dependent on a mass-scale. The current quark-masses are the
only evident scales in QCD. The main experimental confirmations of
QCD have been investigated at high energies and high momentum
transfers (or short distances) where the QCD coupling is small and
correspondingly the forces are weak. This situation, which is known
as asymptotic freedom property, enables us to detect the composite
structure of hadrons by scattering high energy electrons. The most
difficult challenge in this model can be observed when
perturbation theory fails to describe the short range static
potential obtained from quenched lattice simulations where the
difference between the non-perturbatively determined potential and
perturbation theory at short distance has been parameterized by a
linear term. In addition, there are also so many difficulties for
the study of the asymptotic behavior of QCD-perturbative series
beyond the two-loop level where the original effort is to find a way
to subtract perturbative contributions to a given physical process
in order to isolate non-perturbative terms. In the domain of
relatively low energies and momentum transfers such as $Q^{2} \sim 1
- 5 \ {\rm GeV}^{2}$ while the proton's mass is approximately $1 \
{\rm GeV}$, the QCD coupling constants are larger where many loops
perturbative calculations should be applied. Because of the nontrivial
vacuum structure of QCD, in the domain of lower energies and
momentum transfers (or large distances) such as $Q^{2} \le 1 \ {\rm
GeV}^{2}$, the QCD coupling constants are stronger than one. Under this condition,
the analytic calculations do not useful but there are some methods
such as the chiral effective theory, lattice calculations, large $N$
limits and Dyson--Schwinger equations to provide some algorithmic
computations. This situation, which is actually the failure to
directly observe coloured excitations in a detector, is the origin
of the concept of confinement as a fundamental challenge that we do not
see free quarks or gluons in nature but rather we only see
colourless. The analytic description of confinement is one difficult
task for the understanding of continuum QCD. The phenomenology
of confinement can be studied in the context of Dyson--Schwinger
equations. \cite{bali-1, calaque-strobl-1, delamotte-1, marino-1,
marino-2, marino-3, roberts-schmidt-1}

The situations beyond perturbation boundary deal with divergencies
originated from Green's functions under strong running coupling constants where we might need to consider Feynman diagrams with many numbers or infinite number of loops. What does this class of diagrams means and
how can we deal with infinite formal expansions of this class of diagrams?
Applying advanced mathematical structures is helpful for the
better understanding of this class of divergencies. It is obvious
that having strong mathematical modelings is the initial step
to make fundamental progress in dealing with the complicated behavior
of elementary particles inside of the nuclei. We can address the mathematical
foundations of the (modified) Standard Model as a good theoretical
methodology which led scientists to obtain strong experimental
investigations about elementary particles. String Theory is another
powerful mathematical platform which aims to provide a unified
theory for the study of elementary particles with respect to all
fundamental forces. The basic philosophy of this research work is to
build and develop the mathematical foundations of QFT models with strong couplings where as the consequence, we expect
to provide some new mathematical tools for the better understanding
of physical parameters beyond perturbation theory. Our study
suggests a new interpretation from the phenomenology of strong
couplings in the context of combinatorial, geometric and categorical
settings.

\section{\textsl{Mathematical backgrounds}}

The contributions of mathematical tools to Quantum Field Theories
have been extraordinary developed when the (Bogoliubov--)Zimmermann
forest formula was reinterpreted by Kreimer in the context of
(co)algebraic combinatorial tools. This reinterpretation had been
concerned by Connes and Kreimer to build a new modern formulation for the
Bogoliubov--Parasiuk--Hepp--Zimmermann (BPHZ) perturbative
renormalization in Quantum Field Theory on the basis of the theory
of Hopf algebras and the Riemann--Hilbert problem. The
Connes--Kreimer approach has become the main foundation in many
research efforts for the study of complicated Feynman integrals,
Green's functions and Renormalization Group methods where it has led the
Theoretical Physics's community to achieve some new mathematical
tools for the description of  physical parameters in
(renormalizable) gauge field theories under algebraic, combinatorial
and geometric settings. It is now possible to encapsulate the machinery of perturbative renormalization in terms of a connected graded free commutative
non-cocommutative (finite type) Hopf algebraic structure $H_{{\rm
FG}}(\Phi)$ on Feynman diagrams of a physical theory $\Phi$ which
has a Lie algebraic nature determined by the insertion operator. The
compatibility of the fundamental identities such as Slavnov--Taylor
and Ward identities in QCD and QED with the renormalization
coproduct have been shown in the language of Hopf ideals. The
phenomenology of counterterms has been concerned underlying a
geometric treatment to provide some alternative methods for the computation of these physical values in the
language of (singular) differential equations. In this setting, a
new class of equi-singular flat connections can govern the values of
counterterms with respect to the $\beta$-function. This setting has
been lifted onto a universal Tannakian formalism where a
renormalizable Quantum Field Theory is studied via a category
of geometric objects which can be recovered by the neutral Tannakian category of
finite dimensional representations of the affine group scheme
$\mathbb{G}_{\Phi}:={\rm Hom}(H_{{\rm FG}}(\Phi),-).$
\cite{agarwala-1, agarwala-2, borinsky-1, brouder-1,
brouder-frabetti-1, brouder-frabetti-2, brouder-frabetti-menous-1,
broadhurst-kreimer-1, brown-kreimer-1, brown-yeats-1,
connes-kreimer-1, connes-kreimer-2, connes-marcolli-1,
figueroa-graciabondia-1, figueroa-graciabondia-2, kreimer-1,
kreimer-7, kreimer-8, kreimer-9, kreimer-panzer-1,
mencattini-kreimer-1, suijlekom-1, suijlekom-2, weinzierl-1,
yeats-1}

One of the most fundamental result in this direction is actually the
discovery of a very deep interrelationship between Feynman integrals
and the theory of motives where a motivic
renormalization machinery has been formulated to deal with
divergencies in the language of Picard--Fuchs equations and other
powerful tools. The theory of motives in Algebraic Geometry aims to concern the
existence of a universal cohomology theory for algebraic varieties
defined over a base field $k$ while taking values into an abelian
tensor category. The construction of a category of motives (mixed
motives) related to general varieties is a difficult task. The
noncommutative version of motivic objects provides the motivic
cohomology applied in the construction of a universal cohomology
theory. The structure of mixed Tate motives as elements of the
subring $\mathbb{Z}[\mathbb{L}]$ of the Grothendieck group
$K_{0}({\rm Var}_{k})$ of $k-$varieties has been considered where
$\mathbb{L}:=[\mathbb{A}^{1}]$ is the Grothendieck class of the
affine line. The application of motives enables us to develop a
unified setting underlying different cohomology theories such as
Betti, de Rham, l-adic, crystalline and etale. For this purpose, the
construction of an abelian tensor category, which provides a
linearization of the category of algebraic varieties, has been
studied to provide some fundamental requirements of standard
conjectures of Grothendieck. The importance of motives in Quantum Field Theory have been discussed in
different settings. The Bloch--Esnault--Kreimer approach which
informs interesting applications of Hodge type structures in the
calculation processes of Feynman integrals underlying graph
polynomials \cite{bloch-esnault-kreimer-1, bloch-kreimer-1,
kreimer-9, weinzierl-2}. The Aluffi--Marcolli approach, which builds
the motivic version of Feynman rules characters,
applied Kirchhoff--Symanzik polynomials to formulate a new version of
algebro-geometric (dimensionally regularized) Feynman rules
characters. These abstract characters send classes in the
Grothendieck ring of conical immersed affine varieties to the
classes in the Grothendieck ring of varieties spanned by the classes
$[X_{\Gamma}]$. This formalism, which is on the basis of the
deletion--contraction operators and the Tutte--Grothendieck polynomial,
enables us to relate Feynman diagrams with periods of algebraic
varieties. This framework provides a motivic treatment in the study
of perturbative renormalization process at the level of the
universal motivic Feynman rules character \cite{aluffi-marcolli-1,
aluffi-marcolli-2}. The Connes--Marcolli approach deals with the
geometric interpretation of counterterms on the basis of flat
equi-singular connections such that these geometric objects have
been organized in a categorical structure $\mathcal{E}^{\Phi}$ which
is recovered by the neutral Tannakian category of finite
dimensional representations of the affine group scheme
$\mathbb{G}_{\Phi}$. This category has been embedded (as a
sub-category) into the universal category $\mathcal{E}^{{\rm
CM}}$ of flat equi-singular vector bundles with the corresponding
universal affine group scheme $\mathbb{U}$. Objects of this
universal category, which has Tannakian nature, address mixed Tate
motives which contribute to divergencies of renormalizable physical
theories. In addition, $\mathcal{E}^{{\rm CM}}$ determines the
universal singular frame as the unique loop with values in
$\mathbb{U}$ which provides the universal counterterm. The Lie
algebra of the universal affine group scheme leads us to formulate a
particular shuffle type Hopf algebra. \cite{connes-marcolli-1,
marcolli-2, marcolli-1}

A single Feynman diagram reports only a small piece of information
about a finite number of possible interactions among (virtual)
elementary particles where its on-shell part (i.e. incoming and
outgoing particles) obeys the mass-energy equation and the conservation
of momenta while its off-shell part (i.e. virtual particles) obeys
no special rules or measurements. The iterated integral
corresponding to a given Feynman diagram might contain nested or overlapping sub-divergencies. Infinite formal expansions of Feynman diagrams (as polynomials with
respect to coupling constants), which are organized in Green's functions, have an essential role to
encapsulate various possible chains of interactions which could or
might happen among elementary particles in a physical theory. These formal
expansions can be studied in terms of the self-similar nature of Green's functions which allows us
to formulate fixed point equations known as Dyson--Schwinger
equations. Numerical methods, large $N$ limit,
Borel resummation, lattice models and theory of instantons can provide some approximations for strongly coupled Dyson--Schwinger equations (\cite{marino-1, marino-2, marino-3}) while the real time dynamics of these non-perturbative equations require other advanced mathematical tools. In addition, these non-perturbative equations in physical theories
with strong couplings have also been considered in the context of the
Connes--Kreimer renormalization Hopf algebra to provide some new
advanced mathematical tools for the computation of their solutions.
This Hopf algebraic formalism is one of the original motivations of
this research program. In this direction we need to formulate the Hochschild cohomology of (commutative) bialgebras
as the dual notion of the Hochschild cohomology of (commutative)
algebras. For a given commutative Hopf algebra $H$, consider linear
maps $T:H \rightarrow H^{\otimes n}$ as $n$-cochains where the
coboundary operator is defined by
\begin{equation} \label{grafting-1}
\bold{b}T:= ({\rm id} \otimes T) \Delta + \sum^{n}_{i=1} (-1)^{i}
\Delta_{i} T + (-1)^{n+1} T \otimes \mathbb{I}
\end{equation}
such that $\Delta_{i}$ is the coproduct $\Delta$ of $H$ applied to
the $i$-th factor in $H^{\otimes n}$. The Kreimer's renormalization
coproduct on Feynman diagrams can be reformulated recursively in
terms of the linear operator $B^{+}$ on Feynman diagrams, known as the
grafting operator, as the following way
\begin{equation} \label{cop-rec-1}
\Delta_{{\rm FG}} B^{+} = ({\rm id} \otimes B^{+}) \Delta_{{\rm FG}}
+ B^{+} \otimes \mathbb{I}.
\end{equation}
The operator $B^{+}$, as a linear homogeneous endomorphism of degree
one, replaces a vertex in a given Feynman diagram with a whole graph
in terms of the type of the targeting vertex and the types of external edges of
the second graph. Thanks to (\ref{grafting-1}) and
(\ref{cop-rec-1}), the grafting operator could determine some generators of the first
rank Hochschild cohomology group of the Connes--Kreimer Hopf algebra
of Feynman diagrams. For each primitive Feynman diagram
$\gamma$, $B^{+}_{\gamma}$ is a Hochschild one cocycle.
\cite{connes-kreimer-1, ebrahimifard-kreimer-1,
ebrahimifard-manchon-1, kreimer-3}

The first importance of the grafting operator is its role for the rooted tree
representation of any complicated Feynman diagram in terms of its primary components (or its primitive Feynman subdiagrams). It is possible to embed the Connes--Kreimer renormalization Hopf algebra of Feynman diagrams of a given gauge field theory into a decorated version of the Connes--Kreimer Hopf algebra of non-planar rooted trees.
The grafting operator acts on each forest
$t_{1}...t_{n}$ to deliver a rooted tree by adding a new vertex $r$
as the root and $n$ new edges which connect the roots of $t_{n}$s
to $r$. It is shown that the pair $(H_{{\rm CK}},B^{+})$ has a universal property with respect to the Hochschild cohomology theory. This pair is actually the initial object for a particular category of objects $(H,T)$
consisting of a commutative Hopf algebra $H$ and a Hochschild one
cocycle $T$ on $H$. The Hopf algebra homomorphisms which commute
with the cocycles are morphisms of this category. Decorations on trees enable us to update $H_{{\rm CK}}$ with
respect to each physical theory. For a given physical theory $\Phi$, each 1PI Feynman diagram without any
sub-divergencies is a primitive element in the renormalization Hopf algebra
$H_{{\rm FG}}(\Phi)$. We can encode this class of objects via vertices in rooted trees where edges can determine the positions of those primitive subgraphs inside of a more complicated Feynman diagram. It is shown the existence of an injective Hopf algebraic
homomorphism from $H_{{\rm FG}}(\Phi)$ to the Hopf algebra $H_{{\rm
CK}}(\Phi)$ of decorated non-planar rooted trees. \cite{brouder-1,
brouder-frabetti-1, brouder-frabetti-2, ebrahimifard-guo-kreimer-1,
ebrahimifard-guo-kreimer-2, hoffman-1, kreimer-1, kreimer-10,
kreimer-11}

The second importance of the grafting operator is its fundamental
role in the reconstruction of Dyson--Schwinger equations under a
combinatorial setting. For a given family $\{\gamma_{n}\}_{n \ge 1}$
of primitive (1PI) Feynman diagrams with the corresponding
Hochschild one cocycles $\{B^{+}_{\gamma_{n}}\}_{n \ge 1}$, a class
of combinatorial Dyson--Schwinger equations in $H_{{\rm FG}}(\Phi)[[\lambda g]]$ is defined
by
\begin{equation} \label{dse-1}
X = \mathbb{I} + \sum_{n \ge 1} (\lambda g)^{n} \omega_{n} B^{+}_{\gamma_{n}}
(X^{n+1})
\end{equation}
such that $g$ is the bare coupling constant. This class of equations
accepts a unique solution $X= \sum_{n \ge 0} (\lambda g)^{n} X_{n}$ as the formal
expansion of finite Feynman diagrams where for each $n>0$, we have
\begin{equation} \label{dse-4}
X_{n} = \sum_{j=1}^{n} \omega_{j} B^{+}_{\gamma_{j}}(\sum_{k_{1}+
... + k_{j+1}=n-j, \ k_{i} \ge 0} X_{k_{1}}...X_{k_{j+1}}).
\end{equation}
$X_{0}$ is the empty graph, each
$X_{n}$ is an object in the Hopf algebra $H_{{\rm FG}}(\Phi)$, and
the unique solution $X$ lives in a completion of $H_{{\rm FG}}(\Phi)[[\lambda g]]$ with respect
to the $n$-adic topology. The Cartier--Quillen--Milnor--Moore
theorem shows us that the unique solution of each Dyson--Schwinger equation DSE can
determine the generators of a Faa di Bruno type Hopf subalgebra
$H_{{\rm DSE}}$ of the Connes--Kreimer renormalization Hopf algebra. It is a free commutative unitial counital
(non-)cocommutative connected graded finite type Hopf subalgebra where
its coproduct on generators $X_{n}$ does not depend on the
parameters $\omega_{j}$. The Mellin transform allows us to deform
these combinatorial type of equations to their original integral
versions. It is possible to lift this formalism onto the level of
systems of Dyson--Schwinger equations where we deal with a system
$(S)$ of a finite collection of equations with the general form
\begin{equation}
(S): \ \ \ \forall i \in I, \ \ \ \ x_{i} = \sum_{j \in J_{i}}
B^{+}_{(i,j)}(f^{(i,j)}(x_{k}, \ k \in I))
\end{equation}
such that $I:=\{1,...,n\}$, $J_{i}$ is a graded connected set,
$B^{+}_{(i,j)}$s are Hochschild one cocycles and $f^{(i,j)}$s are
formal series in $\mathbb{K}[[\alpha_{1},...,\alpha_{n}]]$. It is
shown that the system $(S)$ has a unique solution such that under
some conditions it can determine the Hopf subalgebra $H_{(S)}$ originated from the Hopf subalgebras $H_{1},...,H_{n}$ generated by
combinatorial Dyson--Schwinger equations in the system.
\cite{bergbauer-kreimer-1, cartier-1, ebrahimifard-fauvet-1,
foissy-4, foissy-1, foissy-2, kreimer-3, kreimer-5}

The main skeleton of a combinatorial Dyson--Schwinger equation is
actually a family of Hochschild one cocycles. There exists a
surjective map from the first rank Hochschild cohomology group to
the space of primitive Feynman diagrams of the renormalization Hopf
algebra. It means that each family $\{\gamma_{n}\}_{n \ge 1}$ of
primitive Feynman diagrams determines the corresponding family
$\{B^{+}_{\gamma_{n}}\}_{n \ge 1}$ of Hochschild one cocycles. It is
important to note that each 1PI Feynman diagram, which is free of
sub-divergencies, is a primitive element but they are not the only
primitives in the renormalization Hopf algebra. In other words,
there are primitive Feynman diagrams in higher degrees which can
determine Hochschild one cocycles. \cite{bergbauer-kreimer-1,
connes-kreimer-1, kreimer-11}

\section{\textsl{Recent progress and objectives}}

The combinatorial reformulation of Dyson--Schwinger equations in
terms of the renormalization Hopf algebra and Hochschild cohomology
theory has applied recently for the construction of some new
mathematical treatments in dealing with the computations of non-perturbative parameters. These efforts have already developed our knowledge about the mathematical foundations of the phenomenology of non-perturbative
parameters under different settings. \cite{broadhurst-kreimer-2,
ebrahimifard-fauvet-1, foissy-2, foissy-3, kreimer-3, kreimer-4,
kreimer-5, kreimer-yeats-1, kruger-kreimer-1, marie-yeats-1,
shojaeifard-8, submitted-1, submitted-2, submitted-3,
tanasa-kreimer-1, vanbaalen-kreimer-uminsky-yeats-1, weinzierl-3}

The Milnor--Moore theorem (\cite{milnor-moore-1}) allows us to
determine the infinite dimensional complex graded pro-unipotent Lie group
$\mathbb{G}_{\Phi}(\mathbb{C})$ which is actually the complex points of the affine
group scheme $\mathbb{G}_{\Phi} = {\rm Hom}(H_{{\rm FG}}(\Phi),-)$.
This Lie group, which is the projective limit of linear algebraic groups $G_{n}$ embedded as Zariski closed
subsets in some ${\rm GL}_{m_{n}}$s, is rich enough to encode (dimensionally regularized) Feynman
rules characters with respect to the scale and angle dependence of
amplitudes \cite{brown-kreimer-1}. In addition, we can also
determine the infinite dimensional complex graded pro-unipotent Lie group
\begin{equation}
\mathbb{G}_{{\rm DSE}}(\mathbb{C}):= {\rm Hom}(H_{{\rm DSE}},
\mathbb{C})
\end{equation}
for each given Dyson--Schwinger equation DSE. There exists a natural
injective Hopf algebra homomorphism $\rho: H_{{\rm DSE}}
\rightarrow H_{{\rm FG}}(\Phi)$. If we apply ${\rm Spec}$ as a
contravariant functor, then we can obtain a surjective morphism
$\tilde{\rho}: {\rm Spec}(H_{{\rm FG}}(\Phi)) \rightarrow {\rm
Spec}(H_{{\rm DSE}})$ between spaces of prime ideals in the
commutative algebras $H_{{\rm FG}}(\Phi)$ and $H_{{\rm DSE}}$
equipped with the Zariski topology. This map can be lifted onto the
surjective group homomorphism
\begin{equation} \label{dse-lie-group-1}
\overline{\rho}: \mathbb{G}_{\Phi}(\mathbb{C}) \longrightarrow
\mathbb{G}_{{\rm DSE}}(\mathbb{C}).
\end{equation}
The Hopf subalgebra $H_{{\rm DSE}}$ is actually a Hopf ideal in $H_{{\rm FG}}(\Phi)$ and we can consider its corresponding quotient Hopf algebra $H_{{\rm FG}}(\Phi) / H_{{\rm DSE}}$. In the dual setting, the associated sub-group scheme of this quotient Hopf algebra can determine a Lie sub-group of $\mathbb{G}_{\Phi}(\mathbb{C})$. However the existence of Lie (sub)groups $\mathbb{G}_{{\rm DSE}}(\mathbb{C})$
corresponding to Dyson--Schwinger equations have been applied to
bring a new geometric setting for the study of non-perturbative
parameters in the context of differential systems together with
singularities. The construction of a category of flat equi-singular
$\mathbb{G}_{{\rm DSE}}(\mathbb{C})$-connections with respect to
each equation DSE has been addressed to encode the BPHZ
renormalization of the unique solution $X_{{\rm DSE}}$ in the
context of differential Galois theory. In other words, the
Connes--Marcolli geometric interpretation of counterterms and the
Connes--Marcolli universal Tannakian machinery in dealing with
renormalizable physical theories have been developed for the study
of Dyson--Schwinger equations. We can study the geometry of any given equation DSE in terms of finite dimensional representations of the Lie group
$\mathbb{G}_{{\rm DSE}}(\mathbb{C})$ which is organized in the neutral
Tannakian category ${\rm Rep}_{\mathbb{G}^{*}_{{\rm DSE}}}$. This class of categories has
been embedded as subcategories into the Connes--Marcolli universal
category $\mathcal{E}^{{\rm CM}}$. Thanks to these backgrounds, we
already have the construction of a differential Galois theory for
the computation of fundamental non-perturbative parameters such as
global $\beta$-functions and non-perturbative counterterms in the
language of Picard--Fuchs equations. In addition, it is now possible
to identify a class of mixed Tate motives with respect to each
Dyson--Schwinger equation. \cite{shojaeifard-1, shojaeifard-3,
shojaeifard-4, shojaeifard-5}

On the first hand, under the Hopf algebraic renormalization platform, Dyson--Schwinger
equations are the fundamental tools to determine Hopf subalgebras in Quantum Field Theory. On the second hand, ''substructure'' is a mutual fundamental concept in Galois Theory and Theory of Computation where intermediate algorithms have been studied recently in terms of algebraic methods in Galois Theory.
On the third hand, the Manin's program for the interpretation of the
Halting problem in the context of the Connes--Kreimer BPHZ
renormalization machinery has initiated the foundations of a
brilliant interrelationship between the amount of computability and
the computation of counterterms encoded via the renormalization
Hopf algebra of the Halting problem. The combination of these achievements have already led us to describe intermediate algorithms in terms of (systems of) Dyson--Schwinger equations where we can understand the amount of non-computability of non-perturbative renormalized values and their corresponding non-perturbative counterterms in the language of the Halting problem, Operad Theory and the theory of Hall sets.
Results in this direction show a new multidisciplinary bridge between Theoretical Computer Science and non-perturbative Quantum Field Theory. \cite{delaney-marcolli-1, manin-2, manin-3,
manin-4, shojaeifard-8, yanofsky-1, yanofsky-2}

Applications of the theory of graphons to Quantum Field Theory is
another output of the renormalization Hopf algebraic platform. The foundations of
a new combinatorial interpretation of Feynman diagrams and their
infinite formal expansions have been studied recently where we
embed (large) Feynman diagrams into a compact topological space
obtained by an enrichment of the Connes--Kreimer renormalization
Hopf algebra with respect to the cut-distance topology. The immediate
consequence of this new topological-combinatorial setting is the
formulation of a generalization of the BPHZ renormalization for
solutions of Dyson--Schwinger equations. In this direction, thanks
to some tools in Measure Theory, a new differential calculus
machinery on Feynman diagrams was built which has led us to
study the evolution of Dyson--Schwinger equations in terms of their
partial sums. \cite{shojaeifard-9,
shojaeifard-10}

Applications of non-commutative differential graded algebras to
Quantum Field Theory were also considered to study the geometry of
quantum motions where some new models of gauge theories have been
obtained. In this direction, the structure of a non-perturbative
version of the Connes--Kreimer renormalization group has been
described in the language of integrable systems. In addition, recently, we have formulated a new class of infinite dimensional spectral triples which encode the geometry of Dyson--Schwinger equations. These observations enable us to clarify the essential role of Noncommutative Geometry in the future of non-perturbative Quantum Field Theory models.
\cite{duboisviolette-1, duboisviolette-2, duboisviolette-3,
duboisviolette-kerner-madore-1, shojaeifard-2, shojaeifard-7, submitted-1}

Combinatorial Dyson--Schwinger
equations, Connes--Kreimer--Marcolli Hopf algebraic
renormalization platform and the theory of graphons for sparse graphs are the main motivational topics for us in this
research to study non-perturbative Quantum Field Theory. Our main
attempt in this work is to develop mathematical structures
originated from Dyson--Schwinger equations to discover some new
information about the complicated behavior of Quantum Field Theories under
strong coupling constants. This work aims also to bring some new
mathematical tools to deal with the computation of non-perturbative
parameters.

Under a combinatorial setting, we plan to apply the theory of
graphon representations of Feynman diagrams together with other
combinatorial and topological tools to provide a new Hopf algebraic
machinery for the renormalization of solutions of Dyson--Schwinger equations under different running coupling constants. In a more general platform, we consider the cut-distance topological vector space $\mathcal{S}^{\Phi,g}$ generated by all fixed point equations of Green's functions of a given strongly coupled gauge field theory $\Phi$ under different running coupling constants $\lambda g$ with respect to the bare coupling constant $g$ such that $0 < \lambda \le 1$. It is possible to topologically complete the space $\mathcal{S}^{\Phi,g}$ to achieve a new Banach space where we need to work on equivalence classes of Dyson--Schwinger equations up to the weakly isomorphic as an equivalence relation. Equations ${\rm DSE}_{1}, {\rm DSE}_{2}$ are called weakly isomorphic (or weakly equivalent) if their corresponding solutions $X_{{\rm DSE}_{1}}, X_{{\rm DSE}_{2}}$ have weakly isomorphic Feynman graphon models. In other words,
\begin{equation}
{\rm DSE}_{1} \approx {\rm DSE}_{2} \iff W_{X_{{\rm DSE}_{1}}} \sim W_{X_{{\rm DSE}_{2}}}.
\end{equation}
Up to this equivalence relation, we can define the distance
\begin{equation}
d(X_{{\rm DSE}_{1}},X_{{\rm DSE}_{2}}):= d_{{\rm cut}}([W_{X_{{\rm
            DSE}_{1}}}],[W_{X_{{\rm DSE}_{2}}}])
\end{equation}
between Dyson--Schwinger equations such that $d_{{\rm cut}}$ is the metric structure on unlabeled Feynman graphon classes. Each unlabeled Feynman graphon class contains all weakly isomorphic Feynman graphons and all relabeled Feynman graphons generated by invertible measure preserving transformations of the base measure space $(\Omega,\mu_{\Omega})$ of our graphon model. We then address some new applications of
combinatorial polynomials such as Kirchhoff--Symanzik and Tutte
polynomials to formulate a new parametric representation theory for solutions of Dyson--Schwinger equations. This study is useful for the construction of
algebro-geometric Feynman rules on the topological Hopf algebra
$\mathcal{S}^{\Phi}_{{\rm graphon}}$ of Feynman graphons. We also concern
the concept of complexity for the description of non-perturbative parameters where we explain the construction of a new multi-scale Renormalization Group machinery on $\mathcal{S}^{\Phi,g}$ which is useful on two levels. Firstly, it can provide a mathematical machinery for the approximation of Dyson--Schwinger equations in strong couplings via equations in weaker couplings. Secondly, it helps us to initiate a new version of the Kolmogorov complexity in dealing with Dyson--Schwinger equations.

Under a geometric setting, we
show some new applications of Noncommutative Geometry, Measure
Theory and Functional Analysis to study the geometry of
non-perturbative Quantum Field Theory. We build the G\^{a}teaux differential calculus on the space $\mathcal{S}^{\Phi,g}$ which provides a global differential geometry for the geometric description of strongly coupled gauge field theories. In addition, we explain the concept of evolution on
$\mathcal{S}^{\Phi,g}$ with respect to a generalized version of the
Fourier transformation. We also build the theory of spectral triples for solutions of Dyson--Schwinger equations which provides a local differential geometry in dealing with the quantum motions in strongly coupled gauge field theories.

Under a foundational setting, we apply Feynman graphon models to consider some intrinsic foundations of Quantum Field Theory models under strong running coupling constants such as quantum entanglement and logical concepts. In this
direction, we offer a new mathematical methodology for the
description of quantum entanglement in interacting quantum physical theories in the context of the theory of lattices and intermediate
substructures in the Theory of Computation. This mathematical formalism
enables us to explain information flow in physical theories with
strong couplings on the basis of lattices of topological Hopf
algebras and Lie subgroups. We lift this mathematical modeling onto a
categorical setting to show that the universal category
$\mathcal{E}^{{\rm CM}}$ is suitable to encode the quantum entanglement
process. At this level, we expect to show a new application of
motives in dealing with information flow. In addition, we put forward
the construction of a new topos of presheaves on a particular base
category which encodes the logical evaluations of propositions about the cut-distance topological regions of Feynman
diagrams. We try to show that how this topos model concerns the strength of running coupling constants for the logical study of gauge field theories. For this purpose, objects of the base category of the non-perturbative topos encode all Dyson--Schwinger equations of a given gauge field theory under different running coupling constants. The non-perturbative topos can evaluate the logical propositions about cut-distance topological regions which are generated by solutions of Dyson--Schwinger equations.


\chapter{\textsf{A theory of renormalization for Dyson--Schwinger equations}}

\vspace{1in}

$\bullet$ \textbf{\emph{Quantum Field Theory}} \\
$\bullet$ \textbf{\emph{Hochschild cohomology of the renormalization
Hopf algebra}} \\
$\bullet$ \textbf{\emph{Renormalization Hopf algebra of Feynman
graphons and filtration of large Feynman diagrams}} \\
$-$ \textbf{\emph{Graphons}} \\
$-$ \textbf{\emph{Feynman graphons}} \\
$\bullet$ \textbf{\emph{A generalization of the BPHZ renormalization machinery for large Feynman diagrams via Feynman graphons}}

\newpage

Having no comprehensive physical description of infinite formal
expansions of Feynman diagrams which contribute to polynomials with
respect to strong running coupling constants such as
(\ref{coupling-poly-1}), it is indeed difficult to analyze a
renormalization program for these infinite expansions. The major
discourse in this situation is to find some meaningful mathematical
insights associated to fixed point equations of Green's
functions. Then these mathematical reasonings serve to compel
technical terms and models for the explanation of the
renormalization process for Dyson--Schwinger equations.

The original task in this chapter is to explain a renormalization
program on the space $\mathcal{S}^{\Phi,g}$ which consists of
Dyson--Schwinger equations with respect to running couplings
$\lambda g$ in a given physical theory $\Phi$ with the bare coupling
constant $g$. For this purpose, we apply the theory of graphons for sparse graphs to
build a new theory of graph function representations for Feynman diagrams and large Feynman diagrams which contribute to solutions of Dyson--Schwinger equations. These new analytic generalizations of Feynman diagrams can be organized into a new graded topological Hopf algebra
$\mathcal{S}^{\Phi}_{{\rm graphon}}$ which can lead us to formulate a new modification of the
Connes--Kreimer BPHZ renormalization program for the computation of non-perturbative counterterms and non-perturbative renormalized values. The complex
Lie group associated to the Hopf algebra $\mathcal{S}^{\Phi}_{{\rm
        graphon}}$ is the key tools for the formulation of a non-perturbative
Renormalization Group which encodes the required $\beta$-functions to govern the dynamics of strongly coupled Dyson--Schwinger equations.

\section{\textsl{Quantum Field Theory}}

Suppose we have a quantized field theory with the general Lagrangian
$\mathcal{L}=\mathcal{L}(\phi,\partial_{\mu} \phi)$ which can contain the free part and the interaction part. The typical free Lagrangian density is given
by
\begin{equation}
 \mathcal{L}_{{\rm free}}:= \frac{1}{2}\big( (\partial_{\mu} \phi)(\partial^{\mu}\phi) - m^{2}\phi^{2} \big)
\end{equation}
which leads us to the free Klein--Gordon equation of motion
\begin{equation}
\big(\partial_{\mu}\partial^{\mu} + m^{2} \big) \phi =0.
\end{equation}
The interaction part $\mathcal{L}_{{\rm int}}$ encodes interactions
of elementary particles in the physical theory. The transition
amplitudes from initial states to all finite states can be formulated
under the S-Matrix setting. It is possible to calculate these matrix elements in
terms of a class of correlation functions with the general form
\begin{equation}
G_{n}(x_{1},...,x_{n}):= <0|T\phi(x_{1})...\phi(x_{n})|0>
\end{equation}
such that $|0>$ is the vacuum ground state. These equations, known
as Green's functions, allow us to formulate perturbative Quantum
Field Theory in terms of formal expansions with the general form
$$G_{n}(x_{1},...,x_{n})= $$
\begin{equation}
\sum_{j=1}^{\infty} \frac{(-1)^{j}}{j!} \int d^{4}y_{1}...d^{4}y_{j}
<0|T\phi_{{\rm in}}(x_{1}) ... \phi_{{\rm in}}(x_{n})
\mathcal{L}_{{\rm int}}(y_{1})...\mathcal{L}_{{\rm int}}(y_{j})|0>
\end{equation}
such that $\phi_{{\rm in}}$ is the initial state of $\phi$ in the
infinite past. If we apply the Wick's Theorem and normal ordering, then
the vacuum expectation value can be described as the integrals of
propagators that typically depend on differences of space-time
vectors. The rigorous challenge is the existence of divergencies in
these integrals with respect to the domains of integrations where
applying regularization machineries (such as Dimensional
Regularization) help us to study these integrals in the context of
Laurent series with finite pole parts. \cite{breitenlohner-maison-1,
calaque-strobl-1}

Feynman diagrams in Quantum Field Theory are useful combinatorial tools to encapsulate the summation over probability amplitudes
corresponding to all possible exchanges of virtual particles
compatible with a process at a given (loop) order. These decorated
diagrams, as a set of edges and a set of vertices, can simplify
the description of interactions of elementary particles in terms of the time parameter in a quantum system. Their decorations
are determined by fundamental parameters of the physical
theory. Momentum and position are actually Fourier
transforms of each other where we can translate diagrams with
respect to momentum space to their corresponding iterated integrals
via Feynman rules. Each closed loop
associates to an integrate over the corresponding momentum. The whole diagram obeys the conservation of
momenta which tells us that the amount of momenta of input
particles in an interaction procedure is the same as the amount of
momenta of output particles.

For example, thanks to the Schwinger parameter $t$, we can consider
\begin{equation}
\frac{1}{p^{2}+m^{2}} = \int_{0}^{\infty} dt \ {\rm
exp}(-t(p^{2}+m^{2}))
\end{equation}
as the propagator for each edge and
\begin{equation}
\int d^{4}x \ {\rm exp}(i \sum_{j}p_{j}x)
\end{equation}
as the propagator for each vertex. In this setting, each edge has a
factor
\begin{equation}
G(x,y;t) = \int \frac{d^{4}p}{(2\pi)^{4}} \ {\rm exp}(ip.(x-y) -
t(p^{2}+m^{2})).
\end{equation}

We can combinatorially reformulate Green's functions in terms of formal expansions of Feynman diagrams. It clarifies the self-similar property of these fundamental functions in Quantum Field Theory. We have
$$\mathcal{G} := 1 + \int I_{\gamma} + \int  \int  I_{\gamma} I_{\gamma} + \int  \int \int I_{\gamma} I_{\gamma}  I_{\gamma} + ...$$
\begin{equation} \label{dse-2}
= 1 + \int I_{\gamma} (1+ \int I_{\gamma} + \int \int I_{\gamma}
I_{\gamma} + ...)
\end{equation}
such that $I_{\gamma}$ is the Feynman integral corresponding to the
primitive (1PI) Feynman diagram $\gamma$. This can be encapsulated via equations such as
\begin{equation} \label{dse-3}
\mathcal{G} = 1 + \int I_{\gamma} \mathcal{G}
\end{equation}
such that its fixed point equations determine Dyson--Schwinger
equations. This formulation of Quantum Field Theory is the result of the path
integral approach to Lagrangian formalism where we study the behavior
of an elementary particle in a system with infinite degrees of
freedom in terms of the sum over all possible situations (such as
trajectories, interactions) which could be selected by the particle. In
terms of some conditions dictated by physical theory, each possible
situation has a particular weight which should be considered in
computational processes of Feynman integrals. \cite{breitenlohner-maison-1, calaque-strobl-1, nair-1,
weinzierl-1}

For example, in QED we deal with
interactions of electron and positron (as matter) with photons (as
electromagnetic waves with different quantized sizes of energies).
There exist six different fundamental interactions namely, the emission of
photon from electron or positron, absorbing a photon via electron or
positron, the creation of a photon via annihilation of the pair
(electron, positron), creation of a pair (electron, positron) via the
annihilation of a photon. All Feynman diagrams and Dyson--Schwinger equations in QED, which
might contain complicated off-shell interactions of virtual
particles, are built on the basis of those six fundamental
interactions. There exists a class of elementary graphs which
play the role of building blocks to make all possible Feynman diagrams in a
physical theory. These graphs, which are called one particle
irreducible Feynman diagrams, remain connected after removing one
internal edge. By induction we can define
$n$-particle irreducible Feynman diagrams which remain connected
after removing $n$ internal edges. It is easy to see
that each $n$-particle irreducible graph is a $(n-1)$-particle
irreducible graph.

The coupling constants in Quantum Field Theory show the
strength of the interactions among elementary particles. The
regularization of UV divergent integrals and the renormalization
procedure generate a scale dependence. The UV cut-off dependence
of the couplings can be eliminated by allowing the couplings and masses
(which appear in the Lagrangian) to acquire a scale dependence. Then
we normalize them to a measured value at a given scale.
Generally speaking, there are two classes of couplings namely, the
bare coupling constant as the original strength of a fundamental
force and the running coupling constants or effective couplings as
the result of renormalization procedures. Quantum Chromodynamics
(QCD) is known as the most successful fundamental gauge theory of
strong interactions. It studies the hadronic interactions involving
quarks and gluons at both long and short distances. Its symmetry
group is ${\rm SU}(3)$ where it includes $N_{f}$ family of quarks
$\psi_{f}^{i}$ and gluons $A_{\mu}^{i}$. Some experimental evidences
inform us that at a critical temperature around $T_{c} \approx 170 \
{\rm MeV}$, the QCD matter undergoes a deconfining phase transition into
quark-gluon plasma. Perturbative QCD is a method based on expanding
different physical quantities with respect to the gauge coupling
constant $g$ which is applied in the region $T \gg T_{c}$
where $g$ is small. The phenomenology of the (bare and running)
coupling constants have been discussed in terms of the uncertainties
in their values at short distances which leads us to a total
theoretical uncertainty in Physics at large hadron collider such as
Higgs production via gluon fusion. In this situation we can still
have hope to apply asymptotic freedom and perturbative calculations
of Renormalization Group equations. However at high perturbative
orders it becomes necessary to evaluate large numbers of multi-loop
Feynman diagrams in the effective theory. \cite{bali-1,
deur-brodsky-deteramond-1, marino-2}

However the behavior of running coupling constants of the physical systems at long
distances such as the scale of the proton mass in order to
understand hadronic structure, quark confinement and hadronization
processes should be analyzed under non-perturbative platforms. In these cases we can study the phenomenology of strong bare or running couplings in terms
of the mathematical foundations of quantum motions of physical systems namely, Dyson--Schwinger equations. We can address recent
theoretical progress for the computation of non-perturbative
parameters derived from solutions of Dyson--Schwinger equations in the context of Combinatorics, Geometry and Category
Theory. \cite{bergbauer-kreimer-1, broadhurst-kreimer-2, kreimer-4,
marino-1, marino-3, roberts-schmidt-1, shojaeifard-3, shojaeifard-4,
shojaeifard-5, vanbaalen-kreimer-uminsky-yeats-1, weinzierl-3}

The running of a coupling constant originates from the
renormalization procedure while predictions for observables should
be determined independent of the choice of renormalization map and
regularization scheme. This invariance with respect to the choice of
renormalization program is encoded via a symmetry group. The
running coupling is an expansion parameter in the perturbative
series describing an observable and there exists the Landau pole as
the point where the perturbative expression of the running coupling
diverges. It means that the full perturbative expression is actually a
non-observable quantity. The observable is independent of the
renormalization scheme while the series's coefficients and the running
coupling are related to the renormalization scheme. Under asymptotic
freedom behavior at short distances, we can get the first
coefficient series as an independent parameter while at very large
distances dependency can not be neglected. This discussion tells
us that the running couplings are not observables because they are
strongly depended on the renormalization scheme at large distances.
In short, the running couplings have weak scale dependence at
distances smaller than $10^{-16}$ m such that this controllable weak
behavior tends to a strong scale dependence larger than a tenth of a
Fermi. This dependency on the scale is restored at larger distances
due to the confinement of quarks and gluons. \cite{bali-1,
deur-brodsky-deteramond-1, marino-1}

The Ward--Takahashi identities on Feynman diagrams tell us that the
photon propagator is the only propagator in QED which contributes to
the running of the coupling constant. The Slavnov--Taylor identities on
Feynman diagrams tell us that intermediate gauge-dependent
quantities in non-abelian gauge theories provide final
gauge-independent results for observables
\cite{kreimer-sars-suijlekom-1, suijlekom-1}. Thanks to these facts,
it is possible to rewrite Dyson--Schwinger equations in terms of
some running couplings to achieve some intermediate quantities which
are useful to simplify the original complicated non-perturbative
type of equations by some approximations. In a general
configuration, Dyson--Schwinger equations are polynomials with
respect to bare or running coupling constants which means that any
change in the amount of running couplings will make direct influence
on the behavior of these equations. Therefore these non-perturbative
type of equations are good tools for the study of the phenomenology
of strong couplings.

Dimensional Regularization was introduced by 't Hooft, Veltman,
Bollini and Gambiagi as a method to regularize ultraviolet
divergencies in a gauge invariant way to complete the proof of
renormalizability. The method works in $D=4-2\epsilon$ space-time
dimensions where divergencies for $D \rightarrow 4$ appears as poles
in $1/\epsilon$. This method also regulates infrared singularities
where if we remove the auxiliary IR regulator, the IR divergencies
appear as poles in $1/\epsilon$. For $\epsilon>0$, we can obtain a
well-defined result which we can be analytically extended to the
whole complex D-plane. The only essential change in the structure of
Feynman rules is to replace the couplings in the Lagrangian via the
transformation $g \longmapsto g \mu^{\epsilon}$ such that $\mu$ is
an arbitrary mass scale. Dimensional Regularization together with
Minimal Subtraction can provide a practical renormalization
program for Feynman integrals with nested sub-divergencies. These
ill-defined parts can be eliminated step by step under a forest
formula setting where the Bogoliubov--Parasiuk--Hepp preparation
allows us to generate some finite values. This particular
renormalization program was reconsidered by Connes and Kreimer under
a modern Hopf algebraic setting to generate counterterms and renormalized
values in terms of the Riemann--Hilbert problem and the Birkhoff
factorization. In this context, Dimensional Regularization is
encapsulated by the space of loops with the domain of a punctured
infinitesimal disk around zero and with values in a pro-unipotent
complex Lie group associated to the renormalization Hopf algebra of
Feynman diagrams. \cite{connes-kreimer-2, connes-kreimer-3,
ebrahimifard-manchon-1, figueroa-graciabondia-2, kreimer-panzer-1,
kreimer-yeats-3}

While working on the applications of the BPHZ procedure to the level
of many-loop graphs is one of the interesting topics in Quantum
Field Theory, the main information of a physical theory are encoded
in (infinite) formal expansions of Feynman diagrams. Thanks to the
Connes--Kreimer--Marcolli theory, the required mathematical tools
for an extension of the BPHZ procedure to the level of
Dyson--Schwinger equations has already been considered under
geometric and algebraic settings. This platform enables us to study any Dyson--Schwinger equation DSEin terms of its corresponding
complex Lie group $\mathbb{G}_{{\rm DSE}}(\mathbb{C})$ where the
existence of the Hopf--Birkhoff factorization on this Lie group has
led us to determine counterterms (which contribute to the unique
solution of the equation DSE) in the language of differential
systems together with irregular singularities. These differential
systems, which are on the basis of equi-singular flat
$\mathbb{G}_{{\rm DSE}}(\mathbb{C})$-connections, determine a new
class of systems of Picard--Fuchs equations with regular
singularities. \cite{shojaeifard-5, shojaeifard-8}

\section{\textsl{Hochschild cohomology of the renormalization Hopf algebra}}

The basic elements of the path integral method in Quantum Field
Theory are (divergent) iterated Feynman integrals over the momentum
space such that the integrands are determined from a definite
collection of rules originated from the physical theory. We can encode these
integrals in terms of a class of combinatorial decorated
finite diagrams which are known as Feynman diagrams where
sub-divergencies in the original integral are presented in terms of
the existence of nested or overlapping loops in the main diagram.
The Kreimer's coproduct, which highlights the combinatorics of
removing sub-divergencies from integrals, enables us to factorize
the original complicated Feynman diagram into its basic
sub-divergencies (as sub-graphs). This factorization reduces several
layers of complications in the computational processes of
perturbative renormalization in terms of a certain graded
commutative non-cocommutative Hopf algebra denoted by $H_{{\rm
FG}}(\Phi)$. It is a graded Hopf algebra with respect to the first
Betti number of Feynman diagrams which means that $H_{{\rm
FG}}(\Phi)=\bigoplus_{n \ge 0} \mathcal{H}_{n}$ such that
$\mathcal{H}_{0}=\{\mathbb{I}\}$ and for each $n$, $\mathcal{H}_{n}$
is the vector space of divergent 1PI $n$-loop Feynman diagrams and
products of Feynman diagrams with overall loop number $n$. There is
also another graduation parameter to build a graded Hopf algebra. We
can show that $H_{{\rm FG}}(\Phi)$ is a graded Hopf algebra with
respect to the number of internal edges of Feynman graphs such that
the components of this grading have finite dimensions as the vector
spaces. \cite{borinsky-1, connes-marcolli-1, kreimer-1, kreimer-7}

The original version of the Kreimer's coproduct was defined in the
language of parenthesized words to characterize nested,
independent or overlapped sub-divergencies via sequences of letters and their linear combinations. It encapsulates the
Bogoliubov--Zimmermann forest formula based on the formal expansion
\begin{equation} \label{copro-1}
\Delta_{{\rm FG}}(\Gamma) = \Gamma \otimes \mathbb{I} + \mathbb{I}
\otimes \Gamma + \sum_{\gamma \subset \Gamma} \gamma \otimes \Gamma
/ \gamma
\end{equation}
for each Feynman diagram $\Gamma$ such that the sum is over all
disjoint unions of 1PI divergent proper Feynman subgraphs.
\cite{connes-kreimer-1, figueroa-graciabondia-1,
figueroa-graciabondia-2, kreimer-1}

Generally speaking, for a unital algebra $(A,m,e)$ and a counital
coalgebra $(C,\Delta,\varepsilon)$ over a field $\mathbb{K}$ of
characteristic zero, let ${\rm Hom}(C,A)$ be the vector space of all
$\mathbb{K}$-linear maps from $C$ to $A$. Equip this space with a
convolution product defined in terms of the following composition
\begin{equation}
C \longrightarrow^{\Delta} C \otimes C \longrightarrow^{f*g} A
\otimes A \longrightarrow^{m} A
\end{equation}
to achieve an algebra with the unit $e \circ \varepsilon$. A
bialgebra $(H,m,e,\Delta,\varepsilon)$ in which the identity map
${\rm id}_{H}$ is invertible under the convolution product is a Hopf
algebra. This particular inverse $S$ which obeys the following
property
\begin{equation}
{\rm id}_{H} * S = S * {\rm id}_{H} = e \circ  \varepsilon
\end{equation}
is called the antipode such that it is  a unital algebra counital
coalgebra antihomomorphism. \cite{cartier-1,
figueroa-graciabondia-2}

There are natural graduation parameters on Feynman diagrams such as
number of internal edges or number of independent loops. The
graduation parameter and the coproduct (\ref{copro-1}) determine the
required antipode for the construction of a free commutative
connected graded finite type Hopf algebra on Feynman diagrams of a
given physical theory $\Phi$. The antipode deforms Feynman rules
characters to obtain renormalized values. The Hopf algebra $H_{{\rm
FG}}(\Phi)$ has a Lie algebraic origin and in addition, it can be
simplified via decorated rooted trees to provide a universal model
for perturbative renormalization. The rooted tree version of the
renormalization coproduct (\ref{copro-1}) can be defined in terms of
the notion of ''admissible cuts'' on trees. \cite{broadhurst-kreimer-1,
ebrahimifard-kreimer-1, ebrahimifard-guo-kreimer-1,
ebrahimifard-guo-kreimer-2, figueroa-graciabondia-1,
mencattini-kreimer-1}

The factorization of a Feynman diagram into its primitive components
can be reversed under some conditions via the insertion operator
which enables us to glue subdiagrams. It is important to note that
in gauge field theories we should work on a quotient of the
renormalization Hopf algebra with respect to Ward identities and
Slavnov--Taylor identities to achieve a unique factorization for
each Feynman diagram with respect to the insertion operator. The
insertion operator provides a Lie algebraic structure on Feynman
diagrams such that the graded dual of its universal enveloping
algebra will be equivalent to the renormalization Hopf algebra.
\cite{figueroa-graciabondia-2, kreimer-11, mencattini-kreimer-1,
suijlekom-2}

Theory of Hochschild cohomology for bialgebras is useful for us to
formulate the Hochschild equation on Feynman diagrams which results
a recursive formulation for the renormalization coproduct. This class of equations can provide the key tools for the reformulation of Dyson--Schwinger equations. In this part we review the foundations of the Hochschild equation with respect to the renormalization coproduct and for further details we address  \cite{ebrahimifard-fauvet-1, figueroa-graciabondia-1,
    figueroa-graciabondia-2, kreimer-3, kreimer-11, yeats-1} as the major sources.

For a given bialgebra $H$ such as $H_{{\rm FG}}(\Phi)$, the dual of the
coalgebra $(H,\Delta,\varepsilon)$ is an algebra $H^{*}$ such that
the unit map $\mathbb{I}$ of $H$ transposes to a character
$\mathbb{I}^{t}$ of $H^{*}$. Therefore we can build Hochschild
cohomology groups $H^{n}(H,H^{*})$ such that n-cochains are linear
maps such as $T:H \longrightarrow H^{\otimes n}$. We can transpose
them to $n$-linear maps such as $\rho_{T}:(H^{*})^{n}
\longrightarrow H^{*}$ where we have
\begin{equation}
\rho_{T}(\Gamma_{1},...,\Gamma_{n}):=T^{t}(\Gamma_{1} \otimes ...
\otimes \Gamma_{n}).
\end{equation}
In this setting, the Hochschild coboundary operator ${\rm
\textbf{b}}$ can be determined by the relation
\begin{equation}
<\Gamma_{1} \otimes ... \otimes \Gamma_{n+1},{\rm
\textbf{b}}T(\Gamma)>:=<{\rm \textbf{b}}
\rho_{T}(\Gamma_{1},...,\Gamma_{n+1}),\Gamma>
\end{equation}
for each $\Gamma \in H$. Define $\Delta_{j}:H^{\otimes n}
\longrightarrow H^{\otimes (n+1)}$ as the homomorphism which applies
the coproduct $\Delta$ only on the $j$th factor. Now we can show
that
\begin{equation}
<\rho_{T}(\Gamma_{1},...,\Gamma_{j}\Gamma_{j+1},...,\Gamma_{n+1}),\Gamma>
= <\Gamma_{1} \otimes ... \otimes
\Gamma_{n+1},\Delta_{j}(T(\Gamma))>.
\end{equation}
It leads us to rewrite the Hochschild coboundary operator as the
following way
\begin{equation}
{\rm \textbf{b}}T(\Gamma):= ({\rm id} \otimes T) \Delta(\Gamma) +
\sum_{j=1}^{n} (-1)^{j} \Delta_{j}(T(\Gamma)) + (-1)^{n+1} T(\Gamma)
\otimes \mathbb{I}.
\end{equation}
The resulting cohomology groups
$H^{n}(H^{*},H^{*}_{\mathbb{I}^{t}})$ are indeed the Hochschild
cohomology theory of the bialgebra $H$. It is easy to check that
linear forms on $H$ are $0$-cochains and one cocycles are linear
maps such as $l:H \longrightarrow H$ which obeys the following
relation
\begin{equation}
\Delta(l) = l \otimes \mathbb{I} + ({\rm id} \otimes l)
\Delta.
\end{equation}

The Hochschild cohomology with values in a $H$-bimodule $A$ (such
as the regularization algebra) is defined by working on $n$-cochains via
the vector space $C^{n}:=C^{n}(H,A)$ consisting of $n$-linear maps
$\psi:H^{n} \longrightarrow A$ with the $H$-bimodule structure
\begin{equation}
(\gamma_{1} . \psi . \gamma_{2}) (\Gamma_{1},...,\Gamma_{n}):=
\gamma_{1}.\psi(\Gamma_{1},...,\Gamma_{n}).\gamma_{2}.
\end{equation}
The coboundary map ${\rm \textbf{b}}: C^{n} \longrightarrow C^{n+1}$
is given by
$${\rm \textbf{b}}(\psi)(\Gamma_{1},...,\Gamma_{n+1}) = \Gamma_{1}. \psi(\Gamma_{2},...,\Gamma_{n+1})$$
\begin{equation}
+ \sum_{j=1}^{n}
(-1)^{j}\psi(\Gamma_{1},...,\Gamma_{j}\Gamma_{j+1},...,\Gamma_{n+1})
+ (-1)^{n+1} \psi(\Gamma_{1},...,\Gamma_{n}).\Gamma_{n+1}.
\end{equation}
The resulting cohomology groups are denoted by
$\mathcal{H}^{n}(H,A)$.

Let us consider the Hochschild equation for the algebra
$\mathbb{K}[X]$ which is also equipped with a cocommutative
coalgebra structure by considering the indeterminate $X$ as the
primitive object where $\varepsilon(X)=0$. For all $k \ge 2$, by
induction, we can show that
\begin{equation}
\Delta(X^{k}) = (\Delta X)^{k} = \sum_{j=0}^{k} \binom {k} {j}
X^{k-j} \otimes X^{j}.
\end{equation}
For any linear form $\varrho$ on $\mathbb{K}[X]$, we have
\begin{equation}
{\rm \textbf{b}} \varrho(X^{k}) = ({\rm id} \otimes \varrho)
\Delta(X^{k}) - \varrho(X^{k}) \otimes \mathbb{I} = \sum_{j=1}^{k}
\binom {k} {j} \varrho(X^{k-j})X^{j}.
\end{equation}
${\rm \textbf{b}} \varrho$ is a linear transformation of polynomials
which does not increase the degree. It shows that the integration
map $T(X^{k}):= X^{k+1} / (k+1)$ is not a 1-coboundary but it is an
one cocycle.

Thanks to this Hochschild cohomology theory, it is possible to
define the renormalization coproduct under a recursive setting.

A graded bialgebra $H$ over a field $\mathbb{K}$ is graded as an
algebra and as a coalgebra. It is called connected if the degree
zero component of the graduation structure consists of scalars (i.e.
elements of the field $\mathbb{K}$). We have
\begin{equation}
H= \bigoplus_{n \ge 0} H_{n}, \ \ H_{m}H_{n} \subseteq H_{m+n}, \ \
\Delta(H_{n}) \subseteq \bigoplus_{p+q=n}H_{p} \otimes H_{q}.
\end{equation}

The coproduct in the connected graded bialgebra of Feynman diagrams can be presented in
terms of the Sweedler notation such that for $\Gamma \in \mathcal{H}_{n}$,
we have
\begin{equation}
\Delta(\Gamma) = \Gamma \otimes \mathbb{I} + \mathbb{I} \otimes
\Gamma + \sum \Gamma'_{1} \otimes \Gamma'_{2}
\end{equation}
where terms $\Gamma'_{1}$ and $\Gamma'_{2}$ all have degrees between
$1$ and $n-1$. The counit equations
\begin{equation}
\sum \varepsilon(\Gamma'_{1}) \Gamma'_{2} = \sum \Gamma'_{1} \varepsilon(\Gamma'_{2}) = \Gamma, \\
\sum S(\Gamma'_{1}) \Gamma'_{2} = \sum \Gamma'_{1} S(\Gamma'_{2}) =
\varepsilon(\Gamma)\mathbb{I}
\end{equation}
tell us that $\Delta(\Gamma)$ must contain terms $\Gamma
\otimes \mathbb{I} \in \mathcal{H}_{n} \otimes \mathcal{H}_{0}$ and $\mathbb{I}
\otimes \Gamma \in \mathcal{H}_{0} \otimes \mathcal{H}_{n}$ and the remaining terms
which have intermediate bidegrees. They address the equation
\begin{equation}
\Gamma = (\varepsilon \otimes {\rm id}) (\Delta (\Gamma)) =
\varepsilon (\Gamma) \mathbb{I} + \Gamma + \sum
\varepsilon(\Gamma'_{1}) \Gamma'_{2}
\end{equation}
which leads us to the relation $\varepsilon(\Gamma)=0$ for every
non-trivial Feynman diagram.

In general, the augmentation ideal in a graded bialgebra is
given by
\begin{equation}
{\rm Ker} \varepsilon = \bigoplus_{n=1}^{\infty} H_{n}.
\end{equation}
For $P:= {\rm id} - \mathbb{I}\varepsilon$ as the projector onto the
augmentation ideal, define ${\rm Aug}^{m}:= (P \otimes ...^{m \ {\rm
times}}...\otimes P) \Delta^{m-1}$ and then set
\begin{equation}
H^{(m)}:= {\rm Aug}^{m+1} / {\rm Aug}^{m}
\end{equation}
for all $m \ge 1$. It determines the
bigraded structure on our bialgebra given by
\begin{equation}
H= \bigoplus_{n \ge 0} H_{n} = \bigoplus_{m \ge 0} H^{(m)}
\end{equation}
such that for all $k \ge 1$
\begin{equation}
H_{k} \subset \bigoplus_{j=1}^{k} H^{(j)}, \ \ \ H_{0} \simeq H^{(0)}
\simeq \mathbb{K}.
\end{equation}

In addition, number of internal edges or number of
independent loops can be applied as the graduation parameters on the
renormalization bialgebra of Feynman diagrams to formulate its corresponding antipode
inductively. We have the recursive formulation
\begin{equation}
S(\Gamma)= - \Gamma - \sum S(\Gamma'_{1}) \Gamma'_{2}
\end{equation}
for the antipode of each Feynman diagram $\Gamma$ with respect to its renormalization coproduct. This formula can be applied to show the existence of a convolution inverse
for the identity map.

\section{\textsl{Renormalization Hopf algebra of Feynman graphons and filtration of large Feynman diagrams}}

The combinatorial interpretation of Feynman diagrams in terms of decorated rooted trees and the
recursive nature of the renormalization
coproduct are the key tools to build a new Hopf algebra $H_{{\rm CK}}$ of non-planar rooted trees. This combinatorial Hopf algebra together with the grafting operator $B^{+}$ has the universal property
with respect to the Hochschild cohomology theory of commutative Hopf algebras. Each Feynman diagram with nested loops can be represented
by a labeled rooted tree where the root is the symbol for the
original graph and other vertices are symbols of nested loops. Edges
among vertices determine the positions of nested loops with respect
to each other. In addition, it is possible to represent Feynman
diagrams with overlapping divergencies with rooted trees where we
should deal with linear combinations of decorated rooted trees
\cite{figueroa-graciabondia-1, figueroa-graciabondia-2, kreimer-2,
    kreimer-11}. In this section we present a new interrelationship
between the theory of infinite graphs and the fundamental structure of
Green's functions in Quantum Field Theory on the basis of rooted trees. We provide a new analytic generalization of Feynman diagrams organized in to a new topological enriched version of the Connes--Kreimer renormalization Hopf algebra. For this purpose, we use the representation of Feynman diagrams via rooted trees
to define a new graph function representation formalism in dealing with expansions of Feynman diagrams. We show the importance of analytic graphs for the study of solutions of strongly coupled Dyson--Schwinger equations in terms of sequences of partial sums and the notion of convergence with respect to the topology of graphons. We then work on the combination of the renormalization Hopf algebra and graphon models of Feynman diagrams to build a new Hopf algebra of graphons which contribute to the analytic representation of Feynman diagrams and solutions of Dyson--Schwinger equations. This Hopf algebra of graphons is useful for us to formulate a new topological generalization of the Connes--Kreimer BPHZ renormalization to generate non-perturbative counterterms and related renormalized values. We also use this graphon model approach to Dyson--Schwinger equations to explain the structure of a filtration treatment on the space of solutions of these non-perturbative equations.

On the one hand, Dyson--Schwinger equations are reformulated in
terms of the renormalization Hopf algebra and the grafting operator which act on Feynman diagrams.
The unique solution of each equation DSE with the general form
(\ref{dse-1}) is an infinite formal expansion of Feynman diagrams together with powers of running couplings.
In physical theories with weak coupling constants, we can expect to study the renormalization of these expansions
by applying many-loop computation
techniques under the perturbative setting. In physical theories with
strong couplings, these expansions contain infinite number of terms such that their renormalization generate infinite number of counterterms. In this situation we need to deal with non-renormalizable gauge field theories. However we can encapsulate all these formal expansions in terms of fixed point equations of Green's functions of the given physical theory \cite{kreimer-4}. On the other hand, graph limits, as a modern branch in infinite
combinatorics, study limits of finite combinatorial objects such as
weighted or directed graphs, multi or hyper graphs, bipartite graphs
and posets. Theory of graphons and random graphs is one of the
recent progress in infinite combinatorics where we deal with
(symmetric) measurable functions such as $W$ defined on the
probability space $\Omega$. Actually, a graph limit, as the
convergent limit of an infinite sequence of graphs, can be
represented by a graphon which does not have necessarily the unique
representation. The key tool which allows us to concern convergence
and equivalence of graphons is the concept of cut-metric
\cite{lovasz-1}. It is reasonable to think about any
relation between solutions of combinatorial Dyson--Schwinger
equations namely, formal expansions such as $\sum_{n \ge 0} (\lambda g)^{n}
X_{n}$, such that $g$ is the bare coupling constant while $\lambda g$ is any running coupling
or any rescaled version of the bare coupling, and graphons generated by infinite sequences of sparse graphs. This idea has already been
discussed in \cite{shojaeifard-9, shojaeifard-10} where we have shown that the unique solution of each equation DSE can be obtained as the cut-distance convergent limit of a sequence of random graphs generated in terms of the combinatorial information of partial sums $Y_{m}=\sum_{i=1}^{m} (\lambda g)^{i}
X_{i}$. In this work we plan to explain the structure of new Hopf algebra of graphons which can encode solutions of all Dyson--Schwinger equations in a given (strongly coupled) gauge field theory. This Hopf algebra can be equipped with the cut-distance topology which enables us to describe infinite expansions such as $X_{{\rm DSE}(\lambda g)} = \sum_{n \ge 0} (\lambda g)^{n} X_{n}$ as the unique solution of the equation ${\rm DSE}(\lambda g) \in \mathcal{S}^{\Phi}(\lambda g)$ in terms of objects in the boundary region of the compact topological space $\mathcal{S}^{\Phi}_{{\rm graphon}}$. In other words, this new platform enables us to study $\mathcal{S}^{\Phi,g} = \bigcup_{\lambda} \mathcal{S}^{\Phi}(\lambda g)$, as the collection of all Dyson--Schwinger equations under different running coupling constants in the physical theory $\Phi$, in the context of a subspace of the topological Hopf algebra $\mathcal{S}^{\Phi}_{{\rm graphon}}$.

\subsection{\textsl{Graphons}}

Generally speaking, the theory of graph limits aims to assign a limit to
a sequence of finite graphs such as $\{G_{n}\}_{n \ge 0}$ when
number of vertices of graphs in the sequence tends to infinity.
There are some different approaches to define the concept of
convergence at this level but the one approach which is based on
random graphs and cut-distance topology is very useful. We can say that a sequence $\{G_{n}\}_{n \ge 0}$ of finite
graphs is convergent when $|G_{n}|$ tends to infinity, if for each
fixed value $k$, the distribution of the random graphs $G_{n}[k]$ is
convergent when $n$ tends to infinity. In this setting, $G_{n}[k]$ is
a labeled subgraph of $G_{n}$ with vertices $1,...,k$ obtained by
selecting $k$ distinct vertices $v_{1},...,v_{k} \in G_{n}$ under a
uniformly random process. Graph limits can be generated by infinite sequences of dense or sparse graphs. We can generalize graph limits to graph functions (i.e. graphons) in an arbitrary probability space where we can study them in terms of equivalence classes of convergent sequences of finite
graphs with respect to (invertible) measure preserving transformations of the ground probability space. \cite{janson-1,lovasz-1}

For a given probability space $\Omega$,
graphons are bounded measurable symmetric functions such as $W: \Omega
\times \Omega \rightarrow [0,1]$. The symmetric condition can be
removed when we work on another class of graphons known as
bigraphons. Bigraphons are useful for us if we want to rebuild Feynman diagrams. We associate graph functions to Feynman diagrams via decorated rooted trees where we need to fix an orientation on trees to identify the positions of nested loops with respect to each other in Feynman diagrams. The choice of the orientation allows us to work on the upper triangular or lower triangular versions of the adjacency matrix of rooted trees which are not clearly symmetric.

The graphon representation of any given graph limit is not unique. In other words, we can produce different graphon representations in terms of changing the ground probability space or changing presentation parameters. However for a fixed probability space, we can associate different labeled graphons to a given graph limit by applying invertible measure preserving transformations of the ground probability space.

For a given finite graph, the pixel pictures derived from the adjacency matrix are the simplest examples of labeled graphons. For a given sequence of finite simple graphs, the corresponding sequence of pixel pictures derived from the adjacency matrices can provide labeled graphon representations for the graph limit of the initial sequence.

A map $\rho: \Omega_{1} \rightarrow \Omega_{2}$ between
probability spaces $(\Omega_{1},\mathcal{F}_{1},\mu_{1})$ and
$(\Omega_{2},\mathcal{F}_{2},\mu_{2})$ is called measure preserving
if it is measurable and $\mu_{1}(\rho^{-1}(A))=\mu_{2}(A)$ for each
measurable set $A \in \mathcal{F}_{2}$. $\rho$ is called measure
preserving bijection if it is a bijection map and $\rho, \
\rho^{-1}$ are measure preserving. It is easy to check that for a
given measure preserving map $\rho$, the map $\rho \otimes \rho:
\Omega_{1}^{2} \rightarrow \Omega_{2}^{2}$ defined by $\rho
\otimes \rho (x,y):= (\rho(x),\rho(y))$ is also a measure preserving
map. If $\rho$ is a bijection, then $f^{\rho}, \ W^{\rho}$ are
called rearrangements of $f$ (as a function on $\Omega_{2}$) and $W$
(as a function on $\Omega^{2}_{2}$). Actually, relabeling of labeled
graphons can be understood as a kind of rearrangement. In other
words, for a given measure preserving map $\rho$, the pull backs of
$f$ and $W$ are defined by
\begin{equation}
f^{\rho}(x):=f(\rho(x)), \ W^{\rho}(x,y):= W(\rho(x),\rho(y)).
\end{equation}
If $f \in L^{1}(\Omega_{2})$ and $W \in L^{1}(\Omega^{2}_{2})$, then
$\parallel f^{\rho} \parallel_{1} = \parallel f \parallel_{1}$ and
$\parallel W^{\rho} \parallel_{1} = \parallel W \parallel_{1}$.

\begin{defn} \label{graphon-1}
(i) For the probability space $\Omega:=[0,1]$ together with the Lebesgue measure, an unlabeled graphon is a graphon up to relabeling such that a relabeling is defined by an invertible measure preserving
transformation of the closed unit interval.

(ii) Graphons $W_{1},W_{2}$ are called weakly isomorphic (or weakly equivalent i.e. $W_{1} \approx W_{2}$) if there exists a graphon $W$ and measure preserving transformations $\rho_{1},\rho_{2}$ on the probability space $\Omega$ such that $W^{\rho_{1}}=W_{1}$ and $W^{\rho_{2}}=W_{2}$ almost everywhere.
\end{defn}

Relabeled graphons can be organized into an equivalence class, namely an unlabeled graphon. For a given labeled graphon $W$, its corresponding unlabeled graphon class $[W]$ is given by
\begin{equation} \label{unlabeled-graphon}
[W]:= \{W^{\rho}: (x,y) \mapsto W(\rho(x),\rho(y)): \ \rho \
{\rm is \ an \ arbitrary \ rearrangement}\}.
\end{equation}

Set $\mathcal{W}(\Omega)$ as the set of all labeled graphons on a
given probability space $\Omega$. It is not difficult to see that if
$\Omega:=[0,1]$ is the probability space, then $\mathcal{W}(\Omega)$ is the
subspace of symmetric functions in $L_{\infty}([0,1]^{2})$. By
defining a suitable equivalence relation on labeled graphons, which
encodes exchanging decorations (i.e. (\ref{unlabeled-graphon})), it is possible to associate a unique
graphon class to each graph limit. This graphon class is called
unlabeled graphon. Set $[\mathcal{W}](\Omega)$ as the family of all
unlabeled graphons on a given probability space $\Omega$.

Graphons, as edge weighted graphs on the vertex set $[0,1]$, can provide
a generalization of discrete graphs. Each finite
simple graph $G$ defines naturally an unlabeled graphon class
$[W_{G}]$ in terms of its adjacency matrix. First we build a labeled graphon $W_{G}$ with respect
to the information of the adjacency matrix. Consider $V(G)$ as a
probability space such that each vertex has probability
$\frac{1}{|G|}$. Define the map $W_{G}^{1}: V(G) \times V(G)
\longrightarrow [0,1]$  as follows
\begin{equation} \label{graphon-graph-1}
W_{G}^{1}(u,v):= \{^{1, \ \ {\rm if \ u \ and \ v \ are \
adjacent}}_{0, \ {\rm otherwise}}.
\end{equation}
It is easy to see that $W_{G}^{1}$ is a symmetric measurable
function. In an alternative setting, we can also consider
$\Omega=(0,1]$  as the probability space which is equipped with a
partition $\{I_{i}^{n}\}_{i=1,...,n}$ such that
$I_{i}^{n}:=(\frac{i-1}{n},\frac{i}{n}]$. If the vertices of $G$ is
labeled by $1,..,n$, then the corresponding graph function
$W^{2}_{G}$ is given by
\begin{equation} \label{graphon-graph-2}
W^{2}_{G}(x,y)=W^{1}_{G}(i,j)=1, \ \ \  x \in I^{n}_{i}, y\in
I^{n}_{j}
\end{equation}
whenever there exists an edge between $i$ and $j$.

The homomorphism density is an useful tool in dealing with the
probability of the existence of a subgraph in an extremely large
graph. For a given finite graph $G$, the homomorphism density of
each subgraph $H$ in $G$ is given by
\begin{equation} \label{hom-den-1}
t(H,G):= \frac{{\rm hom}(H,G)}{|V(G)|^{|V(H)|}}
\end{equation}
such that ${\rm hom}(H,G)$ is the number of graph homomorphisms from
$H$ to $G$. It is possible to generalize this idea for the level of
graph limits where this parameter informs the density of $H$ as a
subgraph in $G$ asymptotically when the number of vertices of $G$
tends to infinity.  For a given graphon $W: \Omega^{2}
\rightarrow [0,1]$ and a simple finite graph $H$, the
homomorphism density is defined by
\begin{equation} \label{hom-den-2}
t(H,W):= \int_{\Omega^{|H|}} \prod_{ij \in E(H)} W(x_{i},x_{j})
d\mu(x_{1}) ... d\mu(x_{|H|}).
\end{equation}
If the graphon $W_{G}$ is a labeled graphon with respect to a
given graph $G$, then we have $t(H,G):= t(H,W_{G})$. Homomorphism densities provide an alternative way of defining convergence of sequences of dense graphs. It is shown that the sequence $\{G_{n}\}_{n \ge 0}$
converges to the labeled graphon $W$ iff the sequence
$\{t(H,G_{n})\}_{n \ge 0}$ of homomorphism densities converges to
the homomorphism density $t(H,W)$ for each simple subgraph $H$. Therefore any convergent sequence of finite sparse graphs should be convergent to the graphon with zero density which is almost everywhere the $0$-graphon. However the theory of graphons of sparse graphs has been developed to obtain non-zero graphons for sequences of sparse graphs in terms of the rescaled versions of the canonical graphons or other normalization methods applied on the ground probability space. \cite{br-1, janson-1, lovasz-1, shojaeifard-10}

The space of graphons provides the completion of the space of finite
graphs with respect to a topology generated by a particular metric
namely, cut-distance. The cut-distance between labeled graphons
$W,U$ is defined by
\begin{equation} \label{cut-1}
\delta_{{\rm cut}}(W,U):= {\rm inf}_{\rho, \tau} \ {\rm sup}_{S,T}
|\int_{S \times T} W(\rho(x),\rho(y)) - U(\tau(x),\tau(y)) dxdy |
\end{equation}
such that the infimum is taken over all relabeling $\rho$ on $W$
and $\tau$ on $U$ and the supremum is taken over all measurable
subsets $S,T$ of the closed interval. The infimum over relabeling
allows us to define the cut-distance on the space of unlabeled
graphons.

Weakly isomorphic graphons and relabeled graphons with respect to a given graphon
have the same corresponding symmetric
measurable functions almost everywhere. The distance (\ref{cut-1})
does not distinguish between weakly isomorphic graphons. We can see that graphons $W,U$ are weakly isomorphic whenever
for any finite subgraph $H$, we have $t(H,W)=t(H,U)$.

\begin{thm}
Each graphon is the cut-distance convergent limit of a sequence of
finite graphs. In addition, the cut-distance $\delta_{{\rm cut}}$
determines a compact topological structure on the quotient space
$\mathcal{W}(\Omega)/\approx$ of labeled graphons with respect to the weakly
isomorphic relation. \cite{janson-1,lovasz-1}
\end{thm}

\subsection{\textsl{Feynman graphons}}

It is the place to deal with a new application of graph limits to
Quantum Field Theory where we aim to achieve a new interpretation of
Feynman diagrams and their corresponding formal expansions. In
general, any arbitrary Feynman diagram, as a weighted graph decorated by some
physical parameters, might have many nested or overlapping loops. In higher order perturbation expansions in (strongly) coupled gauge field theories, we can investigate the appearance of many loops Feynman diagrams. We can encode these loops as vertices of the decorated rooted tree representation of the original Feynman diagram and then consider this sparse graph to generate its corresponding graphon model. Decorated rooted trees, as simple graphs with low densities, are useful to simplify Feynman diagrams and their complicated formal expansions which have more densities than sparse graphs. Generating non-zero graphon models for the infinite sequence of rooted trees requires to apply renormalization methods on the canonical graphons and the ground probability space. \cite{br-1,shojaeifard-10, shojaeifard-18, ihes-submitted}

\begin{lem} \label{feynman-graphon-1}
For a fixed probability space $(\Omega,\mu_{\Omega})$, the algebraic combinatorics of each Feynman diagram can be
encoded by a unique unlabeled graphon class.
\end{lem}

\begin{proof}
A rooted tree $t$ is a finite, connected oriented graph without
loops in which every vertex has exactly one incoming edge, except one namely, the
root which has no incoming but only outgoing edges. We can put two classes of
decorations on each tree namely, vertex-labeled and edge-labeled.
The rooted tree representations of Feynman diagrams can be defined
via the grafting operator. The free commutative algebra generated by
isomorphism classes of non-planar rooted trees is actually the polynomial algebra generated by symbols $t$ where each symbol represents one isomorphism class. The concatenation is the product and the empty tree is the unit for this polynomial algebra. In addition, this polynomial algebra can be equipped by a modified version of the
renormalization coproduct given by
\begin{equation}
\Delta_{{\rm CK}}(t) = \mathbb{I} \otimes t + t \otimes \mathbb{I} +
\sum_{c} R_{c}(t) \otimes P_{c}(t)
\end{equation}
such that the sum is taken over all admissible cuts $c$ on $t$ which
divides the tree into two parts. The part $R_{c}(t)$ contains the
original root of $t$ and the part $P_{c}(t)$ is a forest of
subtrees. The resulting Hopf algebra $H_{{\rm CK}}$ of non-planar
rooted trees is connected graded free commutative non-cocommutative
finite type. Decorations enable us to adapt this combinatorial Hopf
algebra with respect to physical theories. Each Feynman diagram $\Gamma$,
which might contain divergent Feynman subgraphs, is encoded by a
decorated non-planar rooted tree $t_{\Gamma}$ such that the root
represents the full graph and each leaf is a divergent subgraph
which has no further subdivergencies. If the original graph has
overlapping subdivergencies, then we can replace the single rooted
tree by a sum of decorated rooted trees after disentangling the
overlaps. Thanks to these rules, we can embed the Hopf algebra
$H_{{\rm FG}}(\Phi)$ of Feynman diagrams of $\Phi$ into the
decorated Connes--Kreimer Hopf algebra $H_{{\rm CK}}(\Phi)$ as a
closed Hopf subalgebra. This embedding is encapsulated by the
injective Hopf algebra homomorphism
\begin{equation} \label{Feynman-tree}
\Gamma \longmapsto \Xi(\Gamma):= \sum_{j=1}^{r}
B^{+}_{\Gamma_{j},G_{j,i}} \big(\prod_{i=1}^{k_{j}}
\Xi(\gamma_{j,i}) \big)
\end{equation}
such that $\Gamma=\prod_{i=1}^{k_{j}} \Gamma_{j} \star_{j,i}
\gamma_{j,i}$ and each $G_{j,i}$ is the gluing information. For a
given decorated non-planar rooted tree $t$, if the longest path from
the root to a leaf contains $k$ edges, then the renormalization
coproduct $\Delta_{{\rm CK}}(t)$ is a sum of at least $k+1$ terms.
In other words, the decorated non-planar rooted tree $t$ represents
an iterated integral with $k$ nested sub-divergencies while each
vertex corresponds to a sub-integral without any sub-divergencies.
\cite{ebrahimifard-kreimer-1, ebrahimifard-guo-kreimer-1,
ebrahimifard-guo-kreimer-2, figueroa-graciabondia-1, hoffman-1}

For any Feynman diagram $\Gamma$ without overlapping sub-divergencies,
the decorated tree $t_{\Gamma}:= \Xi(\Gamma)$ is a simple finite
weighted graph where thanks to its corresponding adjacency matrix we
can determine the labeled graphons $W_{t_{\Gamma}}$ of the form
(\ref{graphon-graph-1}) or (\ref{graphon-graph-2}). Set
$[W_{t_{\Gamma}}]$ as the unlabeled graphon class associated to
$t_{\Gamma}$. Up to the rearrangement, weakly isomorphic relation and the embedding (\ref{Feynman-tree}), we can show that the unlabeled graphon $[W_{t_{\Gamma}}]$.

For any Feynman diagram $\Gamma$ which has some overlapping
sub-divergencies, $u_{\Gamma}:= \Xi(\Gamma)$ is a linear combination
of decorated non-planar rooted trees such as $u_{\Gamma}= \alpha_{1} t_{1} + ... + \alpha_{n} t_{n}$. In this situation, the labeled
graphons $W_{u_{\Gamma}}$ can be determined by normalizing or rescaling methods used on $W_{t_{i}}$s. For each $1 \le i \le n$, the labeled graphon $W_{t_{i}}$ is projected or embedded into the subinterval $I_{i}$ of $[0,1]$ and then we have
\begin{equation}
W_{u_{\Gamma}}(x,y):=\frac{W_{t_{1}} + ... + W_{t_{n}}}{|W_{t_{1}} +
...+ W_{t_{n}}|}.
\end{equation}
The subintervals $I_{i}$ are determined in terms of the grading value of rooted trees and cut-distance topology.
\end{proof}

\begin{rem}
For each Feynman diagram $\Gamma$ and each natural number $n \ge 2$, the graphon $W_{n \Gamma}$ is actually $n$ copies of $W_{\Gamma}$ inside the closed unital interval. We can generate this graphon by measure preserving transformation $\rho_{n}: x \mapsto nx$ on $[0,1]$ which is not invertible with respect to the Lebesgue measure. Therefore $W_{\Gamma}$ and $W_{n \Gamma}$ are weakly isomorphic.
\end{rem}

We use the phrase ''Feynman graphons'' for this class of graphons which provide an analytic generalizations for Feynman diagrams.

\begin{defn}  \label{feynman-graphon-2}
A sequence $\{\Gamma_{n}\}_{n \ge 0}$ of non-trivial Feynman diagrams is called
convergent when $n$ tends to infinity, if the corresponding sequence
$\{[W_{t_{\Gamma_{n}}}]\}_{n \ge 0}$ of non-zero unlabeled graphon classes is
convergent to a unique non-zero Feynman graphon class with respect to the cut-distance topology when $n$ tends
to infinity.
\end{defn}

Non-zero graphons, as the cut-distance convergent limit of the sequences of sparse graphs, can be built in terms of applying rescaling or renormalization methods to the ground probability space or canonical graphons. \cite{br-1, ihes-submitted}

Suppose the unlabeled graphon class $[W]$ is the convergent limit
for the sequence $\{[W_{t_{\Gamma_{n}}}]\}_{n \ge 0}$. If we
consider the pixel picture representation of the graphon $[W]$, then
we can associate an infinite tree or forest $t$ such that $W_{t} \in [W]$ and
$W \in [W_{t}]$. Therefore $[W]=[W_{t}]$. Thanks to the homomorphism
(\ref{Feynman-tree}), it is possible to build an extremely large
Feynman diagram $\Gamma_{t}$ with respect to the infinite tree or forest $t$. This
$\Gamma_{t}$ can be described as the convergent limit of the
sequence $\{\Gamma_{n}\}_{n \ge 0}$ with respect to the cut-distance
topology. We can also show that this limit is unique up to the rearrangement, weakly isomorphic relation and the embedding (\ref{Feynman-tree}). \cite{shojaeifard-10}

Graphon models for Feynman diagrams in a given gauge field theory can lead us to study formal expansions of Feynman diagrams in the context of the theory of random graphs.

The study of random graphs was begun by Erdos, Renyi and Gilbert
when they were working on a probabilistic construction of a graph
with large girth and large chromatic number. After a short period of
time, work on random graphs $G_{n,m}$ has been concerned by many
mathematicians in Combinatorics and Discrete Mathematics. Nowadays
it is not difficult to observe various applications of these
combinatorial objects in many fields in Mathematics and other
applied sciences. Generally speaking, the theory of random graphs aims
to provide some results such as ''a combinatorial property A almost
always implies another combinatorial property B''. For any integer value $n$ and $0 \le p \le 1$, a random graph
$G(n,p)$ is defined by taking $n$ nodes and connecting any two of
them with the probability $p$, making an independent decision about
each pair. There are alternative ways to build random graphs. As
an example, consider $\mathcal{L}_{n,m}$ as the collection of all
labeled graphs with the vertex set $V=[n]=\{1,2,...,n\}$ and $m$ edges
such that $0 \le m \le \binom {n} {2}$. To each $G \in
\mathcal{L}_{n,m}$, assign the probability
\begin{equation} \label{random-1}
\mathbb{P}(G) = \frac{1}{\binom {\binom {n} {2}} {m}}.
\end{equation}
In other words, start with an empty graph on the set $[n]$ and
insert $m$ edges in such a way that all possible $\binom {\binom {n}
{2}} {m}$ choices are equally likely. The resulting graph $G_{n,m}:=
([n],E_{n,m})$ is known as the uniform random graph. As other
example, fix $0 \le p \le 1$ and for each graph $G$ with the vertex set
$[n]$ and $0 \le m \le \binom {n} {2}$ edges, assign the
probability
\begin{equation} \label{random-2}
\mathbb{P}(G)=p^{m}(1-p)^{\binom {n} {2} - m}.
\end{equation}
In other words, start with an empty graph with the vertex set $[n]$
and consider $\binom {n} {2}$ to insert edges independently with the
probability $p$. The resulting graph $G_{n,p}:= ([n],E_{n,p})$ is
known as the binomial random graph.

It is shown that the random graph $G_{n,p}$ with $0 \le m \le \binom
{n} {2}$ edges is the same as one of the $\binom {\binom {n} {2}}
{m}$ graphs that have $m$ edges. For enough large $n$, random graphs
$G_{n,m}$ and $G_{n,p}$ have the same behavior whenever the number
of edges $m$ in $G_{n,m}$ is very close to the expected number of
edges of $G_{n,p}$, namely,
\begin{equation}
m = \binom {n} {2} p \approx \frac{n^{2}p}{2}.
\end{equation}
It is equivalent to say that the edge probability in $G_{n,p}$
should be $p \approx \frac{2m}{n^{2}}$. \cite{frieze-karonski-1}

\begin{lem}  \label{feynman-graphon-3}
Each labeled graphon determines a class of random graphs.
\end{lem}

\begin{proof}
If we have a simple weighted graph $G$, then we can build a random
simple graph $R(G)$ by including the edge with probability equal to
its weight. Thanks to this idea, suppose we have a labeled graphon
$W$ and finite subset ${\it S}:=\{s_{1},...,s_{n}\}$ in $[0,1]$. We
can make a weighted graph $G({\it S},W)$ with $|{\it S}|=n$ nodes
such that the edge $s_{i}s_{j}$ has the weight $W(s_{i},s_{j})$. In
general, the random graph $R(n,W):=R(G({\it S},W))$ with respect to
the weighted graph $G({\it S},W)$ is our promising graph such that
${\it S}$ is a set of $n$ points which are selected independently
from the closed interval.
\end{proof}

The random graphs $R(n,W)$ are useful to approximate graphons
$W$ associated to large numbers of points in the closed interval. It
is shown that with the probability $1$, the sequence $\{R(n,W)\}_{n \ge
0}$ is convergent to the graphon $W$ with respect to the
cut-distance topology when $n$ tends to infinity. \cite{lovasz-1}

Thanks to the discussed topics, it is time to observe some new
applications of the graph function representation theory of Feynman
diagrams in dealing with expansions of these physical graphs in
Quantum Field Theory. At the first application, we address a new
interpretation of solutions of Dyson--Schwinger equations in the context of
sequences of random graphs derived from partial sums of solutions.

\begin{thm}  \label{feynman-graphon-4}
The unique solution of each combinatorial Dyson--Schwinger equation
can be described as the cut-distance convergent limit of a sequence
of finite Feynman diagrams.
\end{thm}

\begin{proof}
The full proof is given in \cite{shojaeifard-10} and here we only
address the main idea. Suppose DSE be a combinatorial Dyson--Schwinger equation with the
general form (\ref{dse-1}) such that its unique solution is given by
\begin{equation}
X_{{\rm DSE}}=\sum_{n \ge 0} (\lambda g)^{n} X_{n}
\end{equation}
while $g$ is the
bare coupling constant and the generators $X_{n}$ are determined by the
recursive relations (\ref{dse-4}). Make the new sequence
$\{Y_{m}\}_{m \ge 1}$ of partial sums of the expansion $\sum_{n \ge
0} (\lambda g)^{n} X_{n}$ such that we have
\begin{equation}
Y_{m}:= (\lambda g)^{1}X_{1}+...+(\lambda g)^{m}X_{m}.
\end{equation}
It is shown in \cite{bergbauer-kreimer-1, foissy-4} that the
large Feynman diagram $X_{{\rm DSE}}$ belongs to a completion of $H_{{\rm
FG}}[[g]]$ with respect to the $n$-adic topology. We claim that the
sequence $\{Y_{m}\}_{m \ge 1}$ of finite graphs converges to the
large Feynman diagram $X_{{\rm DSE}}$ with respect to the
cut-distance topology. For this purpose, we can apply the $n$-adic
metric and the graphon representations of the components $X_{n}$ to
build a random graph with respect to each graph $Y_{m}$. It leads us
to associate a sequence $\{R(Y_{m})\}_{m \ge 1}$ of random graphs
with respect to the sequence $\{Y_{m}\}_{m \ge 1}$ which is
cut-distance convergent to $X_{{\rm DSE}}$.
\end{proof}

The structure of a modification of the Connes--Kreimer BPHZ
renormalization for large Feynman diagrams has been formulated in
\cite{shojaeifard-10} where we worked on a topological completion of
the renormalization Hopf algebra of Feynman diagrams with respect to
the cut-distance topology. As the second application, we work on Feynman graphon formalism to build a
renormalization program on the collection $\mathcal{S}^{\Phi,g}$
under a topological Hopf algebraic setting. This new approach enables us to
proceed our knowledge about non-perturbative versions of Feynman
rules which act on large Feynman diagrams. For this purpose we
explain the structure of a new Hopf algebra derived from the
renormalization coproduct on graphons.

\begin{thm}  \label{feynman-graphon-5}
Thanks to the renormalization coproduct, there exists a topological
Hopf algebraic structure on the collection $\mathcal{S}^{\Phi}_{{\rm
graphon}}$ of all unlabeled graphons which contribute to represent (large) Feynman diagrams of a physical theory
$\Phi$.
\end{thm}

\begin{proof}
We plan to equip $\mathcal{S}^{\Phi}_{{\rm graphon}}$ with an enriched version of the renormalization Hopf algebra which is completed with respect to the cut-distance topology.

Thanks to Lemma \ref{feynman-graphon-1}, for each finite Feynman
diagram $\Gamma$, we associate the unlabeled graphon class
$[W_{\Gamma}]$. In addition, the unique solution of each
combinatorial Dyson--Schwinger equation DSE in
$\mathcal{S}^{\Phi,g}$ determines a unique large Feynman diagram
$X_{{\rm DSE}}$ such that thanks to Theorem \ref{feynman-graphon-4},
this infinite graph can be interpreted as the convergent limit of
the sequence of partial sums with respect to the cut-distance
topology. Therefore it does make sense to replace objects of
$\mathcal{S}^{\Phi,g}$ with large Feynman diagrams such as $X_{{\rm DSE}}$
as the unique solution of the equation DSE. Thanks to Lemma
\ref{feynman-graphon-1}, we associate a unique unlabeled
graphon class $[W_{t_{X_{{\rm DSE}}}}]$ to the large Feynman diagram
$X_{{\rm DSE}}$ via its rooted tree (forest) representation. For simplicity in the presentation, from now we use the
notation $[W_{X_{{\rm DSE}}}]$ for this graphon class.

It is possible to lift the renormalization coproduct (\ref{copro-1})
onto the level of unlabeled graphons which contribute to the
description of (large) Feynamn diagrams. For a given finite Feynman
diagram $\Gamma$ with the corresponding unlabeled graphon
$[W_{\Gamma}]$, define
\begin{equation} \label{cop-graphon-1}
\Delta_{{\rm graphon}}([W_{\Gamma}]):= \sum [W_{\gamma}] \otimes
[W_{\Gamma / \gamma}]
\end{equation}
such that the sum is taken over all unlabeled graphon classes such as
$[W_{\gamma}]$ associated to $\gamma$ as the disjoint union of 1PI superficially
divergent subgraphs of $\Gamma$.

Thanks to Theorem \ref{feynman-graphon-4}, for the unlabeled graphon
class $[W_{X_{{\rm DSE}}}]$ corresponding to the large Feynman
diagram $X_{{\rm DSE}}$, define its coproduct as the convergent
limit of the sequence $\{\Delta_{{\rm graphon}}([W_{Y_{m}}])\}_{m
\ge 1}$ of the coproducts of the finite partial sums with respect to
the cut-distance topology.

Now we can adapt (\ref{cop-graphon-1}) for the level of large
Feynman diagrams and define
\begin{equation} \label{cop-graphon-2}
\Delta_{{\rm graphon}}([W_{X_{{\rm DSE}}}]):= \sum [W_{\Upsilon}]
\otimes [W_{X_{{\rm DSE}} / \Upsilon}]
\end{equation}
such that the sum is taken over all unlabeled graphon classes such as
$[W_{\Upsilon}]$ associated to $\Upsilon$ as the disjoint union of 1PI superficially
divergent subgraphs of $X_{{\rm DSE}}$.

If we consider objects of $\mathcal{S}^{\Phi}_{{\rm graphon}}$ as generators of a free commutative algebra, then thanks to (\ref{cop-graphon-1}) we obtain a bialgebra structure on Feynman graphons which is graded
in terms of the number of independent loops of the corresponding
Feynman diagrams. The unlabeled graphon class $[W_{\mathbb{I}}]$
corresponding to the empty graph is the unit for this bialgebra. The
counit is also defined by
\begin{equation}
\tilde{\varepsilon}([W_{\Gamma}])=\{^{1, \ \ \ \
[W_{\Gamma}]=[W_{\mathbb{I}}]}_{0, \ \ \ \ \ {\rm else}}.
\end{equation}

The existence of the graduation parameter is the key tool to define
an antipode map. For each finite Feynman diagram $\Gamma$, we have
\begin{equation} \label{antipode-graphon-1}
S_{{\rm graphon}}([W_{\Gamma}]) = - [W_{\Gamma}] - \sum
S_{{\rm graphon}}([W_{\gamma_{(1)}}]) [W_{\gamma_{(2)}}]
\end{equation}
such that $\Delta_{{\rm graphon}}([W_{\Gamma}])=\sum
[W_{\gamma_{(1)}}]  \otimes [W_{\gamma_{(2)}}]$.

Thanks to Theorem \ref{feynman-graphon-4}, for the unlabeled graphon
class $[W_{X_{{\rm DSE}}}]$ corresponding to the large Feynman
diagram $X_{{\rm DSE}}$, define its antipode as the convergent limit
of the sequence $\{S_{{\rm graphon}}([W_{Y_{m}}])\}_{m \ge 1}$ of
unlabeled graphons of finite partial sums $Y_{m}$ with respect to
the cut-distance topology. Since partial sums are finite graphs,
their corresponding graphon type antipodes $S_{{\rm graphon}}([W_{Y_{m}}])$ can be obtained inductively by the coproduct
$\Delta_{{\rm graphon}}$ where we have
\begin{equation} \label{antipode-graphon-2}
S_{{\rm graphon}}([W_{Y_{m}}]) = - [W_{Y_{m}}] - \sum
S_{{\rm graphon}}([W_{\Gamma_{(1)}}]) [W_{\Gamma_{(2)}}]
\end{equation}
such that $\Delta_{{\rm graphon}}([W_{Y_{m}}])=\sum
[W_{\Gamma_{(1)}}] \otimes [W_{\Gamma_{(2)}}]$.

Now we can adapt the antipode (\ref{antipode-graphon-1}) for the
level of large Feynman diagrams and define
\begin{equation} \label{antipode-graphon-3}
S_{{\rm graphon}}([W_{{\rm DSE}}]) = - [W_{{\rm DSE}}] - \sum
S_{{\rm graphon}}([W_{\Upsilon_{(1)}}]) [W_{\Upsilon_{(2)}}]
\end{equation}
such that $\Delta_{{\rm graphon}}([W_{{\rm DSE}}])=\sum
[W_{\Upsilon_{(1)}}]  \otimes [W_{\Upsilon_{(2)}}]$.

Therefore $\mathcal{S}^{\Phi}_{{\rm graphon}}$ becomes a connected
graded free commutative non-cocommutative (not necessarily finite type) Hopf algebra.
In addition, the recursive structure of the coproduct
(\ref{cop-graphon-2}) and the antipode (\ref{antipode-graphon-3}) together with their linear property allow us to show the compatibility of this Hopf algebraic structure with
the cut-distance topology.
\end{proof}

Feynman graphon models of Feynman diagrams allow us to show that the space $\mathcal{S}^{\Phi}_{{\rm graphon}}$ (given by Theorem \ref{feynman-graphon-5}) is completed with respect to the cut-distance topology.

\begin{cor}
Let $V$ be a complex vector space with a basis labeled by coupling
constants of a given Quantum Field Theory $\Phi$, and suppose ${\rm
Diff}(V)$ be the group of formal diffeomorphisms of $V$ tangent to
the identity at $0 \in V$ and $H_{{\rm diff}}(V)$ be its
corresponding Hopf algebra. The complex Lie group $\mathbb{G}^{\Phi}_{{\rm
graphon}}(\mathbb{C})$ of characters on Feynman graphons can be
represented by ${\rm Diff}(V)$.
\end{cor}

\begin{proof}
The Hopf algebra $H_{{\rm diff}}(\mathbb{C})$ of formal
diffeomorphisms of $\mathbb{C}$ tangent to the identity has
generators such as $a_{n}$ which play the role of coordinates of
\begin{equation}
\phi(x) = x + \sum_{n \ge 2} a_{n}(\phi)x^{n}
\end{equation}
such that $\phi$ is a formal diffeomorphism satisfying $\phi(0)=0$,
$\phi'(0)={\rm id}$. Its coproduct is given by
\begin{equation}
\Delta(a_{n})(\phi_{1} \otimes \phi_{2}) = a_{n}(\phi_{2} \circ
\phi_{1}).
\end{equation}

We can define a Hopf algebra homomorphism $\Psi: H_{{\rm diff}}(V)
\rightarrow H_{{\rm FG}}(\Phi)$ with the corresponding dual
group homomorphism $\hat{\Psi}: \mathbb{G}_{\Phi}(\mathbb{C})
\rightarrow {\rm Diff}(V)$. The map $\Psi$ maps the coefficients
of the expansion of formal diffeomorphisms to the coefficients in
the renormalization Hopf algebra of the expansions of the effective (or running)
coupling constants of the physical theory as formal power series in the bare
coupling constants. As the consequence, for each Dyson--Schwinger
equation DSE with the corresponding Hopf subalgebra $H_{{\rm DSE}}$
and Lie (sub)group $\mathbb{G}_{{\rm DSE}}(\mathbb{C})$, we can define
a group homomorphism $\hat{\Psi}_{{\rm DSE}}$ from $\mathbb{G}_{{\rm
DSE}}(\mathbb{C})$ to ${\rm Diff}(V)$. \cite{connes-marcolli-1,
shojaeifard-5}

Thanks to Lemma \ref{feynman-graphon-1}, we can embed $H_{{\rm
FG}}(\Phi)$ into the renormalization Hopf algebra
$\mathcal{S}^{\Phi}_{{\rm graphon}}$ of Feynman graphons. This
allows us to lift the map $\Psi$ onto the level of Feynman graphons
and build a new Hopf algebra homomorphism $\overline{\Psi}: H_{{\rm
diff}}(V) \rightarrow \mathcal{S}^{\Phi}_{{\rm graphon}}$ with
the corresponding dual group homomorphism $\hat{\overline{\Psi}}:
\mathbb{G}^{\Phi}_{\rm graphon}(\mathbb{C}) \rightarrow {\rm Diff}(V)$.
\end{proof}

The construction of a canonical filtration on terms $X_{n}$s of the
unique solution of a given Dyson--Schwinger equation has been explained in \cite{kruger-kreimer-1} where each filtered term
maps to a certain power of $L$ in the log-expansion. The original
idea is to filter images of Feynman diagrams in a particular
universal enveloping algebra which generates a quasi-shuffle type
Hopf algebra. Thanks to Theorem \ref{feynman-graphon-4} and Theorem
\ref{feynman-graphon-5}, we aim to adapt this filtration for large
Feynman diagrams.

\begin{thm} \label{graphon-filtration-1}
Renormalized Feynman rules characters of the Hopf algebra
$\mathcal{S}^{\Phi}_{{\rm graphon}}$ filtrate large Feynman
diagrams.
\end{thm}

\begin{proof}
Set $H_{{\rm word}}$ as the vector space of words which contains
$H_{{\rm letter}}$ as the subspace of letters. Set a commutative associative map $\Theta: H_{{\rm
letter}} \times H_{{\rm letter}} \rightarrow H_{{\rm letter}}$ as the Hoffman pairing which
sends two generators $a,b$ to another generator $\Theta(a,b)$ and adds degrees. Define the
generalized quasi-shuffle product $\ominus_{\Theta}$ on $H_{{\rm
word}}$ as follows
\begin{equation}
au \ominus_{\Theta} bv:= a(u \ominus_{\Theta} bv) + b(au
\ominus_{\Theta} v) + \Theta(a,b)(u \ominus_{\Theta} v)
\end{equation}
which builds a commutative associative algebra with empty word
$\mathbb{I}$ as the unit. We can equip this algebra with the
following coproduct structure
\begin{equation} \label{shuffle-cop-1}
\Delta_{{\rm word}}(w) = \sum_{vu=w} u \otimes v
\end{equation}
which gives us a bialgebra structure on $H_{{\rm word}}$ with the
counit $\hat{\mathbb{I}}_{{\rm word}}$. The length of each word
determines a natural graduation parameter on this bialgebra which
leads us to define an antipode recursively. As the consequence,
$(H_{{\rm word}}, \ominus_{\Theta}, \mathbb{I}, \Delta_{{\rm word}},
\hat{\mathbb{I}}_{{\rm word}}, S_{{\rm word}})$ is a graded
connected commutative unital non-cocommutative counital Hopf algebra
\cite{ebrahimifard-guo-1, hoffman-2}. We have
\begin{equation}
\ominus_{\Theta} \circ (S_{{\rm word}} \otimes {\rm id}) \circ
\Delta_{{\rm word}} = \ominus_{\Theta} \circ ({\rm id} \otimes
S_{{\rm word}}) \circ \Delta_{{\rm word}} = \mathbb{I}_{{\rm word}}
\circ \hat{\mathbb{I}}_{{\rm word}}.
\end{equation}
In this setting, the grafting operator on words allows us to add a
letter to the first place
\begin{equation}
B_{a}^{+}(u):= au.
\end{equation}
We can check that for each $a$, the grafting operators are
Hochschild one-cocycles. It is possible to embed the renormalization
Hopf algebra of Feynman diagrams into the Hopf algebra of words.
This embedding is defined in terms of the homomorphism
$\nu: H_{{\rm FG}}(\Phi) \rightarrow H_{{\rm word}}$ determined by the relations
$$\nu(\mathbb{I})=\mathbb{I}_{{\rm word}}, \ \ \
\ominus_{\Theta} \circ (\nu \otimes \nu) = \nu \circ m, \ \ \
\hat{\mathbb{I}}_{{\rm word}} \circ \nu = \nu \circ
\hat{\mathbb{I}},$$
\begin{equation}
\Delta_{{\rm word}} \circ \nu = (\nu \otimes \nu) \circ \Delta_{{\rm
FG}}, \ \ \ S_{{\rm word}} \circ \nu = \nu \circ S, \ \ \
B_{a_{n}}^{+} \circ \nu = \nu \circ B^{+}_{\gamma_{n}}.
\end{equation}
The morphism $\nu$ sends each primitive Feynman graph $\gamma_{n}$
to a letter $a_{n}$. Thanks to Theorem \ref{feynman-graphon-5}, it
is possible to lift the embedding $\nu$ onto a new homomorphism
$\overline{\nu}$ which embeds the renormalization Hopf algebra of
graphons $\mathcal{S}^{\Phi}_{{\rm graphon}}$ into the Hopf algebra
of words. It is enough to replace each Feynman diagram $\Gamma$ with
its corresponding unlabeled graphon class $[W_{\Gamma}]$.

Consider Dyson--Schownger equations for 1PI Green's functions with
the general form
\begin{equation} \label{dse-green-1}
\Gamma^{\bar{n}} = 1 + \sum_{\gamma, {\rm res}(\gamma)=\bar{n}}
\frac{g^{|\gamma|}}{{\rm Sym}(\gamma)} B^{+}_{\gamma}
(X^{\gamma}_{\mathcal{R}})
\end{equation}
such that $B^{+}_{\gamma}$ are Hochschild closed one-cocycles of the
Hopf algebra of Feynman diagrams indexed by Hopf algebra primitives
$\gamma$ with external legs $\bar{n}$, $X^{\gamma}_{\mathcal{R}}$ is
a monomial in superficially divergent Green's functions which dress
the internal vertices and edges of $\gamma$. If we apply the
renormalized Feynman rules character $\phi_{r}$ to a Feynman graph
which contributes to this class of equations, then we can obtain a
polynomial in a suitable external scale parameter $L={\rm
log}S/S_{0}$ such that $S_{0}$ fixes a reference scale for the
renormalization process. At the end of the day, we can get a
renormalized version $G_{r}(g,L,\theta)$ of Green's functions. Lemma
\ref{feynman-graphon-1} and Theorem \ref{feynman-graphon-5} are
useful to reformulate the equation (\ref{dse-green-1}) in the
language of graphons as an equation in the Hopf algebra
$\mathcal{S}^{\Phi}_{{\rm graphon}}$. The embedding $\overline{\nu}$
enables us to lift this graphon model Dyson--Schwinger equations
onto their corresponding equations in the Hopf algebra of words. We
have
\begin{equation}
X_{{\rm DSE},{\rm word}} = \overline{\nu}([W_{X_{{\rm DSE}}}])=
\mathbb{I}_{{\rm word}} + \sum_{n \ge 1} g^{n}B^{+}_{l_{n}}(X_{{\rm
DSE},{\rm word}}^{\ominus_{\Theta} (n+1)})
\end{equation}
such that $X_{{\rm DSE},{\rm word}}$ is the word representation of
the unlabeled graphon class $[W_{X_{{\rm DSE}}}]$ with respect to
the large graph $X_{{\rm DSE}}$. We have $X_{{\rm DSE},{\rm word}} =
\sum_{n \ge 0} g^{n}z_{n}$ such that each $z_{n}=\nu(X_{n})$ is
determined recursively by the relations
\begin{equation}
z_{n} = \sum_{m=1}^{n}B^{+}_{l_{m}}(\sum_{k_{1}+...+k_{m+1}=n-m, \ \
k_{i} \ge 0} z_{k_{1}} \ominus_{\Theta} ... \ominus_{\Theta}
z_{k_{m+1}}).
\end{equation}

We plan to explain the filtration structure on words and then by
applying the inverse of the embedding $\nu$, we can adapt it for
the level of Feynman graphon representations of large Feynman diagrams which contribute to solutions of Dyson--Schwinger equations.

The canonical candidate for the filtration on words is built
in terms of the lower central series at the Lie algebra level where
we need to apply the theory of Hall sets and Hall basis. The
Milnor--Moore theorem (\cite{milnor-moore-1}) allows us to build the
graded dual Hopf algebra to $H_{{\rm word}}$ in terms of the
universal algebra of a particular Lie algebra.

A bilinear anti-symmetric map $[.,.]$ on a vector space
$\mathcal{L}$ over the field $\mathbb{K}$ with characteristic zero
defines a Lie algebra structure if it obeys the conditions \\
$$\forall x \in \mathcal{L}: \ \ [x,x]=0,$$
\begin{equation}
\forall x,y,z \in \mathcal{L}: \ \
[x,[y,z]]+[y,[z,x]]+[z,[x,y]]=0.
\end{equation}
The lexicographical ordering enables us to build the Hall basis for
the Lie algebra $\mathcal{L}$ \cite{hall-1, hall-2}. For a given
ordering  $x_{1} < x_{2} < ... < [x_{1},x_{2}] < ...$ on
$\mathcal{L}$, define $[x,x']$ as an element of a Hall basis for
$\mathcal{L}$ iff \\
- $x,x' \in \mathcal{L}$ are Hall basis elements with $x<x'$, \\
- if $x'=[x_{1},x_{2}]$, then $x' \ge x_{2}$. \\
The unique universal enveloping algebra associated to $\mathcal{L}$
is defined in terms of the tensor algebra
$(T(\mathcal{L}),\otimes,1)$ such that
\begin{equation}
T(\mathcal{L}):= \bigoplus_{n \ge 0} \mathcal{L}^{\otimes n}.
\end{equation}
Set
\begin{equation}
I:= \{s \otimes (x \otimes y - y \otimes x - [x,y]) \otimes t: \ \
x,y \in \mathcal{L}; \ \ s,t \in T(\mathcal{L})\}
\end{equation}
as a two sided ideal and then define the equivalent classes of the
form
\begin{equation}
[t]:= \{s \in T(\mathcal{L}): \ \ s-t \in I\}.
\end{equation}
$T(\mathcal{L})/I$ is actually the unique universal enveloping
algebra $\mathcal{U}(\mathcal{L})$ generated by $\mathcal{L}$ where
the product of this algebra is given by
\begin{equation}
m_{\mathcal{U}(\mathcal{L})}([s] \otimes [t]):= [s \otimes t]
\end{equation}
and $[1]$ is the unit of this algebra. In addition, we can equip
$\mathcal{U}(\mathcal{L})$ by a graded Hopf algebra structure with
the coproduct
\begin{equation}
\Delta_{\mathcal{U}(\mathcal{L})}([x])=[x] \otimes \mathbb{I} +
\mathbb{I} \otimes [x]
\end{equation}

Set $\mathcal{L}_{{\rm word}}$ as the Lie algebra corresponding to
the Hopf algebra $H_{{\rm word}}$. Define the decreasing sequence
\begin{equation}
\mathcal{L}_{{\rm word}}=\mathcal{L}_{1} \ge \mathcal{L}_{2} \ge
\mathcal{L}_{3} \ge ...
\end{equation}
such that $\mathcal{L}_{n+1}$ is generated by all objects $[x,y]$
with $x \in \mathcal{L}$ and $y \in \mathcal{L}_{n}$. For letters
$a_{1},a_{2},...,\Theta(a_{1},a_{2}),... \in H_{{\rm letter}}$, set
$x_{1},x_{2},...,\Theta(x_{1},x_{2}), ... \in \mathcal{L} /
\mathcal{L}_{2}$. The duality between $H_{{\rm word}}$ and
$\mathcal{U}(\mathcal{L}_{{\rm word}})$ can be determined by the
unique linear invertible map $\nu$ such that
\begin{equation}
\nu(a_{i})=[x_{i}], \ \ \nu(\Theta(a_{i},a_{j})) =
[\Theta(x_{i},x_{j})], \ \ \ \nu(a_{i}a_{j})=[x_{i} \otimes x_{j}],
....
\end{equation}
The universal enveloping algebra $\mathcal{U}(\mathcal{L}_{{\rm
word}})$ is a filtered bialgebra.

Thanks to the structure of the quasi-shuffle product, we can build a
filtration algorithm where it requires to consider all words with
length $k$ into the lexicographical order in terms of the
concatenation commutator with respect to the Hall basis which
generates words with the length $k-1$. This procedure, which starts
with the maximal length of words, should be repeated for the full
quasi-shuffle products of the $k$ corresponding letters and then
insert them into the expression \cite{kruger-kreimer-1}. Now if we
apply the inverse of the embedding $\overline{\nu}$, then this
filtration can be defined on Feynman graphons and large Feynman
diagrams which live in $\mathcal{S}^{\Phi}_{{\rm graphon}}$.

Let us now apply renormalized Feynman rules characters on large
Feynman diagrams. The Hopf algebra homomorphism $\nu$ sends the
renormalized Feynman rules character $\phi_{r}$ to
\begin{equation}
\psi_{r}=\phi_{r} \circ \nu^{-1}.
\end{equation}
In \cite{kruger-kreimer-1}, it is shown that for each word $w \in
H_{{\rm word}}$ with the corresponding $[x] \in
\mathcal{U}(\mathcal{L}_{{\rm word}})$, if $x \in
T(\mathcal{L}_{{\rm word}})$ is also an element of the Lie algebra
$\mathcal{L}_{{\rm word}}$, then $\psi_{r}(u)$ maps to the
$L$-linear part of the log-expansion of the renormalized Green's
functions. In addition, we have
\begin{equation}
\psi_{r}(u \circleddash_{\Theta} v) = \psi_{r}(u).\psi_{r}(v).
\end{equation}
For a given Feynman diagram $\Gamma$ with the coradical degree
$r_{\Gamma}$, we have
\begin{equation} \label{feynman-rules-1}
\phi_{r}(\Gamma) = \sum_{j=1}^{r_{\Gamma}}
c_{j}^{\Gamma}(\theta)L^{j}
\end{equation}
such that
\begin{equation} \label{feynman-rules-2}
c_{j}^{\Gamma} = c_{1}^{\otimes j} \tilde{\Delta}_{{\rm
FG}}^{j-1}(\Gamma)
\end{equation}
while $c_{1}^{\otimes j}: H_{{\rm FG}}(\Phi) \otimes ...^{j \ {\rm
times}} \otimes H_{{\rm FG}}(\Phi) \rightarrow \mathbb{C}$ is a
symmetric function. Thanks to the Hopf algebra homomorphism $\nu$
which preserves the co-radical degree, for any word $u \in H_{{\rm
word}}$, we have
\begin{equation}
\psi_{r}(u) = \sum_{j=1}^{r_{u}} d_{j}^{u} L^{j}
\end{equation}
such that $d_{j}^{u}=c_{j}^{\nu^{-1}(u)}$.

Thanks to the graphon representations of Dyson--Schwinger equations
(Theorem \ref{feynman-graphon-4} and Theorem
\ref{feynman-graphon-5}), we want to lift the Feynman rules
character (\ref{feynman-rules-1}) onto the level of large Feynman
diagrams.

In \cite{kruger-kreimer-1}, it is shown that $\psi_{r}$ maps the
shuffle product $u_{1} \ominus_{\Theta} ... \ominus_{\Theta} u_{n}$
to the $L^{n}$-term in the log expansion such that as the result,
this process filtrates coefficients $X_{n}$ in the unique solution
of each Dyson--Schwinger equation. We lift this story onto the
level of Feynman graphons where we have the correspondences $Y_{m} \mapsto [W_{Y_{m}}]$ for each $m \ge 1$ and $X_{{\rm DSE}} \mapsto [W_{X_{{\rm DSE}}}]$. Now the renormalized character
$\tilde{\psi}_{r}:=\tilde{\phi}_{r} \circ \overline{\nu}^{-1}$ can help us to map
the formal expansions $\sum_{1}^{m} u_{i_{1}} \ominus_{\Theta} ...
\ominus_{\Theta} u_{i_{k}}$ of shuffle products of words
corresponding to the partial sums $Y_{m}$ of $X_{{\rm DSE}}$ to a
certain term in the expansion
\begin{equation}  \label{feynman-rules-6}
\tilde{\phi}_{r}([W_{Y_{m}}]) = \sum_{j=1}^{r_{Y_{m}}}
c_{j}^{Y_{m}}(\theta)L^{j}
\end{equation}
such that
\begin{equation} \label{feynman-rules-7}
c_{j}^{Y_{m}} = c_{1}^{\otimes j} \tilde{\Delta}_{{\rm
graphon}}^{j-1}([W_{Y_{m}}]).
\end{equation}
When $m$ tends to infinity the sequence $\{c_{j}^{Y_{m}}\}_{m \ge
1}$ of coefficients converges to $c_{j}^{X_{{\rm DSE}}}$ (for each
$j$) with respect to the cut-distance topology. In addition, Feynman
rules characters are linear homomorphisms which means that when $m$
tends to infinity, the sequence $\{\tilde{\phi}_{r}([W_{Y_{m}}])\}_{m \ge
1}$ is cut-distance convergent to $\tilde{\phi}_{r}([W_{X_{{\rm DSE}}}])$.

Therefore for the infinite graph $X_{{\rm DSE}}$, we can obtain the
following formal expansion as the result of the application of the
renormalized Feynman rules character $\tilde{\phi}_{r}$.
\begin{equation} \label{feynman-rules-3}
\tilde{\phi}_{r}([W_{X_{{\rm DSE}}}]) = \sum_{j=1}^{r_{X_{{\rm DSE}}}}
c_{j}^{X_{{\rm DSE}}}(\theta)L^{j}
\end{equation}
such that
\begin{equation}
c_{j}^{X_{{\rm DSE}}} = c_{1}^{\otimes j} \widetilde{\Delta}_{{\rm
graphon}}^{j-1}([W_{X_{{\rm DSE}}}]).
\end{equation}

Suppose $\mathcal{S}^{\Phi}_{{\rm graphon}, (i)}$ is the vector
space generated by those Feynman graphons derived from solutions of
Dyson--Schwinger equations. The considered Feynman graphons are filtered in
terms of the canonical filtration on their corresponding words.
Namely, the filtration $(i)$ can be defined by applying
$\overline{\nu}$ and $\tilde{\psi}_{r}$ while the associated words
map to a similar term $i$ in the log-expansion
(\ref{feynman-rules-3}). Set
\begin{equation} \label{filter-1}
\mathcal{S}^{\Phi}_{{\rm graphon}, (0)} \preceq
\mathcal{S}^{\Phi}_{{\rm graphon}, (1)} \preceq ... \preceq
\mathcal{S}^{\Phi}_{{\rm graphon}, (i)} \preceq ... \preceq
\mathcal{S}^{\Phi,g}_{{\rm graphon}}
\end{equation}
as the resulting filtration on all Feynman graphons which contribute
to solutions of Dyson--Schwinger equations such that
$\mathcal{S}^{\Phi,g}_{{\rm graphon}} := \bigcup_{i \ge 0}
\mathcal{S}^{\Phi}_{{\rm graphon}, (i)} \subset
\mathcal{S}^{\Phi}_{{\rm graphon}}$. It defines the graded vector space $\mathcal{G}^{\Phi}$ given by
$$\mathcal{G}^{\Phi}_{[0]}=\mathcal{S}^{\Phi}_{{\rm
graphon}, (0)}$$
\begin{equation}
\mathcal{G}^{\Phi}_{[i]}:= \mathcal{S}^{\Phi}_{{\rm graphon}, (i)} /
\mathcal{S}^{\Phi}_{{\rm graphon}, (i-1)}, \ \ \ \forall i
>0
\end{equation}
where we have
\begin{equation}
\mathcal{G}^{\Phi}=\bigoplus_{i \ge 0} \mathcal{G}^{\Phi}_{[i]}.
\end{equation}
We can show that $\mathcal{G}^{\Phi}$ and
$\mathcal{S}^{\Phi,g}_{{\rm graphon}}$ are isomorphic as vector
spaces.
\end{proof}

Theory of words and quasi-shuffle products were studied in
\cite{hoffman-2} where the existence of Hopf algebra structures on
words have been addressed. The applications of shuffle type of
products to Hopf algebraic renormalization have been addressed in
different settings \cite{ebrahimifard-guo-1, kreimer-11,
shojaeifard-1, shojaeifard-6}. There is also another alternative
machinery (\cite{shojaeifard-1}) to lift Dyson--Schwinger equations
in $\mathcal{S}^{\Phi,g}$ onto their corresponding equations in the
Hopf algebra of words. According to this approach, we apply the
rooted tree representation of the Connes--Marcolli shuffle type
renormalization Hopf algebra $H_{\mathbb{U}}$ and then embed
$H_{\mathbb{U}}$ into an adapted version of the Hopf algebra
$H_{{\rm CK}}$ decorated by a particular Hall set. In this setting,
we can address the question about the existence of another
filtration on Dyson--Schwinger equations originated from Hopf
algebra of words.

Thanks to the explained Hopf algebraic formalism we are ready to
formulate a generalization of the BPHZ renormalization machinery for
non-perturbative QFT in the context of Feynman graphons which will
be discussed in the next part.

\section{\textsl{A generalization of the BPHZ renormalization machinery for large Feynman diagrams via Feynman graphons}}

In gauge field theories with strong couplings such as QCD, the
size of the coupling constant even at rather large values of the
exchanged momentum is in the range that the convergence of the perturbative
expansion is slow. Although in higher energy levels, the theory
enjoys the asymptotic freedom property, several orders of
perturbation theory should be concerned to provide a greater
accuracy where we need to deal with the evaluation of a large class
of higher order Feynman diagrams. We can address the corrections to
the quark self-energy as a complicated example in this setting. The
situation goes stranger when we deal with QCD in relatively lower
energy levels where non-perturbative aspects do appear. This is the
level that we need to build a powerful theoretical model for the
study of interactions of elementary particles. Thanks to the Hopf
algebra structure $\mathcal{S}^{\Phi}_{{\rm graphon}}$, which
encodes Dyson--Schwinger equations of a given physical theory
$\Phi$, in this part we plan to build a topological Hopf algebraic
renormalization program for large Feynman diagrams in the context of
the Riemann--Hilbert problem. We describe the construction of the
Connes--Kreimer renormalization group at the level of Feynman
graphons.

The renormalization Hopf algebra $H_{{\rm FG}}(\Phi)$ of Feynman
diagrams of the physical theory $\Phi$ encodes enough mathematical
tools to explain the step by step removal of sub-divergencies. The
graded dual of this Hopf algebra identifies an infinite dimensional
complex pro-unipotent Lie group denoted by
$\mathbb{G}_{\Phi}(\mathbb{C})$. Feynman rules, which allow us to
replace Feynman diagrams with their corresponding Feynman integrals,
are encoded by some elements of $\mathbb{G}_{\Phi}(\mathbb{C})$.
This Lie group has been applied by Connes and Kreimer to describe
perturbative renormalization as a special instance of a general
mathematical procedure of multiplicative extraction of finite values
in the context of the Riemann--Hilbert problem. According to this
approach, the BPHZ perturbative renormalization can be described
as the existence of the Birkhoff factorization for loops such as
$\gamma_{\mu}$ defined on an analytic curve $C \subset
\mathbb{C}P^{1}$ (as the domain) which has values in
$\mathbb{G}_{\Phi}(\mathbb{C})$. It is shown that
\begin{equation}
\gamma_{\mu}(z)=\gamma_{-}(z)^{-1}\gamma_{\mu,+}(z)
\end{equation}
such that $\gamma_{\mu,+}(z)$ is the boundary value of a holomorphic
map from the inner domain of $C$ to the group
$\mathbb{G}_{\Phi}(\mathbb{C})$ and $\gamma_{-}(z)$ is the boundary
value of the outer domain of $C$ to the group
$\mathbb{G}_{\Phi}(\mathbb{C})$. In addition, $\gamma_{-}$ is
normalized by $\gamma_{-}(\infty)=1$. The renormalized theory is the
evaluation of the holomorphic part $\gamma_{\mu,+}$ of
$\gamma_{\mu}$ as a product of two holomorphic maps $\gamma_{\pm}$
from the connected components $C_{\pm}$ of the complement of the
circle $C$ in the Riemann sphere $\mathbb{C}P^{1}$. For Dimensional
Regularization, we are interested in an infinitesimal disk around
$z=0$ while the curve $C$ is the boundary of this disk. Then we
have $0 \in C_{+}$ and $\infty \in C_{-}$. \cite{connes-kreimer-2,
connes-kreimer-3}

Each regularized Feynman integral $U^{z}(\Gamma(p_{1},...,p_{N}))$
defines a loop $\gamma_{\mu}(z)$ which allows us to lift the
analytic formulation of the Birkhoff factorization onto the level of
affine group schemes. Set
\begin{equation}
K=\mathbb{C}\{z\}[z^{-1}], \ \ O_{1}=\mathbb{C}\{z\}, \ \
O_{2}=\mathbb{C}[z^{-1}].
\end{equation}
It is shown that each character $\phi \in \mathbb{G}_{\Phi}(K)$ has
a unique Birkhoff factorization $\phi = (\phi_{-} \circ S) *
\phi_{+}$ such that $\phi_{+} \in \mathbb{G}_{\Phi}(O_{1})$,
$\phi_{-} \in \mathbb{G}_{\Phi}(O_{2})$ and $\varepsilon_{-} \circ
\phi_{-} = \varepsilon$. The BPHZ renormalization procedure has been
encapsulated in terms of the deformed version of the antipode with respect to the Minimal Subtraction scheme. We have
\begin{equation}
\Gamma \longmapsto S_{R_{ms}}^{\phi} * \phi (\Gamma)
\end{equation}
such that
\begin{equation} \label{bphz-1}
S_{R_{ms}}^{\phi}(\Gamma) = - R_{ms}(\phi(\Gamma)) - R_{ms}
(\sum_{\gamma \subset \Gamma} S_{R_{ms}}^{\phi}(\gamma) \phi(\Gamma
/ \gamma)).
\end{equation}
Therefore
\begin{equation} \label{bphz-2}
S_{R_{ms}}^{\phi} * \phi (\Gamma) = \overline{R}(\Gamma) +
S_{R_{ms}}^{\phi}(\Gamma)
\end{equation}
such that
\begin{equation} \label{bphz-3}
\overline{R}(\Gamma) = U_{\mu}^{z}(\Gamma) + \sum_{\gamma \subset
\Gamma} c(\gamma) U_{\mu}^{z}(\Gamma / \gamma) = \phi(\Gamma) +
\sum_{\gamma \subset \Gamma} S_{R_{ms}}^{\phi}(\gamma) \phi(\Gamma /
\gamma)
\end{equation}
is the Bogoliubov's operation. For any given Feynman integral
$U_{\mu}(\Gamma)$, the mathematical term $S_{R_{ms}}^{\phi}(\Gamma)$
generates the counterterm and the mathematical term
$S_{R_{ms}}^{\phi}
* \phi (\Gamma)$ generates the corresponding renormalized value.
\cite{connes-kreimer-2, connes-kreimer-3, ebrahimifard-kreimer-1,
figueroa-graciabondia-2}

Each Dyson--Schwinger equation determines a commutative graded Hopf
subalgebra $H_{{\rm DSE}}$ of $H_{{\rm FG}}(\Phi)$ where the
morphism (\ref{dse-lie-group-1}) allows us to have a projection of the complex Lie group
$\mathbb{G}_{{\rm DSE}}(\mathbb{C})$ inside
$\mathbb{G}_{\Phi}(\mathbb{C})$. The Lie group $\mathbb{G}_{{\rm DSE}}(\mathbb{C})$ has been applied to
formulate the foundations of a new differential Galois theory for the computation of counterterms which contribute to the renormalization of formal expansions in the solutions of Dyson--Schwinger equations. Work on this Lie group has already provided global $\beta$-functions with respect to solutions of Dyson--Schwinger equations to relate renormalized values generated via different regularization schemes. The study of Dyson--Schwinger equations in the context of this class of Lie groups have also been lifted onto a categorical setting where the renormalization of any given equation DSE can be encoded by objects of its corresponding neutral Tannakian category ${\rm Rep}_{\mathbb{G}_{{\rm DSE}}^{*}}$. In this setting, we can compute non-perturbative parameters derived from the equation DSE in terms of a class of equi-singular differential equations organized as a subcategory of the universal Connes--Marcolli category $\mathcal{E}^{{\rm CM}}$ of flat equi-singular vector bundles. \cite{shojaeifard-1, shojaeifard-3, shojaeifard-4, shojaeifard-5}

Thanks to Feynman graphon models, it is shown in \cite{shojaeifard-10} that the unique solution of any given Dyson--Schwinger equation in $\mathcal{S}^{\Phi,g}$ is the cut-distance convergent limit of the sequence of its corresponding partial sums. It allows us to embed $\mathcal{S}^{\Phi,g}$ as a topological sub-vector space in $\mathcal{S}^{\Phi}_{{\rm graphon}}$. In this part we plan to provide a new interpretation of the
renormalization of (large) Feynman diagrams in the context of the
Hopf algebra $\mathcal{S}^{\Phi}_{{\rm graphon}}$ of Feynman graphons. We determine a new class of singular differential equations which contribute in the computation of non-perturbative counterterms corresponding to renormalized Dyson--Schwinger equations.

\begin{thm}\label{bphz-dse-graphon-2}
The Hopf--Birkhoff factorization process provides a renormalization
program for each large Feynman diagram in $\mathcal{S}^{\Phi,g}$.
\end{thm}

\begin{proof}
We build a renormalization program for Feynman graphons in the Hopf algebra
$\mathcal{S}^{\Phi}_{{\rm graphon}}$ and then we pull back the results to the level of Feynman diagrams and solutions of Dyson--Schwinger equations.

Thanks to Milnor--Moore Theorem (\cite{milnor-moore-1}), the
commutative graded Hopf algebra $\mathcal{S}^{\Phi}_{{\rm graphon}}$
(given by Theorem \ref{feynman-graphon-5}) determines the complex infinite
dimensional pro-unipotent Lie group $\mathbb{G}^{\Phi}_{{\rm
graphon}}(\mathbb{C})$. Choose Dimensional Regularization and
Minimal Subtraction as the renormalization program where the
commutative algebra $A_{{\rm dr}}$ of Laurent series with finite
pole parts encodes the regularization scheme and the linear map
$R_{{\rm ms}}$, which projects series onto their pole parts, encodes
the renormalization scheme. Set ${\rm Loop}(\mathbb{G}^{\Phi}_{{\rm
graphon}}(\mathbb{C}), \mu)$ as the space of loops $\gamma_{\mu}$ on
the infinitesimal punctured disk ${\bf \Delta}^{*}$, around the
origin in the complex plane, with values in $\mathbb{G}^{\Phi}_{{\rm
graphon}}(\mathbb{C})$. These loops can represent unrenormalized
regularized Feynman rules characters in $\mathbb{G}^{\Phi}_{{\rm
graphon}}(\mathbb{C})$ which act on Feynman graphons. The
Rota--Baxter property of $(A_{{\rm dr}},R_{{\rm ms}})$ supports the
existence of a unique Birkhoff factorization
$(\gamma_{-},\gamma_{+})$ which can be lifted onto the level of
Feynman rules characters of Feynman graphons to achieve the factorization
$(\tilde{\phi}_{-},\tilde{\phi}_{+})$ for Feynman rules character
$\tilde{\phi}$. We have
\begin{equation}
\tilde{\phi}^{z}([W_{\Gamma}]):= \phi^{z}(\Gamma)
\end{equation}
as the modified version of the regularized Feynman rules character
$\phi^{z}$ which acts on Feynman graphons.

For a given Feynman graphon $[W_{X}]$ corresponding to the unique
solution of an equation DSE, set $[W_{Y_{m}}]$ (for each $m$) as the
unlabeled graphon classes with respect to the partial sums $Y_{m}$
of the infinite graph $X$. Since formal expansions $Y_{m}$s and their corresponding rooted tree representations are sparse graphs, then we need to apply rescaling methods to obtain non-zero Feynman graphon model for the solution $X$. Details about this subject has been discussed in \cite{shojaeifard-10,ihes-submitted}.
Now if we apply the renormalization coproduct formulas (\ref{cop-graphon-1}), (\ref{cop-graphon-2}) on Feynman graphons, the renormalization
antipode formulas (\ref{antipode-graphon-2}),
(\ref{antipode-graphon-3}) on Feynman graphons and also the Hopf algebraic BPHZ process
given by (\ref{bphz-1}), (\ref{bphz-2}), (\ref{bphz-3}), then we can
build the sequence $\{S_{R_{{\rm
ms}}}^{\tilde{\phi}}([W_{Y_{m}}])\}_{m \ge 1}$ of Feynman graphons
which is convergent with respect to the cut-distance topology. We
have
$$S_{R_{{\rm ms}}}^{\tilde{\phi}}([W_{X}]) = {\rm lim}_{m \rightarrow \infty} S_{R_{{\rm ms}}}^{\tilde{\phi}}([W_{Y_{m}}])
= {\rm lim}_{m \rightarrow \infty} \sum_{i=1}^{m} S_{R_{{\rm
ms}}}^{\tilde{\phi}}([W_{X_{i}}])$$
\begin{equation}
= {\rm lim}_{m \rightarrow \infty} \sum_{i=1}^{m} -R_{{\rm ms}}
(\tilde{\phi}([W_{X_{i}}])) - R_{{\rm ms}} (\sum S_{R_{{\rm
ms}}}^{\tilde{\phi}}([W_{\gamma}])
\tilde{\phi}([W_{X_{i}/\gamma}])).
\end{equation}
The functional $S_{R_{{\rm ms}}}^{\tilde{\phi}}$ is the negative
component of the Birkhoff factorization of $\tilde{\phi}$. The term $S_{R_{{\rm ms}}}^{\tilde{\phi}}([W_{X}])$
addresses the counterterm with respect to the Feynman
graphon $[W_{X}]$. We can also build the sequence $\{S_{R_{{\rm
ms}}}^{\tilde{\phi}}*\tilde{\phi}([W_{Y_{m}}])\}_{m \ge 1}$ of
Feynman graphons which is convergent with respect to the
cut-distance topology. We have
\begin{equation}
S_{R_{{\rm ms}}}^{\tilde{\phi}}*\tilde{\phi}([W_{X}]) = {\rm lim}_{m
\rightarrow \infty} S_{R_{{\rm
ms}}}^{\tilde{\phi}}*\tilde{\phi}([W_{Y_{m}}]) = {\rm lim}_{m
\rightarrow \infty} \sum_{i=1}^{m} S_{R_{{\rm
ms}}}^{\tilde{\phi}}*\tilde{\phi}([W_{X_{i}}])
\end{equation}
such that the convolution product $*$ is defined with respect to the
coproduct $\Delta_{{\rm graphon}}$. The functional
$S_{R_{{\rm ms}}}^{\tilde{\phi}}*\tilde{\phi}$ is the positive
component of the Birkhoff factorization of $\tilde{\phi}$. The term $S_{R_{{\rm
ms}}}^{\tilde{\phi}}*\tilde{\phi}([W_{X}])$ addresses
the renormalized value with respect to the Feynman graphon
$[W_{X}]$.
\end{proof}

Thanks to the filtration parameter on Feynman graphons on the basis
of the Hopf algebra of words given by Theorem
\ref{graphon-filtration-1}, it is possible to define a new one-parameter group
$\{\theta_{t}\}_{t}$ of automorphisms on $\mathbb{G}^{\Phi}_{{\rm
graphon}}(\mathbb{C})$ with the infinitesimal generator
\begin{equation}
\frac{d}{dt}|_{t=0}\theta_{t} = Y
\end{equation}
such that $Y$ sends each Feynman graphon $[W_{\Gamma}]$ to its
corresponding  filtration rank $n_{[W_{\Gamma}]}$. In other words,
for each character $\tilde{\phi}$,
\begin{equation}
<\theta_{t}(\tilde{\phi}),[W_{\Gamma}]>:=<\tilde{\phi},\theta_{t}([W_{\Gamma}])>.
\end{equation}

\begin{lem} \label{renorm-group-1}
Suppose the loop $\gamma_{\mu} \in {\rm Loop}(\mathbb{G}^{\Phi}_{{\rm
graphon}}(\mathbb{C}), \mu)$ encodes the regularized unrenormalized
Feynman rules character $\tilde{\phi}$ on Feynman graphons. Then we have
$$\gamma_{e^{t}\mu}(z) = \theta_{tz}(\gamma_{\mu}(z)).$$
In addition, the limit
$$F_{t} = {\rm lim}_{z \rightarrow 0} \gamma_{-}(z) \theta_{tz}(\gamma_{-}(z)^{-1})$$
defines a new one-parameter subgroup of $\mathbb{G}^{\Phi}_{{\rm
graphon}}(\mathbb{C})$ such that for each $t \in \mathbb{R}$,
$$\gamma_{e^{t}\mu^{+}}(0) = F_{t} \gamma_{\mu^{+}}(0).$$
\end{lem}

\begin{proof}
Thanks to the construction of the renormalization Hopf algebra of
Feynman graphons (given by Theorem \ref{feynman-graphon-5}) and Proposition
1.47 in \cite{connes-marcolli-1}, we have the proof.
\end{proof}

More details about non-perturbative parameters generated by Feynman graphon representations of solutions of Dyson--Schwinger equations have been discussed in \cite{shojaeifard-10,shojaeifard-18,ihes-submitted}.

It is possible to lift the Connes--Marcolli geometric approach onto
the level of the renormalization Hopf algebra of Feynman graphons.
For this purpose we need to adapt the regularization bundle and then
classify equi-singular flat connections, which encode
counterterms, in terms of the renormalization of Feynman graphons.

\begin{prop} \label{graphon-geimetric-1}
There exists a bijective correspondence (independent of the choice
of a local regular section $\sigma: {\bf \Delta} \rightarrow B$)
between equivalence classes of flat equi-singular $\mathbb{G}^{\Phi}_{{\rm
graphon}}(\mathbb{C})$-connections on the regularization bundle and
elements of the Lie algebra $\mathfrak{g}^{\Phi}_{{\rm
graphon}}(\mathbb{C})$.
\end{prop}

\begin{proof}
The regularization parameter in Dimensional Regularization can be
encoded by the punctured version of an infinitesimal disk ${\bf
\Delta}$ around the origin $z=0$. Set $P_{{\rm graphon}}:= ({\bf
\Delta} \times \mathbb{C}^{*}) \times \mathbb{G}^{\Phi}_{{\rm
graphon}}(\mathbb{C})$ as the trivial principal bundle over the base
space ${\bf \Delta} \times \mathbb{C}^{*}$. Remove the fiber over
$z=0$ to obtain the bundle $P^{0}_{{\rm graphon}}= {\bf \Delta}
\times \mathbb{C}^{*} - \pi^{-1}(\{0\}) \times \mathbb{G}^{\Phi}_{{\rm
graphon}}(\mathbb{C})$ as the regularization bundle modified with
respect to the renormalization Hopf algebra of Feynman graphons.

A flat $\mathbb{G}^{\Phi}_{{\rm graphon}}(\mathbb{C})$-connection $\varpi$
on $P^{0}_{{\rm graphon}}$ is a $\mathfrak{g}^{\Phi}_{{\rm
graphon}}(\mathbb{C})$-valued one form such that $\mathfrak{g}^{\Phi}_{{\rm
graphon}}(\mathbb{C})$ is the Lie algebra of the Lie group
$\mathbb{G}^{\Phi}_{{\rm graphon}}(\mathbb{C})$ which contains all linear
maps $l: \mathcal{S}^{\Phi}_{{\rm graphon}} \rightarrow
\mathbb{C}$ with the property
\begin{equation}
l([W_{\Gamma_{1}}][W_{\Gamma_{2}}]) = l([W_{\Gamma_{1}}])
\tilde{\varepsilon}([W_{\Gamma_{2}}]) +
\tilde{\varepsilon}([W_{\Gamma_{1}}])l([W_{\Gamma_{2}}]).
\end{equation}
The Lie bracket is given by the formula
\begin{equation}
[l_{1},l_{2}]([W_{\Gamma}]) = <l_{1} \otimes l_{2} - l_{2} \otimes
l_{1}, \Delta_{{\rm graphon}}([W_{\Gamma}])>.
\end{equation}
The one-parameter group $\{\theta_{t}\}_{t \in \mathbb{C}}$ of
automorphisms of $\mathbb{G}^{\Phi}_{{\rm graphon}}(\mathbb{C})$ can be
lifted onto the level of this Lie algebra.

A flat $\mathbb{G}^{\Phi}_{{\rm graphon}}(\mathbb{C})$-connection $\varpi$
on $P^{0}_{{\rm graphon}}$ is called equi-singular if it is
$\mathbb{G}_{m}$-invariant and for any solution $f$ of the
differential equation ${\bf D}f = \varpi$ with respect to the
logarithmic derivative, the restrictions of $f$ to sections
$\sigma:{\bf \Delta} \rightarrow {\bf \Delta} \times
\mathbb{C}^{*}$ have the same type of singularity.

Thanks to \cite{connes-marcolli-1}, the negative component of the
Birkhoff factorization of each $\gamma_{\mu} \in {\rm
Loop}(\mathbb{G}^{\Phi}_{{\rm graphon}}(\mathbb{C}), \mu)$ determines a
unique element $\beta$ in $\mathfrak{g}^{\Phi}_{{\rm graphon}}(\mathbb{C})$
where we have
\begin{equation} \label{geometric-counterterms-1}
\gamma_{-}(z) = T{\rm exp}(\frac{-1}{z} \int_{0}^{\infty}
\theta_{-t}(\beta)dt).
\end{equation}

We can show that for each $\varpi \in \mathfrak{g}^{\Phi}_{{\rm
graphon}}(K)$ with the trivial monodromy, there exists a solution
$\tilde{\psi} \in \mathbb{G}^{\Phi}_{{\rm graphon}}(K)$ for the
differential equation ${\bf D}\tilde{\psi}=\varpi$.

Two connections $\varpi_{1}, \varpi_{2}$ with the trivial monodromy
are called equivalent if they are gauge conjugate by an element
regular at $z=0$. It leads us to show that for equivalent
connections $\varpi_{1}, \varpi_{2}$, the solutions of the
differential equations ${\bf D}\tilde{\psi}^{\varpi_{1}} =
\varpi_{1}$ and ${\bf D}\tilde{\psi}^{\varpi_{2}} = \varpi_{2}$ have
the same negative components of the Birkhoff factorization, namely,
\begin{equation} \label{geometric-counterterms-2}
\tilde{\psi}^{\varpi_{1}}_{-} = \tilde{\psi}^{\varpi_{2}}_{-}.
\end{equation}

Thanks to (\ref{geometric-counterterms-1}),
(\ref{geometric-counterterms-2}) and the Connes--Marcolli
Classification Theorem (given by Theorem 1.67 in \cite{connes-marcolli-1}),
each element $\beta \in \mathfrak{g}^{\Phi}_{{\rm graphon}}(\mathbb{C})$
determines a unique class $\varpi$ of flat equi-singular
$\mathbb{G}^{\Phi}_{{\rm graphon}}(\mathbb{C})$-connections on $P^{0}_{{\rm
graphon}}$ in terms of a differential equation with the general form
${\bf D} \gamma_{\mu} = \varpi$ such that
\begin{equation}
\gamma_{\mu}(z,v) = T{\rm exp}(\frac{-1}{z} \int_{0}^{v}
u^{Y}(\beta)\frac{du}{u})
\end{equation}
where $u=tv, \ t\in[0,1]$ and $u^{Y}$ is the action of
$\mathbb{G}_{m}$ on $\mathbb{G}^{\Phi}_{{\rm graphon}}(\mathbb{C})$.
\end{proof}

\begin{prop}
The category $\mathcal{E}^{{\rm CM}}$ encodes the renormalization
group corresponding to the BPHZ renormalization of large Feynman
diagrams.
\end{prop}

\begin{proof}
Thanks to Proposition \ref{graphon-geimetric-1}, for the Hopf
algebra $\mathcal{S}^{\Phi}_{{\rm graphon}}$, we can determine a
family of flat equi-singular $\mathbb{G}^{\Phi}_{{\rm
graphon}}(\mathbb{C})$-connections which encode counterterms on the
basis of the $\beta$-functions. Thanks to \cite{connes-marcolli-1},
these geometric objects form a new category $\mathcal{E}^{\Phi}_{{\rm
graphon}}$ which is recovered by the category ${\rm
Rep}_{\mathbb{G}^{\Phi, *}_{{\rm graphon}}}$ of finite dimensional
representations of the affine group scheme $\mathbb{G}^{\Phi, *}_{{\rm
graphon}}$. In addition, the renormalization Hopf algebra
$H_{{\rm FG}}(\Phi)$ of Feynman diagrams of the physical theory
$\Phi$ determines the category $\mathcal{E}^{\Phi}$ of geometric objects recovered by the
category ${\rm Rep}_{\mathbb{G}^{*}_{\Phi}}$ of finite dimensional
representations of the affine group scheme $\mathbb{G}^{*}_{\Phi}$.
Thanks to the explained categorical formalism in
\cite{connes-marcolli-1}, we can embed ${\rm
Rep}_{\mathbb{G}^{*}_{\Phi}}$ as a sub-category into
$\mathcal{E}^{{\rm CM}}$. It is shown that $\mathcal{E}^{{\rm CM}}$ is isomorphic to the
category ${\rm Rep}_{\mathbb{U}^{*}}$ such that the complex Lie
group $\mathbb{U}(\mathbb{C})$ can be described in terms of the Lie
algebra $L_{\mathbb{U}}$ generated by elements $e_{-n}$ of degree
$-n$ for each $n>0$ such that the sum $e=\sum e_{-n}$ is an element
of this Lie algebra. We can lift $e$ onto the morphism ${\rm
rg}:\mathbb{G}_{a} \rightarrow \mathbb{U}$. The
universality of $\mathcal{E}^{{\rm CM}}$ supports the existence of a new class of graded
representations such as
\begin{equation}
\vartheta:\mathbb{U}(\mathbb{C}) \rightarrow
\mathbb{G}^{\Phi}_{{\rm graphon}}(\mathbb{C}).
\end{equation}
Now the composition $\vartheta
\circ {\rm rg}$ determines the renormalization group
$\{F_{t}\}_{t\in \mathbb{C}}$ at the level of Feynman graphons (i.e.
Lemma \ref{renorm-group-1}).
\end{proof}

Lemma \ref{feynman-graphon-1} and Theorem \ref{feynman-graphon-5}
enable us to embed $H_{{\rm FG}}(\Phi)$ into
$\mathcal{S}^{\Phi}_{{\rm graphon}}$ which leads us to define an
epimorphism of affine group schemes from $\mathbb{G}^{\Phi, *}_{{\rm
graphon}}$ to $\mathbb{G}^{*}_{\Phi}$. In addition, the renormalization Hopf
algebra of Feynman graphons includes solutions of all
Dyson--Schwinger equations in the physical theory $\Phi$. Therefore the category $\mathcal{E}^{\Phi}_{{\rm graphon}}$ is rich enough to recover the category $\mathcal{E}^{\Phi}$ and also, non-perturbative information of the physical theory.

As the summary, we have shown that the renormalization topological
Hopf algebra of Feynman graphons is capable to encode the
renormalization of Feynman diagrams and solutions of
Dyson--Schwinger equations. We have also embedded this
graphon model of renormalization into the universal Connes--Marcolli
categorical setting where as the result, we can study Feynman graphons under the differential Galois theory. In final, these achievements suggest the existence of a new unexplored
interconnection between motivic renormalization and
Dyson--Schwinger equations in the context of the theory of graphons.


\chapter{\textsf{Non-perturbative computational complexity}}

\vspace{1in}

$\bullet$ \textbf{\emph{A parametric representation for large Feynman diagrams}} \\
$\bullet$ \textbf{\emph{The optimization of non-perturbative complexity via a multi-scale Renormalization Group}} \\
$-$ \textbf{\emph{A Renormalization Group program on $\mathcal{S}^{\Phi,g}$}}\\
$-$ \textbf{\emph{Kolmogorov complexity of Dyson--Schwinger equations}}

\newpage

The original motivation for the introduction of Feynman graphons is to clarify and study infinities originated from
non-perturbative aspects in Quantum Field Theory with strong
couplings via a new topological Hopf algebraic renormalization program. From a physicist's perspective, these infinities do not
acceptable and applying some approximation methods are useful for the production of some intermediate values such as running couplings, N large
methods, etc. Then the Physics of elementary particles and its
phenomenology shall be described in terms of those approximations.
From a mathematician's perspective, we have a different story where
it is possible to deal with infinities under different settings
instead of only removing them. The Cartier's cosmic Galois group as
a universal group of symmetries is useful for the analysis of divergencies in perturbative gauge field theories to generate some new data about the behavior of these physical theories in higher order perturbation terms. The motivic
Galois group associated to the renormalization Hopf algebra
$\mathcal{S}^{\Phi}_{{\rm graphon}}$ of Feynman graphons is useful to analyze
divergencies originated from solutions of Dyson--Schwinger equations in the language of singular differential equations on cut-distance topological regions of Feynman diagrams. This motivic Galois group could be a practical candidate to encode non-perturbative gauge field theories. In this chapter
we plan to show some new applications of Feynman graphon models in dealing with the non-perturbative computations of physical parameters derived from Dyson--Schwinger equations. Here we focus on finding a new parametric representation in the context of Tutte polynomials and other combinatorial polynomials for the description of solutions of Dyson--Schwinger equations. This new combinatorial setting allows us to understand the graph complexity of large Feynman diagrams in terms of the graph complexities of partial sums. Then we explain the elementary foundations of the concept of Kolmogorov complexity on the space of all Dyson--Schwinger equations of a given strongly coupled gauge field theory. We define this complexity in terms of a new multi-scale non-commutative non-perturbative Renormalization Group on $\mathcal{S}^{\Phi,g}$ which governs the changing the scales of the bare coupling constant $g$ and running coupling constants. This platform is useful to provide a new algorithm for the study of Dyson--Schwinger equations under strong couplings in terms of sequences of Dyson--Schwinger equations with lesser complexities. We try to show that this Renormalization
Group machinery can optimize the complexity of non-perturbative
computations. We show how this notion of complexity is related to the Halting problem in the Theory of Computation.
This study suggests a new contextualization for the description of non-perturbative situations
and their complexity.

\section{\textsl{A parametric representation for large Feynman diagrams}}

The original task in Quantum Field Theory is to compute correlation functions
(i.e. Green's functions) in a (non-)perturbative expansion setting whose terms
are decorated by Feynman diagrams. Each term in this class of
expansions consists of a multiple ill-defined integral such that the
integrand is codified by the combinatorial information of its
corresponding Feynman diagram.  Generally speaking, we can work in
momentum space of $D$ dimensions such that a preliminary count of
the powers of the momenta in the integrands can report the possible
superficially divergence in the integral. In this situation, the
renormalization program associates a counterterm to each
superficially divergent subgraph to finally produce a finite result
by subtraction treatment. All superficially divergent Feynman subgraphs should be
considered under a recursive setting to assign a final finite value to the
full Feynman diagram. Studying Feynman diagrams via tree representations
enables us to formulate perturbative renormalization theory under a
simplified universal setting. Furthermore, it provides also a
combinatorial reformulation of Dyson--Schwinger equations where we
can study solutions of these non-perturbative type of equations in
the context of partial sums of decorated non-planar rooted trees.
\cite{bergbauer-kreimer-1, foissy-4, foissy-1, kreimer-4, kreimer-9,
tanasa-kreimer-1, yeats-1}

In the previous sections we have shown that the unique solution of
each Dyson--Schwinger equation is described as the convergent
limit of a sequence of Feynman graphons with respect to the
cut-distance topology. This formalism has been applied to lift the
BPHZ renormalization program onto the level of large Feynman graphs
to generate some new expressions for the description of counterterms
and renormalized values associated to fixed point equations of
Green's functions. Using graph polynomials for the study of Feynman
integrals has played an important role in the computational
processes where this class of combinatorial polynomials can
bring some powerful algorithms for the analysis of the behavior of
these divergent integrals (\cite{aluffi-marcolli-2,
bloch-esnault-kreimer-1, krajewski-moffatt-tanasa-1,
krajewski-rivasseau-tanasa-1, kreimer-yeats-2, marcolli-1,
nakanishi-1, tanasa-1, weinzierl-1}). In this section we show
another application of this graphon representation of
non-perturbative parameters where we deal with the concept of
parametric representation of large Feynman diagrams. We study
solutions of Dyson--Schwinger equations in the language of Tutte
polynomial and Kirchhoff--Symanzik polynomials.

The Tutte polynomial, as a two variables graph polynomial, enjoys a
universal property which enables us to evaluate any multiplicative
graph invariant with a deletion/contraction reduction machinery
\cite{tutte-4, tutte-1, tutte-2, tutte-3, welsh-1}. This fundamental
property provides the opportunity to demonstrate how graph
polynomials can be specialized or generalized. The Aluffi--Marcolli
approach has clarified the practical importance of Tutte polynomials
in dealing with Feynman rules characters and Feynman integrals under
an algebro-geometric setting where a motivic perspective on
perturbative renormalization program has been formulated very nicely.
\cite{aluffi-marcolli-1, aluffi-marcolli-2, marcolli-2, marcolli-1}

We first review the basic structure of Tutte polynomials
on finite graphs, its different reformulations and its universal property (\cite{marcolli-1, tutte-1, tutte-2, tutte-3, welsh-1}) and then we explain graph polynomials which can contribute to the combinatorial representations of Feynman graphons of large Feynman diagrams.

A given (finite) graph $G$ has a set $V(G)$ of vertices and a set $E(G)$ of edges. Graphs $G_{1}, G_{2}$ are isomorphic when there exists a bijection such as $\rho$ between the sets $V(G_{1})$ and $V(G_{2})$ such that for each edge $uv$ in $G_{1}$, $\rho(uv)$ is an edge in $G_{2}$ and vice versa. Subsets of the set of vertices or the set of edges can give us subgraphs. For any subset $A \subset E(G)$ of edges, the rank $r(A)$ and the nullity $n(A)$ are defined by the relations
\begin{equation}
r(A):= |V(G)| - \kappa(A), \ \ \ n(A):= |A| - r(A)
\end{equation}
such that $\kappa(A)$ is the number of connected components of the graph. In general, finite graphs can be classified in terms of their numbers of non-trivial connected components. A graph is called $n$-connected, if we should remove at least $n$ edges from the graph to obtain a disconnected graph. Rooted trees, as fundamental tools for us, are connected graphs which have no cycles or loops. They are non-trivial connected components of forests as more complicated graphs. Sometimes working on subgraphs of a given complicated (finite) graph enables us to clarify some fundamental properties of the original graph. Spanning subgraphs are applied as one important class of subgraphs for the construction of graph polynomials. A spanning subgraph covers all vertices of the original graph with the optimum number of edges.

The notion of "dual" in Graph Theory enables us to build the algebraic combinatorics of graphs.  If we can embed a graph into the plane without any crossing in edges, then the graph is called planar. Each planar graph can separate the plane into regions known as faces. Faces are key tools for the construction of the dual of a graph. For a given planar graph $G$, its corresponding connected dual graph is built by assigning a vertex to each face where there exist $m$ edges between two vertices in the dual graph if the corresponding faces of the original graph have $m$ edges in their boundaries. We denote $G^{*}$ as the dual of the connected planar graph $G$ and it can be seen that
\begin{equation}
(G^{*})^{*} = G.
\end{equation}

There are two fundamental commutative operations on graphs namely, deletion and contraction which enable us to build the algebraic combinatorics of graphs. For a given finite graph $G$, we can build a new graph $G \setminus e$ as the result of deleting an edge $e \in E(G)$. This new graph has the same set of vertices $V(G)$ and the set of edges $E(G) - \{e\}$. We can also build another new graph $G / e$ as the result of contracting an edge $e$ in terms of identifying the endpoints of the edge $e$ by shrinking this edge. It is easy to check that the deletion and the insertion on a self-loop edge determine the same resulting graph.

\begin{lem}
(i) For any given different edges $e_{1},e_{2}$ of a given planar graph $G$, the graph $(G \setminus e_{1}) / e_{2}$ is isomorphic to the graph $(G / e_{2}) \setminus e_{1}$.

(ii) A planar graph and its dual have the same numbers of spanning trees.

(iii) The rank of a dual graph is well-defined. \cite{tutte-1, welsh-1}
\end{lem}

The deletion or contraction of an edge determines a minor of a graph. In more general setting, if a graph $H$ is isomorphic to $G \setminus A_{1}/A_{2}$ for some choice of disjoint subsets $A_{1},A_{2}$ of $E(G)$, then it is called a minor graph. In this setting, a class of graphs is called minor closed if whenever the graph $G$ is in the class, then any minor of $G$ is also in the class.

Graph invariants are useful tools for the characterization of graphs in terms of some particular properties. A graph invariant is a function on the class of all graphs such that it has the same output on isomorphic graphs. Graph polynomials (such as Tutte polynomials) are indeed some graph invariants such that their images belong to some polynomial rings.

There are several different (but equivalent) (re)formulations for Tutte polynomials in terms of rank--nullity generating function method, linear recursion machinery and spanning tree expansion method which was originally applied by Tutte. The linear recursion form can be described as a collection of reduction rules to rewrite a graph as a weighted formal sum of graphs that are less complicated than the original graph. This method is useful to identify a collection of simplest or irreducible graphs. \cite{marcolli-1, tutte-1, tutte-2, tutte-3, welsh-1}

\begin{defn} \label{tutt-polynomial-1}
The Tutte polynomial $T(G;x,y)$ of a given (finite) graph $G$ is a two variables polynomial with respect to the independent variables $x,y$ which is defined in terms of the following recursive machinery: \\
- If $G$ has no edge, then $T(G;x,y)=1$; otherwise, for any edge $e
\in E(G)$, \\
- $T(G;x,y)=T(G \setminus e;x,y) + T(G / e; x,y)$, \\
- $T(G;x,y)=xT(G / e;x,y)$, if $e$ is a coloop, \\
- $T(G;x,y)=yT(G \setminus e; x,y)$, if $e$ is a loop. \\
\end{defn}

We can see that if $G$ has $i$ bridges and $j$ loops, then its corresponding Tutte polynomial is given by
\begin{equation}
T(G;x,y) = x^{i}y^{j}.
\end{equation}
In addition, Definition \ref{tutt-polynomial-1} shows us that the Tutte polynomial of the disjoint union of finite number of graphs can be defined in terms of the Tutte polynomials of each graph in the union. In other words,
\begin{equation}
T(G_{1} \sqcup G_{2} \sqcup ... \sqcup G_{n};x,y)=T(G_{1};x,y) ... T(G_{n};x,y).
\end{equation}

We can redefine Tutte polynomials in the language of the rank--nullity generating functions. They are (infinite) polynomials with coefficients which can count structures which are encoded by the exponents of variables. In this setting we have
\begin{equation}
T(G;x,y) = \sum_{A \subset E(G)} (x-1)^{r(E(G))-r(A)}(y-1)^{n(A)}.
\end{equation}

\begin{rem}
The Tutte polynomials of a planar graph $G$ and its dual graph $G^{*}$ can determine each other. This means that
\begin{equation}
T(G;x,y)=T(G^{*};y,x).
\end{equation}
\end{rem}

We can also redefine Tutte polynomials in terms of spanning trees. In this setting, we need to define a total order $\prec$ on the set of edges $E(G)=\{v_{1},...,v_{n}\}$ of a given graph $G$ such as
\begin{equation}
v_{i} \prec v_{j} \leftrightarrow i > j.
\end{equation}
For a given tree $t$, an edge $e$ is called internally active if $e$ is an edge of $t$ and it is the smallest edge in the cut defined by $e$. We can lift this concept onto the dual level where an edge $u$ is called externally active if $u \not\in t$ and it is the smallest edge in the cycle defined by $u$. Now the Tutte polynomial of the totally ordered graph $G$ can be defined (independent of the chosen total order) by the formal expansion
\begin{equation} \label{tutt-polynomial-2}
T(G;x,y)=\sum_{i,j} t_{ij}x^{i}y^{j}
\end{equation}
such that $t_{ij}$ counts spanning trees with internal activity $i$ and external activity $j$.
\cite{aluffi-marcolli-2, krajewski-moffatt-tanasa-1, marcolli-1,
tutte-4, tutte-1, welsh-1}

The most fundamental property of the Tutte polynomial is its universality under a graph invariant setting. This means that any multiplicative graph invariant on disjoint unions and one-point joins of graphs
which is formulated via a deletion/contraction reduction can be described as an evaluation of the Tutte polynomial. There are different notions for the generalization of the Tutte polynomials and here we address the one which is useful for us. \cite{aluffi-marcolli-2, marcolli-1}

\begin{defn} \label{general-tutte-poly-1}
Let $\mathfrak{G}$ be the set of isomorphism classes of finite graphs. A graph invariant $F$ from $\mathfrak{G}$ to a commutative ring such as the polynomial ring $\mathbb{C}[\alpha,\beta,\eta,x,y]$ is called Tutte--Grothendieck invariant of graphs, if it has the following properties: \\
- $F(G)= \eta^{\#\ V(G)}$ if the set of edges is empty, \\
- $F(G)= xF(G / e)$ if the edge $e \in E(G)$ is a bridge, \\
- $F(G)= yF(G \setminus e)$ if the edge $e \in E(G)$ is a looping edge, \\
- For any ordinary edge, which is not a bridge nor a looping edge,
\begin{equation}
F(G)=\alpha F(G / e) + \beta F(G \setminus e), \\
\end{equation}
- For every $G,H \in \mathfrak{G}$, if $G \cup H \in \mathfrak{G}$
or $G \bullet H \in \mathfrak{G}$, then $F(G \cup H) = F(G)F(H)$ and
$F(G \bullet H)=F(G)F(H)$ such that the one-point join $G\bullet H$ is defined by identifying a vertex of
$G$ and a vertex of $H$ into a new single vertex of $G\bullet H$.
\end{defn}

The induction machinery can show that
\begin{equation}
T(G\bullet H)=T(G)T(H),
\end{equation}
which means that the Tutte polynomial does not distinguish between the one-point join of two graphs and their disjoint union.

The Tutte polynomial is a special version of the Tutte--Grothendieck invariant which is independent of the choice of any ordering of edges of the graph. We can show that for any given map
$f:\mathfrak{G}\longrightarrow R$, if there exist $a,b \in R$ such that $f$ is another Tutte--Grothendieck invariant, then $f$ can be presented in terms of the Tutte polynomial such that we have
\begin{equation}
f(G) = a^{|E(G)| -
r(E(G))}b^{r(E(G))}T(G;\frac{x_{0}}{b},\frac{y_{0}}{a}).
\end{equation}

Here we want to apply Feynman graphon representations of Feynman diagrams in a given gauge field theory to build a new class of Tutte polynomials which contribute to the combinatorial presentations of solutions of Dyson--Schwinger equations. Our idea is to implement an efficient
algorithm for the computation of the Tutte polynomial $T(X_{{\rm DSE}};x,y)$ of the unique solution of a given Dyson--Schwinger equation DSE in terms of handling intermediate graphs (i.e. partial
sums) and their corresponding Tutte polynomials to avoid unnecessary
recomputations.

\begin{thm} \label{large-graph-parametric-1}
There exists a new class of Tutte polynomials with respect to large
Feynman diagrams.
\end{thm}

\begin{proof}
We work on the unique solution $X_{{\rm DSE}}=\sum_{n \ge 0} X_{n}$ of a given
Dyson--Schwinger equation DSE under the coupling constant $\lambda g=1$.
Thanks to Theorem \ref{feynman-graphon-4}, the sequence
$\{Y_{m}\}_{m \ge 1}$ of partial sums is convergent to $X_{{\rm
DSE}}$ with respect to the cut-distance topology.

Thanks to Definition
\ref{tutt-polynomial-1} and the formula (\ref{tutt-polynomial-2}), for each $m \ge
1$, the Tutte polynomial $T(Y_{m};x,y)$ of the finite disjoint union
graph $Y_{m}:= \mathbb{I} + X_{1} + ... + X_{m}$ is defined in terms of the Tutte polynomials of the components $X_{k}$s. We have
\begin{equation}
T(Y_{m};x,y) = \prod_{k=1}^{m}T(X_{k};x,y) = \prod_{k=1}^{m}
\sum_{i_{k},j_{k}} t_{i_{k}j_{k}}x^{i_{k}}y^{j_{k}}
\end{equation}
such that $t_{i_{k}j_{k}}$ is the number of spanning trees in
$X_{k}$ with internal activity $i_{k}$ and external activity
$j_{k}$.

We know that ${\rm lim}_{m \rightarrow \infty} Y_{m}=X_{{\rm
DSE}}$ with respect to the cut-distance topology. It means that for
each $\epsilon >0$, there exists $N_{\epsilon}$ such that for each
$m_{1}, m_{2} \ge N_{\epsilon}$, we have
\begin{equation}
d(Y_{m_{1}},Y_{m_{2}}) = d_{{\rm
cut}}([W_{Y_{m_{1}}}],[W_{Y_{m_{2}}}]) < \epsilon.
\end{equation}
Therefore
\begin{equation} \label{limit-cut-11}
d_{{\rm cut}}([W_{Y_{m_{1}}}],[W_{Y_{m_{2}}}]) = 0 \Leftrightarrow
[W_{Y_{m_{1}}}] \approx [W_{Y_{m_{2}}}].
\end{equation}
For each $m$, the Feynman graphon class $[W_{Y_{m}}]$ is determined in
terms of the rooted tree representations of Feynman diagrams
$X_{1},...,X_{m}$ where decorated rooted trees (or forests)
$t_{X_{1}},...,t_{X_{m}}$ are the only spanning trees (or forests) in themselves.
Thanks to the relation (\ref{limit-cut-11}), for enough large orders,
unlabeled graphon classes corresponding to partial sums are going to be weakly
isomorphic while they converge to the unique Feynman graphon class
$[W_{X_{{\rm DSE}}}]$. It means that spanning forests of partial
sums for enough large orders tend to the spanning forest $t_{X_{{\rm
DSE}}}$ of the unique graph limit $X_{{\rm DSE}}$.

In addition, the
Tutte polynomial for each arbitrary rooted tree $t$ is given by
\begin{equation} \label{tutte-tree-1}
T(t;x,y) = \sum_{s \in R(t)} x^{|E(s)|}(y+1)^{|E(s)|-|L(s)|}
\end{equation}
such that $R(t)$ is the set of all subtrees of $t$, $|E(s)|$ is the
number of edges of a subtree $s$ and $|L(s)|$ is the number of
leaves of a subtree $s$.

Now for the collection $\{\prod_{k=1}^{m} T(t_{X_{k}};x,y)\}_{m \ge 1}$
of Tutte polynomials, we can define a collection $\{p_{m}:
\prod_{k=1}^{\infty} T(t_{X_{k}};x,y) \longrightarrow
\prod_{k=1}^{m} T(t_{X_{k}};x,y)\}_{m \ge 1}$ of projections. Thanks
to the universal property of the Tutte polynomial, for any graph
invariant $T$ (which enjoys the properties in Definition
\ref{tutt-polynomial-1}) together with the collection $\{f_{m}: T
\longrightarrow \prod_{k=1}^{m} T(t_{X_{k}};x,y)\}_{m \ge 1}$, we
can define the unique map
\begin{equation}
\digamma: T \longrightarrow \prod_{k=1}^{\infty} T(t_{X_{k}};x,y)
\end{equation}
such that $f_{m}=p_{m} \circ \digamma$. As the consequence, we can
consider the direct product $\prod_{k=1}^{\infty} T(t_{X_{k}};x,y)$
as the Tutte polynomial for the infinite tree (or forest) $t_{X_{{\rm DSE}}}$.

If we replace rooted tree representations with the original Feynman
diagrams, then we can build the Tutte polynomial for the large
Feynman diagram $X_{{\rm DSE}}$ in terms of the direct product over
the Tutte polynomials for simpler finite graphs (i.e. partial sums)
$\{T(X_{k};x,y)\}_{k \ge 1}$. We have
\begin{equation}
T(X_{{\rm DSE}};x,y) = \prod_{k=1}^{\infty} T(X_{k};x,y).
\end{equation}
\end{proof}

We can address here some interesting applications of this class of
Tutte polynomials in dealing with the complexity of non-perturbative
parameters.

As the first application, it is possible to describe the complexity
of a large Feynman diagram $X_{{\rm DSE}}$ in terms of the
complexity of finite Feynman diagrams which live in partial sums
$Y_{m}, \ (m \ge 1)$. The complexity of $Y_{m}$ is interpreted
in terms of the number of different spanning trees (or forests) which live in the
graph. We can compute the complexity of $Y_{m}$ under a recursive
algorithm where at each stage of the algorithm, only an edge
belonging to the proper cycle is chosen. The algorithm starts with a
given graph and produces two graphs at the end of the first stage.
By applying the elementary contraction to a multiple edge, the
resulting graph can have a loop and therefore the procedure can be
still continued. At each subsequent stage one proper cyclic edge
from each graph is chosen (if it exists) for applying the
recurrence. On termination of the algorithm, we get a set of graphs
(or general graphs) none of which have a proper cycle. The
complexity of $Y_{m}$ is the sum of the number of these graphs. If
we perform this recursive algorithm for each $Y_{m}$ when $m$ tends
to infinity, then we can get a sequence which presents the behavior
of complexities when the partial sums converge to $X_{{\rm DSE}}$.

As the second application, we can interpret Feynman rules characters
of the renormalization Hopf algebra of Feynman graphons in the
context of deletion and contraction operators. This approach leads
us to formulate a universal motivic Feynman rule character on large
Feynman diagrams.

\begin{cor}
The Tutte polynomial invariant defines an abstract version of
Feynman rules characters on the renormalization Hopf algebra of Feynman
graphons.
\end{cor}

\begin{proof}
Feynman graphon classes in $\mathcal{S}^{\Phi}_{{\rm graphon}}$ can recover
finite Feynman diagrams and their finite or infinite formal
expansions which contribute to Dyson--Schwinger equations in the physical theory $\Phi$.

For each unlabeled graphon class $[W_{\Gamma}]$ corresponding to a
finite Feynman diagram $\Gamma$, we can define its corresponding Tutte polynomial
$T([W_{\Gamma}];x,y)$ via
\begin{equation} \label{graphon-tutte-1}
T([W_{\Gamma}];x,y):= T(\Gamma;x,y).
\end{equation}
Thanks to Proposition 2.2 in \cite{aluffi-marcolli-2}, the Tutte
polynomial is multiplicative over disjoint unions of finite
(Feynman) diagrams. To see this property requires to describe each connected Feynman diagram $\Gamma$ as a tree $t_{\Gamma}$ with 1PI graphs inserted at the vertices of that tree. Then we can compute the Tutte polynomials of the
resulting trees (i.e. formula (\ref{tutte-tree-1})) to show that the Tutte polynomial of the disjoint union of Feynman graphons
$[W_{\Gamma_{1}}], [W_{\Gamma_{2}}]$ corresponding to finite Feynman
diagrams $\Gamma_{1}, \Gamma_{2}$ can be determined by the Tutte
polynomial of the disjoint union of decorated trees
$t_{\Gamma_{1}}$ and $t_{\Gamma_{2}}$ which is multiplicative. So we have
\begin{equation} \label{tutte-tree-2}
T([W_{\Gamma_{1}}] \sqcup [W_{\Gamma_{2}}];x,y) = T(\Gamma_{1}
\sqcup \Gamma_{2};x,y)
\end{equation}
$$\sum_{s=(s_{1},s_{2})} (x-1)^{b_{0}(s_{1}) + b_{0}(s_{2}) -
b_{0}(\Gamma_{1} \sqcup \Gamma_{2})} (y-1)^{b_{1}(s_{1}) +
b_{1}(s_{2})}$$

$$= T(\Gamma_{1};x,y)T(\Gamma_{2};x,y)= T([W_{\Gamma_{1}}];x,y)T([W_{\Gamma_{2}}];x,y)$$
such that the sum is taken over all pairs $s=(s_{1},s_{2})$ of
subgraphs of $\Gamma_{1}$ and $\Gamma_{2}$, respectively where
$V(s_{i}) \subseteq V(\Gamma_{i}),$ $E(s_{i}) \subset
E(\Gamma_{i}),$ $b_{0}(s) = b_{0}(s_{1}) + b_{0}(s_{2})$.

Furthermore, for a finite connected Feynman diagram $\Gamma$, we
have
\begin{equation} \label{graph-1}
\Gamma = \bigcup_{v \in V(t_{\Gamma})} \Gamma_{v}
\end{equation}
such that $\Gamma_{v}$s are 1PI Feynman diagrams inserted at the
vertices of the tree $t_{\Gamma}$. The internal edges of the tree
$t_{\Gamma}$ are all bridges in the resulting graph and thus
\begin{equation} \label{tutte-feynman-1}
T(\Gamma;x,y) = x^{|E_{{\rm int}}(t_{\Gamma})|}T(\Gamma /
\cup_{e \in E_{{\rm int}}(t_{\Gamma})}e;x,y).
\end{equation}
It is possible to lift this property onto the level of Feynman
graphons where the decomposition (\ref{graph-1}) can be described by
the disjoint unions of Feynman graphons. In other words, for each
$v_{1},...,v_{r} \in  V(t_{\Gamma})$, set $[W_{\Gamma_{v}}]$ as the
unlabeled Feynman graphon class with respect to the graph $\Gamma_{v}$ such that
\begin{equation}
[W_{\Gamma}] = [W_{\Gamma_{v_{1}}}] \sqcup ... \sqcup
[W_{\Gamma_{v_{r}}}].
\end{equation}
It means that the Feynman graphon $[W_{\Gamma}]$ can be defined in terms of a rescaled version of the linear combination of Feynman graphons $[W_{\Gamma_{v_{j}}}]$. In other words,
\begin{equation}
W_{\Gamma_{v_{1}} \sqcup ... \sqcup \Gamma_{v_{r}}} =
\frac{\sum_{j=1}^{r} W_{\Gamma_{v_{j}}}}{|\sum_{j=1}^{r}
W_{\Gamma_{v_{j}}}|}.
\end{equation}

Thanks to (\ref{graphon-tutte-1}), the Tutte polynomial of each
$[W_{\Gamma_{v_{j}}}]$ is defined in terms of the Tutte polynomial of the graph $\Gamma_{v_{j}}$. Then we have
\begin{equation} \label{feynman-graphon-tutte-1}
T([W_{\Gamma}];x,y) = \prod_{j=1}^{r} T([W_{\Gamma_{v_{j}}}];x,y)
\end{equation}
which leads us to a Feynman graphon version of the relation
(\ref{tutte-feynman-1}).

Theorem \ref{large-graph-parametric-1} describes the Tutte
polynomial of a large Feynman diagram $X_{{\rm DSE}}$ on the basis
of the Tutte polynomials of the partial sums $\{Y_{m}\}_{m \ge 1}$.
The cut-distance convergence of the sequence of partial sums to
$X_{{\rm DSE}}$ and the universality of the Tutte polynomial enable
us to lift the properties (\ref{tutte-tree-2}) and
(\ref{tutte-feynman-1}) onto the Feynman graphon $[W_{X_{{\rm
DSE}}}]$. It allows us to define the abstract Feynman rules
characters on large Feynman diagrams in terms of the Tutte
polynomial where we have
\begin{equation}
\tilde{U}([W_{X_{{\rm DSE}}}]):= T([W_{X_{{\rm DSE}}}];x,y) =
T(X_{{\rm DSE}};x,y).
\end{equation}
\end{proof}

Now we explain the construction of another important class of
combinatorial polynomials namely, the first Kirchhoff--Symanzik
polynomials for large Feynman diagrams.

The Feynman parametric representation of a given Feynman integral
$U(\Gamma)$ can be described by the integration theory over a
topological simplex such as $\sigma_{n}$ with respect to Feynman
parameters $w=(w_{1},...,w_{n}) \in \sigma_{n}$ such that $n$ is the
number of internal edges of the corresponding Feynman diagram
$\Gamma$. If $l=b_{1}(\Gamma)$ be the first Betti number of $\Gamma$
(as the maximum number of independent loops in the graph) and an
orientation were fixed on the graph, then we can define the
circuit matrix $\hat{\eta}=(\eta_{ik})_{ik}$ such that $e_{i} \in
E(\Gamma)$ and $k$ ranges over the chosen basis of loops. If an edge
$e_{i}$ belongs to a loop $l_{k}$ with the same/reverse
orientations, then $\eta_{ik}=1$ and $\eta_{ik}=-1$, respectively. If
the edge $e_{i}$ does not belong to a loop $l_{k}$, then
$\eta_{ik}=0$. The arrays of the corresponding $l \times l$
Kirchhoff--Symanzik matrix $M_{\Gamma}(w)$ are given by
\begin{equation}
(M_{\Gamma}(w))_{kr} = \sum_{i=1}^{n} w_{i} \eta_{ik} \eta_{ir}
\end{equation}
which defines a function $M_{\Gamma}: \mathbb{A}^{n} \longrightarrow
\mathbb{A}^{l^{2}}$, $w=(w_{1},...,w_{n}) \longmapsto M_{\Gamma}(w)$
over higher dimensional affine spaces. The first Kirchhoff--Symanzik
polynomial of the graph $\Gamma$ is then defined by the equation
\begin{equation}
\Psi_{\Gamma}(w) = {\rm det} (M_{\Gamma}(w))
\end{equation}
which is independent of the choice of an orientation on the graph
and the basis of loops. This function on $\mathbb{A}^{n}$, which is
a homogeneous polynomial of degree $l$, can be formulated in the
language of spanning trees. We have
\begin{equation}
\Psi_{\Gamma}(w) = \sum_{T \subset \Gamma} \prod_{e \not\in E(T)}
w_{e}
\end{equation}
such that the sum is over all spanning trees $T$ of the graph
$\Gamma$ and for each spanning tree, the product is over all edges of
$\Gamma$ that are not in the selected spanning tree. We can show
that this product is multiplicative over connected components.

Now consider a large Feynman diagram $X$ with the corresponding
sequence $\{Y_{m}\}_{m \ge 1}$ of partial sums. For each $m$, we
know that the first Kirchhoff--Symanzik polynomial of $Y_{m}$ is the
product of the polynomials of each of its components which means
that
\begin{equation}
\Psi_{Y_{m}}(w) = \prod_{j=1}^{m} \Psi_{X_{j}}
\end{equation}
where
\begin{equation}
\Psi_{X_{j}}(w) = \sum_{T_{j} \subset X_{j}} \prod_{e \not\in
E(T_{j})}w_{e}
\end{equation}
such that the sum is taken over all the spanning forests $T_{j}$ of
$X_{j}$ and for each spanning forest, the product is taken over all
edges of $X_{j}$ that are not in that spanning forest.

Thanks to the cut-distance convergence of the sequence $\{Y_{m}\}_{m
\ge 1}$ to $X$, for each $\epsilon >0$, there exists $N_{\epsilon}$
such that for each $m_{1}, m_{2} \ge N_{\epsilon}$, we have
\begin{equation}
d(Y_{m_{1}},Y_{m_{2}}) = d_{{\rm
cut}}([W_{Y_{m_{1}}}],[W_{Y_{m_{2}}}]) < \epsilon.
\end{equation}
Therefore
\begin{equation} \label{limit-cut-1}
d_{{\rm cut}}([W_{Y_{m_{1}}}],[W_{Y_{m_{2}}}]) = 0 \Leftrightarrow
[W_{Y_{m_{1}}}] \approx [W_{Y_{m_{2}}}],
\end{equation}
which means that for enough large $m$, spanning forests of $Y_{m}$ tend to the
spanning forests of the unique graph limit $X$.

\begin{defn}
The first Kirchhoff--Symanzik polynomial $\Psi_{X}(w)$ of the large
Feynman diagram $X$ is defined as the convergent limit of the
sequence $\{\Psi_{Y_{m}}(w)\}_{m \ge 1}$ of the first
Kirchhoff--Symanzik polynomials of finite graphs $Y_{m}= X_{1}
\sqcup ... \sqcup X_{m} = X_{1} + ... + X_{m}$ with respect to the cut-distance topology.
\end{defn}

We can present this polynomial by the expansion
\begin{equation}
\Psi_{X}(w) = \prod_{j=1}^{\infty} \Psi_{X_{j}} = \sum_{T \subset X}
\prod_{e \not\in E(T)} w_{e}
\end{equation}
such that the sum is taken over all the spanning forests $T$ of $X$
and for each spanning forest, the product is taken over all edges of
$X$ that are not in that spanning forest.

\begin{lem}
The first Kirchhoff--Symanzik polynomial $\Psi_{X}(w)$ of the large
Feynman diagram $X$ can be defined recursively in terms of the
deletion and the contraction operators.
\end{lem}

\begin{proof}
Set
\begin{equation} \label{contraction-large-graph-1}
F:= \frac{\partial \Psi_{X}}{\partial w_{n}} = \Psi_{X} \setminus e
\end{equation}
as the deletion operator which is the result of deleting the edge
$e=e_{n}$ from the original graph. In addition, set
\begin{equation}
G:= \Psi_{X}|_{w_{n}=0} = \Psi_{X} / e
\end{equation}
as the contraction operator which is the result of contracting the
edge $e=e_{n}$ to a point in the original graph.

For each edge $e$ which is not a bridge or self-loop in the large
Feynman diagram $X$, we can show that
\begin{equation}
\Psi_{X} = w_{e}F+G
\end{equation}
such that $w_{e}F$ collects the monomials corresponding to spanning
forests that do not include $e$.
\end{proof}

At the end of this section, we address a new application of the
first Kirchhoff--Symanzik polynomial for the study of polynomial
invariants of large Feynman diagrams and Feynman rules characters
which act on Feynman graphons.

For a given large Feynman diagram $X$ with the corresponding first
Kirchhoff--Symanzik polynomial $\Psi_{X}(w)$, define
\begin{equation}
\hat{V}_{X} = \{w \in \mathbb{A}^{\infty}:=\prod_{i=1}^{\infty}\mathbb{A}^{n_{i}}:
\Psi_{X}(w)=0\}
\end{equation}
such that the affine hypersurface complement $\mathbb{A}^{\infty}
\setminus \hat{V}_{X}$ enjoys the multiplicative property. We have
\begin{equation}
\mathbb{A}^{\infty} \setminus \hat{V}_{X} = \prod_{i=1}^{\infty}
\mathbb{A}^{n_{i}} \setminus \hat{V}_{X_{i}}
\end{equation}
such that $n_{i}$ is the number of internal edges of the Feynman
diagram $X_{i}$.

Consider the Grothendieck ring $\mathcal{F}$ of immersed conical
varieties generated by the equivalence classes $[\hat{V}]$ up to
linear changes of coordinates of varieties $\hat{V} \subset
\mathbb{A}^{\infty}$ embedded into some affine space. These varieties are defined
in terms of homogeneous ideals with the usual inclusion--exclusion relation
\begin{equation}
[\hat{V}] = [\hat{R}] + [\hat{V} \setminus \hat{R}]
\end{equation}
for the closed embedding. Now we can define algebro--geometric Feynman
rules characters on the renormalization Hopf algebra of Feynman
graphons. It is an abstract Feynman rules character $\hat{U}:
\mathcal{S}^{\Phi}_{{\rm graphon}} \rightarrow A_{{\rm dr}}$
with the general form
\begin{equation}
\hat{U}([W_{X}]) = I([\mathbb{A}^{\infty} \setminus \hat{V}_{X}])
\end{equation}
such that $[\mathbb{A}^{\infty} \setminus \hat{V}_{X}]$ is the class
in $\mathcal{F}$ and $I: \mathcal{F} \rightarrow A_{{\rm dr}}$
is a ring homomorphism.

We can also define a new invariant of infinite Feynman diagrams in
terms of a generalization of the Chern--Schwartz--MacPherson (CSM)
characteristic classes of singular varieties. The algebro--geometric
Feynman rules have been constructed in terms of a polynomial
invariant originated from the CSM characteristic classes
(\cite{aluffi-marcolli-1, aluffi-marcolli-2, marcolli-1}) and here
we plan to lift that study onto the level of Feynman graphons.

\begin{cor}
There exists an extension of the CSM homomorphism for the level of
large Feynman diagrams generated by Dyson--Schwinger equations.
\end{cor}

\begin{proof}
The existence of the CSM-homomorphism $I^{\infty}_{{\rm CSM}}$ is
another consequence of the cut-distance topology and graphon
representations of Feynman diagrams.

For a given large Feynman diagram $X_{{\rm DSE}}$ as the unique
solution of an equation DSE,  suppose $\Psi_{X_{{\rm DSE}}}(w)$ is
the first Kirchhoff--Symanzik polynomial and $\hat{V}_{X_{{\rm
DSE}}}$ is its associated hypersurface. In addition, suppose
$1_{\hat{V}_{X_{{\rm DSE}}}}$ is the function for $\hat{V}_{X_{{\rm
DSE}}} \subset \mathbb{A}^{\infty}$ and $A(\mathbb{P}^{\infty})$ is
the associated Chow group. The natural transformation
\begin{equation}
1_{\hat{V}_{X_{{\rm DSE}}}} \longmapsto a_{0}[\mathbb{P}^{0}] + a_{1}[\mathbb{P}^{1}] + a_{2}
[\mathbb{P}^{2}] + ... \in A(\mathbb{P}^{\infty})
\end{equation}
allows us to define
\begin{equation}
G_{\hat{V}_{X_{{\rm DSE}}}}(T):= a_{0} + a_{1}T + a_{2}T^{2} + ... +
a_{N}T^{N} + ... \ .
\end{equation}
Now define
\begin{equation}
I^{\infty}_{{\rm CSM}}: \mathcal{F} \longrightarrow \mathbb{Z}[T], \
\ \ [\hat{V}_{X_{{\rm DSE}}}] \longmapsto G_{\hat{V}_{X_{{\rm
DSE}}}}(T)
\end{equation}
and extend it by linearity to achieve a group homomorphism.
\end{proof}

\section{\textsl{The optimization of non-perturbative complexity via a multi-scale Renormalization Group}}

In Complexity Theory, the efficiency of an algorithm against a
problem is judged in terms of the algorithm's capability in
dealing with computational demands about quantities originated from
the intrinsic complexity of that problem. An algorithm is known as
feasible if it has a polynomial-time asymptotic scaling and it is
known as infeasible if it has a super-polynomial (typically,
exponential) scaling. The calculations of quantum field-theoretical
scattering amplitudes at high precision or strong couplings are
infeasible on classical computers but recently, there are some
research efforts which aim to show that these calculations can be
feasible on quantum computers. Traditional calculations of
scattering amplitudes in Quantum Field Theory is on the basis of a
series expansion in powers of the coupling constant (i.e. the
coefficients of the interaction terms) such that the running
coupling constant is taken to be small. Feynman diagrams provide an
intuitive way to organize this class of perturbative expansions
where the loop number is associated with the power of the
(running) coupling constant. The number of this class of combinatorial
diagrams gives us a reasonable measure to evaluate the computational
complexity of perturbative calculations. This measure increases
factorially in terms of the number of loops and the number of external
particles. Furthermore, if the amount of the coupling constant is
insufficiently small, then the perturbative machinery can not
provide correct results while the series expansions are
divergent or asymptotic even at weak coupling constants. Indeed, if
we include higher-order terms beyond a certain point, then the
approximations can be inappropriate. In fact, by increasing the
coupling constant, one eventually reaches a quantum phase transition
at some critical couplings such that in the parameter space near
this phase transition perturbative methods become unreliable. This
region can be studied under strong-coupling regimes.

Generally speaking, limits of computations and the efficiently
computing of things are the most important topics in Theory of Computation and Information
Theory where people deal with the Halting problem as an undecidable
type of problem which determines whether the program will finish
running or continue to run forever. Thanks to rooted trees decorated
by primitive recursive functions, Manin discovered a new
reinterpretation of the Halting problem in the context of the BPHZ
perturbative renormalization. He encapsulated the amount of (non-)computability
in terms of the existence of the
Birkhoff factorization at the level of the renormalization Hopf
algebra of the Halting problem \cite{delaney-marcolli-1, manin-2,
manin-3, manin-4}.

Algorithms belong to the intermediate steps between programs and
functions which means that they are classified as substructures
in the context of Galois theory. This fundamental fact has already
been applied to describe the foundations of a new
categorical--geometric setting for the study of (systems) of
Dyson--Schwinger equations (as the generators of intermediate steps)
in the renormalization Hopf algebra of the Halting problem under
Dimensional Regularization and the global $\beta$-functions. As the
consequence of this treatment, we already have the construction of a
new class of neutral Tannakian subcategories of the universal
Connes--Marcolli category $\mathcal{E}^{{\rm CM}}$ which encode
intermediate algorithms in the context of systems of differential
equations together with singularities. In addition, these
subcategories can address the existence of a new interrelationship
between mixed Tate motives and the Halting problem in Theory of Computation. Furthermore,
thanks to the combinatorial reformulation of the universal
counterterm, some new computational techniques for the study of the
amount of non-computability in the language of the theory of Hall
words have been obtained. It is now possible to analyze infinities or non-computability in the
Theory of Computation in terms of a renormalization theory on
(systems) of Dyson--Schwinger equations and vice versa.
\cite{delaney-marcolli-1, manin-2, shojaeifard-8}

It is so difficult to have an optimal solution when we want to
consider a complex problem under a limited period of time. In this
situation we work on the construction of anytime algorithms by
computing an initial potentially highly suboptimal solution and then
we improve the computed suboptimal solution as time allows.

The Kolmogorov complexity, as an uncomputable concept, aims to
determine the length of the shortest algorithm which produces an
object as the output of a procedure. For a given set $\Sigma$ of
alphabets or letters, let $f$ be a computable function on the set of
all possible strings generated by elements in $\Sigma$. A
description of a string $\sigma$ is some string $\tau$ with
$f(\tau)=\sigma$. The Kolmogorov complexity $K_{f}$ is defined by
\begin{equation}
K_{f}(\sigma):= \big \{^{{\rm
min}\{|\tau|:f(\tau)=\sigma\}}_{\infty, \ \ \ {\rm otherwise}}.
\end{equation}
It is possible to modify this definition independent of choosing $f$
where we need to apply a universal Turing machine. In fact, there
exists a Turing machine $U$ such that for all partial computable
functions $f$, there exists a program $p$ such that for all $y$, we
have $U(p,y)=f(y)$. It enables us to define $K(\sigma)$ as the
Kolmogorov complexity of $\sigma$. It is shown that for all $n$, there
exists some $\sigma$ with $|\sigma|=n$ such that $K(\sigma) \ge n$.
Such $\sigma$ is called Kolmogorov random.

In this section, we plan to apply the graphon representation theory of solutions of Dyson--Schwinger equations
to build a new concept of complexity of non-perturbative parameters in terms of the Kolmogorov complexity. We will show that optimal algorithms in
dealing with non-perturbative parameters can be achieved by
working on a multi-scale non-perturbative Wilsonian renormalization
group defined on the space $\mathcal{S}^{\Phi,g}$.

\subsection{\textsl{A Renormalization
Group program on $\mathcal{S}^{\Phi,g}$}}

One important method for the study of the dynamics of quantum
systems is changing the scales of fundamental parameters of the
physical theory such as momentum, energy and mass. Theory of
Renormalization Group aims to describe the behavior of quantum
systems under this class of re-scalings where the possibility of
exchanging information from scale to scale is considered under
the fundamental principles of Quantum Mechanics. The interpretation
of the concept of mass in the context of time and distance by using
the Planck constant and the interpretation of the concept of time in
the context of distance by using the speed of light enable us to
study the dynamics of relativistic quantum systems under the
re-scaling of the distance parameter. In this situation, small
distances and times are equivalent to large momenta, energies and
masses which produce divergencies in Quantum Field Theory.

There is another important parameter in Quantum Field Theory which
encodes the strength of the fundamental forces. This parameter,
which is known as the coupling constant, appears in the interaction
part of the Lagrangian where we encode information of physical
theory in the language of Green's functions and Feynman integrals.
The amount of the coupling constant has direct influence on the
complexity of Green's functions. As the basic fact, in QED we deal
with couplings smaller than $1$ while in low energy QCD we deal with couplings
at the size of $1$ or larger than $1$. In theoretical and
experimental studies we study coupling constants under two settings
namely, the bare couplings and the running couplings. Running
coupling constants are the outputs of (dimensional) regularization
and renormalization schemes and they have been applied in high
energy levels to generate some intermediate quantities which are
useful for the approximation of non-perturbative parameters.
Running couplings guide us to deal with changing the scale of
the momentum parameter where the Wilsonian model of the
Renormalization Group has been formulated.

Generally speaking, there are two different well-known approaches
for the formulation of non-perturbative Renormalization Group in
Theoretical and Mathematical Physics namely, Wilson--Polchinski framework and
effective average action. We can address following references \cite{delamotte-1, krajewski-toriumi-1,
    marino-1, morris-1, morris-2, newman-riedel-1, newman-riedel-muto-1,
    polchinski-1, wetterich-1} for the study of these methods and here we only review their general facts.

In Wilson--Polchinski framework, Physics at very small scale
corresponds to a scale $\Lambda$ in momentum space which is actually
the inverse of a microscopic length where the partition function is
given by
\begin{equation}
Z[B]=\int d\mu_{C_{\Lambda}}(\phi) {\rm exp} \big(- \int V(\phi) +
\int B\phi \big)
\end{equation}
such that $d\mu_{C_{\Lambda}}$ is a functional Gaussian measure with
a cut-off at scale $\Lambda$. Now if we separate the field
$\phi_{p}=\phi(p)$ into rapid and slow modes $\phi_{p,<},
\phi_{p,>}$ with respect to a scale $k$, then we can rewrite the
partition function in terms of these components which lead us to
define a running potential $V_{k}$ at scale $k$ via performing the
integration on $\phi_{>}$. So we can have
\begin{equation}
Z=\int d\mu_{C_{k}}(\phi_{<}){\rm exp}(-\int V_{k}(\phi_{<}))
\end{equation}
such that when $k \le \Lambda$, $V_{k}$ involves derivative terms
with any power of the derivatives of $\phi_{<}$. The
Wilson--Polchinski equation is indeed a differential equation for
the evolution of $V_{k}$ with $k$ such that the flow of potentials
$V_{k}(\phi_{<})$ do not contain all information on the initial
theory and in addition, $V_{k}(\phi)$ involves infinitely many
couplings contrarily to perturbation theory that involves only the
renormalizable ones. In this method of non-perturbative
Renormalization Group there is no general achievement about the
convergence of the series of approximations that are used. In
addition, the anomalous dimension is depended on the choice of
cut-off parameters that separate the rapid and the slow modes
whereas it should be independent of it.

In effective average action method, the basic idea is to build the
1-parameter family of models for which a momentum depended mass
term is added to the original Hamiltonian where we have
\begin{equation}
Z_{k}[B] = \int \mathcal{D} \phi(x) {\rm exp} \big(-H[\phi] - \Delta
H_{k}[\phi] + \int B \phi \big)
\end{equation}
\begin{equation}
\Delta H_{k}[\phi] = \frac{1}{2} \int_{q} R_{k}(q) \phi_{q}
\phi_{-q}.
\end{equation}
For $0<k<\Lambda$, the rapid modes are almost unaffected by the
cut-off function $R_{k}(q)$ (as a homogeneous to a mass square)
which means that $R_{k}(|q|>k) \simeq 0$. Set
\begin{equation}
W_{k}[B]:= {\rm log} Z_{k}[B]
\end{equation}
with the corresponding Legendre transformation
$\Gamma_{k}[M(x)]=\Gamma_{k}[\frac{\delta W_{k}}{\delta B(x)}]$. The
Renormalization Group equation on $\Gamma_{k}$ is the differential
equation of the type
\begin{equation}
\partial_{k} \Gamma_{k} = f(\Gamma_{k}).
\end{equation}
It is shown that by working on dimensionless and renormalized
quantities, the resulting non-perturbative Renormalization Group can
be written independently of the scales $k$ and $\Lambda$. The
geometry of the resulting Renormalization Group flow from this
framework supports the universality of self-similarity and
decoupling of massive modes.

In this part we address an alternative Renormalization Group
platform for the study of non-perturbative parameters in terms of
changing the scales of Dyson--Schwinger equations in terms of the rescaling of
bare and running couplings. Our study can provide a mathematical
setting to exchange information among non-perturbative aspects
under different scales. For this purpose, we plan to build a new
multi-scale Renormalization Group machinery on the space
$\mathcal{S}^{\Phi,g}$ of all Dyson--Schwinger equations in the
physical theory $\Phi$ under different scales $\lambda g$ of the
bare coupling constant $g$.

The bare couplings are independent of any regularization and
renormalization schemes and therefore their rescaling can be helpful for us to approximate non-perturbative parameters originated
from Dyson--Schwinger equations under a universal setting. Using running coupling constants makes some fundamental issues. On the one hand, running couplings are
unobervable. On the other hand, the most important mission of the
Renormalization Group is to show that the predictions for the
observables do not depend on theoretical conventions such as
renormalization or regularization schemes, the initial state, the
choice of effective charge or the choice of running coupling
constants. Therefore different choices of these couplings should be
related to each other which means that search for an optimal choice
is very important. Our promising non-perturbative Renormalization
Group enables us to study Dyson--Schwinger equations under changing
the scales of bare couplings and running couplings not simultaneously but related to each other.
We expect that this alternative machinery is helpful to provide a
theoretical algorithm for the determination of effective couplings
where the complexity of the corresponding Dyson--Schwinger equations
will be controlled in terms of changing the scale of the bare
coupling constant.

Let $\mathcal{X}^{g}$ be the collection of all interacting
Lagrangians with coefficients in the ring $\mathbb{R}[[g]]$,
invariant under the change $\phi \rightarrow -\phi$, and
interaction parts with the general form
\begin{equation}
I(\phi):= \sum_{k \ge 2} I_{k}(\phi)
\end{equation}
such that for all $k$, $I_{k}=O(g)$ with respect to the bare
coupling constant $g$. Changing the scale of $g$ allows us to obtain
an effective Lagrangian at the scale $\tau \le \lambda$ of a
Lagrangian $L$ at the initial scale $\lambda$. The interaction part
is the original source of Dyson--Schwinger equations and therefore
the re-scaling of the coupling constants lead us to rescale these
non-perturbative type of equations. The Renormalization Group with
respect to this class of momentum type rescaling enables us to
discuss about the possibility of exchanging information among
rescaled Dyson--Schwinger equations.

For each $k$, set $F_{k}$ as the set of all smooth functions on the
hyperplane $\sum_{i=1}^{k} v_{i}=0$ in $(V^{*})^{\oplus k}$ such
that $V$ is the Euclidean 4-dimensional space-time. Define $F:=
\prod_{i=1} ^{\infty} F_{2i}$ to formulate Green's functions
$\mathcal{G}$ given by
\begin{equation}
\mathcal{G}: \mathcal{X}^{g} \times M_{m} \longrightarrow F, \ \
\mathcal{G}:=(\mathcal{G}_{2},\mathcal{G}_{4},...)
\end{equation}
such that \\
- $M_{m}$ is the set of scales of the momentum parameter, \\
- The value $\mathcal{G}_{k}$ at $(L,\lambda)$ is called the $k$-point correlation function of the Lagrangian $L$ at the scale $\lambda$, \\
- For each $k$, $\mathcal{G}_{k}$ is the formal expansion of
amplitudes of all Feynman diagrams with $k$ external edges.

Dyson--Schwinger equations are actually formulated as the fixed point
equations of $\mathcal{G}$ with the general form
\begin{equation}
\mathcal{G} = 1 + \int I_{\gamma} \mathcal{G}
\end{equation}
such that $I_{\gamma}$ is the integral kernel with respect to the
(IPI) primitive Feynman diagram $\gamma$.

\begin{defn} \label{effective-dse}
An equation DSE' in $\mathcal{S}^{\Phi,g}$ is called effective at the
scale $\tau$ of the original Dyson--Schwinger equation DSE at the
initial scale $\lambda$, if the fixed point equation of the Green's
function $\mathcal{G}(L^{\Phi},\lambda)$ corresponding to the equation DSE coincides with the fixed
point equation of the Green's function
$\mathcal{G}(L^{\Phi},\tau)$ corresponding to the equation DSE'.
\end{defn}

It is shown in \cite{submitted-1} that we can build a unique
effective equation at the scale $\tau$ for any equation DSE in
$\mathcal{S}^{\Phi,g}$ at the original scale $\lambda$ of the
momentum parameter. It enables us to change the scale of the momenta
of internal edges of each term in the formal expansion of the
solution of DSE.

In higher orders in perturbation theory we should deal with a large
number of Feynman diagrams which cost us exponentially growing of
the momentum scale. Therefore all orders in perturbation theory do
not accessible for any scale of the momentum parameter. The
asymptotic freedom behavior of QCD at very high energies enables us
to study Physics of hadrons under perturbative setting. At a
relatively low energy scale, coupling constants
become too large where non-perturbative situations of the physical system can be observed.
Running coupling constants, as the functions of the momentum
parameter, describe the strength of the interactions among quarks
and gluons. The determination of this class of couplings has very
uncertainty nature which makes so many computational and
phenomenological difficulties. Dimensional Regularization allows us
to replace the bare coupling constant with a class of scale
depended couplings. The ultraviolet divergencies are eliminated
by normalizing the coupling at a specific momentum scale. In
addition, the ultraviolet cut-off dependency is removed by
allowing the couplings and masses in the Lagrangian
to have a scale dependency. Therefore we can produce running couplings in terms of normalizing them to a measured value at a given scale.
This normalization of the coupling to a measured value makes the
running coupling to not have sensitivity to the ultraviolet cut-off.
The scale dependency of the strong coupling can be controlled by the
$\beta$-function as the infinitesimal generator of the
Renormalization Group. \cite{deur-brodsky-deteramond-1, marino-2,
marino-3}

However we plan to apply effective Dyson--Schwinger equations under different rescaling of the bare coupling constant and running coupling constants to define a new
multi-scale Renormalization Group on the space $\mathcal{S}^{\Phi,g}$ to
analyze the behavior of rescaled Dyson--Schwinger equations. This new Renormalization Group is also useful to rescale running couplings in terms of changing the scale of the bare coupling constant independent of any regularization method.

\begin{thm} \label{multi-scale-RG}
There exists a Renormalization Group machinery on
$\mathcal{S}^{\Phi,g}$ which encodes the dynamics of
Dyson--Schwinger equations under changing the scales of the bare and running
coupling constants.
\end{thm}

\begin{proof}
Set $M_{{\rm running}}$ as the set of scales of the running
couplings. For scales $\Lambda_{1}, \Lambda_{2}, \Lambda_{3} \in
M_{{\rm running}}$ such that $\Lambda_{1} < \Lambda_{2} <
\Lambda_{3}$, define the scale map $R^{{\rm running}}_{--}$ on
$\mathcal{S}^{\Phi,g}$ which satisfies the property
\begin{equation}
R^{{\rm running}}_{\Lambda_{1}\Lambda_{2}} R^{{\rm
running}}_{\Lambda_{2}\Lambda_{3}} = R^{{\rm
running}}_{\Lambda_{1}\Lambda_{3}}.
\end{equation}
For each equation DSE, $R^{{\rm
running}}_{\Lambda_{1}\Lambda_{2}}{\rm DSE}$ is the effective
Dyson--Schwinger equation at the scale $\Lambda_{2}$ of the equation
DSE at the original scale $\Lambda_{1}$. Now define an action of the
semigroup $\mathbb{R}^{+}_{\le 1}$ on the space
$\mathcal{S}^{\Phi,g} \times M_{{\rm running}}$ given by
\begin{equation} \label{running-rg-1}
r \circ ({\rm DSE},\Lambda):= (R^{{\rm running}}_{\Lambda,
r\Lambda}{\rm DSE},r\Lambda).
\end{equation}
The equation $R^{{\rm running}}_{\Lambda, r\Lambda}{\rm DSE}$ is the rescaled version of the Dyson--Schwinger equation DSE under the running coupling $r\Lambda$ while the equation
\begin{equation}
R^{{\rm running}}_{r \Lambda \ \Lambda}{\rm DSE}:=(R^{{\rm running}}_{\Lambda, r\Lambda}{\rm DSE},r\Lambda)
\end{equation}
is the corresponding unique effective Dyson--Schwinger equation in the effective Lagrangian $L^{\Phi}_{r \Lambda}(g)$. The resulting Renormalization Group allows us to study the dynamics
of Dyson--Schwinger equations under the rescaling of the running
couplings.

Set $M_{{\rm bare}}$ as the set of scales of the bare coupling
constant $g$. For scales $\tau_{1}, \tau_{2}, \tau_{3} \in M_{{\rm
bare}}$ such that $\tau_{1} < \tau_{2} < \tau_{3}$, define the scale
map $R^{{\rm bare}}_{--}$ on $\mathcal{S}^{\Phi,g}$ which satisfies
the property
\begin{equation}
R^{{\rm bare}}_{\tau_{1}\tau_{2}} R^{{\rm bare}}_{\tau_{2}\tau_{3}}
= R^{{\rm bare}}_{\tau_{1}\tau_{3}}.
\end{equation}
For each equation DSE in $\mathcal{S}^{\Phi,g}$ define a new
equation $R^{{\rm bare}}_{\tau_{1}\tau_{2}} {\rm
DSE}$ which is the effective Dyson--Schwinger equation at the scale $\tau_{2}$ of the
equation DSE at the initial scale $\tau_{1}$. Now define an action
of the semigroup $\mathbb{R}^{+}_{\le 1}$ on the space
$\mathcal{S}^{\Phi,g} \times M_{{\rm bare}}$ given by
\begin{equation} \label{bare-rg-1}
r \circ ({\rm DSE},\tau):= (R^{{\rm bare}}_{\tau, r\tau}{\rm
DSE},r\tau).
\end{equation}
The equation $R^{{\rm bare}}_{\tau, r\tau}{\rm DSE}$ is the rescaled version of the Dyson--Schwinger equation DSE under the rescaled bare coupling constant $r\tau$ while the equation
\begin{equation}
R^{{\rm bare}}_{r \tau \ \tau}{\rm DSE}:= (R^{{\rm bare}}_{\tau, r\tau}{\rm DSE},r\tau)
\end{equation}
is the corresponding unique effective Dyson--Schwinger equation in the effective Lagrangian $L^{\Phi}(r\tau g)$. The resulting Renormalization Group allows us to study the dynamics
of Dyson--Schwinger equations under the rescaling of the bare
coupling constant $g$.

Thanks to (\ref{running-rg-1}) and (\ref{bare-rg-1}), we can define
a new multi-scale renormalization group on $\mathcal{S}^{\Phi,g}$
where it is possible to rescale the bare coupling constant $g
\longmapsto \tau g$ before the application of regularization
schemes.

Each triple $({\rm DSE},\tau g,\Lambda_{\tau})$ in
$\mathcal{S}^{\Phi,g} \times M_{{\rm bare}} \times M_{{\rm
running}}$ presents an effective Dyson--Schwinger equation such that its unique solution is a polynomial with respect to the rescaled bare and running
coupling constants. This polynomial is an infinite
formal expansion of Feynman integrals together with the powers of the rescaled bare coupling constant
bare coupling constant $\tau g$ (as the initial scale) such that
each Feynman integral in the expansion is defined in terms of the
momentum parameter at the initial scale $\Lambda_{\tau}$. Now we can define
a new action of the semi-group $\mathbb{R}^{+}_{\le 1}$ on
$\mathcal{S}^{\Phi} \times M_{{\rm bare}} \times M_{{\rm running}}$
as the following way
\begin{equation}
\lambda \circ ({\rm DSE},\tau g,\Lambda_{\tau}):= (R^{{\rm
multi}}_{(\tau g,\Lambda_{\tau}),(\lambda \tau g, \lambda
\Lambda_{\tau})} {\rm DSE}, (\lambda \tau g,\lambda \Lambda_{\tau})).
\end{equation}
The equation $R^{{\rm multi}}_{(\tau g,\Lambda_{\tau}),(\lambda \tau
g, \lambda \Lambda_{\tau})} {\rm DSE}$ is the multi-rescaled version of the Dyson--Schwinger equation DSE obtained by changing scales $\tau g \mapsto \lambda \tau g$ and $\Lambda_{\tau} \mapsto \lambda \Lambda_{\tau}$ of the bare and running coupling constants. The equation
\begin{equation}
R^{{\rm multi}}_{(\lambda \tau g, \lambda \Lambda_{\tau}) \ (\tau g,\Lambda_{\tau})} {\rm DSE}:= (R^{{\rm multi}}_{(\tau g,\Lambda_{\tau}),(\lambda \tau g, \lambda
\Lambda_{\tau})} {\rm DSE}, (\lambda \tau g,\lambda \Lambda_{\tau}))
\end{equation}
in $\mathcal{S}^{\Phi,g}$ is the corresponding unique effective Dyson--Schwinger equation in the effective Lagrangian $L^{\Phi}_{\lambda \Lambda_{\tau}}(\lambda \tau g)$.
\end{proof}

In Theorem \ref{multi-scale-RG}, the equation $R^{{\rm multi}}_{(\lambda \tau g, \lambda \Lambda_{\tau}) \ (\tau g,\Lambda_{\tau})} {\rm DSE}$ can be seen as the unique effective Dyson--Schwinger equation at the multi-scale $(\tau g,\Lambda_{\tau})$ of the equation DSE at the original multi-scale $(\lambda \tau g, \lambda \Lambda_{\tau})$. This means that this new multi-scale Renormalization Group can generate observable running coupling constants which are independent of any regularization or renormalization schemes. This fundamental property clarifies the universality of this non-perturbative Renormalization Group with respect to generating observable intermediate values for a given strong bare coupling constant.

Roughly speaking, the renormalization machinery enables us to redefine the unrenormalized constants which exist in the Lagrangian in such a way that the observable quantities remain finite when the ultraviolet cut-off is removed. This machinery requires a new quantity $\mu$ with the dimension of a mass where all intermediate quantities are depended on $\mu$. The confinement in QCD does not allow us to determine a natural scale for $\mu$. The $\mu$ dependence of the coupling constant and various quark masses in QCD force us to define running coupling constants and running masses where the Renormalization Group equations can control the $\mu$ dependence of the resulting renormalized quantities. The running coupling constant $g(\mu^{2})$ can be studied in terms of the equation
\begin{equation}
\mu^{2} \frac{dg(\mu^{2})}{d\mu^{2}} = \beta(g(\mu^{2}))
\end{equation}
which leads us to
\begin{equation}
g(\mu^{2}) = \frac{1}{\beta_{0}{\rm ln}(\mu^{2}/\Lambda^{2})}
\end{equation}
such that the dimensional scale $\Lambda$ is the scale at which the coupling diverges and perturbation theory becomes meaningless. When the cut-off parameter tends to infinity, $\beta(g(\mu^{2}))$ remains finite such that in perturbation theory we have
\begin{equation}
\beta(g(\mu^{2})) = - g(\mu^{2})^{2} (\beta_{0} + \beta_{1}g(\mu^{2}) + \beta_{2}g(\mu^{2})^{2} + ...).
\end{equation}

The Renormalization Group machinery defined by Theorem
\ref{multi-scale-RG} is non-commutative because the scale of the
momentum parameter is completely depended on the chosen rescaling
of the bare coupling. It encodes the dynamics of non-perturbative
aspects of quantum systems by Dyson--Schwinger
equations under different running coupling constants derived from changing the scales of couplings.

Dimensional Regularization or other regularization schemes can change
the nature of the bare couplings to describe QFT under a
perturbative setting but it fails to be functional in higher
orders. Theorem \ref{multi-scale-RG} enables us to generate a new
class of running couplings in terms of the rescaled bare coupling
which are independent of any regularization process. Therefore the
resulting running couplings preserve the nature of the bare coupling
which means that they have physical meanings.

The one important application of this multi-scale Renormalization Group together with Feynman graphon models is a new way for the description of any strongly coupled Dyson--Schwinger equation ${\rm DSE}(g)$ in terms of a cut-distance convergent sequence of Dyson--Schwinger equations under weaker couplings $\lambda g$ in
the space $\mathcal{S}^{\Phi,g}$. This means that we can approximate the complexity of non-perturbative parameters under the strong coupling constant $g$ in terms of a sequence of complexities of solutions of Dyson--Schwinger equations with lesser rate of complexities. In this setting,
we can compute the unique solution $X(\lambda g)$ of the equation ${\rm DSE}(\lambda g)$
for couplings $\lambda g <1$ as intermediate values for the approximation of
the large Feynman diagram $X(g)$.

\begin{cor} \label{rg-graphon-1}
For each large Feynman diagram $X(g)=\sum_{m=0}^{\infty}g^{m}X_{m}$ derived from an equation DSE in $\mathcal{S}^{\Phi,g}$ at the strong bare coupling constant $g \ge 1$, there exists a sequence of large Feynman
diagrams at weaker effective couplings which converges to $X(g)$
with respect to the cut-distance topology.
\end{cor}

\begin{proof}
Thanks to Theorem \ref{multi-scale-RG}, we can build the sequence $\{R^{{\rm
bare}}_{g,\frac{n}{n+1}g}{\rm DSE}\}_{n \ge 1}$ of Dyson--Schwinger
equations with respect to the rescaled bare coupling constants
$\frac{n}{n+1}g$ for each $n\ge1$ where the initial scale of the
equation DSE is at least 1.

For each $n$, $R^{{\rm bare}}_{g,\frac{n}{n+1}g}{\rm DSE}$ is an
equation in $\mathcal{S}^{\Phi,g}$ which has the unique solution
\begin{equation}
Y(\frac{n}{n+1}g) = \sum_{m=0}^{\infty} (\frac{n}{n+1}g)^{m}
X_{m}.
\end{equation}
The scales $\frac{n}{n+1}$ for each $n \ge 1$ provide an
increasing sequence of effective couplings derived from the bare
coupling constant $g$ where $\frac{n}{n+1}g < g$. Therefore for each
$n$, the solution $Y(\frac{n}{n+1}g)$ of the equation $R^{{\rm
bare}}_{g,\frac{n}{n+1}g}{\rm DSE}$ is actually a disjoint union of
multi-loop Feynman diagrams which can be handled by higher order
perturbation methods. It remains to show that the sequence $\{Y(\frac{n}{n+1}g)\}_{n \ge
1}$ is convergent to $X(g)$ with respect to the cut-distance
topology. Thanks to Lemma \ref{feynman-graphon-1}, for each $n \ge
1$, we can associate a unique unlabeled graphon class
$[W_{Y(\frac{n}{n+1}g)}]$ with respect to each large Feynman diagram
$Y(\frac{n}{n+1}g)$. Thanks to Definition \ref{feynman-graphon-2},
it is enough to show that the sequence
$\{[W_{Y(\frac{n}{n+1}g)}]\}_{n \ge 1}$ is convergent to the
unlabeled graphon class $[W_{X(g)}]$.

On the one hand, we can show that when $n$ goes to infinity, the labeled Feynman graphons $W_{(\frac{n}{n+1}g)^{m} X_{m}}$ and $W_{g^{m}X_{m}}$ are weakly isomorphic for each $m \ge 0$. On the other hand, if we replace the Lebesgue measure with the Gaussian measure in the ground probability measure of our Feynman graphon model, then we can show that for a fixed $n \ge 1$, the labeled Feynman grahons $W_{
(\frac{n}{n+1}g)^{m} X_{m}}$ and $W_{g^{m}X_{m}}$ are weakly isomorphic (or equivalent) for each $m \ge 0$.

Therefore when $m$ tends to infinity, labeled Feynman
graphons $W_{Y(\frac{n}{n+1}g)}$ and $W_{X(g)}$ are also weakly
isomorphic (or equivalent).
\end{proof}

\begin{cor} \label{optimization-1}
For any given strongly coupled Dyson--Schwinger DSE in
$\mathcal{S}^{\Phi,g}$, there exists a sequence of many-loop Feynman
diagrams such that their corresponding BPHZ counterterms and
renormalized values converge to the counterterm and the renormalized value generated by the renormalization of the equation DSE.
\end{cor}
\begin{proof}
Thanks to Corollary \ref{rg-graphon-1}, there exists a sequence of Dyson--Schwinger equations under weaker rescaled bare coupling constants $\lambda_{n}g$ and running couplings $\Lambda_{\lambda_{n}}$ in
$\mathcal{S}^{\Phi,g}$ with the corresponding sequence $\{\Gamma_{n}\}_{n \ge 1}$ of their solutions which is cut-distance convergent to the unique solution $X_{{\rm DSE}}(g)$. Now apply Theorem
\ref{bphz-dse-graphon-2} to each large Feynman diagram $\Gamma_{n}$ to build
sequences
\begin{equation}
\{S^{\tilde{\phi}}_{R_{{\rm ms}}}([W_{\Gamma_{n}}])\}_{n
\ge 1}
\end{equation}
and
\begin{equation}
\{S^{\tilde{\phi}}_{R_{{\rm ms}}}
* \tilde{\phi}([W_{\Gamma_{n}}])\}_{n \ge 1}
\end{equation}
which are cut-distance convergent to
$S^{\tilde{\phi}}_{R_{{\rm ms}}}([W_{X_{{\rm DSE}}(g)}])$ and
$S^{\tilde{\phi}}_{R_{{\rm ms}}}
* \tilde{\phi}([W_{X_{{\rm DSE}}(g)}])$, respectively
\end{proof}

Thanks to these investigations, the multi-scale
Renormalization Group defined by Theorem \ref{multi-scale-RG} is capable to optimize the computational procedures in dealing
with non-perturbative parameters. We consider this topic in the
next section where our main effort is to determine an order of complexity on Dyson--Schwinger equations in a given strongly coupled gauge field theory.

\subsection{\textsl{Kolmogorov complexity of Dyson--Schwinger equations}}

In this part, we plan to build a new concept of complexity on the space $\mathcal{S}^{\Phi,g}$ of Dyson--Schwinger equations of a given strongly coupled gauge field theory $\Phi$ in terms of our new multi-scale Renormalization Group (i.e. Theorem \ref{multi-scale-RG}). Then we use the Manin renormalization Hopf algebra of the Halting problem for non-perturbative Feynman rules characters on Feynman graphons to formulate a new way of computing non-perturbative parameters in the context of the Halting problem for partial recursive functions on the new constructive world $\mathcal{S}^{\Phi,g}$. The required structural numbering for this new constructive world can be determined via our multi-scale Renormalization Group (i.e. Theorem \ref{multi-scale-RG}) and the density of rational numbers in real numbers.

We define the Kolmogorov complexity of each Dyson--Schwinger
equation DSE in $\mathcal{S}^{\Phi,g}$ in terms of changing the
scale of the bare coupling constant where exchanging information
among equations at different scales have been encoded by our new
non-perturbative multi-scale Renormalization Group. We can generate partial recursive (or semi-computable) functions in terms of increasing sequences of rational numbers which provide different rescaling of the bare coupling constant and the momentum parameter. For example, thanks to Corollary \ref{rg-graphon-1}, define
\begin{equation} \label{semi-comp-1}
u^{g}: \mathbb{Z}_{+} \times \mathcal{S}^{\Phi,g}  \longrightarrow
\mathcal{S}^{\Phi,g}, \ \ (n, {\rm DSE}(g)) \longmapsto {\rm
DSE}(\frac{n}{n+1}g)
\end{equation}
as a semi-computable function. It means that there exists an
algorithm which encodes the application of $u^{g}$ on
Dyson--Schwinger equations to generate effective equations under different running couplings. The equation ${\rm DSE}(\frac{n}{n+1}g)$ is the effective Dyson--Schwinger equation $R^{{\rm bare}}_{\frac{n}{n+1}g \ g}$ generated by changing the scale of the bare coupling constant in terms of the sequence $\{\frac{n}{n+1}\}_{n \ge 1}$.

\begin{defn}  \label{complexity-dse-2}
The Kolmogorov complexity of an equation ${\rm DSE}(\lambda g)$ at
the scale $\lambda g$ with respect to the function $u^{g}$ (\ref{semi-comp-1}) is
determined by the relation
\begin{equation}
K_{u^{g}}({\rm DSE}(\lambda g)):= {\rm Min} \{ n \in \mathbb{Z}_{+}:
\ \ u^{g}(n,{\rm DSE}'(g)) \subseteq {\rm DSE}(\lambda g) \}
\end{equation}
such that the inclusion means that, up to the weakly isomorphic relation, the large Feynman diagram $X_{{\rm
DSE}'(\frac{n}{n+1}g)}$ can be embedded as a subgraph into the large Feynman diagram
$X_{{\rm DSE}(\lambda g)}$.
\end{defn}

\begin{lem} \label{complexity-dse-1}
There exists the Kolmogorov total order on $\mathcal{S}^{\Phi,g}$.
\end{lem}

\begin{proof}
The Kolmogorov order of $\mathcal{S}^{\Phi,g}$ is defined as a
bijection $\mathbf{K}_{u^{g}}: \mathcal{S}^{\Phi,g} \rightarrow
\mathbb{Z}_{+}$ which arranges all Dyson--Schwinger equations of the physical theory $\Phi$ in the increasing order of their complexities
$\mathbf{K}_{u^{g}}({\rm DSE}(\lambda g))$. Define
\begin{equation}
{\rm DSE}_{1}(\lambda_{1} g) \preceq {\rm DSE}_{2}(\lambda_{2} g)
\Longleftrightarrow K_{u^{g}}({\rm DSE}_{1}(\lambda_{1} g)) <
K_{u^{g}}({\rm DSE}_{2}(\lambda_{2} g)).
\end{equation}

It is possible to determine some constants $c_{0}>0$ such that for all
Dyson--Schwinger equations such as ${\rm DSE}(\lambda g)$,
\begin{equation} \label{complexity-optimal-1}
c_{0} K_{u^{g}}({\rm DSE}(\lambda g)) \le \mathbf{K}_{u^{g}}({\rm
DSE}(\lambda g)) \le K_{u^{g}}({\rm DSE}(\lambda g)).
\end{equation}
\end{proof}

This is actually the most simple example of this class of semi-computable functions and we can generate other examples in terms of any arbitrary increasing sequence of rational numbers, which provides some rescaling for the bare coupling constant $g$ and the running couplings, to define more general semi-computable functions in terms of our new multi-scale Renormalization Group. We have addressed a more general setting for the structure of the Kolmogorov complexity in other research work and here we only focus on the partial recursive (or semi-computable) function (\ref{semi-comp-1}).

Thanks to Definition \ref{complexity-dse-2} and Definition
\ref{complexity-dse-1}, it is now possible to consider
$\mathcal{S}^{\Phi,g}$ as a poset such that for any given partial
recursive map $\sigma:\mathcal{S}^{\Phi,g} \rightarrow
\mathcal{S}^{\Phi,g}$ which generates a permutation, we can define a
new map
\begin{equation} \label{complexity-dse-8}
\sigma_{\mathbf{K}_{u^{g}}}:= \mathbf{K}_{u^{g}} \circ \sigma \circ
\mathbf{K}_{u^{g}}^{-1}
\end{equation}
where it provides a permutation of the subset
\begin{equation} \label{complexity-dse-7}
{\rm D}(\sigma_{\mathbf{K}_{u^{g}}}):= \mathbf{K}_{u^{g}}({\rm
Dom}(\sigma)) \subseteq \mathbb{Z}_{+}.
\end{equation}
Consider the equation ${\rm DSE}(\lambda g) \in {\rm Dom}(\sigma)$ such
that its corresponding orbit $\sigma^{\mathbb{Z}}({\rm DSE}(\lambda
g))$ is infinite. Set
\begin{equation} \label{complexity-19}
\mathbf{K}_{u^{g}} ({\rm DSE}(\lambda g)):= k^{\lambda}_{{\rm DSE}}
\end{equation}
such that for each $n >0$, we have
\begin{equation}
\sigma^{n}_{\mathbf{K}_{u^{g}}}(k^{\lambda}_{{\rm DSE}}) =
\mathbf{K}_{u^{g}}(\sigma^{n}({\rm DSE}(\lambda g))) \le c
\mathbf{K}_{u^{g}}(n).
\end{equation}
In \cite{manin-1} it is discussed that for any partial recursive
function $f:\mathbb{Z}^{+}\rightarrow \mathbb{Z}^{+}$ and $x \in
{\rm Dom}(f)$ we have
\begin{equation} \label{complex-1}
K(f(x)) \le c_{f}K(x) \le c_{f}^{'}x.
\end{equation}
We want to apply the inequality (\ref{complex-1}) for the Kolmogorov
complexity of Dyson--Schwinger equations defined by Definition
\ref{complexity-dse-2} and bijection $\mathbf{K}_{u^{g}}$. Define
\begin{equation}
Y:=\{\sigma^{n}({\rm DSE}(\lambda g)): \ \ n \in \mathbb{Z}_{+}\}
\end{equation}
as a recursively enumerable subset of $\mathcal{S}^{\Phi,g}$ which
plays the role of the domain for a partial recursive function
$A:\mathcal{S}^{\Phi,g} \rightarrow \mathbb{Z}_{+}$ given by
\begin{equation}
A({\rm DSE}(\tau g)) = n, \ \ \ {\rm if} \ \sigma^{n}({\rm
DSE}(\lambda g)) \approxeq {\rm DSE}(\tau g).
\end{equation}
We then have
\begin{equation}
\mathbf{K}^{-1}_{u^{g}}(n) = \mathbf{K}^{-1}_{u^{g}}(A({\rm DSE}(\tau
g))) \le c' \mathbf{K}_{u^{g}}({\rm DSE}(\tau g)) =
c'\mathbf{K}_{u^{g}}(\sigma^{n}({\rm DSE}(\lambda g)))
\end{equation}
such that as the consequence, it is possible to determine some upper and lower
boundaries for the permutation $\sigma_{\mathbf{K}_{u^{g}}}$ such as
\begin{equation} \label{complexity-dse-18}
c_{1} \mathbf{K}^{-1}_{u^{g}}(n) \le
\sigma^{n}_{\mathbf{K}_{u^{g}}}(k^{\lambda}_{{\rm DSE}}) \le c_{2}
\mathbf{K}^{-1}_{u^{g}}(n).
\end{equation}

\begin{lem} \label{halting-1}
Consider $\mathcal{S}^{\Phi,g}$ as the constructive world and equip
this collection with a total recursive structure of additive group
without torsion with the zero element ${\bf 0}$. The Halting problem
for any partial recursive function $f: \mathbb{Z}_{+} \times
\mathcal{S}^{\Phi,g} \rightarrow \mathcal{S}^{\Phi,g}$ can be
described in the language of fixed points of some permutations on $\mathbb{Z}_{+} \times
\mathcal{S}^{\Phi,g}$ derived from $f$.
\end{lem}

\begin{proof}
Feynman graphon models of solutions of Dyson--Schwinger equations enable us to define a total recursive structure of the additive group without torsion on the constructive world $\mathcal{S}^{\Phi,g} \sqcup \{{\bf 0}\}$ such that ${\bf 0}$ is the zero element. Now extend $f$ to a new function
\begin{equation}
g_{f}: \mathbb{Z}_{+} \times
(\mathcal{S}^{\Phi,g} \sqcup \{{\bf 0}\}) \longrightarrow
(\mathcal{S}^{\Phi,g} \sqcup \{{\bf 0}\})
\end{equation}
such that
\begin{equation}
g_{f}((n,X)):= {\bf 0}, \ \ {\rm if} \ \ (n,X) \not\in {\rm Dom}(f).
\end{equation}
Now define a new permutation
$$\tau_{f}: \mathbb{Z}_{+} \times (\mathcal{S}^{\Phi,g} \sqcup
\{{\bf 0}\}) \times (\mathcal{S}^{\Phi,g} \sqcup \{{\bf 0}\})
\longrightarrow \mathbb{Z}_{+} \times (\mathcal{S}^{\Phi,g} \sqcup \{{\bf 0}\}) \times
(\mathcal{S}^{\Phi,g} \sqcup \{{\bf 0}\}),$$
\begin{equation}
\tau_{f}(n,(X,Y)):= (g_{f}(n,{\bf 0}),X+g_{f}((n,Y)),Y).
\end{equation}
We can check that finite orbits of $\tau_{f}$ are fixed points of the permutation. It
leads us to build a new partial recursive permutation $\sigma_{f}$
with the domain
\begin{equation}
{\rm Dom}(\sigma_{f}):= (\mathcal{S}^{\Phi,g} \sqcup \{{\bf 0}\})
\times {\rm Dom}(f).
\end{equation}
Thanks to \cite{manin-1, manin-3} and the definition of $g_{f}$, we can show that the complement
to ${\rm Dom}(\sigma_{f})$ in the constructive world
$(\mathcal{S}^{\Phi,g} \sqcup \{{\bf 0}\}) \times
(\mathcal{S}^{\Phi,g} \sqcup \{{\bf 0}\})$ covers the fixed points
of $\tau_{f}$. This process reduces the Halting problem for $f$ to
the determination of the fixed points of $\tau_{f}$.
\end{proof}

For the constructive world $\mathcal{S}^{\Phi,g}$, the map $u^{g}$
given by (\ref{semi-comp-1}), the map $\sigma_{\mathbf{K}_{u^{g}}}$
given by (\ref{complexity-dse-8}) and (\ref{complexity-dse-7}), the
integer value $k^{\lambda}_{{\rm DSE}}$ given by
(\ref{complexity-19}), define
\begin{equation} \label{complexity-dse-4}
\Psi(k^{\lambda}_{{\rm DSE}}, \sigma,u^{g},z):=
\frac{1}{(k^{\lambda}_{{\rm DSE}})^{2}} + \sum_{n \ge 1}
\frac{z^{\mathbf{K}_{u^{g}}({\rm DSE}(\lambda \frac{n}{n+1}
g))}}{(\sigma^{n}_{\mathbf{K}_{u^{g}}}(k^{\lambda}_{{\rm
DSE}}))^{2}}.
\end{equation}

\begin{cor} \label{complexity-dse-5}
- If the $\sigma$-orbit of the equation ${\rm DSE}(\lambda g) \in
{\rm Dom}(\sigma)$ is finite, then $\Psi(k^{\lambda}_{{\rm DSE}},
\sigma,u^{g},z)$ is a rational function in the complex variable $z$.
All poles of this formal series, which are of the first order, live
at the roots of unity.

- If the $\sigma$-orbit of the equation ${\rm DSE}(\lambda g) \in
{\rm Dom}(\sigma)$ is infinite, then $\Psi(k^{\lambda}_{{\rm DSE}},
\sigma,u^{g},z)$ is the Taylor series of an analytic function on the
region $|z|<1$ which is continuous at the boundary of this region.
\end{cor}

\begin{proof}
It is a direct result of the discussions in \cite{manin-1, manin-3}
where we need to replace the constructive world $\mathbb{Z}_{+}$
with $\mathcal{S}^{\Phi,g}$.
\end{proof}

Thanks to the Manin's reconstruction of the Halting problem in the
language of the BPHZ renormalization program (\cite{manin-2,
manin-3, manin-4}) and the explained machinery with respect to the
constructive world $\mathcal{S}^{\Phi,g}$, now it is possible to relate
the Halting problem for a given partial recursive map
$f:\mathbb{Z}_{+} \times \mathcal{S}^{\Phi,g} \rightarrow
\mathcal{S}^{\Phi,g}$ to our new version of the Kolmogorov complexity of Dyson--Schwinger equations. For this purpose we reduce $f$ to a partial
recursive permutation
\begin{equation}
\sigma_{f}: {\rm Dom}(\sigma_{f}) \subset \mathcal{S}^{\Phi,g}
\longrightarrow {\rm Dom}(\sigma_{f}) \subset \mathcal{S}^{\Phi,g}
\end{equation}
to interpret the problem of recognizing whether a positive
integer number $k$ belongs to the domain ${\rm Dom}(\sigma_{f})$ or not to
the problem of whether the corresponding analytic function
$\Psi(k,\sigma_{f},u^{g},z)$ of a complex parameter $z$ has a pole
at $z=1$ or not.

\begin{thm} \label{halting-graphon-1}
The BPHZ renormalization of Feynman graphons encodes the Halting
problem for a given partial recursive map $f:\mathbb{Z}_{+} \times
\mathcal{S}^{\Phi,g} \longrightarrow \mathcal{S}^{\Phi,g}$.
\end{thm}
\begin{proof}
Thanks to the construction of the renormalization Hopf algebra of
Feynman graphons $\mathcal{S}^{\Phi}_{{\rm graphon}}$ and the BPHZ
renormalization of large Feynman diagrams, consider the character
\begin{equation}
\varphi_{k}: \mathcal{S}^{\Phi}_{{\rm graphon}} \longrightarrow
A_{{\rm dr}}, \ \ \varphi_{k}([W_{X_{{\rm DSE}}}]):=
\Psi(k^{\lambda}_{{\rm DSE}}, \sigma_{f},u^{g},z).
\end{equation}
Thanks to the Birkhoff factorization on the regularization algebra
$A_{{\rm dr}}$, we have $A_{{\rm dr}}= \mathcal{A}_{+} \oplus
\mathcal{A_{-}}$ such that $\mathcal{A}_{+}$ is the unital algebra
of analytic functions in the region $|z|<1$ which are continuous on
the boundary $|z|=1$ and $\mathcal{A}_{-}:= (1-z)^{-1}
\mathbb{C}[(1-z)^{-1}]$.

If we apply Lemma \ref{halting-1}, Corollary \ref{complexity-dse-5} and
discussion about the existence of a pole at $z=1$ for the analytic
function $\Psi(k^{\lambda}_{{\rm DSE}}, \sigma_{f},u^{g},z)$, then we can determine whether $k^{\lambda}_{{\rm DSE}}$ lives in $D(\sigma_{f})$ or not.
\end{proof}

The main reason of this important result is the existence of a class
of semi-computable maps such as $u^{g}$ (for a given strong coupling
$g$) which has led us to define a modified version of the Kolmogorov
complexity for Dyson--Schwinger equations (i.e. Definition
\ref{complexity-dse-2}). The dynamics of the well-defined map
$u^{g}$ (\ref{semi-comp-1}) can be studied by the multi-scale
Renormalization Group machinery which is defined on
$\mathcal{S}^{\Phi,g}$. We can define the Kolmogorov complexity
$K_{w}$ on Dyson--Schwinger equations with respect to other
arbitrary elements $w$ of the set of Kolmogorov optimal functions.
In this setting, the optimality means that for any partial recursive
$v:\mathbb{Z}_{+} \times \mathcal{S}^{\Phi,g} \longrightarrow
\mathcal{S}^{\Phi,g}$ there exists a constant $c_{v,w}>0$ such that
for each Dyson--Schwinger equation ${\rm DSE}(\lambda g)$,
\begin{equation}
K_{w}((n,{\rm DSE}(\lambda g))) \le c_{v,w} K_{v}((n,{\rm
DSE}(\lambda g))).
\end{equation}
Thanks to Corollary \ref{rg-graphon-1}, relations
(\ref{complexity-optimal-1}) and (\ref{complexity-dse-18}) and
Theorem \ref{halting-graphon-1}, which determines the amount of
non-computability via the Halting problem at the level of Feynman
graphons, those semi-computable maps defined in terms of the map
$R^{{\rm multi}}_{--}$ can be considered as the truth candidate to search for
the optimal option.


\chapter{\textsf{The dynamics of non-perturbative QFT in the language of noncommutative geometry}}

\vspace{1in}

$\bullet$ \textbf{\emph{A spectral triple model for quantum motions}} \\
$\bullet$ \textbf{\emph{A noncommutative symplectic geometry model for $\mathcal{S}^{\Phi}_{{\rm graphon}}$}} \\

\newpage

Noncommutative Geometry studies geometric properties of singular
spaces on the basis of suitable coordinate algebras where point
spaces are replaced by (noncommutative) function algebras. The
standard differential and integral calculi have been adapted to a
more general setting in the way compatible with the interpretation
of variable quantities in Quantum Mechanics as operators on the
Hilbert space of states and spectral analysis. The interplay between
Algebra and Topology has been studied conceptually and contextually
under two general settings on the basis of the theory of Hopf algebras
(or quantum groups) and the theory of C*-algebras. In the resulting
dictionary, noncommutative C*-algebras, which are interpreted
as the algebras of continuous functions on some virtual
noncommutative spaces, are the dual arena for noncommutative
topology. As the important consequence of this interrelationship,
the theory of spectral triples and the theory of noncommutative
differential graded algebras enable us to build the
foundations of differential and integral calculi in Noncommutative
Geometry. \cite{connes-marcolli-2, duboisviolette-3}

The idea of applying Noncommutative Geometry to Quantum Field Theory
has already been considered and developed by different groups of
mathematicians and mathematical/theoretical physicists where we can
address new models of gauge field theories or the mathematical foundations
of Standard Model and its modified versions in dealing with
elementary particles \cite{chamseddine-connes-1, connes-marcolli-1,
duboisviolette-1, duboisviolette-4, duboisviolette-kerner-madore-1,
madore-1, mirkasimov-1}. Furthermore, thanks to the renormalization
Hopf algebra, some new applications of noncommutative geometric tools
in dealing with Dyson--Schwinger equations were found where two
classes of differential graded algebras had been formulated to
describe the geometry of quantum motions. The first class of
differential graded algebras was built in the way to determine
a family of connections which encode quantum motions independent of
the chosen regularization scheme \cite{shojaeifard-2}. The second
class of differential graded algebras was built in the way to
encode regularization and renormalization processes of Feynman
diagrams which contribute to solutions of Dyson--Schwinger equations
in the language of noncommutative differential forms. This setting,
which applies shuffle products and Rota--Baxter algebras
(\cite{guo-1}), has provided a new geometric interpretation of the
Connes--Kreimer renormalization group in the context of integrable
systems under a non-perturbative setting \cite{shojaeifard-7}.

In this chapter, we plan to continue our search for some new
applications of Noncommutative Geometry to non-perturbative aspects
(\cite{submitted-1}). At the first step, we explain the construction
of a new class of spectral triples which encodes the geometry of
Dyson--Schwinger equations under an operator theoretic setting. This
study provides the foundations of a theory of spectral geometry for
the description of large Feynman diagrams. At the second step, we
search for a new class of differential graded algebras on Feynman
graphons to build a noncommutative differential geometry machinery
for the description of physical parameters generated by large
Feynman diagrams.

\section{\textsl{A spectral triple model for quantum motions}}

Geometric objects associated to any $n$-dimensional $C^{\infty}$
manifold $M$ such as vector fields, differential forms, general
tensor fields, vector bundles, Riemannian metric, connections,
curvature tensor, etc are encoded via the commutative algebra
$C^{\infty}(M)$ (i.e. infinite times differentiable functions on $M$)
and some extra operators on this algebra. If we replace the algebra
$C^{\infty}(M)$ with a noncommutative algebra $A$ (such as the
algebra generated by some deformation procedures on
$C^{\infty}(M)$), then we can achieve the basic elements of
Noncommutative Geometry as a generalization of the standard
commutative geometry of manifolds. The basic pedagogical example of
a noncommutative space is given via Gelfand--Naimark
Theorem where studying commutative C*-algebras is translated to
studying compact topological (Hausdorff) spaces and vice versa. It
leads us to a general idea that studying noncommutative C*-algebras
becomes to studying "noncommutative" compact topological spaces.
\cite{connes-marcolli-2}

Classical Mechanics can be interpreted as the fundamental example
of a commutative geometry where the phase space of a system of $N$
non-relativistic particles is a $6N$ dimensional symplectic manifold
$M$ and the physical observables, energy, angular momentum, etc are
functions in $C^{\infty}(M)$. Quantum Mechanics can be interpreted
as the fundamental example of a model in Noncommutative Geometry where we
should deal with a noncommutative algebra of quantum observables
consisting of operators on the Hilbert space of states. The position
operator $Q$ and the momentum operator $P$ (as unbounded
self-adjoint operators) satisfy the canonical Heisenberg's
commutation relation
\begin{equation}
PQ - QP = -i\hbar I.
\end{equation}
The physical observables are represented by hermitian operators. If
we apply one-parameter unitary groups $U_{s} = e^{isP}, \ \ V_{t} =
e^{itQ}$, then we have the Weyl form of the commutation relation
namely,
\begin{equation}
U_{s}V_{t} = e^{-i\hbar st} V_{t}U_{s}.
\end{equation}
Set $s=t=1$ and $\lambda=-2 \pi \hbar$ to obtain the unitary bounded
operators $U,V$ on the same Hilbert space which enjoy the property
$UV=e^{2\pi i \lambda} VU$. The noncommutative polynomial algebra
$A_{\lambda}$ generated by $U,V$ together with their corresponding adjoint
operators where equipped with the operator norm is actually a
noncommutative $C^{*}$-algebra derived from Quantum Mechanics.

Deformation quantization focuses on the construction of a
noncommutative algebra of quantum observables in terms of defining
some new noncommutative type of products on the vector space
$C^{\infty}(M)$. The deformation of the coordinates of space-time
with respect to relations such as $[\hat{x}^{\mu},\hat{x}^{\nu}] = i
\theta^{\mu \nu}$ is another machinery in this setting to build a
Noncommutative Geometry model.

In an alternative approach, Connes developed a formulation of
differential geometry in terms of commutative algebras to build a
noncommutative generalization where we can consider a compact
manifold of arbitrary dimension with a well-defined Riemannian structure which
gives rise to a first order differential operator known as the Dirac
operator. It is shown that the manifold, including the metric
tensor, can be completely reconstructed from the discrete
eigenvalues of the Dirac operator such that the properties of the
spectrum can be encoded by a spectral triple which contains some
algebraic information. In summary, an ordinary compact Riemannian
manifold $M$ is reinterpreted in terms of the spectral triple
$(A=C^{\infty}(M), \mathbb{H}= L^{2}(.), D=i \gamma_{\mu}\partial x^{\mu})$
which is called a commutative spectral triple. Thanks to this
setting, Connes achieved a new modified version of the Gelfand--Naimark
Theorem for compact Riemannian manifolds and spectral triples. The
generalization of this approach has led us to the concept of
noncommutative spectral triples where some new applications of
Noncommutative Geometry to the description of relativistic quantum
theory, elementary particles and space-time at the
micro-scale Physics have been discovered by mathematicians and
mathematical/theoretical physicists. As an example we can address
the mathematical foundations of Standard Model, its modified
versions and perturbative renormalization program on these physical theories in terms of noncommutative geometric tools.
\cite{connes-marcolli-2, connes-marcolli-1, suijlekom-1, suijlekom-2}

Here we plan to explain the structure of a new class of spectral
triples originated from solutions of Dyson--Schwinger equations. The
resulting spectral triples encode the geometry of those parts of
Quantum Field Theories under strong running coupling constants where quantum
motions have complicated non-perturbative behaviors.

In general, a spectral triple is a collection $(A,\mathbb{H},D)$ of
related mathematical structures such that $A$ is a (unital)
involutive algebra which is faithfully represented on a given Hilbert
space $\mathbb{H}$ via a representation $\pi$. The operator $D$ is a self-adjoint
operator acting on $\mathbb{H}$ with the compact resolvent. For any
$a \in A$, $\pi(a)$ maps ${\rm dom}(D)$ into itself. The operator
$[D,\pi(a)]$ extends to a bounded operator on $\mathbb{H}$.

Theory of Clifford algebras and spin structures have provided the
foundations of the algebraic reconstruction of the geometry of
smooth (compact) Riemannian manifolds in the context of the theory of
spectral triples. For a given $n$-dimensional (locally) compact
$C^{\infty}$-Riemannian manifold $M$ without boundary, set
$A^{1}(M):= \Gamma(M,T_{\mathbb{C}}^{*}M)$ as the space of sections
of the complex cotangent bundle, which are differentiable $1$-forms
on $M$, with the corresponding dual space $\aleph(M):=
\Gamma(M,T_{\mathbb{C}}M)$, as the space of sections of the tangent
bundle, which are differentiable vector fields on $M$. The metric
$g$ is therefore a $C^{\infty}(M)$-valued symmetric bilinear
positive definite form on $A^{1}(M)$ (or $\aleph(M)$). The Cech cohomology theory of the algebra of Clifford sections
enables us to define ${\rm spin}^{c}$ structures. Then we determine
the corresponding spin structures under Morita equivalent relation.
A ${\rm spin}^{c}$ connection on a spinor module $\Gamma(M,S)$ is
defined (compatible with the action of the algebra of Clifford
sections) as a Hermitian connection
\begin{equation}
\nabla^{S}: \Gamma(M,S) \longrightarrow A^{1}(M)
\otimes_{C^{\infty}(M)} \Gamma(M,S).
\end{equation}
It is called a spin connection, if it commutes with the anti-linear
charge conjugation $c$ for each real vector field. The Riemannian
distance on the manifold $M$ is determined in terms of the
Dirac operator as a complex linear operator such as $D: \Gamma(M,S)
\longrightarrow \Gamma(M,S)$ defined by the composition $-i
\widehat{c} \circ \nabla^{S}$ such that
\begin{equation}
\widehat{c} \in {\rm Hom}_{C^{\infty}(M)}(B \otimes
\Gamma(M,S),\Gamma(M,S))
\end{equation}
is given by $\widehat{c}(\rho_{1},\rho_{2}):= c(\rho_{1})\rho_{2}$
while $B$ is the Clifford algebra bundle. It is also possible to
present this operator under a local setting in terms of the spaces
of vector fields and 1-forms. This explains the Dirac operator as an
essentially self-adjoint operator on its original domain, where we
can see that $[D,f]=-ic(df)$ for any smooth function $f$. Thanks to
this treatment the relation
\begin{equation}
d(x,y) = {\rm sup}\{|f(y)-f(x)|: f \in C^{\infty}(M), ||[D,f]|| \le
1\}
\end{equation}
describes the geodesic distance in terms of an unbounded Fredholm
module over the C*-algebra $C^{\infty}(M)$ \cite{connes-marcolli-2,
connes-marcolli-1}. Therefore all geometric information of the
manifold $M$ can be encapsulated by the spectral triple
\begin{equation}
(C^{\infty}(M), L^{2}(M,S),D).
\end{equation}

\begin{thm} \label{infinite-spectral-triple-1}
Consider $\{(A_{m},\mathbb{H}_{m},D_{m})\}_{m \ge 1}$ as a countable
family of spectral triples with the corresponding family of
representations $\{\pi_{m}\}_{m \ge 1}$. For each $m$, let
$||.||_{m}$ be the norm on $\mathbb{H}_{m}$ and then choose
$\{\alpha_{m}\}_{m \ge 1}$ as a sequence of non-zero real numbers
such that the sequence
$\{||(1+\alpha_{m}^{2}D_{m}^{2})^{\frac{-1}{2}}||_{m}\}_{m \ge 1}$
converges to zero when $m$ goes to infinity. There exists a spectral
triple
\begin{equation}
(A^{\oplus},\mathbb{H}^{\oplus},\overline{D^{\oplus}})
\end{equation}
such that $\mathbb{H}^{\oplus}:= \bigoplus_{m \ge 1}
\mathbb{H}_{m},$ $D^{\oplus}:= \bigoplus_{m \ge 1} \alpha_{m} D_{m}$
with the corresponding self-adjoint extension
$\overline{D^{\oplus}}$. In addition,
$$A^{\oplus}:=\{(a_{m})_{m \ge 1} \in \prod_{m} A_{m}:$$
$${\rm sup}_{m \ge 1} ||\pi_{m}(a_{m})||_{m} < +\infty, \ \ \ {\rm sup}_{m
\ge 1} ||[\alpha_{m}D_{m}, \pi_{m}(a_{m})]||_{m} < + \infty\}$$ such
that for each $a^{\oplus} \in A^{\oplus}$,
$\pi^{\oplus}(a^{\oplus}):= \bigoplus_{m \ge 1} \pi_{m}(a_{m})$.
\cite{falk-1}
\end{thm}

The graduation parameter on the renormalization Hopf algebra and
Hopf subalgebras generated by Dyson--Schwinger equations enable us
to describe the corresponding complex Lie groups
$\mathbb{G}_{\Phi}(\mathbb{C})$ and $\mathbb{G}_{{\rm
DSE}}(\mathbb{C})$ under projective limits of Lie subgroups.

Generally speaking, for a given commutative (graded) Hopf algebra
$H$, let ${\rm Spec}(H)$ be the set of all prime ideals of $H$
equipped with the Zariski topology and the structure sheaf. This
topological space accepts a group structure generated by the
coproduct of $H$. Under a categorical setting, the functional Spec
is a contravariant functor from the category of commutative algebras
to the category of topological spaces which leads us to define
another functional $\mathbb{G}_{H}={\rm Spec}(H)$ as a covariant
representable functor from the category of commutative algebras to
the category of groups. For each commutative algebra $A$, the Lie group
$\mathbb{G}_{H}(A)={\rm Spec}(H)(A)$ is the set of morphisms with
the general form
\begin{equation}
\varphi:H \longrightarrow A, \ \
\varphi(h_{1}h_{2})=\varphi(h_{1})\varphi(h_{2}), \ \
\varphi(1_{h})=1_{A},
\end{equation}
which is equipped with the convolution product
\begin{equation}
\varphi_{1} * \varphi_{2} (h):= m \circ (\varphi_{1} \otimes
\varphi_{2}) \circ \Delta_{H}(h).
\end{equation}

Thanks to Milnor--Moore Theorem (\cite{milnor-moore-1}), the finite
dimensional complex Lie group ${\rm GL}_{n}$ of $n \times n$
matrices with non-zero determinants corresponds to the Hopf algebra
\begin{equation}
H_{{\rm GL}_{n}}=k[x_{i,j},t]_{i,j=1,...,n} / {\rm det}(x_{i,j})t-1
\end{equation}
with the coproduct
\begin{equation}
\Delta(x_{i,j}) = \sum_{s} x_{i,s} \otimes x_{s,j}.
\end{equation}

It is shown that if the Hopf algebra $H$ is finitely generated as an
algebra, then its corresponding affine group scheme is a linear
algebraic group which can be embedded as a Zariski closed subset
of some ${\rm GL}_{n}$. \cite{milne-1}

If we have a graduation parameter on the commutative Hopf
algebra $H$, then there exists a family $\{H_{n}\}_{n \ge
0}$ of commutative Hopf subalgebras such that $H = \bigcup_{n \ge 0}
H_{n}$ and for all $n$ and $m$, we can find some $k$ where $H_{n} \cup
H_{m} \subset H_{k}$. It is called a finite type graded Hopf algebra
if each component of the grading structure is finitely generated as
an algebra which means that for each $n$, there exists the
corresponding linear algebraic group of the form
\begin{equation}
\mathbb{G}_{n}(\mathbb{C})={\rm Spec}(H_{n})(\mathbb{C})< {\rm
GL}_{m_{n}}(\mathbb{C})
\end{equation}
for some $m_{n}$. These algebraic groups can generate the affine
group scheme $\mathbb{G}_{H}$ corresponding to the Hopf algebra $H$
via the projective limit
\begin{equation} \label{alg-group}
\mathbb{G}_{H} = {\rm lim}_{\longleftarrow_{n}} \mathbb{G}_{n}.
\end{equation}

\begin{thm} \label{infinite-spectral-triple-2}
There exists a class of infinite dimensional spectral triples which
describes the geometry of quantum motions in physical theories with
strong coupling constants.
\end{thm}

\begin{proof}
We are going to build a spectral triple with respect to each
Dyson--Schwinger equation in $\mathcal{S}^{\Phi,g}$ such that the
bare coupling constant $g$ is strong enough to produce
non-perturbative situations. For simplicity in notation we set $g=1$ and suppose the large Feynman diagram
$X_{{\rm DSE}}=\sum_{n \ge 0} X_{n}$ is the unique solution of an
equation DSE. It is discussed that terms $X_{n}$ are generators of
the free graded connected commutative finite type Hopf subalgebra
$H_{{\rm DSE}}(\Phi)$ of the Connes--Kreimer renormalization Hopf
algebra $H_{{\rm FG}}(\Phi)$ of Feynman diagrams graded in terms of the number of
internal edges. Present $H_{{\rm DSE}}(\Phi)$ in terms of its graded
components as follows
\begin{equation}
H_{{\rm DSE}}(\Phi)= \bigcup_{n \ge 0} H_{{\rm DSE}}^{(n)}(\Phi).
\end{equation}
For each $n$, the finite dimensional Hopf subalgebra $H_{{\rm
DSE}}^{(n)}(\Phi)$ determines the finite dimensional complex Lie
subgroup $\mathbb{G}_{{\rm DSE}}^{(n)}(\mathbb{C})$ which is
embedded as a closed subset of the linear algebraic group ${\rm
GL}_{m_{n}}(\mathbb{C})$ for some $m_{n}$ with respect to the
Zariski topology. Thanks to (\ref{alg-group}), the complex pro-unipotent graded Lie group
$\mathbb{G}_{{\rm DSE}}(\mathbb{C})$ is the projective limit of
$\mathbb{G}_{{\rm DSE}}^{(n)}(\mathbb{C})$s as closed subsets of
${\rm GL}_{m_{n}}(\mathbb{C})$s.

For each $m_{n}$, ${\rm GL}_{m_{n}}(\mathbb{C})$ is a finite
dimensional Riemannian manifold with the
corresponding spectral triple
\begin{equation}
\mathcal{S}^{(m_{n})}:=(C^{\infty}({\rm
GL}_{m_{n}}(\mathbb{C})),L^{2}({\rm
GL}_{m_{n}}(\mathbb{C}),S),D_{{\rm GL}_{m_{n}}(\mathbb{C})}).
\end{equation}
A restriction of this spectral triple enables us to build the
spectral triple corresponding to the complex Lie group
$\mathbb{G}_{{\rm DSE}}^{(n)}(\mathbb{C})$. We present it by
\begin{equation}
\mathcal{S}^{(n)}_{{\rm DSE}}=(A^{(n)}_{{\rm DSE}},
\mathbb{H}^{(n)}_{{\rm DSE}}, D^{(n)}_{{\rm DSE}}).
\end{equation}
Now consider the family $\{\mathcal{S}^{(n)}_{{\rm DSE}}\}_{n \ge
0}$ of countable number of spectral triples derived from components
of the graduation structure of the Hopf subalgebra $H_{{\rm
DSE}}(\Phi)$ generated by the equation DSE. Let $\{\alpha_{n}\}_{n
\ge 1}$ be a sequence of non-zero real numbers such that
\begin{equation}
\{||(1+\alpha_{n}^{2}(D^{(n)}_{{\rm
DSE}})^{2})^{\frac{-1}{2}}||_{n}\}_{n \ge 1}
\end{equation}
converges to zero when
$n$ tends to infinity where $||.||_{n}$ is the norm on
$\mathbb{H}^{(n)}_{{\rm DSE}}$. Apply Theorem
\ref{infinite-spectral-triple-1} to achieve the infinite dimensional
spectral triple
\begin{equation}
\mathcal{S}_{{\rm DSE}}^{\oplus}:= (A_{{\rm
DSE}}^{\oplus},\mathbb{H}_{{\rm DSE}}^{\oplus},\overline{D_{{\rm
DSE}}^{\oplus}})
\end{equation}
originated from the five-tuples $(A^{(n)}_{{\rm DSE}},
\mathbb{H}^{(n)}_{{\rm DSE}}, D^{(n)}_{{\rm DSE}},\pi^{(n)}_{{\rm
DSE}},\alpha_{n})$ for each $n$. The norm of the Hilbert space $\mathbb{H}_{{\rm DSE}}^{\oplus}$ is
given by
\begin{equation}
||.||^{\oplus}:= {\rm sup}_{n} ||.||_{n}.
\end{equation}
In addition, we can check that the representation
$\pi^{\oplus}_{{\rm DSE}}$ and the commutator $[D_{{\rm
DSE}}^{\oplus},\pi^{\oplus}_{{\rm DSE}}(A_{{\rm DSE}}^{\oplus})]$
are bounded where the sequence $\{\alpha_{n}\}_{n \ge 1}$ controls
the behavior of the sequence $\{D_{{\rm DSE}}^{(n)}\}_{n \ge 1}$. It
means that
\begin{equation}
\sum_{n} {\rm dim}({\rm Ker} D_{{\rm DSE}}^{(n)}) < \infty.
\end{equation}
\end{proof}

It is reasonable to name $\mathcal{S}_{{\rm DSE}}^{\oplus}$ as the
non-perturbative spectral triple with respect to the
Dyson--Schwinger equation DSE.

\begin{rem}
If the coupling constant of a physical theory is weak enough where
higher order perturbation methods can handle solutions of
Dyson--Schwinger equations, then we can describe the geometry of
this class of quantum motions in terms of summing a finite number of
finite dimensional spectral triples.
\end{rem}

\begin{cor}
Each non-perturbative spectral triple has a graphon representation.
\end{cor}
\begin{proof}
For a given spectral triple $\mathcal{S}_{{\rm DSE}}^{\oplus}$ with
respect to the equation DSE, we can associate the unlabeled Feynman graphon
class $[W_{t_{X_{{\rm DSE}}}}]$ determined by the labeled graph
functions of the infinite tree (or forest) $t_{X_{{\rm DSE}}}$ corresponding to the unique solution of DSE.
\end{proof}

The geometry of the underlying manifold determines the spectrum but
the main challenge is the possibility of recovering geometrical
information from the spectrum to determine completely the metric or
the shape of the boundary. While the answer to this challenge is negative but Noncommutative Geometry can provide an operator theoretic setting to deal with the theory of spectral geometry.  The fundamental integral in
Noncommutative Geometry is described as the Dixmier trace
which extends the Wodzicki residue from pseudodifferential operators
on a manifold to a general framework which concerns spectral triples
\cite{connes-marcolli-2}. In other words, for a given spectral triple, we have
\begin{equation}
\overline{\int} T:= {\rm Res}_{s=0} {\rm Tr}(T|D|^{-s}).
\end{equation}
It is possible to adapt this integral to deal with the geometry of
Dyson--Schwinger equations. The construction of the non-perturbative
spectral triple $\mathcal{S}_{{\rm DSE}}^{\oplus}$ (i.e. Theorem
\ref{infinite-spectral-triple-2}) shows that for each $n$,
$\mathcal{S}^{(n)}_{{\rm DSE}}$ is a finite dimensional spectral
triple. Actually, for each $n \ge 1$, $\mathcal{S}^{(n)}_{{\rm DSE}}$ is the result of
the restriction of the spectral triple associated to the complex Lie
group ${\rm Gl}_{m_{n}}(\mathbb{C})$ for some $m_{n}$. Therefore for
each $n \ge 1$, the functional
\begin{equation}
a \longmapsto {\rm Tr}^{+}(a |D^{(n)}_{\rm DSE}|^{-m_{n}})
\end{equation}
determines a differential calculus theory and spectral geometry with
respect to the Riemannian volume form for $\mathcal{S}^{(n)}_{{\rm
DSE}}$. This differential calculus is describing the geometric
behavior of a quantum motion in terms of its approximation with respect to partial sums
of the unique solution $X_{{\rm DSE}}$ of the corresponding equation
DSE. Thanks to this interpretation, we may have chance to search for the existence
of a noncommutative integral with the general form
\begin{equation}
a^{\oplus}\longmapsto {\rm Tr}_{\omega}(a^{\oplus}|D_{{\rm
DSE}}^{^{\oplus}}|^{-p})
\end{equation}
for some $p \ge 1$ and state $\omega$. This noncommutative integral,
which is on the basis of the Connes--Dixmier traces, can lead us
to build a theory of spectral geometry for large Feynman diagrams.

\section{\textsl{A noncommutative symplectic geometry model for $\mathcal{S}^{\Phi}_{{\rm graphon}}$}}

We have discussed that for a given smooth manifold $M$ with the
corresponding complex commutative unital *-algebra $C^{\infty}(M)$,
it is possible to reconstruct $M$ together with its smooth structure
and the objects attached to the manifold (such as smooth vector
fields) in terms of the spaces of characters and derivations of the
algebra $C^{\infty}(M)$. The choice of the generalization method for
the notion of module over a commutative algebra when this algebra is
replaced by a noncommutative algebra is related to the choice of the
noncommutative generalization of the classical commutative case.
There are some approaches to build the algebraic generalizations of
differential geometry such as Koszul framework. This framework is on
the basis of the space ${\rm Der}(A)$ of all derivations of a
commutative associative algebra $A$. A graded differential
algebra (as the generalization of the algebra of differential forms)
determines another graded differential algebra $C_{\wedge}({\rm Der}(A),A)$
of $A$-valued Chevalley--Eilenberg cochains of the Lie algebra ${\rm
Der}(A)$. The Koszul framework admits a generalization to the
noncommutative setting via differential calculus with respect to derivations.
It is actually the suitable differential calculus for Quantum
Mechanics. In this setting, an algebraic version of differential
geometry in terms of a commutative associative algebra $A$,
$A$-modules and connections on these modules have been designed. If
we replace the commutativity of the algebra with non-commutativity,
then different classes of generalizations of the notion of a module
over a noncommutative algebra can be resulted such as the notions of
left or right $A$-modules and left or right $Z(A)$-modules.
\cite{duboisviolette-1, duboisviolette-2, duboisviolette-3,
duboisviolette-kerner-madore-1, madore-1}

The algebraic interpretation of Classical Geometry requires a
commutative setting where we have two options to fix the algebra.
The first one is the real commutative algebra $A_{\mathbb{R}}$ of
smooth real valued functions where its complexified extension is
canonically a complex commutative *-algebra. The second one is the
complex commutative *-algebra $A_{\mathbb{C}}$ of smooth complex
valued functions where the set $A^{{\rm hermitian}}$ of its
hermitian elements is a real commutative algebra and thus
$A_{\mathbb{C}}$ will be the complexification of $A^{{\rm
hermitian}}$.

The algebraic interpretation of Quantum Physics requires a
noncommutative setting where we already have two classes of
generalizations of the algebra of real valued functions. The first
one is the real Jordan algebra $A^{{\rm hermitian}}$ of all
hermitian elements of a complex noncommutative associative *-algebra
$A$. The second one is a real associative noncommutative algebra.
The most important challenge at this level is the choice of the
mathematical machinery to build a differential calculus theory. One
generalization approach has been formulated by Connes in terms of
theory of cyclic cohomology of an algebra where the generalization
of the cohomology of a manifold in Noncommutative Geometry is
actually the reduced cyclic homology of an algebra which replaces
the standard algebra of smooth functions. As we know the computation
of cohomology theory of classical manifolds is not a unique way and
furthermore, we can expect the construction of noncommutative
generalizations of differential geometry for which the
generalization of de Rham theorem fails to be true. These facts show
that any cochain complex, which has the reduced cyclic homology as
cohomology, can not be an acceptable generalization of differential
forms. Thanks to these efforts, the best candidate for the
construction of a noncommutative differential calculus is on the
basis of the space of derivations as generalizations of vector
fields. This platform, which had been initiated and developed by
Kozul and Dubois-Violette, has already provided the foundations of a
noncommutative symplectic geometry for the study of quantum
theories. \cite{duboisviolette-1, duboisviolette-kerner-madore-1,
duboisviolette-kerner-madore-2, koszul-1}

In a different story, the Connes--Kreimer Hopf algebraic
renormalization is the direct result of the existence of the
Hopf--Birkhoff factorization on a class of Lie groups. The original source
of this particular factorization is the multiplicativity of
perturbative renormalization which is encoded by the theory of
Rota--Baxter algebras. The determination of a class of Hopf
subalgebras via Dyson--Schiwnger equations together with the
renormalization of these equations under Dimensional Regularization
had been applied to build a class of Dubois--Violette's differential
graded algebras which encode the geometric information of these
equations in the context of noncommutative differential forms. The
basic idea in this approach is to associate a noncommutative algebra
to each equation DSE and then build a theory of noncommutative
(symplectic) geometry to encode the behavior of infinitesimal characters corresponding to Feynman diagrams which contribute to the solution of DSE under the renormalization process. This platform is also useful to formulate a new interpretation of the
Connes--Kreimer non-perturbative renormalization group in the
context of quantum integrable systems. \cite{shojaeifard-7}

Our main task in
this part is to develop this new Hopf algebraic approach and explain the construction
of a noncommutative differential calculus theory for the topological
Hopf algebra $\mathcal{S}^{\Phi}_{{\rm graphon}}$ of Feynman
graphons which is originated from the BPHZ non-perturbative renormalization (i.i. Theorem
\ref{bphz-dse-graphon-2}) and the theory of Rota--Baxter algebras
(\cite{guo-1}). The basic step is to associate a (noncommutative) algebra to
$\mathcal{S}^{\Phi}_{{\rm graphon}}$ and then build a theory of
noncommutative differential forms on this algebra. Our study can determine a new class of non-perturbative quantum integrable systems generated by solutions of Dyson--Schwinger equations.

The BPHZ renormalization program is on the basis of Dimensional
Regularization and Minimal Subtraction map $R_{{\rm ms}}$ which is an idempotent Rota--Baxter map on the
regularization algebra of Laurent series with finite pole parts. We want to show that the application of each
step of the BPHZ renormalization program to Feynman graphons can determine a theory
of noncommutative differential calculus. These differential calculi are useful to formulate some new geometric tools for the evaluation of solutions of Dyson--Schwinger equations under the renormalization procedure.

\begin{thm} \label{dga-graphon-1}
The Minimal Subtraction map $R_{{\rm ms}}$ in the BPHZ
renormalization of Feynman graphons (i.e. Theorem \ref{feynman-graphon-5}
and Theorem \ref{bphz-dse-graphon-2}) determines a
noncommutative symplectic geometry model for the Hopf algebra
$\mathcal{S}^{\Phi}_{{\rm graphon}}$.
\end{thm}

\begin{proof}
Consider $A_{{\rm dr}}:= \mathcal{A}_{+} \oplus \mathcal{A_{-}}$ as
the algebra of Laurent series with finite pole parts which encodes
Dimensional Regularization (i.e. regularization scheme) and $R_{{\rm
ms}}$ as the linear map on $A_{{\rm dr}}$ which projects a series
onto its corresponding pole parts. The pair $(A_{{\rm dr}},R_{{\rm
ms}})$ satisfies the conditions of a Rota--Baxter algebra which
enables us to define a new family of convolution products on the
space $L(\mathcal{S}^{\Phi}_{{\rm graphon}},A_{{\rm dr}})$ of linear
maps in terms of the following steps.

- Lift the map $R_{\rm ms}$ onto a new map $\mathcal{R}$ on $L(\mathcal{S}^{\Phi}_{{\rm
graphon}},A_{{\rm dr}})$ defined by
\begin{equation}
\mathcal{R}(\phi):= R_{\rm ms} \circ \phi.
\end{equation}
The pair $(L(\mathcal{S}^{\Phi}_{{\rm
graphon}},A_{{\rm dr}}),\mathcal{R})$ is a new Rota--Baxter algebra.

- Set $\hat{\mathcal{R}}:= Id -
\mathcal{R}$ and for each $\lambda \in \mathbb{R}$, define a new
class of Nijenhuis maps $\mathcal{R}_{\lambda}:= \mathcal{R} -
\lambda \hat{\mathcal{R}}$.

- Define a new family of products on
$L(\mathcal{S}^{\Phi}_{{\rm graphon}},A_{{\rm dr}})$ with the general form
\begin{equation}
\phi_{1} \circ_{\lambda} \phi_{2}:= \mathcal{R}_{\lambda} (\phi_{1})
*_{{\rm gr}} \phi_{2} + \phi_{1} *_{{\rm gr}} \mathcal{R}_{\lambda} (\phi_{2}) -
\mathcal{R}_{\lambda} (\phi_{1} *_{{\rm gr}} \phi_{2})
\end{equation}
such that $*_{{\rm gr}}$ is the convolution product with respect to
the coproduct $\Delta_{{\rm graphon}}$ on Feynman graphons
(\ref{cop-graphon-1}) where we have
\begin{equation}
\psi_{1} *_{{\rm gr}} \psi_{2}([W_{\Gamma}])= \sum
\psi_{1}([W_{\Gamma'}])\psi_{2}([W_{\Gamma''}]), \ \ \Delta_{{\rm
graphon}}([W_{\Gamma}]) = \sum [W_{\Gamma'}] \otimes [W_{\Gamma''}].
\end{equation}

The non-cocommutativity of the renormalization Hopf algebra of
Feynman graphons shows that the convolution product $*_{{\rm gr}}$
and new products $\circ_{\lambda}$ are noncommutative.

The Nijenhuis property of $\mathcal{R}_{\lambda}$ shows that
\begin{equation} \label{nijenhuis-1}
\mathcal{R}_{\lambda}(\phi_{1} \circ_{\lambda} \phi_{2}) =
\mathcal{R}_{\lambda}(\phi_{1}) *_{{\rm gr}}
\mathcal{R}_{\lambda}(\phi_{2})
\end{equation}
which supports the associativity of these new products.

Now set
\begin{equation}
C_{\lambda}^{{\rm graphon}}:= (L(\mathcal{S}^{\Phi}_{{\rm
graphon}},A_{{\rm dr}}), \circ_{\lambda})
\end{equation}
as the unital associative noncommutative algebra generated by the
Minimal Subtraction map. For each $\lambda$, the
commutator with respect to $\circ_{\lambda}$ determines a new Lie
bracket $[.,.]_{\lambda}$ on the space $L(\mathcal{S}^{\Phi}_{{\rm
graphon}},A_{{\rm dr}})$ given by
\begin{equation} \label{braket-1}
[\phi_{1},\phi_{2}]_{\lambda} = [\mathcal{R}_{\lambda} (\phi_{1}),
\phi_{2}] + [\phi_{1}, \mathcal{R}_{\lambda} (\phi_{2})] -
\mathcal{R}_{\lambda}[\phi_{1},\phi_{2}].
\end{equation}
This class of Lie brackets is the key tool for us to build a new noncommutative differential calculus on
$C_{\lambda}^{{\rm graphon}}$ in terms of the following steps.

- Set ${\rm Der}^{\lambda}_{{\rm graphon}}$ as the space of all
derivations on $C_{\lambda}^{{\rm graphon}}$. It has all linear maps
such as $\theta:C_{\lambda}^{{\rm graphon}} \longrightarrow
C_{\lambda}^{{\rm graphon}}$ which enjoys the Leibniz rule.

- The Lie bracket $[.,.]_{\lambda}$ determines naturally the Poisson bracket
$\{.,.\}_{\lambda}$ on $C_{\lambda}^{{\rm graphon}}$. For each
$\phi \in C_{\lambda}^{{\rm graphon}}$, its corresponding
Hamiltonian derivation is defined by
\begin{equation}
{\rm ham} (\phi): \psi \mapsto \{\phi,\psi\}_{\lambda}.
\end{equation}
Set ${\rm Ham}^{\lambda}_{{\rm graphon}}$ as the
$Z(C_{\lambda}^{{\rm graphon}})$-module generated by all Hamiltonian
derivations on $C_{\lambda}^{{\rm graphon}}$.

- Define
\begin{equation}
\Omega^{\bullet}_{\lambda,{\rm graphon}}(C_{\lambda}^{{\rm
graphon}}):= (\bigoplus_{n \ge 0} \Omega^{n}_{\lambda,{\rm
graphon}}(C_{\lambda}^{{\rm graphon}}),d_{\lambda})
\end{equation}
as the differential graded algebra on $C_{\lambda}^{{\rm graphon}}$. For each $n \ge 1$, $\Omega^{n}_{\lambda,{\rm graphon}}(C_{\lambda}^{{\rm graphon}})$
is the space of all $Z(C_{\lambda}^{{\rm graphon}})$-multilinear
antisymmetric mappings from ${\rm Ham}^{\lambda}_{{\rm graphon}}
\times ...^{n} \times {\rm Ham}^{\lambda}_{{\rm graphon}}$ into
$C_{\lambda}^{{\rm graphon}}$. The zero component of this
differential graded algebra is the initial algebra
$C_{\lambda}^{{\rm graphon}}$. In addition, for each $\omega \in \Omega^{n}_{\lambda,{\rm
graphon}}(C_{\lambda}^{{\rm graphon}})$ and $\theta_{i} \in {\rm
Ham}^{\lambda}_{{\rm graphon}}$, the anti-derivative degree one
differential operator $d_{\lambda}$ is defined by
$$d_{\lambda}\omega(\theta_{0},...,\theta_{n}):= \sum_{k=0}^{n}
(-1)^{k} \theta_{k}
\omega(\theta_{0},...,\hat{\theta_{k}},...,\theta_{n}) +$$
\begin{equation}
\sum_{0 \le r < s \le n} (-1)^{r+s}
\omega([\theta_{r},\theta_{s}]_{\lambda},\theta_{0},...,\hat{\theta_{r}},...,\hat{\theta_{s}},...,\theta_{n})
\end{equation}
such that we have $d_{\lambda}^{2}=0$.

Thanks to this differential graded (Lie) algebraic machinery, we can determine a new class of symplectic structures
generated by the Lie bracket $[.,.]_{\lambda}$. Define
$$\omega_{\lambda}: {\rm Ham}^{\lambda}_{{\rm graphon}} \times
{\rm Ham}^{\lambda}_{{\rm graphon}} \longrightarrow
C_{\lambda}^{{\rm graphon}}$$
\begin{equation} \label{symplectic-1}
\omega_{\lambda}(\theta,\theta'):= \sum_{i,j} u_{i} \circ_{\lambda}
v_{j} \circ_{\lambda} [f_{i},h_{j}]_{\lambda}
\end{equation}
such that $\{f_{1},...,f_{m},h_{1},...,h_{n}\} \subset
C_{\lambda}^{{\rm graphon}}$, $\{u_{1},...,u_{m},v_{1},...,v_{n}\}
\subset Z(C_{\lambda}^{{\rm graphon}})$ and
\begin{equation}
\theta=\sum_{i} u_{i} \circ_{\lambda} {\rm ham}(f_{i}), \ \ \
\theta'=\sum_{j}v_{j} \circ_{\lambda} {\rm ham}(h_{j}).
\end{equation}

The differential form $\omega_{\lambda}$ is a $Z(C_{\lambda}^{{\rm graphon}})$-bilinear
anti-symmetric non-degenerate closed 2-form in
$\Omega^{2}_{\lambda,{\rm graphon}}(C_{\lambda}^{{\rm graphon}})$. For a given $f \in C_{\lambda}^{{\rm graphon}}$ with the
corresponding symplectic vector field $\theta^{\lambda}_{f}$, we
have \begin{equation} \{f,g\}_{\lambda}:=
i_{\theta_{f}^{\lambda}}(d_{\lambda}g)
\end{equation}
such that
\begin{equation}
i_{\theta}(\omega_{0}d_{\lambda}\omega_{1}...d_{\lambda}\omega_{n})
= \sum_{j=1}^{n} (-1)^{j-1}\omega_{0} d_{\lambda} \omega_{1} ...
\theta(\omega_{j}) ... d_{\lambda} \omega_{n}
\end{equation}
is the super-derivation of degree -1. We can check that
\begin{equation}
\{f,g\}_{\lambda} = i_{\theta_{f}^{\lambda}}
i_{\theta_{g}^{\lambda}} \omega_{\lambda}.
\end{equation}
\end{proof}

\begin{thm} \label{dga-graphon-2}
The Dimensional Regularization in the BPHZ renormalization of
Feynman graphons (i.e. Theorem \ref{feynman-graphon-5} and Theorem
\ref{bphz-dse-graphon-2}) determines a noncommutative
symplectic geometry model for the Hopf algebra $\mathcal{S}^{\Phi}_{{\rm
graphon}}$.
\end{thm}

\begin{proof}
There exists a universal setting for the construction of a new Nijenhuis
algebra from a given commutative unital algebra $A_{{\rm
dr}}$. We present the product of formal series by $m(f,g)=[fg]$ and
consider the graded tensor module $T(A_{{\rm dr}}):= \bigoplus_{n
\ge 0} A_{{\rm dr}}^{\otimes n}$ generated by expressions such as
$f_{1} \otimes f_{2} \otimes ... \otimes f_{n}$. From now we name
each series in $A_{{\rm dr}}$ as a letter and each sequence
$U:=f_{1}f_{2}...f_{n}$ of letters as a word with the length $n$.
The empty word $e$ which has the length zero is the unit object in
$T(A_{{\rm dr}})$.

By induction we can define a new shuffle
product on $T(A_{{\rm dr}})$ given by
\begin{equation} \label{shuffle-1}
fU \circledcirc gV:= f(U \circledcirc gV) + g(fU \circledcirc V) -
e[fg](U \circledcirc V)
\end{equation}
which is unital and associative. The product (\ref{shuffle-1}) defines another new quasi-shuffle product on
$\overline{T}(A_{{\rm dr}}):= \bigoplus_{n \ge 1} A_{{\rm
dr}}^{\otimes n}$ given by
\begin{equation} \label{shuffle-2}
fU \circleddash gV:= [fg](U \circledcirc V)
\end{equation}
which is also unital and associative.

The linear map
$B^{+}_{e}$ on $\overline{T}(A_{{\rm dr}})$ sends each word
$U$ of length $n$ to the new word $eU$ of length $n+1$. Thanks to
investigations discussed in \cite{ebrahimifard-guo-1}, the triple
$(\overline{T}(A_{{\rm dr}}),\circleddash,B^{+}_{e})$ is th universal Nijenhuis
algebra in a category of
Nijenhuis algebras generated by the initial algebra $A_{{\rm dr}}$.

Now we can lift the linear map $B^{+}_{e}$ onto $L(\mathcal{S}^{\Phi}_{{\rm
graphon}},\overline{T}(A_{{\rm dr}}))$ to define the new Nijenhuis
map
\begin{equation} \label{nijenhuis-graphon-1}
\mathcal{N}_{{\rm graphon}}(\psi):= B^{+}_{e} \circ \psi.
\end{equation}
The resulting Nijenhuis algebra is the key tool for us to build a
new product $\circ_{u}$ on $L(\mathcal{S}^{\Phi}_{{\rm
graphon}},\overline{T}(A_{{\rm dr}}))$ defined by
\begin{equation} \label{nijenhuis-graphon-10}
\psi_{1} \circ_{u} \psi_{2}:= \mathcal{N}_{{\rm graphon}}(\psi_{1})
*_{\circleddash} \psi_{2} +
\psi_{1} *_{\circleddash} \mathcal{N}_{{\rm graphon}}(\psi_{2}) -
\mathcal{N}_{{\rm graphon}}(\psi_{1}*_{\circleddash} \psi_{2})
\end{equation}
such that $*_{\circleddash}$ is the convolution product with respect
to the coproduct $\Delta_{{\rm graphon}}$ on Feynman graphons
(\ref{cop-graphon-1}) and the product $\circleddash$. We have
\begin{equation}
\psi_{1}  *_{\circleddash} \psi_{2}([W_{\Gamma}])= \sum
\psi_{1}([W_{\Gamma'}])  \circleddash   \psi_{2}([W_{\Gamma''}]), \
\ \Delta_{{\rm graphon}}([W_{\Gamma}]) = \sum [W_{\Gamma'}] \otimes
[W_{\Gamma''}].
\end{equation}

The non-cocommutativity of the renormalization Hopf algebra of
Feynman graphons shows that the convolution product
$*_{\circleddash}$ and the new product $\circ_{u}$ are
noncommutative. In addition, the Nijenhuis property shows that
\begin{equation}
\mathcal{N}_{{\rm graphon}}(\psi_{1} \circ_{u} \psi_{2}) =
\mathcal{N}_{{\rm graphon}}(\psi_{1}) *_{\circleddash}
\mathcal{N}_{{\rm graphon}}(\psi_{2})
\end{equation}
which supports the associativity of this new product.

Set
\begin{equation}
C_{u}^{{\rm graphon}}:=(L(\mathcal{S}^{\Phi}_{{\rm
graphon}},\overline{T}(A_{{\rm dr}})), \circ_{u})
\end{equation}
as the unital associative noncommutative algebra generated by
Dimensional Regularization. In addition, the commutator with respect
to the product $\circ_{u}$ determines a new Lie bracket $[.,.]_{u}$ on
the space $L(\mathcal{S}^{\Phi}_{{\rm graphon}},\overline{T}(A_{{\rm
dr}}))$ given by
\begin{equation}
[\psi_{1},\psi_{2}]_{u}:= [\mathcal{N}_{{\rm
graphon}}(\psi_{1}),\psi_{2}] + [\psi_{1},\mathcal{N}_{{\rm
graphon}}(\psi_{2})] - \mathcal{N}_{{\rm
graphon}}[\psi_{1},\psi_{2}].
\end{equation}
This class of Lie brackets enables us to build a new noncommutative differential calculus on
$C_{u}^{{\rm graphon}}$ in terms of the following steps.

- Set ${\rm Der}^{u}_{{\rm graphon}}$ as the space of all derivations
on $C_{u}^{{\rm graphon}}$. It has all linear maps such as
$\theta:C_{u}^{{\rm graphon}} \longrightarrow C_{u}^{{\rm graphon}}$
which enjoys the Leibniz rule.

- The Lie bracket $[.,.]_{u}$ naturally
determines the Poisson bracket $\{.,.\}_{u}$ on $C_{u}^{{\rm
graphon}}$. For each $\phi \in C_{u}^{{\rm graphon}}$, its
corresponding Hamiltonian derivation is defined by
\begin{equation}
{\rm ham} (\phi): \psi \longmapsto \{\phi,\psi\}_{u}.
\end{equation}
Set ${\rm Ham}^{u}_{{\rm graphon}}$ as the $Z(C_{u}^{{\rm
graphon}})$-module generated by all Hamiltonian derivations on
$C_{u}^{{\rm graphon}}$.

- Define
\begin{equation}
\Omega^{\bullet}_{u,{\rm graphon}}(C_{u}^{{\rm graphon}}):=
(\bigoplus_{n \ge 0} \Omega^{n}_{u,{\rm graphon}}(C_{u}^{{\rm
graphon}}),d_{u})
\end{equation}
as the differential graded algebra on $C_{u}^{{\rm graphon}}$. For each $n \ge 1$,
$\Omega^{n}_{u,{\rm graphon}}(C_{u}^{{\rm graphon}})$ is the space
of all $Z(C_{u}^{{\rm graphon}})$-multilinear antisymmetric mappings
from ${\rm Ham}^{u}_{{\rm graphon}} \times ...^{n} \times {\rm
Ham}^{u}_{{\rm graphon}}$ into $C_{u}^{{\rm graphon}}$. The zero
component of this differential graded algebra is the initial algebra
$C_{u}^{{\rm graphon}}$. In addition, for each $\omega \in \Omega^{n}_{u,{\rm graphon}}(C_{u}^{{\rm
graphon}})$ and $\theta_{i} \in {\rm Ham}^{u}_{{\rm graphon}}$, the
anti-derivative degree one differential operator $d_{u}$ is defined
by
$$d_{u}\omega(\theta_{0},...,\theta_{n}):= \sum_{k=0}^{n} (-1)^{k}
\theta_{k} \omega(\theta_{0},...,\hat{\theta_{k}},...,\theta_{n})
+$$
\begin{equation}
\sum_{0 \le r < s \le n} (-1)^{r+s}
\omega([\theta_{r},\theta_{s}]_{u},\theta_{0},...,\hat{\theta_{r}},...,\hat{\theta_{s}},...,\theta_{n})
\end{equation}
such that we have $d_{u}^{2}=0$.

Thanks to this differential graded (Lie) algebraic machinery, we can determine a new class of symplectic structures
generated by the Lie bracket $[.,.]_{u}$. Define
$$\omega_{u}: {\rm Ham}^{u}_{{\rm graphon}} \times
{\rm Ham}^{u}_{{\rm graphon}} \longrightarrow C_{u}^{{\rm
graphon}}$$
\begin{equation} \label{symplectic-111}
\omega_{u}(\theta,\theta'):= \sum_{i,j} u_{i} \circ_{u} v_{j}
\circ_{u} [f_{i},h_{j}]_{u}
\end{equation}
such that $\{f_{1},...,f_{m},h_{1},...,h_{n}\} \subset C_{u}^{{\rm
graphon}}$, $\{u_{1},...,u_{m},v_{1},...,v_{n}\} \subset
Z(C_{u}^{{\rm graphon}})$,
\begin{equation}
\theta=\sum_{i} u_{i} \circ_{u} {\rm ham}(f_{i}), \ \
\theta'=\sum_{j}v_{j} \circ_{u} {\rm ham}(h_{j}).
\end{equation}

The differential form $\omega_{u}$ is a $Z(C_{u}^{{\rm graphon}})$-bilinear anti-symmetric
non-degenerate closed 2-form in $\Omega^{2}_{u,{\rm
graphon}}(C_{u}^{{\rm graphon}})$. For a given $f \in C_{u}^{{\rm graphon}}$ with the corresponding
symplectic vector field $\theta^{u}_{f}$, we have
\begin{equation}
\{f,g\}_{u}:= i_{\theta_{f}^{u}}(d_{u}g)
\end{equation}
such that
\begin{equation}
i_{\theta}(\omega_{0}d_{u}\omega_{1}...d_{u}\omega_{n}) =
\sum_{j=1}^{n} (-1)^{j-1}\omega_{0} d_{u} \omega_{1} ...
\theta(\omega_{j}) ... d_{u} \omega_{n}
\end{equation}
is the super-derivation of degree -1. We can check that
\begin{equation}
\{f,g\}_{u} = i_{\theta_{f}^{u}} i_{\theta_{g}^{u}} \omega_{u}.
\end{equation}
\end{proof}

The modified version of the Connes--Kreimer
Renormalization Group for Feynman graphons is defined by
Lemma \ref{renorm-group-1} where we should apply the filtration parameter
on Feynman graphons given by Theorem \ref{graphon-filtration-1}. Thanks to
the built noncommutative differential geometry on
$\mathcal{S}^{\Phi}_{{\rm graphon}}$, we can provide a new geometric
interpretation for the behavior of the Connes--Kreimer
Renormalization Group whenever it acts on large Feynman diagrams.

\begin{lem} \label{ren-group-large-diagram-1}
Let $\{F_{t}\}_{t}$ be the Renormalization Group on Feynman graphons defined by Lemma \ref{renorm-group-1}. For each $t$ and any large Feynman diagram $X$, $F_{t}(X)$ is the convergent limit of
the sequence $\{F_{t}(X_{n})\}_{n \ge 1}$ with respect to the
cut-distance topology.
\end{lem}

\begin{proof}
Consider the loop $\gamma_{\mu} \in {\rm Loop}(\mathbb{G}^{\Phi}_{{\rm
graphon}}(\mathbb{C}),\mu)$ which encodes the Feynman rules
characters in the renormalization Hopf algebra of Feynman graphons with respect to a given physical theory $\Phi$. For a given Dyson--Schwinger equation DSE with the unique solution
$X=\sum_{n\ge 0} X_{n}$, we have
\begin{equation}
\gamma_{\mu}(z)([W_{X}]):= U_{\mu}^{z}(X)
\end{equation}
such that $U_{\mu}^{z}(X)$ is a Laurent series as the regularized
large Feynman integral with respect to $X$. The one-parameter group
$\{\theta_{t}\}_{t \in \mathbb{C}}$ sends the unlabeled graphon
class $[W_{X}]$ to the filtration rank of the equation DSE (i.e. Theorem
\ref{graphon-filtration-1}). The resulting Renormalization Group
$\{F_{t}\}_{t}$ (i.e. Lemma \ref{renorm-group-1}) is a subgroup of
$\mathbb{G}^{\Phi}_{{\rm graphon}}(\mathbb{C})$ which means that for each
$t$, $F_{t}$ is a linear homomorphism. On the other hand, thanks to
Theorem \ref{feynman-graphon-4}, we know that the large Feynman
diagram $X$ is the convergent limit of the sequence of its partial
sums with respect to the cut-distance topology. Therefore we have
$$F_{t}(X) = F_{t}({\rm lim}_{m \rightarrow \infty} Y_{m}) =
F_{t}({\rm lim}_{m \rightarrow \infty} \sum_{n=1}^{m} X_{n})=$$
$${\rm lim}_{m \rightarrow \infty} \sum_{n=1}^{m} F_{t}(X_{n}) = {\rm
lim}_{m \rightarrow \infty} \sum_{n=1}^{m} {\rm lim}_{z \rightarrow
0} \gamma_{-}(z)(X_{n}) \theta_{tz}(\gamma_{-}^{-1}(z)(X_{n}))$$
\begin{equation}
= {\rm lim}_{m \rightarrow \infty} {\rm lim}_{z \rightarrow 0}
\sum_{n=1}^{m} \gamma_{-}(z)(X_{n})
\theta_{tz}(\gamma_{-}^{-1}(z)(X_{n}))
\end{equation}
such that according to Proposition 1.47 in \cite{connes-marcolli-1},
for each $t$, $F_{t}(X_{n})$ is a polynomial in $t$.
\end{proof}

\begin{cor}
The non-perturbative Connes--Kreimer Renormalization Group on Feynman graphons can
determine an infinite dimensional integrable system.
\end{cor}
\begin{proof}
We work on the unital associative noncommutative algebra
$C_{0}^{{\rm graphon}}:= (L(\mathcal{S}^{\Phi}_{{\rm
graphon}},A_{{\rm dr}}), \circ_{0})$ generated by the Minimal
Subtraction map for $\lambda=0$. Thanks to Theorem
\ref{dga-graphon-1}, consider the differential graded algebra
\begin{equation}
\Omega^{\bullet}_{0,{\rm graphon}}(C_{0}^{{\rm graphon}}):=
(\bigoplus_{n \ge 0} \Omega^{n}_{0,{\rm graphon}}(C_{0}^{{\rm
graphon}}),d_{0})
\end{equation}
with respect to the Lie bracket $[.,.]_{0}$. Each character $F_{t}$
of the Renormalization Group $\{F_{t}\}_{t}$ given by Lemma
\ref{ren-group-large-diagram-1} is an object in the algebra
$C_{0}^{{\rm graphon}}$. Therefore the motion integral equation with respect to
the character $F_{t_{0}}$ is given by the equation
\begin{equation}
\{f,F_{t_{0}}\}_{0} = 0
\end{equation}
such that $f \in
C_{0}^{{\rm graphon}}$. Thanks to the existence of a noncommutative symplectic
form $\omega_{0}$ on $C_{0}^{{\rm graphon}}$ with respect to the Lie
bracket $[.,.]_{0}$ (i.e. Theorem \ref{dga-graphon-1}), the motion
integral can be determined by the equation
\begin{equation}
\{f,F_{t_{0}}\}_{0} = i_{\theta^{0}_{F_{t_{0}}}} i_{\theta^{0}_{f}}
\omega_{0} = w_{0}(\theta^{0}_{F_{t_{0}}},\theta^{0}_{f}) =
[f,F_{t_{0}}] = 0.
\end{equation}
On the one hand, from the definition of the deformed Lie bracket
$[.,.]_{0}$ and the idempotent Rota--Baxter property of $(A_{{\rm
dr}},R_{{\rm ms}})$, we have
\begin{equation}
\{F_{t},F_{s}\}_{0} = [R_{{\rm ms}}(F_{t}),F_{s}] + [F_{t},R_{{\rm
ms}}(F_{s})] - R_{{\rm ms}}([F_{t},F_{s}]).
\end{equation}
On the other hand, for each $t$, $F_{t}([W_{\Gamma}])$ is a
polynomial in $t$ which means that $R_{{\rm
ms}}(F_{t}([W_{\Gamma}]))=0$ and in addition, for each $s,t$,
$F_{t}*F_{s}=F_{t+s}$. Thanks to these facts, we can observe that
for each $s,t$,
\begin{equation}
\{F_{t},F_{s}\}_{0} = 0.
\end{equation}
\end{proof}


\chapter{\textsf{A theory of functional analysis for large Feynman diagrams}}

\vspace{1in}

$\bullet$ \textbf{\emph{The Haar integration on $\mathcal{S}^{\Phi,g}$ and its application}} \\
$\bullet$ \textbf{\emph{The G\^{a}teaux differential calculus on $\mathcal{S}^{\Phi,g}$ and its application}} \\
$-$ \textbf{\emph{Feynman random graphs via homomorphism densities}} \\
$-$ \textbf{\emph{Differentiability}} \\

\newpage

This chapter aims to provide the foundations of a functional
analysis machinery for the study of large Feynman diagrams which
contribute to solutions of Dyson--Schwinger equations. We build an
integration theory and a differentiation theory for the functionals
on the space $\mathcal{S}^{\Phi,g}$ with respect to a given strongly coupled gauge field theory $\Phi$. The space $\mathcal{S}^{\Phi,g}$ can be embedded into the Hopf
algebra $H_{{\rm FG}}^{{\rm cut}}(\Phi)$ of (large) Feynman diagrams
topologically completed with the cut-distance topology. The space $H_{{\rm
FG}}^{{\rm cut}}(\Phi)$ consists of all Feynman diagrams and their
corresponding finite or infinite formal expansions where solutions
of all non-perturbative Dyson--Schwinger equations belong to the
boundary region of this compact topological Hopf algebra. As we have shown in the previous parts, this
enriched Hopf algebra of Feynman diagrams can be encoded in terms of the renormalization
Hopf algebra $\mathcal{S}^{\Phi}_{{\rm graphon}}$ of Feynman
graphons.

At the first step, we build a measure theory on $\mathcal{S}^{\Phi,g}$ where we equip this space with a new
topological group structure which leads us to a new Haar measure
integration theory for functionals on large Feynman diagrams. Then
we deal with some applications of the resulting measure space where
a new generalization of the classical Johnson--Lapidus Dyson series
for large Feynman diagrams will be obtained. In addition, we work on
the construction of a new Fourier transformation machinery on the
Banach algebra $L^{1}(\mathcal{S}^{\Phi,g}, \mu_{{\rm Haar}})$ which
enables us to describe the evolution of large Feynman diagrams on
the basis of their corresponding partial sums under a functional
setting. At the second step, we concern the G\^{a}teaux
differentiability of real valued functionals on $\mathcal{S}^{\Phi}_{{\rm graphon}}$ where we obtain Taylor expansion representations for
these functionals under some conditions.

Achievements of these steps can be adapted for the level of the topological Hopf algebra $H_{{\rm FG}}^{{\rm cut}}(\Phi)$. In other words, the promising differential calculus and integration theory enable us to
describe the dynamics of topological regions of Feynman diagrams on
the basis of the behavior of functionals with respect to the built
Haar integration theory and G\^{a}teaux differentiation theory on
Feynman graphons.

\section{\textsl{The Haar integration on $\mathcal{S}^{\Phi,g}$ and its application}}

For a given physical theory $\Phi$ with strong coupling constant $g
\ge 1$, set $V(\Phi)$ as the set of all vertices which appear in
Feynman diagrams and their corresponding formal expansions as
interactions among elementary particles. This infinite countable set
allows us to count interactions independent of their physical types.
Set $K_{V(\Phi)}$ as the complete graph with $V(\Phi)$ as the set of
vertices and all possible edges among these vertices except
self-loops. The collection $\{0,1\}^{K_{V(\Phi)}}$, as the family of
all functions from $K_{V(\Phi)}$ to $\{0,1\}$, allows us to
characterize Feynman diagrams which contribute to physical theory
$\Phi$. Therefore, each (large) Feynman diagram $\Gamma$ can be
determined by its corresponding characteristic function
$\chi_{\Gamma}$ which sends vertices $v \in K_{V(\Phi)}$ to $1$ if
$v \in \Gamma$ and sends other vertices to $0$. For each edge $e \in K_{V(\Phi)}$, if $V_{e}$ be the set of vertices
in $V(\Phi)$ which are attached to the edge $e$, then the infinite
Cartesian product $\times_{e \in K_{V(\Phi)}}
\mathcal{P}(K_{V_{e}})$ has enough vertices and edges to contain all
Feynman graphs in $H_{{\rm FG}}^{{\rm cut}}(\Phi)$.

\begin{lem} \label{topo-group-graph-1}
$\mathcal{S}^{\Phi,g}$ can be equipped with an abelian compact
Hausdorff topological group structure.
\end{lem}

\begin{proof}
For a given Dyson--Schwinger DSE in $\mathcal{S}^{\Phi,g}$ with the unique solution $X_{{\rm DSE}}$, we can associate the characteristic function $\chi_{X_{{\rm DSE}}} \in \{0,1\}^{K_{V(\Phi)}}$ to DSE. The function $\chi_{X_{{\rm DSE}}}$ allows us to identify vertices and edges which
contribute to the large Feynman diagram $X_{{\rm DSE}}$. It means that we can embed the collection $\mathcal{S}^{\Phi,g}$
into $\{0,1\}^{K_{V(\Phi)}}$ which is useful to define new addition
and multiplication operators on Dyson--Schwinger equations in terms of
the pointwise addition and multiplication of their corresponding
characteristic functions. These operators provide a vector space structure generated by $X_{{\rm DSE}}$ for each DSE where as the result, we have a commutative
$\mathbb{Z}_{2}$-algebra structure on $\mathcal{S}^{\Phi,g}$.

In addition, we can also define a new binary operation on $\mathcal{S}^{\Phi,g}$
in terms of the symmetric difference operator
\begin{equation}
({\rm DSE}_{1}, {\rm DSE}_{2}) \longmapsto X_{{\rm DSE}_{1}} \triangle
X_{{\rm DSE}_{2}}.
\end{equation}
By adding the empty graph $\mathbb{I}$ to $\mathcal{S}^{\Phi,g}$ as the zero element, the pair $(\mathcal{S}^{\Phi,g},\triangle)$ is an abelian
group which can be equipped with a compatible topology to
obtain a compact topological group. For this purpose,
suppose $\alpha$ be a bijection between $K_{V(\Phi)}$ and the set of
natural numbers $\mathbb{N}$. For the fixed coupling constant $g \ge
1$ and each $\epsilon>0$, define a new map  $d_{g,\alpha,\epsilon}:
\mathcal{S}^{\Phi,g} \times \mathcal{S}^{\Phi,g} \longrightarrow
[0,\infty)$ given by
\begin{equation} \label{metric-graph-1}
d_{g,\alpha,\epsilon}({\rm DSE}_{1}, {\rm DSE}_{2}):= \sum_{e \in
X_{{\rm DSE}_{1}} \triangle
X_{{\rm DSE}_{2}}} (g+\epsilon)^{-\alpha(e)}
\end{equation}
such that the sum is taken over vertices such as $e$ which belongs to
only one of the large Feynman diagrams $X_{{\rm DSE}_{1}}$ or $X_{{\rm DSE}_{2}}$. The map $d_{g,\alpha,\epsilon}$ is a translation invariant metric such that
$d_{g,\alpha,\epsilon_{1}}$ and $d_{g,\alpha,\epsilon_{2}}$ have the
equivalent topology.

The space $\mathcal{S}^{\Phi,g}$ together
with the symmetric difference operator and the topology generated by
the metric $d_{g,\alpha,\epsilon}$ is a compact Hausdorff abelian
topological group.
\end{proof}

Thanks to the translation-invariant metric $d_{g,\alpha,\epsilon}$
defined by Lemma \ref{topo-group-graph-1}, for each equation DSE in $\mathcal{S}^{\Phi,g}$ with the corresponding large Feynman
diagram $X_{{\rm DSE}}$ define
\begin{equation} \label{norm-graph-1}
\parallel X_{{\rm DSE}} \parallel_{g,\alpha,\epsilon}:=
d_{g,\alpha,\epsilon}(\mathbb{I},X_{{\rm DSE}}).
\end{equation}

In this setting, a
sequence $\{\Gamma_{n}\}_{n \ge 1}$ of large Feynman diagrams in
$\mathcal{S}^{\Phi,g}$ is convergent to a unique large Feynman
diagram $\Gamma$, if each indicator sequence $\{\textbf{1}_{e \in
\Gamma_{n}}\}_{n \ge 1}$ converges to the indicator $\textbf{1}_{e
\in \Gamma}$ for any $e \in K_{V(\Phi)}$.

\begin{thm} \label{measure-infinite-graph-1}
The topological group $\mathcal{S}^{\Phi,g}$ can be equipped with the Haar measure $\mu_{{\rm Haar}}$.
\end{thm}

\begin{proof}
Lemma \ref{topo-group-graph-1} supports the existence of the unique
Haar measure $\mu_{{\rm Haar}}$ on $\mathcal{S}^{\Phi,g}$ originated
from the compact topological structure. We build this measure which is actually of the type
Bernoulli probability.

Consider the product $\sigma$-algebra $\sum_{{\rm prod}}$ on
$\mathcal{S}^{\Phi,g}$ generated by cylinder sets
\begin{equation} \label{sigma-algbe-1}
S_{\Gamma_{0}}:= \times_{\Gamma \neq \Gamma_{0}} \{\mathbb{I},
\{\Gamma\}\} \times \{\Gamma_{0}\}
\end{equation}
for each large Feynman diagram $\Gamma_{0}$ corresponding to an equation ${\rm DSE}_{0}$ in
$\mathcal{S}^{\Phi,g}$. For each large Feynman diagram $\Gamma$, the characteristic function $\chi_{\Gamma}$ is useful to see $\Gamma$ as an infinite
countable subset of vertices in $K_{V(\Phi)}$ which contribute to the unique
solution of the equation ${\rm DSE}_{\Gamma}$. Therefore
each function $P \in \{0,1\}^{K_{V(\Phi)}}$ can identify a new
function $\tilde{P}: \mathcal{S}^{\Phi,g} \rightarrow [0,1]$ which can be applied to
define the measure $\mu_{\tilde{P}}$ on the $\sigma$-algebra $\sum_{{\rm prod}}$ in terms of the following steps.

- For finite intersections of cylinder sets $S_{\Gamma_{1}}, ...,
S_{\Gamma_{n}}$, we have
\begin{equation}
\mu_{\tilde{P}}(S_{\Gamma_{1}} \cap S_{\Gamma_{2}} \cap ... \cap
S_{\Gamma_{n}}) = \prod_{i=1}^{n} \tilde{P}(\Gamma_{i})
\end{equation}
for large Feynman diagrams $\Gamma_{1}, ..., \Gamma_{n}$ as solutions of Dyson--Schwinger equations ${\rm DSE}_{1},...,{\rm DSE}_{n}$ in $\mathcal{S}^{\Phi,g}$.

- The function $\mu_{\tilde{P}}$ is
a probability measure on $\mathcal{S}^{\Phi,g}$ which can be presented with the general form
\begin{equation}
\mu_{\tilde{P}}:= \prod_{X \in \mathcal{S}^{\Phi,g}}
\mu_{\tilde{P},X}.
\end{equation}

Now we need to show that the measure $\mu_{\tilde{P}}$ is the Haar
measure. In other words, we claim that the Haar measure is
equal with $\mu_{\tilde{P}}$ where $P(X)=1/2$ for each large Feynman
diagram $X \in \mathcal{S}^{\Phi,g}$.

For subsets $Z_{1},Z_{2}$ of $K_{V(\Phi)}$, define
\begin{equation}  \label{sigma-algbe-2}
I(Z_{1},Z_{2}):= \{{\rm DSE} \in \mathcal{S}^{\Phi,g}: Z_{1} \subset
X_{{\rm DSE}}, \ \ Z_{2} \subset K_{V(\Phi)} \backslash X_{{\rm
DSE}} \}
\end{equation}
and then consider the $\sigma$-algebra $\sum_{\mathcal{I}}$
generated by all sets $I(Z_{1},Z_{2})$. The $\sigma$-algebra
$\sum_{\mathcal{I}}$ is the same as the $\sigma$-algebra generated
by all sets $I(Z_{1},Z_{2})$ for disjoint sets $Z_{1},Z_{2}$. In
addition, we have $S_{\Gamma} = I(\mathbb{I},\{\Gamma\})$ which
leads us to show that $\sum_{{\rm prod}} = \sum_{\mathcal{I}}$.

Thanks to some standard methods in Analysis \cite{reed-simon-1}, we can determine the unique
translation-invariant probability measure $\mu_{{\rm Haar}}$ on the compact topological group
$\mathcal{S}^{\Phi,g}$. For a given
large Feynman diagram $X$ and a subset $Z$ of $K_{V(\Phi)}$, define
\begin{equation}
Z+X:= \{\gamma \sqcup X: \ \ \gamma \in Z \}.
\end{equation}
Then we can show that
\begin{equation}
I(Z_{1},Z_{2}) = I(\mathbb{I}, Z_{1} \sqcup Z_{2}) + Z_{2}.
\end{equation}
Thanks to this fact, for given large Feynman diagrams $\Gamma_{1},
\Gamma_{2}$, set $\Gamma=\Gamma_{1} \sqcup \Gamma_{2}$. Then we have
\begin{equation}
\mu_{{\rm Haar}}(I(\Gamma_{1},\Gamma_{2})) = \mu_{{\rm
Haar}}(I(\Gamma_{1},\Gamma_{2}) + \Gamma_{2}) =  \mu_{{\rm Haar}}
(I(\mathbb{I},\Gamma))
\end{equation}
which informs the translation-invariance.

In general, if $(\Omega,A)$ ba $\sigma$-algebra generated by a
subset $C \subset A$ which is closed under finite intersections,
then two probability measures on $A$ are equal if and only if they
agree on $C$. If $C$ has an algebraic structure which is equipped by
a probability measure $\mu$, then we can extend $\mu$ to a unique
measure on $A$ \cite{reed-simon-1}. Now let the
subset $Z$ of $K_{V(\Phi)}$ can determine a finite number of large
Feynman diagrams $\Gamma_{1},...,\Gamma_{n}$. Then as
the set we have
\begin{equation}
\mathcal{S}^{\Phi,g} = \bigsqcup_{Z_{0} \subset Z} I(Z_{0},Z\backslash
Z_{0})
\end{equation}
which can be used to show that $\mu_{{\rm Haar}}$ agrees with $\mu_{1/2}$ on $\sigma_{\mathcal{I}}$. In other words,
\begin{equation}
\mu_{{\rm Haar}}(I(Z_{1},Z_{2})) = 2^{-n} =
\mu_{1/2}(I(Z_{1},Z_{2})).
\end{equation}

It is also possible to check that $\sigma_{\mathcal{I}}$ is equal with the
Borel $\sigma$-algebra generated by all open sets in
$\mathcal{S}^{\Phi,g}$ with respect to the metric
$d_{g,\alpha,\epsilon}$. For a given bijection $\alpha: K_{V(\Phi)}
\rightarrow \mathbb{N}$, set
\begin{equation}
E_{n}(\alpha):= \{e \in K_{V(\Phi)}: \ \ \alpha(e) \le n\}.
\end{equation}
For each large Feynman diagram $X$ in $\mathcal{S}^{\Phi,g}$, set
$B(X,r,\parallel . \parallel_{g,\alpha,\epsilon})$ as the open ball
in $(\mathcal{S}^{\Phi,g},d_{g,\alpha,\epsilon})$ with the center $X$
and the radius $r$. In addition, for given disjoint finite subsets
$N_{1}, N_{2} \subset \mathbb{N}$ such that $N_{1} \cup N_{2} =
\{1,2,...,n\}$, define
\begin{equation}
I(N_{1},N_{2}):= I(\alpha^{-1}(N_{1}),\alpha^{-1}(N_{2})).
\end{equation}
Then we have
\begin{equation}
I(N_{1},N_{2}) = B(\alpha^{-1}(N_{2}),2^{-n},\parallel .
\parallel_{g,\alpha,2}) \bigcup B(K_{V(\Phi)} \backslash
\alpha^{-1}(N_{1}),2^{-n},\parallel . \parallel_{g,\alpha,2})
\end{equation}
such that $B(\alpha^{-1}(N_{i}),2^{-n},\parallel .
\parallel_{g,\alpha,2})$ is the open ball in
$(\mathcal{S}^{\Phi,g},d_{g,\alpha,g+\epsilon=2})$ with the center
$\alpha^{-1}(N_{i})$ and the radius $2^{-n}$.

Now a large Feynman diagram $X \in \mathcal{S}^{\Phi,g}$ can be
described as the convergent limit of the sequence $\{\Gamma_{n}\}_{n
\ge 1}$ such that $\Gamma_{n}:= X \cap E_{n}(\alpha)$.

Furthermore, we know that $\mathcal{S}^{\Phi,g}$ is a compact
Hausdorff topological group (i.e. Lemma \ref{topo-group-graph-1}).
Therefore $K \subset \mathcal{S}^{\Phi,g}$ is compact iff
$\mathcal{S}^{\Phi,g} \backslash K$ is open. It shows that the
$\sigma$-algebra generated by all compact sets is the same as the
Borel $\sigma$-algebra generated by all open sets. As the result, $\mu_{{\rm Haar}}$ on $\sigma_{{\rm prod}}$
determines uniquely the Haar measure. Therefore $\mu_{{\rm
Haar}}=\mu_{1/2}$.
\end{proof}

\begin{thm} \label{measure-infinite-graph-2}
For a given bijection $\alpha$, the Haar measure of any ball of the
radius $0 \le r \le 1$ in the normed vector space
$(\mathcal{S}^{\Phi,g},\parallel.\parallel_{g,\alpha,g+\epsilon=2})$
is $r$.
\end{thm}

\begin{proof}
The proof is a direct result of Theorem
\ref{measure-infinite-graph-1} and the proof of Theorem A in
\cite{khare-rajaratnam-1}.
\end{proof}

The resulting measure space enables us to initiate an
integration theory on the family of (large) Feynman diagrams which contribute to solutions of Dyson--Schwinger equations of a given strongly coupled gauge field theory. This integration theory
can be interpreted in the context of the Riemann--Lebesgue
integration theory on the measure space $(\mathbb{R},
\mathcal{B}(\mathbb{R}))$.

\begin{thm} \label{measure-feynman-1}
The integration theory on the measure space
$(\mathcal{S}^{\Phi,g},\mu_{{\rm Haar}})$ can be formulated in terms of the
Riemann--Lebesgue integration theory on real numbers with respect to
the Borel $\sigma$-algebra generated by all open sets.
\end{thm}

\begin{proof}
Thanks to the structure of the topological group
$\mathcal{S}^{\Phi,g}$ (i.e. Lemma \ref{topo-group-graph-1}) where the
norm $\parallel . \parallel_{g,\alpha,2}$ (given by (\ref{norm-graph-1})) and
the Haar measure $\mu_{{\rm Haar}}$ (given by Theorem
\ref{measure-infinite-graph-1}) are defined on large Feynman
diagrams, we can adapt the proofs of Lemma 3.22 and Proposition 3.23
in \cite{shojaeifard-9} for large Feynman diagrams to obtain the
following results.

(i) We can show that the norm $\parallel .
\parallel_{g,\alpha,2}$ (as a real valued function on
$\mathcal{S}^{\Phi,g}$) is the Haar measure-preserving map.

(ii) We can show that for any Lebesgue integrable real valued
function $f$ on $[0,1]$,
\begin{equation}
\mathbb{E}_{\mu_{{\rm Haar}}}[f(\parallel . \parallel_{g,\alpha,2})]
= \int_{0}^{1} f(x)dx.
\end{equation}

(iii) We can show that for any integrable real valued function $h$
on $\mathcal{S}^{\Phi,g}$,
\begin{equation} \label{graph-function-1}
\mathbb{E}_{\mu_{{\rm Haar}}}[h] = \int_{0}^{1} h((\parallel .
\parallel_{g,\alpha,2})^{-1}(x)) dx.
\end{equation}

Therefore the Haar integration theory on
$(\mathcal{S}^{\Phi,g},\mu_{{\rm Haar}})$ can be described by
transferring the Riemann--Lebesgue integration theory from the unit
interval to the space of large Feynman diagrams.
\end{proof}

The rest of this section provides some applications of this
integration theory for functionals on the measure space
$(\mathcal{S}^{\Phi,g},\mu_{{\rm Haar}})$.

Consider a single quantum particle which moves in a given potential
such that its behavior can be studied by a class of functionals on
$C[0,t]$ given by
\begin{equation}
Z(y):= {\rm exp}\{\int_{(0,t)} \theta(s,y(s))ds\}
\end{equation}
where the complex valued function $\theta$ on $[0,t] \times
\mathbb{R}^{n}$ is the given potential. This formulation is on the
basis of the standard Lebesgue--Stieltjes measure while under some
conditions it is possible to formulate these functionals with
respect to other complex Borel measures. It has been shown that for
each complex number with positive real part $\lambda$, the operators
$K_{\lambda}(Z_{n})$ exist for each $n$ such that $Z_{n}(y):=
(\int_{(0,t)} \theta(s,y(s))d\eta)^{n}$ and
$K_{\lambda}(Z)=\sum_{n\ge0} a_{n}K_{\lambda}(Z_{n})$. The central
motivation of this formulation was to deal with Feynman's
operational calculus in QED and other quantum theories.
\cite{johnson-lapidus-1}

A modification of the Johnson--Lapidus Dyson series for a measure
space of graphs which contribute to the topological Hopf
algebra $H_{{\rm FG}}^{{\rm cut}}(\Phi)$ has been obtained in
\cite{shojaeifard-9}. Now we want to apply the Haar integration theory on the
topological group $\mathcal{S}^{\Phi,g}$ to formulate
the Johnson--Lapidus Dyson series on $\mathcal{S}^{\Phi,g}$. This class of series allows us to explain the evolution of (strongly coupled) Dyson--Schwinger equations in terms of sequences of partial sums or sequences of large Feynman diagrams corresponding to weakly coupled Dyson--Schwinger equations.

\begin{thm} \label{evolution-1}
Let $\theta$ be a complex valued function on $\mathcal{S}^{\Phi,g}
\times \mathbb{R}^{2}$ and $v(z)=\sum_{n\ge0} a_{n}z^{n}$ with the
radius of convergence strictly grater than
$||\theta||_{\infty;\mu_{{\rm Borel}}}$. For a functional $Z$ on the
measure space $L^{1}(\mathcal{S}^{\Phi,g},\mu_{{\rm Haar}})$ of all
complex valued $\mu_{{\rm Haar}}$-integrable functions on
$\mathcal{S}^{\Phi,g}$ given by
\begin{equation}
Z(F):=v(\int_{\mathcal{S}^{\Phi,g}} \theta(X,F(X))d\mu_{{\rm
Borel}})
\end{equation}
, there exists a family of operators $\{K_{\lambda}(Z_{n})\}_{n \in
\mathbb{N}}$ such that

- parameters $\lambda$ are complex numbers with positive
real parts,

- $Z_{n}(F):= (\int_{\mathcal{S}^{\Phi,g}}
\theta(X,F(X))d\mu_{{\rm Borel}})^{n}$,

- $K_{\lambda}(Z)= \sum_{n\ge0} a_{n}K_{\lambda}(Z_{n})$.
\end{thm}

\begin{proof}
Theorem \ref{measure-infinite-graph-1} and Theorem
\ref{measure-feynman-1} enable us to understand the Haar integration
theory on $\mathcal{S}^{\Phi,g}$ in terms of the
Riemann--Lebesgue integration theory for real valued functions on
the closed interval. In addition, we have discussed the equivalence
between the product $\sigma$-algebra $\sum_{{\rm prod}}$ on
cylinders determined by large Feynman diagrams and the Borel
$\sigma$-algebra of open balls with respect to the norm $\parallel .
\parallel_{g,\alpha,2}$. It is useful to determine uniquely the Borel
measure $\mu_{{\rm Borel}}$ on $\mathcal{S}^{\Phi,g}$ corresponding
to the Haar measure $\mu_{{\rm Haar}}$ (i.e. Theorem
\ref{measure-infinite-graph-1}). Therefore we can now extend the classical Johnson-Lapidus Dyson series to the
level of the Haar measure $\mu_{{\rm Haar}}$ on
$\mathcal{S}^{\Phi,g}$.

In addition, we have shown the existence of a compact Hausdorff topological group
structure on $\mathcal{S}^{\Phi,g}$. Thanks to standard methods in Analysis
(\cite{reed-simon-1}), it is easy to show that the topological space
$C_{c}(\mathcal{S}^{\Phi,g})$ consisting of continuous functions on
$\mathcal{S}^{\Phi,g}$ with compact support is dense in
$L^{1}(\mathcal{S}^{\Phi,g},\mu_{{\rm Haar}})$. Apply
(\ref{graph-function-1}) to transfer the Haar measure integral
\begin{equation}
\mathbb{E}_{\mu_{{\rm Haar}}}[h] = \int_{\mathcal{S}^{\Phi,g}}h(X)
d\mu_{{\rm Haar}}
\end{equation}
of each $h \in C_{c}(\mathcal{S}^{\Phi,g})$ to its corresponding
Riemann--Lebesgue integral. It remains only to lift the proof of the classical Johnson-Lapidus
generalized Dyson series given in \cite{johnson-lapidus-1} onto
$L^{1}(\mathcal{S}^{\Phi,g},\mu_{{\rm Haar}})$.
\end{proof}

\begin{cor} \label{evolution-2}
(i) The Johnson--Lapidus Dyson series can describe the behavior of a
combinatorial Dyson--Schwinger equation in a given potential.

(ii) The functionals $K_{\lambda}(Z)$ (determined by Theorem
\ref{evolution-1}) enable us to describe the evolution of each large
Feynman diagram $X$ in terms of a sequence of Dyson--Schwinger
equations in $\mathcal{S}^{\Phi,g}$.
\end{cor}

\begin{proof}
(i) Suppose $X_{{\rm DSE}}(g)= \sum_{n \ge 0} g^{n}X_{n}$ is a large Feynman diagram as the unique
solution of an equation DSE in the normed vector space
$(\mathcal{S}^{\Phi,g},\parallel .
\parallel_{g,\alpha,2})$. Thanks to the Hahn--Banach Theorem
(\cite{reed-simon-1}), there exists a continuous linear map
$\psi_{{\rm DSE}}: \mathcal{S}^{\Phi,g} \longrightarrow \mathbb{R}$
such that
\begin{equation}
\psi_{{\rm DSE}}(X_{{\rm DSE}}) = \parallel X_{{\rm DSE}}
\parallel_{g,\alpha,2}, \ \ \ \parallel \psi_{{\rm DSE}} \parallel
\le 1
\end{equation}
where the operator norm $\parallel \psi_{{\rm DSE}} \parallel$ is
defined by
\begin{equation}
\parallel \psi_{{\rm DSE}} \parallel:= {\rm inf} \{c \ge 0: \ |\psi_{{\rm DSE}}(X)| \le c \parallel
X \parallel_{g,\alpha,2}, \ \forall X \in \mathcal{S}^{\Phi,g}\}.
\end{equation}
Now apply Theorem \ref{evolution-1} for $\psi_{{\rm DSE}} \in
L^{1}(\mathcal{S}^{\Phi,g},\mu_{{\rm Haar}})$.

(ii) For any given Dyson--Schwinger equation DSE, if we apply the multi-scale
Renormalization Group given by Theorem \ref{multi-scale-RG}, then we can build a sequence $\{{\rm
DSE}_{n}(\frac{n}{n+1}g)\}_{n \ge 1}$ of Dyson--Schwinger equations under rescaled values of the bare coupling constant
where we have
\begin{equation}
X_{{\rm DSE}_{n}}(\frac{n}{n+1}g) \subset X_{{\rm
DSE}_{n+1}}(\frac{n+1}{n+2}g).
\end{equation}
For each $n$, set $\chi^{n}_{{\rm DSE}}$ as the characteristic
function with respect to the large Feynman diagram $X_{{\rm
DSE}_{n}}(\frac{n}{n+1}g)$ on $\mathcal{S}^{\Phi,g}$ such that
$\chi^{n}_{{\rm DSE}} \in L^{1}(\mathcal{S}^{\Phi,g},\mu_{{\rm
Haar}})$. Now apply Theorem \ref{evolution-1} to the sequence
$\{\chi^{n}_{{\rm DSE}}\}_{n \ge 1}$ to obtain a description for the
evolution of $X_{{\rm DSE}}(g)$ in terms of large sub-graphs. A free
evolution from $X_{{\rm DSE}}(0)=\mathbb{I}$ (i.e. the empty graph) to
$X_{{\rm DSE}_{1}}(\frac{1}{2}g)$, interactions of particles in $X_{{\rm
DSE}_{1}}(\frac{1}{2}g)$ with the potential $\theta$, free evolution
from $X_{{\rm DSE}_{1}}(\frac{1}{2}g)$ to $X_{{\rm DSE}_{2}}(\frac{2}{3}g)$,
and so on up to $n^{\rm th}$ integration with $\theta$ at the level
$X_{{\rm DSE}_{n}}(\frac{n}{n+1}g)$ followed by a free evolution from
$X_{{\rm DSE}_{n}}(\frac{n}{n+1}g)$ to $X_{{\rm DSE}}(g)$ when $n$ tends
to infinity.
\end{proof}

\begin{lem}
Thanks to the symmetric difference as a binary operation on large
Feynman diagrams, there exists a complex commutative Banach algebra
structure on $L^{1}(\mathcal{S}^{\Phi,g},\mu_{{\rm Haar}})$.
\end{lem}

\begin{proof}
We have seen that the binary operation $\triangle$ determines an
abelian compact Hausdorff topological group structure on
$\mathcal{S}^{\Phi,g}$ such that the empty graph $\mathbb{I}$ is the
zero element of this group. Now define the following convolution
product on $L^{1}(\mathcal{S}^{\Phi,g},\mu_{{\rm Haar}})$
$$F_{1}*_{\triangle}F_{2}(\Gamma_{1})$$
\begin{equation} \label{conv-banach-1}
= \int_{\mathcal{S}^{\Phi,g}} F_{1}(\Gamma_{2})
F_{2}(\Gamma_{2}^{-1} \triangle \Gamma_{1}) d\mu_{{\rm
Haar}}(\Gamma_{2}), \ \ F_{1},F_{2} \in
L^{1}(\mathcal{S}^{\Phi,g},\mu_{{\rm Haar}})
\end{equation}
such that  $\Gamma_{2}^{-1}$ is the inverse of the graph with
respect to the group structure $\triangle$. The compatibility
between the product (\ref{conv-banach-1}) and the $L^{1}$-norm
\begin{equation}
\parallel F \parallel_{1} := \int_{\mathcal{S}^{\Phi,g}} \mid F(X) \mid d \mu_{{\rm
Haar}} (X)
\end{equation}
provides our promising Banach algebra. The abelian property of the
group $(\mathcal{S}^{\Phi,g}, \triangle)$ guarantees the
commutativity of this Banach algebra. In addition, we add the
infinitesimal delta function $\delta$ as the multiplicative unit for
this Banach algebra. Then we have
\begin{equation}
\int_{\mathcal{S}^{\Phi,g}} F(X) \delta(X) d\mu_{{\rm Haar}}(X) =
F(\mathbb{I})
\end{equation}
for each $F\in L^{1}(\mathcal{S}^{\Phi,g},\mu_{{\rm Haar}})$ and
each large Feynman diagram $X \in \mathcal{S}^{\Phi,g}$.
\end{proof}

\begin{prop}
(i) Each functional $F \in L^{1}(\mathcal{S}^{\Phi,g},\mu_{{\rm
Haar}})$ has a non-empty spectrum.

(ii) The space $\Omega(L^{1}(\mathcal{S}^{\Phi,g},\mu_{{\rm
Haar}}))$ of all characters of the complex Banach algebra
$L^{1}(\mathcal{S}^{\Phi,g},\mu_{{\rm Haar}})$ is a compact
Hausdorff topological space.
\end{prop}

\begin{proof}
We need only to adapt the standard procedures in Functional Analysis
(\cite{reed-simon-1}) to achieve the results.

(i) We show that
\begin{equation}
{\rm sp}(F):= \{\lambda \in \mathbb{C}: F - \lambda \delta \ \ {\rm
not \ invertible}\}
\end{equation}
is non-empty. If $F=0$, then thanks to the definition of the
infinitesimal delta function, we have the result. If $F$ be a
non-zero functional, suppose its spectrum is empty which means that
the new function
\begin{equation}
 R:\mathbb{C} \longrightarrow
L^{1}(\mathcal{S}^{\Phi,g},\mu_{{\rm Haar}}), \ \ \lambda
\longmapsto (F- \lambda \delta)^{-1}
\end{equation}
is well-defined, holomorphic, non-constant and bounded.
For any bounded linear functional $\Upsilon$ on
$L^{1}(\mathcal{S}^{\Phi,g},\mu_{{\rm Haar}})$, define a new function
$\tilde{\Upsilon}$ on $\mathbb{R}^{2}$ given by
\begin{equation}
\tilde{\Upsilon}(x,y):=\Upsilon(R(xe^{iy})).
\end{equation}
We can show that $\tilde{\Upsilon}$ is continuously differentiable
with respect to variables $x$ and $y$. Now by differentiation under
the integral sign from the holomorphic bounded function $K(x):=
\int_{0}^{2\pi} \tilde{\Upsilon}(x,y)dy$, we have $K'(x)=0$.
Therefore $K$ is a constant function which is a contradiction with
the initial assumption.

(ii) The ideal generated by kernel of any
character provides a natural correspondence between the set of
maximal ideals of the Banach algebra
$L^{1}(\mathcal{S}^{\Phi,g},\mu_{{\rm Haar}})$ and the set of
characters on the space $L^{1}(\mathcal{S}^{\Phi,g},\mu_{{\rm
Haar}})$.

Each character $\psi \in \Omega(L^{1}(\mathcal{S}^{\Phi,g},\mu_{{\rm
Haar}}))$ is actually an algebra homomorphism from
$L^{1}(\mathcal{S}^{\Phi,g},\mu_{{\rm Haar}})$ to $\mathbb{C}$ such
that $\psi(\delta)=1$. We can show that $\psi$ is continuous of norm
$1$, otherwise there exists a function $F\in
L^{1}(\mathcal{S}^{\Phi,g},\mu_{{\rm Haar}})$ such that $||F||<1$
and $\psi(F)=1$. Apply the convolution product to define $G:=\sum_{n
\ge 1} G^{n}$. From the equation $G=F+FG$ we have
\begin{equation}
\psi(G)=\psi(F) + \psi(F)\psi(G)=1+\psi(G)
\end{equation}
which shows a contradiction. So the norm of $\psi$ is less than or
equal to $1$ and $\psi(\delta)=1$ which implies that $||\psi||=1$.
Thanks to this fact, we can see that
$\Omega(L^{1}(\mathcal{S}^{\Phi,g},\mu_{{\rm Haar}}))$ is a closed
subset of the unit ball of the dual space
$L^{1}(\mathcal{S}^{\Phi,g},\mu_{{\rm Haar}})^{*}$ which is a
compact Hausdorff space with respect to the weak-$\star$ topology. As the
consequence, $\Omega(L^{1}(\mathcal{S}^{\Phi,g},\mu_{{\rm Haar}}))$
is a compact Hausdorff topological space.
\end{proof}

Thanks to the Gelfand transform, define
\begin{equation} \label{gelfand-1}
L^{1}(\mathcal{S}^{\Phi,g},\mu_{{\rm Haar}}) \longrightarrow
C_{0}(\Omega(L^{1}(\mathcal{S}^{\Phi,g},\mu_{{\rm Haar}})))
\end{equation}
$$F \longmapsto \widetilde{F}, \ \ \widetilde{F}(\psi):=\psi(F).$$
It is a norm decreasing algebraic homomorphism such that its image
separates $\mu_{{\rm Haar}}$-integrable functions on
$\mathcal{S}^{\Phi,g}$. It can be seen that
\begin{equation}
||\widetilde{F}||_{\infty}={\rm max} \{|\lambda|: \ \ \lambda \in
{\rm sp}(F)\}.
\end{equation}

Thanks to the Pontryagin duality Theorem \cite{reed-simon-1}, we can
obtain a correspondence between elements of the topological space
$\Omega(L^{1}(\mathcal{S}^{\Phi,g},\mu_{{\rm Haar}}))$ and elements
of the Pontryagin dual. In this situation, the canonical isomorphism
\begin{equation} \label{ev-1}
{\rm ev}_{L^{1}(\mathcal{S}^{\Phi,g},\mu_{{\rm
Haar}})}(X)(\rho)=\rho(X) \in S^{1} \subset \mathbb{C}
\end{equation}
can be applied to define a modification of the Fourier
transformation on $L^{1}(\mathcal{S}^{\Phi,g},\mu_{{\rm Haar}})$.

\begin{defn}
For a given Quantum Field Theory $\Phi$ with the strong coupling
constant $g \ge 1$ and the corresponding collection
$\mathcal{S}^{\Phi,g}$ of all large Feynman diagrams generated by
Dyson--Schwinger equations, the Fourier transformation $\mathcal{F}$
on the complex commutative unital Banach algebra
$L^{1}(\mathcal{S}^{\Phi,g},\mu_{{\rm Haar}})$ is well-defined. For
$F \in L^{1}(\mathcal{S}^{\Phi,g},\mu_{{\rm Haar}})$, we have
\begin{equation} \label{Fourier-1}
\widehat{F}(\rho)=\int_{\mathcal{S}^{\Phi,g}} F(X)
\overline{\rho(X)}d\mu_{{\rm Haar}}(X).
\end{equation}
for any character $\rho$.
\end{defn}

For functionals $G,H \in L^{1}(\mathcal{S}^{\Phi,g},\mu_{{\rm
Haar}})$, we have
\begin{equation}
\mathcal{F}\{G *_{\triangle} H\} = \mathcal{F}\{G\}\mathcal{F}\{H\}.
\end{equation}

The original motivation to formulate the Gelfand transform
(\ref{gelfand-1}) is providing a way to separate functionals in
$L^{1}(\mathcal{S}^{\Phi,g},\mu_{{\rm Haar}})$. Thanks to this idea,
the Fourier transformation (\ref{Fourier-1}) encodes the
mathematical procedure for the decomposition of the functional $F$
in terms of large Feynman diagrams which contribute to
Dyson--Schwinger equations in $\mathcal{S}^{\Phi,g}$. In particular,
if we restrict our discussion to a fixed Dyson--Schwinger equation
DSE with the unique solution $X_{{\rm DSE}}$, then our generalized
Fourier transformation describes the evolution of the large Feynman
diagram $X_{{\rm DSE}}$ with respect to $\mu$-integrable functions
originated from large subdiagrams (or partial sums) which converge
to $X_{{\rm DSE}}$.

As the final note, we have developed a theory of Haar
integration on the space of all Dyson--Schwinger equations of a given (strongly coupled) gauge field theory in the language of the
classical Riemann--Lebesgue integral. This new measure theoretic approach is useful to build an analogous version of the classical
Newton--Leibniz differentiation theory with respect to
metrics $d_{g,\alpha,\epsilon}$ for the study of functionals on large Feynman diagrams.

\section{\textsl{The G\^{a}teaux differential calculus on $\mathcal{S}^{\Phi,g}$ and its application}}

In the second section of the previous chapter we have explained the
construction of a noncommutative differential geometry model for the
Hopf algebra $\mathcal{S}^{\Phi}_{{\rm graphon}}$ which is derived
from the BPHZ non-perturbative renormalization process of Feynman graphons. In
\cite{shojaeifard-9} we have applied the analysis of linear spaces to
discuss the construction of a theory of differentiation on the space
of Feynman diagrams in terms of the graph function representations of these physical diagrams and
G\^{a}teaux differentiability under the cut-distance topology where
we worked on admissible directions to define well-defined
differentiations. It has led us to obtain Taylor series type
representations for continuous functionals on Feynman graphons on
the basis of homomorphism densities. The homomorphism densities can be
considered as functionals $W_{\Gamma} \longmapsto t(G,W_{\Gamma})$
on Feynman graphons such that $G$ is an arbitrary finite graph. If
$G$ is a simple graph (such as rooted trees), then we can show that
the corresponding homomorphism density is continuous with respect to
the cut-distance topology and it is also $L^{1}$-integrable. Thanks
to the disjoint union operator, we can build an algebraic structure on
the linear span of homomorphism densities with respect to finite
simple graphs which determine a dense subset in the space
$C(\mathcal{S}^{\Phi}_{{\rm graphon}})$ of all continuous functions
on $\mathcal{S}^{\Phi}_{{\rm graphon}}$ with respect to the topology
of uniform convergence. In addition, we can compute the G\^{a}teaux
derivatives of homomorphism densities where their Fr\'{e}chet
differentiability can be achieved with respect to simple graphs
(i.e. decorated non-planar rooted trees) under some conditions where
we might need to remove the symmetric condition of Feynman graphons
and work on Feynman bigraphons. These observations can clarify the importance of
homomorphism densities to study
graphons under a functional analysis setting.

In this section, we concern the question of how to endow with a
differential calculus on large Feynman diagrams independent of any
renormalization program. We consider the real or complex vector
space $\mathcal{S}^{\Phi,g}$ generated by all Dyson--Schwinger
equations in the physical theory $\Phi$ and equip this space with
the cut-distance topology determined by
\begin{equation}
d({\rm DSE}_{1},{\rm DSE}_{2}):= d_{{\rm cut}}([f^{X_{{\rm
DSE}_{1}}}],[f^{X_{{\rm DSE}_{2}}}]).
\end{equation}
It defines the cut-norm for each large Feynman diagram. We have
\begin{equation} \label{norm-1}
\parallel X_{{\rm DSE}} \parallel_{{\rm cut}} = {\rm sup}_{A,B \subseteq [0,1]} \mid \int_{A \times B} f^{X_{{\rm DSE}}}(x,y) dxdy \mid
\end{equation}
such that the supremum is taken over Lebesgue measurable subsets
$A,B$ of the closed interval. The resulting space can be interpreted
as a closed topological subspace of the compact topological space of
all unlabeled graphons. It means that we can consider
$\mathcal{S}^{\Phi,g}$ as a new Banach space to build a new G\^{a}teaux
differential calculus on the space of solutions of Dyson--Schwinger equations of a given physical theory. Then we apply the
functional analysis of graphons
(\cite{diao-guillot-khare-rajaratnam-1}) to obtain a new G\^{a}teaux
differential calculus machinery for the study of
$C(\mathcal{S}^{\Phi,g})$.

Total derivative and directional derivatives are the most common
differentiation machineries in finite dimensions such that their
analogous versions in infinite dimensions are Fr\'{e}chet derivative
and G\^{a}teaux derivative.

For a given function $F: X \longrightarrow Y$ between two Banach
spaces (or normed vector spaces), the G\^{a}teaux derivative at
$x_{0} \in X$ is by definition a bounded linear operator
$T_{x_{0}}:X\longrightarrow Y \in \mathcal{B}(X,Y)$ such that for
every $u \in X$,
\begin{equation}  \label{der-1}
{\rm lim}_{t \longrightarrow 0} \frac{F(x_{0}+tu) - F(x_{0})}{t} =
T_{x_{0}}u.
\end{equation}
If for some fixed $u$ the limit
\begin{equation}
\delta_{u}F(x):= \frac{d}{dt}|_{t=0} F(x+tu) = {\rm lim}_{t
\rightarrow 0} \frac{F(x+tu) - F(x)}{t}
\end{equation}
exists, then we call $F$ has a directional derivative at $x$ in the
direction $u$. Therefore $F$ is G\^{a}teaux derivative at $x_{0}$ if
and only if all the directional derivatives $\delta_{u}F(x)$ exist
and form a bounded linear operator
\begin{equation}
DF(x): u \longmapsto \delta_{u}F(x).
\end{equation}
$T_{x_{0}}$ is called the Fr\'{e}chet derivative of $F$ at $x_{0}$,
if the limit (in the sense of the G\^{a}teaux derivative) exists
uniformly in $u$ on the unit ball in $X$. If we set $y=tu$ then $t$
tends to zero is equivalent to $y$ tends to zero. Now $F$ is
Fr\'{e}chet differentiable at $x_{0}$, if for all $y$ we have
\begin{equation} \label{frechet-1}
F(x_{0}+y) = F(x_{0})+ T_{x_{0}}(y) + o(\parallel y \parallel),
\end{equation}
which means that
\begin{equation} \label{frechet-2}
{\rm lim}_{\parallel h \parallel \rightarrow 0} \frac{\parallel
F(x+h) - F(x) - Th \parallel}{\parallel h
\parallel} = 0
\end{equation}
holds. As we can see, the limit in the Fr\'{e}chet derivative only
depends on the norm of $y$ where the operator $T$ defines the
natural linear approximation of $F$ in a neighborhood of the point
$x_{0}$. In this setting, we call $T_{x_{0}} = DF(x_{0})$ as the
derivative of $F$ at $x_{0}$. In addition, we can show that being
Fr\'{e}chet differentiable at a point implies being G\^{a}teaux
differentiable at a point such that in this case the G\^{a}teaux
derivative is equal to the Fr\'{e}chet derivative.

If $F$ is G\^{a}teaux differentiable on $X$, then we have the mean
value formula
\begin{equation} \label{mean-value-1}
\parallel F(y) - F(x) \parallel \le \parallel x-y \parallel {\rm
sup}_{0 \le \theta \le 1} \parallel DF(\theta x + (1-\theta)y)
\parallel.
\end{equation}
This enables us to show that if $F$ is G\^{a}teaux differentiable on an
open neighborhood $U$ of $x$ and $DF(x)$ is continuous, then $F$ is
Fr\'{e}chet differentiable at $x$. \cite{flett-1, reed-simon-1}

We plan to study (smooth) real valued continuous functions on the
Banach space $\mathcal{S}^{\Phi,g}$ in terms of their Taylor series
representation under the higher orders G\^{a}teaux differentiations.
We show that the solution space of the natural generalization of the
equation $\frac{d^{n}}{dx^{n}}F(x) \equiv 0$ to large Feynman
diagrams namely, G\^{a}teaux type differential equations with the
general form
\begin{equation}
d^{N+1}F(X;Z_{1},...,Z_{N+1})=0
\end{equation}
for all large Feynman diagrams $X,Z_{1},...Z_{N+1} \in
\mathcal{S}^{\Phi,g}$, can be described by homomorphism densities.
Solutions of this class of equations enjoy the property
\begin{equation}
F(Z)=F(0) + dF(0;Z) +
\frac{d^{2}F(0;Z,Z)}{2} + ... +
\frac{d^{n}F(0;Z,...,Z)}{n!}.
\end{equation}

\begin{defn} \label{der-5}
For a given function $F: \mathcal{S}^{\Phi,g} \longrightarrow
\mathbb{R}$ and each large Feynman diagram $X$, the G\^{a}teaux
derivative exists at $X$ in the direction $Y \in
\mathcal{S}^{\Phi,g}$, if the limit
\begin{equation}
dF(X;Y)={\rm lim}_{t \longrightarrow 0}\frac{F(X+tY) - F(Y)}{t}
\end{equation}
exists. The higher orders of the G\^{a}teaux differentiability can
be defined by induction where for any $n \ge 2$, $F$ is called
$n$-time G\^{a}teaux differentiable at $X$ in directions
$Z_{1},...,Z_{n}$ if at the first, the higher mixed G\^{a}teaux
derivatives
\begin{equation}
d^{n-1}F(X+\lambda Z_{n};Z_{1},...,Z_{n-1})
\end{equation}
exist for each real number $\lambda$ and at the second, the limit
$$d^{n}F(X;Z_{1},...,Z_{n})=$$
\begin{equation} \label{der-6}
{\rm lim}_{\lambda \longrightarrow 0}\frac{d^{n-1} F(X+\lambda
Z_{n};Z_{1},...,Z_{n-1}) - d^{n-1}F(X;Z_{1},...,Z_{n-1})}{\lambda}
\end{equation}
exists.
\end{defn}

It is easy to check that the G\^{a}teaux derivatives
$d^{n}F(X;Z_{1},...,Z_{n})$ are multilinear maps in $Z_{i}$ and in
addition, for any permutation $\tau \in S_{n}$, we have
\begin{equation}
d^{n}F(X;Z_{1},...,Z_{n}) = d^{n}F(X;Z_{\tau(1)},...,Z_{\tau(n)}).
\end{equation}

\subsection{\textsl{Feynman random graphs via homomorphism densities}}

The description of large Feynman diagrams via Feynman graphons in
$\mathcal{S}^{\Phi}_{{\rm graphon}}$ allows us to think about the concept of ''density'' of infinite Feynman graphs in
the solution of a given Dyson--Schiwnger equation. From the view point of Quantum
Field Theory, any infinite Feynman graph can contain nested loops which present different types of subdivergencies derived from virtual particles and their interactions. One important task is to search for any algorithm which could estimate the appearance
of a particular class of subdivergencies (or subgraphs) in an infinite expansion of Feynman
diagrams. The homomorphism density, as a class function on
$\mathcal{S}^{\Phi}_{{\rm graphon}}$ or $\mathcal{S}^{\Phi,g}$,
is a useful tool to formulate a new mathematical model for this fundamental challenge.

\begin{prop} \label{hom-dens-large-1}
The homomorphism density is a well-defined operator on large Feynman
diagrams.
\end{prop}

\begin{proof}
Proposition 4.6 in \cite{shojaeifard-10} shows that the unique
solution $X_{{\rm DSE}}$ of a given Dyson--Schwinger equation DSE
can be interpreted as the convergent limit of the sequence
$\{Y_{m}\}_{m \ge 1}$ of its partial sums with respect to the
cut-distance topology.

For each unlabeled Feynman graphon class
$[W]$, we can build the homomorphism density $t(X_{{\rm DSE}},W)$
as the "limit" of the sequence $\{t(Y_{m},W)\}_{m \ge 1}$
of homomorphism densities corresponding to finite expansions $Y_{m}$
of finite graphs (which do not have self-loops but have loops).

For each $m \ge 1$, if the partial sum $Y_{m}:= X_{1} + ...+X_{m}$ has $k_{m}$ vertices, then we have
\begin{equation} \label{hom-dens-1}
t(Y_{m},W)= \int_{[0,1]^{k_{m}}} \prod_{(i,j) \in E(Y_{m})}
W(x_{i},x_{j}) dx_{1}...dx_{k_{m}}.
\end{equation}
The induction is useful to show that
\begin{equation}
t(Y_{m+1},W)=t(Y_{m},W)t(X_{m+1},W).
\end{equation}

The condition $d_{{\rm cut}}(W',W)=0$ for weakly isomorphic graphons can be applied to show that
$t(Y_{m},W')=t(Y_{m},W)$ which leads us to $t(X_{{\rm
DSE}},W')=t(X_{{\rm DSE}},W)$. Therefore we can define a poset on
homomorphism densities where a family of injections $\{f_{ij}:
t(Y_{i},-) \rightarrow t(Y_{j},-)\}_{i \le j}$ on the space of
Feynman graphons can be formulated. This gives us an inverse system
where its inverse limit can be identified as a subset of the direct
product of the homomorphism densities $t(Y_{i},-)$s. This inverse limit can be considered
as the homomorphism density with respect to the large Feynman
diagram $X_{{\rm DSE}}$. We have
$$t(X_{{\rm DSE}},-) = {\rm lim}_{\leftarrow_{m}} t(Y_{m},-)$$
\begin{equation}
= \{\overrightarrow{W} \in \prod_{m=1}^{\infty} t(Y_{m},-): \ \
f_{ij}(W_{j})=W_{i}, \ \forall i \le j\} \subseteq
\prod_{m=1}^{\infty} t(Y_{m},-).
\end{equation}
\end{proof}

For $1 \le n \le \infty$, set $[n]:=\{i \in \mathbb{N}: \ \ i \le
n\}$. For a given Feynman graphon $W$, define a random graph
$G(n,W)$ with the vertex set $[n]$ (chosen points $\{x_{1},...,x_{n}\}$
at random from the closed unit interval) by letting $ij$ be an edge
in $G(n,W)$ with the probability $W(x_{i},x_{j})$. It is possible to
build $G(n,W)$ for all $n$ by constructing $G(\infty,W)$ as an
exchangeable random graph namely, its distribution is invariant
under permutations of the vertices and every exchangeable random
graph is a mixture of such graphs. Then take the subgraph defined by
the first $n$ vertices. In general, for any Feynman diagram $\Gamma$
(as a labeled graph), the homomorphism density $t(\Gamma,W)$ equals
the probability that $\Gamma$ is a subgraph of $G(\infty,W)$ or of
$G(n,W)$ for any $n \ge |\Gamma|$. In other words, the family
$\{t(\Gamma,W)\}_{\Gamma}$ and the distribution of $G(\infty,W)$ can
determine each other.

It is important to note that for given Feynman
graphons $W,W'$, $G(\infty,W)$ and $G(\infty,W')$ have the same
distribution if and only if those two graphons are weakly isomorphic or
equivalent.

\begin{lem} \label{hom-dens-large-2}
Homomorphism densities can determine a new class of random graphs with respect
to solutions of Dyson--Schwinger equations.
\end{lem}

\begin{proof}
We know that $X_{{\rm DSE}}={\rm lim}_{m \rightarrow \infty} Y_{m}$
with respect to the cut-distance topology such that $|Y_{m}|
\rightarrow \infty$. For each $[k]$ there exists a random graph
$Y_{m}[k]$ on the vertex set $[k]$ such that $\{Y_{m}[k]\}_{m \ge
1}$ converges to $X_{{\rm DSE}}[k]$ with respect to the metric
\begin{equation}
d_{{\rm dens}}(\Gamma,\Gamma') = \sum_{i}2^{-i}|t(W_{i},\Gamma) -
t(W_{i},\Gamma')|
\end{equation}
which is equivalent to the cut-distance
(\cite{borgs-chayes-lovasz-sos-vesztergombi-1, janson-1}). Therefore
there exists an infinite random graph $\overline{X_{{\rm DSE}}}$ on
$[\infty]$ such that $X_{{\rm DSE}}[k] \equiv \overline{X_{{\rm
DSE}}}|_{[k]}$ with respect to the metric $d_{{\rm dens}}$. This
means that
\begin{equation}
{\rm lim}_{m \rightarrow \infty} Y_{m}[k] = \overline{X_{{\rm
DSE}}}|_{[k]}.
\end{equation}
\end{proof}

\begin{cor} \label{hom-dens-large-4}
The distribution of the random graphs $Y_{m}[k]$ with respect to
partial sums of $X_{{\rm DSE}}$ converges when $m$ tends to
infinity.
\end{cor}

\begin{proof}
Thanks to Lemma \ref{hom-dens-large-2}, for each $k \le |Y_{m}|$,
there exists $Y_{m}[k]$ as the random induced subgraph of $Y_{m}$
with $k$ vertices determined by selecting $k$ separate vertices
$v_{1},...,v_{k}$ of $Y_{m}$ at uniformly random procedures. Now
thanks to graphon representations of Feynman diagrams, it is enough
to apply the definition of convergent sequences in the theory of
graphons via random graphs (\cite{janson-1}) to the sequence
$\{Y_{m}\}_{m \ge 1}$ to show its convergence to $X_{{\rm DSE}}$.
\end{proof}

\begin{lem} \label{hom-dens-large-3}
Dyson--Schwinger equations which generate isomorphic Hopf
subalgebras have the same homomorphism density.
\end{lem}

\begin{proof}
Suppose Hopf subalgebras $H_{{\rm DSE}_{1}}$, $H_{{\rm DSE}_{2}}$
corresponding to the equations  ${\rm DSE}_{1}$ and ${\rm DSE}_{2}$
are isomorphic which means that the unique solutions $X_{{\rm
DSE}_{1}}$ and $X_{{\rm DSE}_{1}}$ are isomorphic infinite graphs.
We can lift the weakly isomorphic relation on graphons onto the
level of large Feynman diagrams. We say that two large Feynamn
diagrams $X_{{\rm DSE}_{1}}, X_{{\rm DSE}_{2}}$ are weakly isomorphic or
equivalent if their corresponding labeled graphons have the same
unlabeled measurable function almost everywhere. In other words,
$X_{{\rm DSE}_{1}}$ and $X_{{\rm DSE}_{2}}$ are weakly isomorphic or equivalent whenever
\begin{equation}
d_{{\rm cut}}(f^{X_{{\rm DSE}_{1}}},f^{X_{{\rm DSE}_{2}}})=0.
\end{equation}
Thanks to Borgs--Chayes--Lovasz Theorem (\cite{lovasz-1}) and
Proposition \ref{hom-dens-large-1}, we can show that two weakly
equivalent large Feynamn diagrams $X_{{\rm DSE}_{1}}, X_{{\rm
DSE}_{2}}$ have the same homomorphism density for all (finite)
simple graphs. In other words,
\begin{equation}
t(H,f^{X_{{\rm DSE}_{1}}}) = t(H,f^{X_{{\rm DSE}_{2}}}).
\end{equation}
\end{proof}

\subsection{\textsl{Differentiability}}

In this part our task focuses on the study of the space of smooth real valued functions on $\mathcal{S}^{\Phi,g}$ in terms of homomorphism densities of Feynman graphons which contribute to solutions of Dyson--Schwinger equations. We show that this class of homomorphism densities can play the role of a basis for the Taylor series representations of smooth functions.

Define a new graduation parameter on the collection
$\mathcal{H}(\Phi)$ of all Feynman diagrams of a physical theory
$\Phi$ in terms of the number of edges. For each $n$, set
$\mathcal{H}_{n}(\Phi)$ as the isomorphism classes of Feynman
diagrams with $n$ internal and external edges, no isolated vertices,
no self-loops but possible multi-edges. Set $\mathcal{H}_{\le
n}(\Phi) = \bigcup_{j \le n} \mathcal{H}_{j}(\Phi)$. The
homomorphism density for all Feynman graphs $\Gamma \in
\mathcal{H}_{j}(\Phi)$ is well-defined.

\begin{thm} \label{hom-dens-large-7}
For each $n \ge 1$, define
\begin{equation}
F_{n}: \mathcal{S}^{\Phi}_{{\rm graphon}} \longrightarrow \mathbb{R}, \
\ F_{n}(W_{\Gamma}):= \sum_{\gamma \in \mathcal{H}_{\le n}(\Phi)}
a_{\gamma} t(\gamma,W_{\Gamma})
\end{equation}
for any finite Feynman diagram $\Gamma$ and some constants $a_{\gamma}$. Then \\
(i) $F_{n}$ is continuous in the $L^{1}$-topology. \\
(ii) It is possible to lift $F_{n}$ onto the space of large Feynman
diagrams and define a new real valued map $\tilde{F}$ on
$\mathcal{S}^{\Phi,g}$ which is continuous in the $L^{1}$-topology.
\end{thm}

\begin{proof}
(i) Define a new real valued multilinear functional $\tau$ on
$(\mathcal{S}^{\Phi}_{{\rm graphon}})^{n}$ given by
\begin{equation}
\tau((W_{\gamma_{t}})_{\gamma_{t}}):= \int_{[0,1]^{\sum |\gamma_{t}|}} \prod_{\gamma_{t}} W_{\gamma_{t}}(x_{i},x_{j}) \prod_{i}dx_{i}.
\end{equation}
We can show that
\begin{equation}
|\tau((W_{\gamma_{t}})_{\gamma_{t}}) - \tau((W'_{\gamma_{t}})_{\gamma_{t}})|
\le \sum_{\gamma_{t}} \parallel W_{\gamma_{t}} - W'_{\gamma_{t}} \parallel_{1}
\end{equation}
which leads us to
\begin{equation} \label{eq-1}
|t(\sqcup \gamma_{t},W) - t(\sqcup \gamma_{t},W')| \le |E(\sqcup \gamma_{t})| \parallel W_{\gamma_{t}} -
W'_{\gamma_{t}} \parallel_{1}.
\end{equation}

(ii) Thanks to Proposition \ref{hom-dens-large-1}, we plan to lift
the above process onto the level of large Feynman diagram $X_{{\rm
DSE}}$ with the partial sums $Y_{m}$, $m \ge 1$. It is enough to
extend the relation (\ref{eq-1}) to $Y_{m+1}=Y_{m}+X_{m+1}$. We have
$$\big|t(Y_{m} \sqcup X_{m+1},W) - t(Y_{m} \sqcup X_{m+1},W') \big |=$$
$$\big|t(Y_{m} \sqcup X_{m+1},W) - t(Y_{m} \sqcup X_{m+1},W') \pm t(Y_{m},W')t(X_{m+1},W) \big|=$$
$$\big|t(X_{m+1},W) \big(t(Y_{m},W) - t(Y_{m},W')\big) + t(Y_{m},W')\big(t(X_{m+1},W) - t(X_{m+1},W')\big)\big|$$
\begin{equation}
\le |t(X_{m+1},W)| |(t(Y_{m},W) - t(Y_{m},W'))| + |t(Y_{m},W')|
|t(X_{m+1},W) - t(X_{m+1},W')|
\end{equation}
$$\le |t(X_{m+1},W)| \big(|E(Y_{m})| \parallel W_{Y_{m}}-W'_{Y_{m}} \parallel_{1}\big) +
|t(Y_{m},W')| \big(|E(X_{m+1})| \parallel W_{Y_{m}}-W'_{Y_{m}} \parallel_{1}\big)$$
such that $\parallel W_{Y_{m}}\parallel_{\infty}, \parallel
W'_{Y_{m}}\parallel_{\infty} \le 1$.

Lift the multilinear operator $\tau$ onto the multilinear operator
$\tilde{\tau}$ defined as a bounded operator on the Banach space
$\mathcal{S}^{\Phi,g}$. Now define the new map $\tilde{F}$ on
$\mathcal{S}^{\Phi,g}$ given by
\begin{equation}
\tilde{F}(X_{{\rm DSE}}):= \prod_{m=1}^{\infty} F_{m}(W_{Y_{m}})
\end{equation}
such that each term $F_{m}(W_{Y_{m}})$ is a $L^{1}$-continuous function.
Therefore $\tilde{F}$, as the product of continuous functions, is
also continuous with respect to the $L^{1}$- topology.
\end{proof}

The space $\mathcal{W}_{[0,1]}$ of all (bi-)graphons can be embedded into
the vector space $\mathcal{W}$ of bounded (symmetric) measurable
functions $f:[0,1]^{2} \rightarrow \mathbb{R}$ which is equipped
by a semi-norm. Under weakly isomorphic relation $\approx$, we can
build a complete metric structure on the quotient space
$\mathcal{W}_{[0,1]}/\approx$. The topological space
$\mathcal{S}^{\Phi}_{{\rm graphon}}$ of all unlabeled graphon classes which
contribute to the representations of Feynman diagrams and
Dyson--Schwinger equations sits inside $\mathcal{W}_{[0,1]}/\approx$.
As we have discussed each Feynman graphon $[W_{\Gamma}] \in \mathcal{S}^{\Phi}_{{\rm graphon}}$ is a class of
bounded (symmetric) measurable functions on $[0,1]^{2}$ up to relabeling and weakly isomorphic relation. This class of graphons are generated by rooted
tree representations of (large) Feynman diagrams. Rooted trees are
simple graphs where their adjacency matrices can determine their
corresponding graphon classes. Orientations on decorated non-planar rooted trees, which encode
positions of nested loops in the main Feynman diagram, inform us that we might need
only the upper part or the lower part of the adjacency matrix for the
reconstruction of any Feynman diagram from its graphon representation.
It means that we do not need the symmetric property of graphons and
we can work only on the bounded measurable functions $f:[0,1]^{2}
\rightarrow [0,1]$ up to the relabeling and weakly isomorphic relation. This class of objects,
which are known as bi-graphons, enables us to have
Fr\'{e}chet differentiability of homomorphism densities of Feynman graphons.

\begin{lem} \label{hom-dens-diff-1}
(i) The homomorphism densities on $\mathcal{S}^{\Phi}_{{\rm
graphon}}$ are Fr\'{e}chet differentiable.

(ii) The homomorphism densities on $\mathcal{S}^{\Phi,g}$ are
Fr\'{e}chet differentiable.
\end{lem}

\begin{proof}
(i) We compute the Fr\'{e}chet derivatives of the homomorphism
densities on $\mathcal{S}^{\Phi}_{{\rm graphon}}$ for ladder trees
$l_{1},l_{2},l_{3}$ and the rooted tree $\bigvee$ where vertices
$2,3$ are adjacent to the root $1$.

For the tree with only one vertex, $t(l_{1},-)\equiv 1$ which is
obviously Fr\'{e}chet differentiable.

For the oriented decorated ladder tree $l_{2}$ with two vertices
$1,2$ and one edge $e_{12}$ from $1$ to $2$, the G\^{a}teaux
derivative can be computed by
$$d(t(H,W_{\Gamma});W_{\Gamma'})=$$
\begin{equation}
\int_{[0,1]^{k}} \sum_{(i_{1},j_{1}) \in E(H)}
W_{\Gamma'}(x_{i_{1}},x_{j_{1}}) \prod_{(i,j) \in E(H)\backslash
(i_{1},j_{1})} W_{\Gamma}(x_{i},x_{j}) dx_{1} ... dx_{k}
\end{equation}
which leads us to compute the unique Fr\'{e}chet derivative by the
linear map
\begin{equation}
W_{\Gamma'} \longmapsto d(t(l_{2},W_{\Gamma});W_{\Gamma'}) =
\int_{[0,1]^{2}} W_{\Gamma'}(x_{1},x_{2}) dx_{1}dx_{2}.
\end{equation}
We have
\begin{equation}
{\rm lim}_{W_{\Gamma'} \rightarrow 0} \frac{\big|
t(H,W_{\Gamma}+W_{\Gamma'}) - t(H,W_{\Gamma}) - \int_{[0,1]^{2}}
W_{\Gamma'}(x_{1},x_{2})dx_{1}dx_{2}\big|}{\parallel W_{\Gamma'}
\parallel_{{\rm cut}}}
\end{equation}
$$= {\rm lim}_{W_{\Gamma'} \rightarrow 0} \frac{0}{\parallel W_{\Gamma'}
\parallel_{{\rm cut}}}=0$$
which approves the Fr\'{e}chet differentiability in terms of the
formula (\ref{frechet-2}).

For the oriented decorated ladder tree $l_{3}$ with three vertices
$1,2,3$ and two edges $e_{12},e_{23}$ which connect the vertices $1$
to $2$ and $2$ to $3$, the unique candidate for the Fr\'{e}chet
derivative of $t(l_{3},-)$ should be
$$W_{\Gamma'} \longmapsto d(t(l_{3},W_{\Gamma});W_{\Gamma'}) =$$
\begin{equation}
\int_{[0,1]^{3}} W_{\Gamma'}(x_{1},x_{2})W_{\Gamma}(x_{2},x_{3}) +
W_{\Gamma}(x_{1},x_{2}) W_{\Gamma'}(x_{2},x_{3})dx_{1}dx_{2}dx_{3}
\end{equation}
$$=2 \int_{[0,1]^{3}} W_{\Gamma}(x_{1},x_{2})W_{\Gamma'}(x_{2},x_{3}) dx_{1}dx_{2}dx_{3}.$$
$t(l_{3},-)$ is Fr\'{e}chet differentiable if the following limit
exists and equals to zero,
$${\rm lim}_{W_{\Gamma'} \rightarrow 0} \frac{\big| t(l_{3},W_{\Gamma}+W_{\Gamma'})
- t(l_{3},W_{\Gamma}) - d(t(l_{3},W_{\Gamma});g) \big|}{\parallel
W_{\Gamma'}
\parallel_{{\rm cut}}}$$
\begin{equation} \label{eq-5}
= {\rm lim}_{W_{\Gamma'} \rightarrow 0}  \frac{\big| \int_{[0,1]^{3}}
W_{\Gamma'}(x_{1},x_{2})W_{\Gamma'}(x_{2},x_{3})dx_{1}dx_{2}dx_{3}
\big|}{\parallel W_{\Gamma'} \parallel_{{\rm cut}}}.
\end{equation}
If this limit is not zero, then there are some below boundaries
$c>0$ which means that we can define a sequence $\{W_{n}\}_{n \ge
1}$ of graphons which converges to zero with respect to the cut-norm
but the limit (\ref{eq-5}) does not zero when $n$ tends to infinity.
In other words, we might have
$$0< c \le \frac{1}{n}= \big | \int_{[0,1]^{2},x_{2} \in [0,1/n]} 1 dx_{1}dx_{2}dx_{3} \big| = $$
\begin{equation}
\big| \int_{[0,1]^{2},x_{2} \in [0,1/n]}
W_{n}(x_{1},x_{2})W_{n}(x_{2},x_{3}) dx_{1}dx_{2}dx_{3} \big| \le
\end{equation}
$$\big| \int_{[0,1]^{3}} W_{n}(x_{1},x_{2})W_{n}(x_{2},x_{3}) dx_{1}dx_{2}dx_{3}  \big|.$$
This situation supports the existence of sequences of graphons such
as $W_{n}={\bf 1}$ on ${\rm min}(x_{1},x_{2}) < 1/n$ which satisfy the
above inequality. On the other hand, we can build the topological
renormalization Hopf algebra $\mathcal{S}^{\Phi}_{{\rm graphon}}$ of
Feynman graphons by measurable bounded functions from $[0,1]^{2}$ to
$[0,1]$ in terms of working only on the upper parts or lower parts
of the adjacency matrices of oriented decorated non-planar rooted
trees. In this non-symmetric setting, the sequences such as
$\{W_{n}\}_{n \ge 1}$ of graphons does not belong to
$\mathcal{S}^{\Phi}_{{\rm graphon}}$. As the consequence, the only
lower boundary for Feynman graphons is zero itself which means that
the limit (\ref{eq-5}) is zero.

By a similar discussion, we can show the existence of Fr\'{e}chet
derivative for other oriented rooted trees $H$ which contains the
tree $\bigvee$ where vertices $2,3$ are adjacent to the root $1$. In
this situation, we need to deal with
\begin{equation} \label{eq-7}
{\rm lim}_{n \rightarrow \infty} \frac{\big|
t(H,W_{\Gamma}+W_{\Gamma'}) - t(H,W_{\Gamma}) -
d(t(H,W_{\Gamma});W_{\Gamma'}) \big|}{\parallel W_{\Gamma'}
\parallel_{{\rm cut}}}
\end{equation}
where if this limit is not zero then we get some lower
boundaries $c>0$  such that
$$c^{|E(H)|-2} \int_{[0,1]^{3}} W_{\Gamma'}(x_{1},x_{2}) W_{\Gamma'}(x_{2},x_{3})dx_{1}dx_{2}dx_{3} \le $$
$$ \int_{[0,1]^{|V(H)|}} W_{\Gamma'}(x_{1},x_{2}) W_{\Gamma'}(x_{2},x_{3})
\prod_{(ij) \in E(H) \backslash \{(1,2),(2,3)\}}
W_{\Gamma}(x_{i},x_{j})dx_{1}...dx_{|V(H)|}$$
\begin{equation}
 \le t(H,W_{\Gamma}+W_{\Gamma'}) - t(H,W_{\Gamma}) - d(t(H,W_{\Gamma});W_{\Gamma'}).
\end{equation}
This situation allows us to determine sequences of graphons which can not
belong to $\mathcal{S}^{\Phi}_{{\rm graphon}}$ whenever we work on Feynman bigraphons. Therefore the
limit (\ref{eq-7}) should be zero.

(ii) Objects in the Banach space $\mathcal{S}^{\Phi,g}$ are large
Feynman diagrams namely, infinite formal expansions of Feynman
diagrams which have nested or overlapping loops. Thanks to Theorem
\ref{hom-dens-large-1}, homomorphism densities on large Feynman
diagrams can be computed in terms of homomorphism densities of
finite partial sums. For large Feynman diagrams $X,Z$, we have
\begin{equation}
t(X,Z)={\rm lim}_{\leftarrow_{m}}t(Y_{m},Z)
\end{equation}
such that for each $m \ge 1$,
\begin{equation}
t(Y_{m},Z)= \prod_{i=1}^{m} t(X_{i},W_{Z}).
\end{equation}
We have $X(g)=\sum_{n \ge 0}g^{n}X_{n}$ and $W_{Z}$ as a Feynman
bigraphon which lives in $\mathcal{S}^{\Phi}_{{\rm graphon}}$ such that it
can not have a non-zero lower boundary.

Furthermore, thanks to the
formula (\ref{frechet-1}), we know that the Fr\'{e}chet
differentiability depends only on the norm $W_{Z}$ where by applying
(i), each $t(X_{m},W_{Z})$ is Fr\'{e}chet differentiable. Thanks to
the product rule, $ t(Y_{m},Z)$ as the product of Fr\'{e}chet
differentiable functions is also Fr\'{e}chet differentiable for each
$m$. Since $t(X,Z)$ can be identified by a subset of the direct
product $\prod_{m=1}^{\infty} t(Y_{m},Z)$, $t(X,Z)$ will be also
Fr\'{e}chet differentiable.
\end{proof}

\begin{lem} \label{hom-dens-diff-7}
For a given $C^{n}$ ($n>0$) class function $G:\mathcal{S}^{\Phi,g}
\rightarrow \mathbb{R}$, $d^{n}G(0;Z_{1},...,Z_{n})$ is
a symmetric $S_{[0,1]}$-invariant multilinear functional.
\end{lem}

\begin{proof}
We can extend the functions
\begin{equation}
\mathcal{G}_{X,Z}(\lambda_{1},...,\lambda_{m}):=G(X+\lambda_{1}Z_{1}+...+\lambda_{m}Z_{m})
\end{equation}
to $C^{n}$ functions on $\mathbb{R}^{m}$ where the equality of mixed
partial derivatives show that for any permutation $\sigma \in
S_{n}$, we have
\begin{equation}
d^{n}G(X;Z_{1},...,Z_{n}) =
d^{n}G(X;Z_{\sigma(1)},...,Z_{\sigma(n)}).
\end{equation}
In addition, we can extend $d^{n}G(0;Z_{1},...,Z_{n})$
multilinearly to each
\begin{equation}
({\rm span}_{\mathbb{R}}(\Gamma_{1},...,\Gamma_{m}))^{n}.
\end{equation}
\end{proof}

\begin{cor} \label{hom-dens-diff-2}
Let $\tilde{F}: \mathcal{S}^{\Phi,g} \rightarrow \mathbb{R}$ be
a $L^{1}$-continuous functional (determined by Theorem
\ref{hom-dens-large-7}) such that for some $N \ge 1$, $\tilde{F}$ is
$N+1$ times G\^{a}teaux differentiable. Then for each large Feynman
diagram $X$ and also $Z_{1},...,Z_{N+1}$ in the Banach space
$\mathcal{S}^{\Phi,g}$,
$$d^{N+1}\tilde{F}(X;Z_{1},...,Z_{N+1}) =
0$$ if and only if there exists a unique family
$\{a_{\gamma}\}_{\gamma}$ of real constants such that $\tilde{F}(X)
= \sum_{\gamma \in \mathcal{H}_{\le N}(\Phi)} a_{\gamma}
t(\gamma,W_{X}).$
\end{cor}

\begin{proof}
Thanks to Definition \ref{der-5}, Proposition
\ref{hom-dens-large-1}, Theorem \ref{hom-dens-large-7}, Lemma
\ref{hom-dens-diff-1}, Lemma \ref{hom-dens-diff-7}, we can apply the
main result in \cite{diao-guillot-khare-rajaratnam-1}.
\end{proof}

\begin{cor} \label{hom-dens-diff-3}
Let $G: \mathcal{S}^{\Phi,g} \rightarrow \mathbb{R}$ be a
continuous functional with respect to the cut-distance topology and
smooth with respect to the G\^{a}teaux derivation. For each large
Feynman diagram $X$ define the following sequence of Taylor
polynomials
\begin{equation}
P_{n}(X):= \sum_{m=0}^{n} \frac{1}{m!} d^{m}G(0;X,...,X), \ \
\forall n \ge o
\end{equation}
which converges to $G(X)$ when $n$ tends to infinity. In addition,
let the Taylor expansion
\begin{equation}
\sum_{m=0}^{\infty} \sum_{\gamma \in \mathcal{H}_{m}(\Phi)}
a_{\gamma}t(\gamma,W_{X})
\end{equation}
is absolutely convergent to $P(G)(X)$ such that
\begin{equation}
\sum_{\gamma \in \mathcal{H}_{m}(\Phi)} a_{\gamma}t(\gamma,W_{X}) =
\frac{1}{m!} d^{m}G(0;X,...,X).
\end{equation}
Then $G(X) = P(G)(X)$.
\end{cor}

\begin{proof}
In \cite{diao-guillot-khare-rajaratnam-1}, the required conditions for the existence of the convergent
Taylor series of a smooth function on the space of unlabeled graphons have been provided. Now it is enough
to adapt that procedure for Feynman graphons (which
is already addressed in \cite{shojaeifard-9}) and then lift it onto the whole space
$\mathcal{S}^{\Phi,g}$ in terms of Definition
\ref{der-5}, Proposition \ref{hom-dens-large-1}, Theorem
\ref{hom-dens-large-7}, Lemma \ref{hom-dens-diff-1}, Lemma
\ref{hom-dens-diff-7} and Corollary \ref{hom-dens-diff-2}.
\end{proof}


\chapter{\textsf{The intrinsic foundations of QFT under a non-perturbative setting}}

\vspace{1in}

$\bullet$ \textbf{\emph{Some historical remarks}} \\
$\bullet$ \textbf{\emph{QFT-entanglement via lattices of Dyson--Schwinger equations}} \\
$-$ \textbf{\emph{A new lattice model}}\\
$-$ \textbf{\emph{Intermediate algorithms for QFT-entanglement}}\\
$-$ \textbf{\emph{Tannakian subcategories as intermediate algorithms}}\\
$\bullet$ \textbf{\emph{Quantum logic via non-perturbative propositional calculus}} \\
$-$ \textbf{\emph{Quantum Topos}} \\
$-$ \textbf{\emph{Non-perturbative Topos}}

\newpage

The first purpose in this chapter is to build a new mathematical model for the description of information flow among particles in (strongly coupled) interacting gauge field theories. This new platform enables us to analyze quantum entanglement via fundamental tools in Category Theory and Theoretical Computer Science. We apply combinatorial
Dyson--Schwinger equations as the building blocks of information
flow among distant elementary particles in a system with infinite
degrees of freedom. The cut-distance topology is applied to construct topological regions around elementary particles which encode passing information. We organize these cut-distance topological regions into a new class of lattices of topological Hopf subalgebras. This setting allows us to understand quantum entanglement in the language of intermediate algorithms which contribute to transferring information among entangled particles \cite{submitted-2}. The second purpose in this chapter is to build a new mathematical
model for the description of logical propositions of non-perturbative aspects in strongly coupled gauge field theories. We explain the basic foundations of a new topos of presheaves which is capable to encode topological regions of elementary particles and the strength of coupling constants. This topos has enough physical information to evaluate logical propositions about infinite formal expansions of Feynman diagrams which contribute to quantum motions \cite{submitted-3}.

\section{\textsl{Some historical remarks}}

Entanglement, non-locality and indeterminism are actually the most complicated and challenging concepts in Quantum Mechanics. These concepts have been originated from pioneering efforts to clarify the complicated methodologies applied in dealing with the outputs of Quantum Mechanics. The study of the behavior of electrons via Quantum Mechanics had shown some results which were against the objections of Einstein and his followers
such as the double slit experiment, the photon box experiment and
the Einstein--Podolsky--Rosen paradox. The completeness of Quantum Mechanics has already been clarified in terms of successful experiments on photons polarization and formulating modern gauge field theories. However there still remain
some philosophical challenges for the interpretation of predications in Quantum Mechanics. \cite{a1,b2,b4,epr1,s1}

On the one hand, there are some rigorous efforts for the interpretation of Quantum Mechanics under a deterministic setting. Theory of many worlds and universal wave function, as the key tools in this setting, can provide predictions completely different from Quantum Mechanics predictions. These tools introduce a different theory for quantum world. Bohmian Mechanics, hidden variables, many-world interpretation and
theory of universal wave function are well-known tools for the deterministic interpretation of Quantum Mechanics.
This class of theories has tried to show that the form of all hidden variable models is capable
to reproduce Quantum Mechanics of a spin-singlet by satisfying both the
assumptions of free will and no signaling, which correspond,
measurement-independence and setting-independence, respectively.
\cite{al1,b2,b6,l2,v5,w1}

On the other hand, Bell's Theorem (\cite{b2,b3}), Kochen--Specker Theorem (\cite{ks1,p2}) and collapse theories such as von Neumann method (\cite{v7}), Ghirardi--Rimini--Weber Theory
(\cite{grw1}), no-collapse theories such as modal interpretations
(\cite{h1})) are rigorous well-known efforts to show the inconsistency of the
Einstein--Podolsky--Rosen paradox and other deterministic
observations with the foundations of Quantum Mechanics. In this direction, the
existence of certain class of observables which can not consistently
be assigned values at all has been considered. Then it is
discussed that a quantum system evolves based on the Schrodinger
equation between measurements at which it collapses to the
eigenstate of the measured variable. In non-measurement
interactions, the evolution of states obeys a linear and unitary
equation of motion such that the particle pair in the
Einstein--Podolsky--Rosen experiment remains in an entangled state.
This class of equations dictates that in a spin measurement, the
pointers of the measurement tools are entangled with the particle
pair in a non-separable state in which the indefiniteness of spins
of particles is transmitted to the pointer's position. The critical challenge in this description is the lack of
explicit definitions for the notions of measurement and time,
duration and nature of state collapses. There are some efforts to
handle this issue such as adding a nonlinear term to the
Schrodinger equation, modal no-collapse interpretations. In addition,
Kochen--Specker Theorem shows the impossibility of the
reproduction of Quantum Mechanics predictions in terms of a hidden variable
model where the hidden variables could assign a value to every
projector deterministically and non-contextually. We
can also address Free Will Theorem which proves that no theory,
whether it extends Quantum Mechanics or not, can correctly predict the results of
future spin experiments. These topics could provide strong reasons
for the intrinsic non-local indeterministic nature of Quantum
Mechanics. \cite{b2,b3,bgh1,ck6,p2}

The notion of entanglement in Quantum Mechanics was derived from
Schrodinger's efforts to describe quantum systems extended over
physically distant parts. Then it was modified by Bohr to
deal with the Einstein--Podolsky--Rosen paradox. Bohr was trying
to solve the question about the spin components of a pair of
particles emitted from a source to move in opposite directions where
no slower than light or light signal can travel between them. Bohr
addressed that since two particles have interacted, they were part
of one whole phenomenon which means that the two particles are
entangled. In other words, these particles are part of one whole
phenomenon or one whole system that has one wave function \cite{b4,s1}. On the other hand, Bell's
Theorem made the free choice of the experimenter as one of the
axioms where it is proved mathematically that certain quantum
correlations violate realism, locality or freedom of choice. The
Bell's inequality is on the basis of the assumption that the quantum
state is not the ultimate limit and additional parameters such as
hidden variables could provide a modified description. Bell
discovered that no local hidden variable theory can reproduce all
possible results of Quantum Mechanics. This framework shows that any local model of
the Bohm's version of the Einstein--Podolsky--Rosen experiment is
committed to certain inequalities about the probabilities of
measurement outcomes which could be incompatible with the predications
of Quantum Mechanics. Bell's Theorem has achieved the non-locality of the quantum realm in terms of an alternative interpretation of the factorization as a locality condition. This perspective postulates
that for each quantum mechanical state there exists a distribution
over all possible pair states which is independent of the settings
of the equipments. However, experiments (such as polarization of
photons) violate the Bell's inequality and moreover, they confirm
the predictions of Quantum Mechanics with high accuracy. The violation of the Bell's inequality is enough to show that there is no underlying classical description for Quantum Mechanics. \cite{b1,b2,b3,bgh1}

Quantum concepts such as entanglement and superposition are fundamental tools for a theory of quantum computation
which performs operations on information in terms of quantum
bits. The state of a quantum bit lives in a superposition of two
orthonormal states such as $|\psi>:= \alpha |0> + \beta |1>$ such that
$\alpha, \beta$ are complex numbers. The measurement of one quantum
bit collapses the wave function of the other quantum bit. Quantum entanglement deals with three
fundamental subjects which can be studied under deterministic
and indeterministic settings. The first challenge is to explain how
we can detect optimally entanglement under theoretical models and
experimental tests. The second challenge is to build theoretical
models and experimental tests which reverse an inevitable process of
degradation of entanglement. The third challenge is to design
computational algorithms which enable us to characterize, control
and quantify entanglement. The main objective in dealing with these
challenges is to find a way to estimate optimally the amount of
quantum entanglement of the compound system in an unknown state if only
incomplete data in the form of average values of some operators
detecting entanglement are accessible. In this direction, a notion
of minimization of entanglement has been formulated under
a chosen measure of entanglement with constrains in the form of
incomplete set of data from experiment. In addition, theory of
positive maps has been developed to provide strong tools for
the detection of entanglement. \cite{calabrese-cardy-tonni-1,
huber-friis-gabriel-spengler-hiesmayr-1, jaeger-1, plenio-virmani-1,
shankar-1}

Entanglement in Quantum Field Theory have also been considered
recently where the measurements of the amount of entanglement in a
quantum system with infinite degrees of freedom were modeled under
some settings such as entropy, kinematic entanglement, particle
mixing and oscillations, theory of neutrino oscillations and
entangled space-time points. Entropy is on the basis of partitioning
an extended quantum system into two complementary subsystems and
calculating the entanglement entropy defined as the von Neumann
entropy of the reduced density matrix of one subsystem. This
treatment does not provide information about the entanglement
between two non-complementary parts of a larger system because of
the existence of a mixed state. Negativity is one interesting tool
to deal with this issue in Quantum Field Theory. Multi-mode
entanglement of single-particle states has been concerned via
particle mixing and flavor oscillations. It is shown that in Quantum
Field Theory these phenomena exhibit a fine structure of quantum
correlations as multi-mode multi-particle entanglement appears.
Quantum information theory is capable to provide appropriate tools
to quantify the content of multi-particle flavor entanglement in
QFT systems. The multi-particle flavor-species entanglement
associated with flavor oscillations of the QFT neutrino system has
been studied in terms of the particle-antiparticle species as
further quantum modes. Neutrino oscillations are due to neutrino
mixing and neutrino mass differences. Theory of entanglement in
neutrino oscillations is another progress in this direction where
mode entanglement can be expressed in terms of flavor transition
probabilities. Charged-current weak-interaction processes together
with their associated charged leptons enable us to identify flavor
neutrinos. Neutrino oscillations and CP violation concern neutrino
mixing such that neutrino masses as corrections to Standard Model
play their essential roles in the procedure \cite{akhmedov-smirnov-1,
blasone-dellanno-desiena-dimauro-Illuminati-1,
blasone-dellanno-desiena-dimauro-Illuminati-2,
blasone-dellanno-desiena-Illuminati-1, calabrese-cardy-tonni-1,
goldman-1, kayser-kopp-robertson-vogel-1}. In this chapter we plan to build a new mathematical model for the description of quantum entanglement in terms of the space of Dyson--Schwinger equations of a given gauge field theory. We show the importance of analytic generalization of solutions of Dyson--Schwinger equations (i.e. Feynman graphon models) for the analysis of quantum informational bridges in terms of lattices of substructures.

Passing from Classical Physics to Quantum Physics changes the
logical foundations of our mathematical frameworks. If we have a rigorous
formulation for the logic of quantum systems with infinite degrees
of freedom, then it definitely helps us to develop our knowledge about
entanglement machinery in QFTs. The logical foundations of Quantum
Mechanics were firstly built in the context of propositional
calculus, Hilbert space of states and the space of observations. The original aim of the propositional calculus in logic is to evaluate
propositions with the general form " {\it the physical quantity such
as $A$ of a given system $S$ has a value in the subset $\Delta$ of
real numbers.}" In this context, the main task is to find what
truth-values such propositions have in a given state of the system
and how the truth-value changes with the state in time. In Classical
Physics, there is a space of states such as points in a topological
space equipped with some additional structures such as Poisson
brackets, symplectic forms, .... In any given state, each
physical quantity has its value and each proposition of the form $A
\in \Delta$, which is represented by some Borel subsets of the
state space, has a truth-value true or false. The Borel subsets of
the state space form a Boolean $\sigma$-algebra which means that the
logic of classical systems can be encoded by a definite Boolean logic.
This description enables us to label Classical Physics as a realist
theory. Quantum Physics does not have this explicit realistic nature
and according to the Kochen--Specker Theorem, there is no state
space of a quantum system analogous to the classical state space. As
the assumptions of this Theorem, the physical quantities are
represented as real-valued functions on the hypothetical state space
of a quantum system. Then it is shown that such a space does not
exist and it is impossible to assign values to all physical
quantities at once. Therefore it is also impossible to assign
true or false values to all propositions. Birkhoff and von Neumann
built the foundations of an instrumentalist approach to quantum
logic where upon measurement of the physical quantity $A$, we could
find the result belong in $\Delta$ with a determined probability. In
this approach, pure states are represented by unit vectors in one particular
Hilbert space and propositions with the general form $A \in \Delta$
are represented by projection operators on this Hilbert space. These
projections form a non-distributive lattice. Set $\hat{E}[A \in
\Delta]$ as the projection which represents the proposition $A \in
\Delta$. The probability of $A \in \Delta$ being true in a given
state $|\psi>$ is determined by
\begin{equation}
P(A\in \Delta|\psi>):= <\psi|\hat{E}[A\in \Delta]|\psi> \in [0,1].
\end{equation}
Non-distribuitivity, dependence on measurement tools and the use of
real numbers as continuum are the most fundamental and conceptual
issues of  this instrumentalist approach and its generalizations. \cite{a1, adelman-corbett-1, b1, birkhoff-vonneumann-1, coecke-1, chiara-giuntini-1}

Category Theory had been applied to formulate a modern topos model for the
logical descriptions of physical phenomena. This modern approach clarified the passing from Classical Physics to Quantum Physics in terms of their corresponding topos models. Study in this direction has provided a new contextual form of quantum logic where it is possible to reconstruct the foundations of
physical theories in the context of search for a suitable
representation in a topos of a certain formal language. Classical Physics is reconstructed via the category of sets while Quantum Physics is reconstructed via the category of presheaves on a particular base category. This categorical machinery has been developed to QFT
models where nowadays we have some topos models for gauge field
theories \cite{baez-dolan-1, dowker-1,doring-isham-1, doring-isham-2, doring-isham-3, doring-isham-4, i1, ib1, ib2, lambek-scott-1, maclane-1, mclarty-1, maclane-moerdijk-1,p1}. In this chapter we plan to build a new topos model for strongly coupled gauge field theories which is capable to recognize the logical differences between perturbative and non-perturbative aspects of the physical theory.

\section{\textsl{QFT-entanglement via lattices of Dyson--Schwinger equations}}

The mathematical formulation of Standard Model in the context of
Noncommutative Geometry can motivate to bring a new approach
for the description of quantum entanglement in gauge field theories.
In Standard Model we have six quarks, six leptons and gauge bosons
which are responsible to carry fundamental forces. Gauge bosons
describe exchanging information between elementary particles
and their interactions in strong, weak and electromagnetic forces.
For example, the exchanging virtual photons (as the gauge boson in
quantum electrodynamics) makes transferring information as the force
between two electrons which is repulsive. Gluons are involved gauge
bosons in strong interactions among hadrons (i.e. six quarks) which
live in the nucleus of an atom. Electrons and neutrinos do not feel
strong nuclear force. Every charge particle feels the
electromagnetic force. $W^{\pm}, Z$ are involved gauge bosons in
weak interactions where everything is effected by the weak nuclear
force. Graviton is the theoretical candidate for gauge bosons of
gravity which effects everything. The modified versions of Standard
Model aim to describe the contribution of gravity. Heavier gauge
bosons are bosons of the fundamental force with the shortest range
of effect. Photons are massless which means that the electromagnetic
force has infinite range. $W^{\pm}, Z$ bosons are extremely heavy
and they have very short range. For example, a neutron can decay
into a proton and the gauge boson $W^{-}$ where at the very short
time, this boson quickly decays into an electron and an antielectron
neutrino. A proton can decay into a neutron and the gauge boson
$W^{+}$ where at the very short time this boson quickly decays into
a positron and a neutrino. Protons and neutrons are built by quarks.
$W^{\pm}$ bosons can contribute to exchanging a type of quark to
another type where as the result a proton converts to a neutron and
vice versa. Since $W^{\pm}$ are heavy, they need to borrow energy to
perform this exchange and then they should pay back the energy by
converting to pairs (positron, neutrino) or (electron, antielectron
neutrino) very quickly. Quarks enjoy the Pauli exclusion principle
which means that quarks should be in different quantum states. This
distinction is encoded by colors. Gluons govern any possible
interactions among quarks which convert or exchange the colors of
quarks by absorbtion or emission of gluons. Gluons can also
produce other gluons and they glue quarks together.  Force between
two quarks is independent of distance between them and therefore we
need infinite amount of energy to separate quarks. This fact, known
as quark confinement, tells us that we can not isolate a quark.
Thanks to gluons, the strong force also governs the existence of protons
and neutrons together inside the nucleus but the force at this level
is not independent of distance. Theoretically, the amount of energy
can be converted to a pair of quark and anti-quark where some
interactions could happen to exchange colours.
\cite{connes-marcolli-1, kreimer-6, nair-1, roberts-1,
roberts-schmidt-1, tanasa-2}

\subsection{\textsl{A new lattice model}}

It is possible to encapsulate all possible interactions in terms of
Green's functions where their self-similar nature enable us to study
interactions in the context of fixed point equations of Green's
functions namely, Dyson--Schwinger equations. The strength of the
fundamental forces dictate the appearance of perturbative,
asymptotic freedom or non-perturbative behaviors to these equations.
It is mentioned that gauge bosons provide information exchange
and here we plan to mathematically describe the existence of
information flow among elementary particles at strong levels of the
coupling constants in interacting gauge field theories via towers of
Dyson--Schwinger equations, cut-distance topological regions of
Feynman diagrams which contribute to solutions of these equations
and the vacuum energy. The vacuum energy guarantees the existence of
virtual particles in the vacuum state which will be used in our
setting.

\begin{rem} \label{vaccum-1}
The vacuum state in free field theory can be described in terms of a
tensor product of the Fock space vacuum states for each independent
field mode where there is no entanglement between the field modes at
different momenta. The full vacuum state in interacting Quantum
Field Theory can be described in terms of a superposition of Fock
basis states where the modes of different momenta are entangled.
\end{rem}

Our promising mathematical model tries to engineer divergencies of non-perturbative aspects of strongly coupled gauge field theories in the context of lattices of topological regions of elementary particles. It shows a deep dependence of the quantum entanglement on the indeterminateness of elementary particles. This platform enables us to understand quantum entanglement as an intrinsic non-local property of non-perturbative QFT-models which approves the indeterministic nature of strongly coupled gauge field theories.

\begin{defn} \label{interaction-dse-1}
For each $n$, consider $\gamma^{p}_{n}$ as a primitive (1PI)
Feynman diagram which presents some interactions of any elementary particle $p$ with other (virtual) particles in the physical theory. For each Feynman
diagram $\Gamma$, $B^{+}_{\gamma^{p}_{n}}(\Gamma)$ builds a new
disjoint union of Feynman diagrams as the result of all possibilities for the
insertion of $\Gamma$ into $\gamma^{p}_{n}$ in terms of the types of
vertices in $\gamma^{p}_{n}$ and types of external edges in $\Gamma$. Each family $\{B^{+}_{\gamma^{p}_{n}}\}_{n \ge 0}$ of this class of Hochschild one cocycles can determine a particular Dyson--Schwinger equation ${\rm DSE}_{p}$ which encodes a collection of
possible interactions between $p$ and other (virtual) particles in the
physical theory.
\end{defn}

We plan to apply lattice structures to describe quantum entanglement. A lattice is a partially ordered set such
that each pair of its elements has a unique join $\vee$ which is the
least upper bound and a unique meet $\wedge$ which is the greatest
lower bound. A lattice is called bounded if there exist the greatest
element and the least element. A lattice is called distributive, if
the operations meet and join obey the distributive conditions.

\begin{thm} \label{entanglement-1}
The information flow between the particle $p$ and all unobserved
intermediate states can be described in terms of a lattice of
topological Hopf subalgebras.
\end{thm}

\begin{proof}
Intermediate states address virtual particles. Dyson--Schwinger
equations are the best tools for us to build Hopf subalgebras of the
Connes--Kreimer renormalization Hopf algebra $H_{{\rm FG}}(\Phi)$ of
Feynman diagrams of a given gauge field theory $\Phi$. Thanks to the cut-distance topology defined on
Feynman graphons (i.e. Theorem \ref{feynman-graphon-5}), we can naturally
equip each Hopf subalgebra $H_{{\rm DSE}}$ with this topology such
that the distance between Feynman diagrams $\Gamma_{1}, \Gamma_{2}$
is given by
\begin{equation} \label{distance-feynman-1}
d(\Gamma_{1},\Gamma_{2}):= d_{{\rm
cut}}([f^{\Gamma_{1}}],[f^{\Gamma_{2}}]).
\end{equation}
The class $[f^{\Gamma_{i}}]$ is the unique unlabeled Feynman graphon with
respect to the Feynman diagram $\Gamma_{i}$ and
\begin{equation}
d_{{\rm cut}}([f^{\Gamma_{1}}],[f^{\Gamma_{2}}])={\rm
inf}_{\phi,\psi}{\rm sup}_{A,B} |\int_{A \times B}
f^{\Gamma_{1}}(\phi(x),\phi(y)) -
f^{\Gamma_{2}}(\psi(x),\psi(y))dxdy|
\end{equation}
where the infimum is taken over all different relabeling $\phi,
\psi$ for the labeled graphons $f^{\phi},f^{\psi}$, respectively.
The supremum is taken over all Lebesgue measurable subsets $A,B$ of
the closed interval.

In addition, the coproduct of $H_{{\rm DSE}}$ is a linear bounded
operator on normed space of Feynman diagrams which leads us
to consider each $H^{{\rm cut}}_{{\rm DSE}}$ as a topological Hopf
subalgebra.

Thanks to Definition \ref{interaction-dse-1}, choose an equation
${\rm DSE}_{p}$ which contains some interactions related to the
particle $p$. For each $j \ge 1$, build a new collection
$\{\Gamma^{(j)}_{{n}}\}_{n \ge 1}$ of primitive graphs
\begin{equation}
\Gamma^{(j)}_{n}:= \Gamma_{1}^{(j-1)} + ... + \Gamma_{n}^{(j-1)}
\end{equation}
such that $\Gamma^{(0)}_{{n}}=\gamma^{p}_{n}$ for each $n\ge 1$.

Thanks to the Milnor--Moore Theorem (\cite{milnor-moore-1}), the
Connes--Kreimer renormalization Hopf algebra $H_{{\rm FG}}(\Phi)$ is
isomorphic to the graded dual of the universal enveloping algebra of
the Lie algebra of primitive elements in $H_{{\rm FG}}(\Phi)^{*}$
where for each Feynman diagram $\Gamma$, we can consider its
corresponding infinitesimal character $Z_{\Gamma}$ in the dual
space. Since the renormalization coproduct is a linear map, we can
check easily that the sum of a finite number of primitive graphs is
primitive. Therefore for each $j \ge 1$ and $n \ge 1$, the operator
$B^{+}_{\Gamma^{(j)}_{{n}}}$ is the corresponding Hochschild one cocyle such that
for each Feynman diagram $\Gamma$, this operator concerns all possible
situations for the insertion of $\Gamma$ into the disjoint unions of
Feynman diagrams $\Gamma_{1}^{(j-1)}, ..., \Gamma_{n}^{(j-1)}$. In
addition, by induction, we can show that for each $j$, the resulting
graph $B_{\Gamma^{(j)}_{{n}}}^{+}(\Gamma)$ covers
$B_{\Gamma^{(j-1)}_{{n}}}^{+}(\Gamma)$ as a subgraph.

For each $j \ge 1$, we can consider the Dyson--Schwinger equation
\begin{equation}
{\rm DSE}_{p}^{(j)}:= <\{B_{\Gamma^{(j)}_{{n}}}^{+}\}_{n \ge 1}>,
\end{equation}
with the corresponding Hopf subalgebra $H_{{\rm DSE}_{p}^{(j)}}$.
There exists a natural injective Hopf algebra homomorphism from
$H_{{\rm DSE}_{p}^{(j)}}$ to $H_{{\rm DSE}_{p}^{(j+1)}}$ which leads
us to build the following increasing chain of Hopf subalgebras
\begin{equation} \label{chain-1}
H_{{\rm DSE}_{p}} \le H_{{\rm DSE}_{p}^{(1)}} \le H_{{\rm
DSE}_{p}^{(2)}} \le ... \le H_{{\rm DSE}_{p}^{(j)}} \le ... .
\end{equation}
If $X_{{\rm DSE}_{p}^{(j)}} = \sum_{n_{(j)} \ge 0} g^{n_{(j)}}
X_{n_{(j)}}^{j}$ is the unique solution of the equation ${\rm
DSE}_{p}^{(j)}$, then for each $n_{(j)} \ge 1$, set
\begin{equation}
H(X_{1}^{j},...,X_{n_{(j)}}^{j})
\end{equation}
as the graded Hopf subalgebra of $H_{{\rm DSE}_{p}^{(j)}}$ free
commutative generated algebraically by graphs $X_{1}^{j}, ...,
X_{n_{(j)}}^{j}$.

We can also equip these Hopf subalgebras with the
cut-distance topology and then work on the topologically completed enrichment of these Hopf subalgebras. We have discussed that solutions of Dyson--Schwinger equations are
actually graph limits of sequences of finite Feynman diagrams with
respect to the cut-distance topology. The Hopf subalgebra $H_{{\rm
DSE}}$ generated by a given Dyson--Schwinger equation DSE is graded
with respect to the number of internal edges or the number of independent loops. We have
\begin{equation}
H_{{\rm DSE}} = \bigoplus_{n \ge 0} H_{{\rm DSE},(n)}
\end{equation}
such that $H_{{\rm DSE},(n)}$ is the homogeneous component of degree
$n$ and
$$H_{{\rm DSE},(p)}H_{{\rm DSE},(q)} \subseteq H_{{\rm DSE},(p+q)},$$
\begin{equation}
\Delta (H_{{\rm DSE},(n)}) \subseteq \bigoplus_{p+q=n} H_{{\rm
DSE},(p)} \otimes H_{{\rm DSE},(q)}, \ \  S(H_{{\rm DSE},(p)})
\subseteq H_{{\rm DSE},(p)}.
\end{equation}
Define a new parameter ${\rm val}$ for Feynman diagrams in $H_{{\rm DSE}}$
given by
\begin{equation}
{\rm val}(\Gamma):= {\rm Max} \{ n \in \mathbb{N}: \ \ \Gamma \in
\bigoplus_{k \ge n} H_{{\rm DSE}, (k)} \}
\end{equation}
which determines the $n$-adic metric on $H_{{\rm DSE}}$ defined by
the formula
\begin{equation}
d_{{\rm adic}}(\Gamma_{1}, \Gamma_{2}):= 2^{-{\rm val}(\Gamma_{1} - \Gamma_{2})}.
\end{equation}
Thanks to Proposition 4.6 in \cite{shojaeifard-10}, the $n$-adic
metric enables us to build a sequence $\{R_{n}({\rm DSE})\}_{n \ge 1}$ of
random graphs which is convergent to the unique solution $X_{{\rm
DSE}}$ with respect to the cut-distance topology. Now if we apply
the graphon representations of Feynman diagrams, then we can embed
$H_{{\rm DSE}}$ into the renormalization topological Hopf algebra of
Feynman graphons $\mathcal{S}^{\Phi}_{{\rm graphon}}$. We can consider the completed enrichment of
$H^{{\rm cut}}_{{\rm DSE}}$ with respect to the cut-distance topology to add graph-limits to this topological Hopf algebra. The
coproduct $\Delta_{H_{{\rm DSE}}}$ is a linear bounded map $H^{{\rm cut}}_{{\rm DSE}}$ which makes it a continuous
operator. Thanks to the graduation parameter, the antipode is also a
continuous operator in this setting. Denote $H_{{\rm DSE}}^{{\rm
cut}}$ as the resulting topological Hopf algebra.

Now for each $j\ge 1$, set $H_{{\rm DSE}_{p}^{(j)}}^{\rm
cut}$ as the resulting topological Hopf subalgebra corresponding to
each equation ${\rm DSE}_{p}^{(j)}$. The family $\mathcal{C}_{p}$ of
these topological Hopf subalgebras can be equipped by the following
binary relation
\begin{equation}
V \preccurlyeq W   \Longleftrightarrow
\end{equation}
there exists a finite sequence of topological Hopf subalgebras
$V_{1},...,V_{r} \in \mathcal{C}_{p}$ together with injective
morphisms $V \rightarrow V_{1} \rightarrow V_{2} \rightarrow ...
\rightarrow V_{r} \rightarrow W$ which connect $V$ to $W$.

We can check that $(\mathcal{C}_{p},\preccurlyeq)$ is a lattice of topological
Hopf subalgebras with the greatest lower bound
$H(X^{(0)}_{1_{(0)}})^{{\rm cut}}$. For each subset $\{V_{1},
V_{2}\}$ of $(\mathcal{C}_{p},\preccurlyeq)$, if $V_{1} \preccurlyeq
V_{2}$ then define
\begin{equation}
V_{1} \wedge V_{2}:= V_{1}, \ \ \ V_{1} \vee V_{2}:= V_{2}.
\end{equation}
The lattice $(\mathcal{C}_{p},\preccurlyeq)$ enables us to relate a
class of Dyson--Schwinger equations derived from the basic equation
${\rm DSE}_{p}$ to each other with respect to morphisms among their
corresponding topological Hopf subalgebras. This means that the
lattice $(\mathcal{C}_{p},\preccurlyeq)$ mathematically describes
the transferring of information from one equation to another. It describes the informational bridge between $p$ and its intermediate states.
\end{proof}

\begin{defn} \label{region-cut-1}
Set $R^{p}$ as the smallest collection of all Feynman diagrams in
$\Phi$ which contribute to the equations ${\rm DSE}_{p}$ and ${\rm
DSE}_{p}^{(j)}$ for all $j\ge1$. Then equip it with the cut-distance
topology and add garph limits to obtain a complete topological
region $\overline{R^{p}}$.
\end{defn}

Thanks to Theorem \ref{entanglement-1}, the lattice
$(\mathcal{C}_{p},\preccurlyeq)$ shows us the information flow
processes among all (virtual) particles which contribute to the
topological region $\overline{R^{p}}$.

\begin{thm} \label{entanglement-2}
There exists a lattice of topological Hopf subalgebras which describes the information flow in
an entangled system of elementary particles in a given interacting Quantum
Field Theory.
\end{thm}

\begin{proof}
At the first step, we are going to show the existence of a class of
Dyson--Schwinger equations for the description of information flow
between two space-time far distant particles which belong to an
entangled system in a given interacting Quantum Field Theory.

Thanks to Theorem \ref{entanglement-1}, we already have described
the entanglement process in a topological region around an
elementary particle on the basis of the cut-distance topology. Here
we need to show the possibility of information flow between two
entangled particles $p,q$ while $p$ does not belong to $\overline{R^{q}}$, $q$ does not belong to $\overline{R^{p}}$ and $\overline{R^{p}} \cap \overline{R^{q}} = \emptyset$.

We have identified topological subspaces $\overline{R^{p}}$ and
$\overline{R^{q}}$ of the topological Hopf algebra $H_{{\rm
FG}}^{{\rm cut}}(\Phi)$ in terms of their contribution to the entanglement of
intermediate states (i.e. Definition \ref{region-cut-1}). Now thanks to
the metric (\ref{distance-feynman-1}), define the distance between
this class of regions in terms of
\begin{equation} \label{region-distance-1}
d(\overline{R^{p}},\overline{R^{q}}):= {\rm inf} \{d (X,Y): \ \ X
\in \overline{R^{p}}, \ Y \in \overline{R^{q}}\}.
\end{equation}
We want to show the existence of topological regions such as $R^{c_{pq}}$ in
$H_{{\rm FG}}^{{\rm cut}}(\Phi)$ with the following conditions
\begin{equation}
\overline{R^{p}} \cap \overline{R^{c_{pq}}} \neq \emptyset, \ \
\overline{R^{q}} \cap \overline{R^{c_{pq}}} \neq \emptyset.
\end{equation}

For $d(\overline{R^{p}},\overline{R^{q}}) >0$, there exist
$j_{1},j_{2} \ge 0$ such that the corresponding equations ${\rm
DSE}_{p}^{(j_{1})}$ and ${\rm DSE}_{q}^{(j_{2})}$ have the following conditions
\begin{equation}
X_{{\rm DSE}_{p}^{(j_{1})}} = {\rm lim}_{n \rightarrow \infty}
\sum_{k=0}^{n} X_{k}^{(j_{1})}, \ \ \  X_{{\rm DSE}_{q}^{(j_{2})}} =
{\rm lim}_{n \rightarrow \infty} \sum_{k=0}^{n} X_{k}^{(j_{2})}
\end{equation}
\begin{equation}
d(\overline{R^{p}},\overline{R^{q}}) = d_{{\rm cut}}(X_{{\rm
DSE}_{p}^{(j_{1})}},X_{{\rm DSE}_{q}^{(j_{2})}})>0.
\end{equation}
For each $\epsilon >0,$ we can determine Hochschild one cocycles
$B^{+}_{\gamma_{n}^{\epsilon},p}$, $n \ge 1$ which fulfills the
following conditions:

- Each $\gamma_{n}^{\epsilon}$ is a finite primitive (1PI) Feynman
diagram such that
\begin{equation}
\forall n \ge 1, \ \ \gamma_{n}^{\epsilon} \notin R^{p}, \ \
\gamma_{n}^{\epsilon} \notin R^{q}, \ \ \gamma_{n}^{\epsilon} \in
H_{\rm FG}(\Phi).
\end{equation}

- The equation ${\rm DSE}_{p}^{(\epsilon)}$ as the Dyson--Schwinger
equation originated from the family
$\{B^{+}_{\gamma_{n}^{\epsilon},p}\}_{n \ge 1}$ with the unique
solution $X_{\epsilon}^{p}= \sum_{n \ge 0} X_{n}^{(\epsilon)p}$ has
the following property that there exists $N_{\epsilon} \in
\mathbb{N}$ such that for each $n > N_{\epsilon}$, we have
$d(X_{n}^{(j_{1})},X_{n}^{(\epsilon)p}) \le \epsilon.$

Thanks to the triangle inequality of the cut-distance metric, we can
obtain
\begin{equation} \label{region-1}
d(X_{{\rm DSE}_{p}^{(j_{1})}},X_{\epsilon}^{p}) \le \epsilon
\end{equation}

In addition, for each $\epsilon >0,$ we can determine Hochschild one
cocycles $B^{+}_{\eta_{n}^{\epsilon},q}$, $n \ge 1$ which fulfills
the following conditions:

- Each $\eta_{n}^{\epsilon}$ is a finite primitive (1PI) Feynman
diagram such that
\begin{equation}
\forall n \ge 1, \ \ \eta_{n}^{\epsilon} \notin R^{p}, \ \
\eta_{n}^{\epsilon} \notin R^{q}, \ \ \eta_{n}^{\epsilon} \in H_{\rm
FG}(\Phi).
\end{equation}

- The equation ${\rm DSE}_{q}^{(\epsilon)}$ as the Dyson--Schwinger
equation originated from the family
$\{B^{+}_{\eta_{n}^{\epsilon},q}\}_{n \ge 1}$ with the unique
solution $X_{\epsilon}^{q}= \sum_{n \ge 0} X_{n}^{(\epsilon)q}$ has
the following property that there exists $N'_{\epsilon} \in
\mathbb{N}$ such that for each $n > N'_{\epsilon}$, we have
$d(X_{n}^{(j_{2})},X_{n}^{(\epsilon)q}) \le \epsilon.$

Thanks to the triangle inequality of the cut-distance metric, we can
obtain
\begin{equation} \label{region-2}
d(X_{{\rm DSE}_{q}^{(j_{2})}},X_{\epsilon}^{q}) \le \epsilon.
\end{equation}

The vacuum in an interacting physical theory can be described as a
homogeneous system of virtual particles where its states are
invariant by all transformations of the invariance group. Some
particles in the vacuum have negative energies where without
violating the conservation laws they can annihilate
(\cite{roberts-1}). Thanks to this fact, we can determine
$R^{c_{pq}}$ as the smallest topological subset of $H^{{\rm
cut}}_{{\rm FG}}(\Phi)$ consisting of Feynman graphs which
contribute to equations of the types ${\rm DSE}^{(\epsilon)}_{p}$
and ${\rm DSE}^{(\epsilon)}_{q}$. This region contains virtual
particles (created by the vacuum energy) which has separate
contributive parts with topological regions $R^{p}$ and $R^{q}$. The
relations (\ref{region-1}) and (\ref{region-2}) guarantee that
\begin{equation}
R^{p} \cap R^{c_{pq}} \neq \emptyset, \ \ \  R^{q} \cap R^{c_{pq}}
\neq \emptyset.
\end{equation}
Now if we apply Theorem \ref{entanglement-1}, then graphs which
belong to the region $\overline{R^{c_{pq}}}$ (as the completion of
$R^{c_{pq}}$ with respect to the cut-distance topology) make
informational bridges between entangled particles $p$, $q$ and their
corresponding intermediate states (virtual particles) which live in
$\overline{R^{p}} \cup \overline{R^{q}} \cup \overline{R^{c_{pq}}}$.

At the second step, we want to encapsulate the above machinery in terms of lattice models.
Suppose ${\rm DSE}_{p}$ and ${\rm DSE}_{q}$ are the basic
Dyson--Schwinger equations corresponding to entangled particles $p$
and $q$. Thanks to the built lattice by Theorem
\ref{entanglement-1}, let $(\mathcal{C}_{p},\preccurlyeq)$ be the
lattice of topological Hopf subalgebras $H^{{\rm cut}}_{{\rm
DSE}_{p}^{(j)}}$ generated by equations of the type ${\rm
DSE}_{p}^{(j)}$ which live in the topological region
$\overline{R^{p}}$ and let $(\mathcal{C}_{q},\preccurlyeq)$ be the
lattice of topological Hopf subalgebras $H^{{\rm cut}}_{{\rm
DSE}_{q}^{(l)}}$ generated by equations of the type ${\rm
DSE}_{q}^{(l)}$ which live in the topological region
$\overline{R^{q}}$. Thanks to the previous part of the proof, we can show the existence
of $j_{1}, j_{2} \ge 0$ such that ${\rm DSE}_{p}^{(j_{1})}$ and
${\rm DSE}_{p}^{(j_{2})}$ contribute in the description of the
distance between two topological regions $\overline{R^{p}}$ and
$\overline{R^{q}}$ (i.e. metric (\ref{region-distance-1})). Set
\begin{equation}
j^{*}_{1}:= {\rm Min} \{j_{1}: {\rm DSE}_{p}^{(j_{1})}\}, \ \ \
j^{*}_{2}:= {\rm Min} \{j_{2}: {\rm DSE}_{q}^{(j_{2})}\}.
\end{equation}
Consider the sub-lattice
$(\mathcal{C}^{j^{*}_{1}}_{p},\preccurlyeq)$ which contains only the
first $j^{*}_{1}$ columns of the original lattice
$(\mathcal{C}_{p},\preccurlyeq)$ and the sub-lattice
$(\mathcal{C}^{j^{*}_{2}}_{q},\preccurlyeq)$ which contains only the
first $j^{*}_{2}$ columns of the original lattice
$(\mathcal{C}_{q},\preccurlyeq)$. These two sub-lattices have the
greatest lower bound and the smallest upper bound.

Thanks to the structure of the topological region $R^{c_{pq}}$, we
can build a new lattice
$(\mathcal{C}^{j_{1}^{*}j_{2}^{*}}_{c_{pq}},\preccurlyeq)$ which is
the result of the disjoint union of the sub-lattices
$(\mathcal{C}^{j^{*}_{1}}_{p},\preccurlyeq)$ and
$(\mathcal{C}^{j^{*}_{2}}_{q},\preccurlyeq)$ which are connected to
each other by topological Hopf algebra homomorphisms associated to
Dyson--Schwinger equations of the types ${\rm DSE}^{(\epsilon)}_{p}$
and ${\rm DSE}^{(\epsilon)}_{q}$. We have
\begin{equation} \label{1}
H^{{\rm cut}}_{{\rm DSE}_{p}} \le H^{{\rm cut}}_{{\rm
DSE}_{p}^{(1)}} \le ... \le H^{{\rm cut}}_{{\rm
DSE}_{p}^{(j^{*}_{1})}}
\end{equation}
\begin{equation} \label{2}
H^{{\rm cut}}_{{\rm DSE}_{q}} \le H^{{\rm cut}}_{{\rm
DSE}_{q}^{(1)}} \le ... \le H^{{\rm cut}}_{{\rm
DSE}_{q}^{(j^{*}_{2})}}
\end{equation}
which belong to the new lattice in terms of one of the following
topological Hopf algebra homomorphisms
\begin{equation}
H^{{\rm cut}}_{{\rm DSE}_{p}^{(j^{*}_{1})}} \longrightarrow H^{{\rm
cut}}_{{\rm DSE}_{p}^{(\epsilon)}} \longrightarrow H^{{\rm cut}}_{{\rm
DSE}_{c}^{(k)}} \longrightarrow H^{{\rm cut}}_{{\rm
DSE}_{q}^{(\epsilon)}} \longrightarrow H^{{\rm cut}}_{{\rm
DSE}_{q}^{(j^{*}_{2})}}
\end{equation}
or
\begin{equation}
H^{{\rm cut}}_{{\rm DSE}_{q}^{(j^{*}_{2})}} \longrightarrow H^{{\rm
cut}}_{{\rm DSE}_{q}^{(\epsilon)}} \longrightarrow H^{{\rm cut}}_{{\rm
DSE}_{c}^{(k)}} \longrightarrow H^{{\rm cut}}_{{\rm
DSE}_{p}^{(\epsilon)}} \longrightarrow H^{{\rm cut}}_{{\rm
DSE}_{p}^{(j^{*}_{1})}}
\end{equation}
such that $H^{{\rm cut}}_{{\rm DSE}_{c}^{(k)}}$ is the topological
Hopf subalgebra associated to the Dyson--Schwinger equation ${\rm
DSE}_{c}^{(k)}$ which lives in the region $R^{c_{pq}}$ and derived
from the virtual particle $c$.
\end{proof}

\subsection{\textsl{Intermediate algorithms for QFT-entanglement}}

The Milnor--Moore theorem provides a correspondence between
pro-unipotent Lie groups and graded commutative Hopf algebras
(\cite{cartier-1, milne-1}). This correspondence enables us to
translate the determination of Hopf subalgebraic structures in the
renormalization Hopf algebra to a problem in Lie groups. One
interesting application of Galois theory is to find a fundamental
interrelationship between subgroups of the group of all
automorphisms and intermediate algorithmic structures which live in
the middle of programs and computable functions \cite{yanofsky-1,
yanofsky-2}. Dyson--Schwinger equations
can be applied for the determination of substructures in the
renormalization Hopf algebra $H_{{\rm FG}}(\Phi)$ and the
determination of Lie subgroups of the complex Lie group
$\mathbb{G}_{\Phi}(\mathbb{C})$ of characters. We can consider each $H_{{\rm DSE}}$ as the Hopf ideal in $H_{{\rm FG}}(\Phi)$ and then consider the resulting quotient Hopf subalgebra. The complex Lie group corresponding to this quotient Hopf algebra is a Lie subgroup of $\mathbb{G}_{\Phi}(\mathbb{C})$ while we have a surjective group homomorphism from $\mathbb{G}_{\Phi}(\mathbb{C})$ to $\mathbb{G}_{{\rm DSE}}(\mathbb{C})$. Therefore these
non-perturbative equations play the middle bridge between Theory of Computation and Quantum Field Theory \cite{delaney-marcolli-1,
shojaeifard-8}.

Theorem \ref{entanglement-1} and Theorem \ref{entanglement-2} have
explained the entanglement of elementary particles in terms of
the existence of substructures originated from Dyson--Schwinger
equtations. This new mathematical platform can address a deep connection
between the concept of information flow in Quantum Field Theory and
the existence of subobjects inside an object determined by the
Galois Fundamental Theorem. The immediate consequence of this
investigation is to recognize a new approach to quantum entanglement
in the language of intermediate algorithms in Theoretical Computer
Science. We deal with this interesting challenge on the basis of the
representation theory of Lie groups.

\begin{thm} \label{entanglement-3}
The intermediate algorithms corresponding to Lie sub-groups of the complex Lie group $\mathbb{G}^{\Phi}_{{\rm graphon}}(\mathbb{C})$ can encode the information flow in an entangled system in a given interacting Quantum Field Theory $\Phi$ in terms of a lattice of Lie subgroups.
\end{thm}

\begin{proof}
At the first step, we plan to show that the information flow between the particle $p$ and all unobserved
intermediate states can be encoded via a lattice of Lie subgroups of $\mathbb{G}_{\Phi}(\mathbb{C})$.
Thanks to Theorem \ref{entanglement-1}, the entanglement of the
particle $p$ and its related virtual particles is encapsulated
by the lattice $(\mathcal{C}_{p},\preceq)$ of topological Hopf
subalgebras. Each pair of Hopf subalgebras $H_{{\rm DSE}_{p}^{(k)}}
\preceq H_{{\rm DSE}_{p}^{(l)}}$ in this lattice determines the
injective Hopf algebra homomorphism
\begin{equation}
i_{kl}: H_{{\rm DSE}_{p}^{(k)}} \longrightarrow H_{{\rm
DSE}_{p}^{(l)}}.
\end{equation}
The passing from Hopf subalgebras to Lie subgroups can be formulated
by applying the contravariant functor ${\rm Spec}$ which sends a
commutative algebra to a topological space. For each object $H_{{\rm
DSE}_{p}^{(k)}} \in (\mathcal{C}_{p},\preceq)$, ${\rm Spec}(H_{{\rm
DSE}_{p}^{(k)}})$ is the set of all prime ideals of the commutative
algebra $H_{{\rm DSE}_{p}^{(k)}}$ equipped with the Zariski topology
and the structure sheaf. The homomorphism $i_{kl}$ can be lifted
onto the surjective homomorphism
\begin{equation}
\widetilde{i}_{kl}:{\rm Spec}(H_{{\rm DSE}_{p}^{(l)}})
\longrightarrow {\rm Spec}(H_{{\rm DSE}_{p}^{(k)}})
\end{equation}
of affine group schemes. For a fixed Hopf subalgebra $H_{{\rm
DSE}_{p}^{(k)}}$, ${\rm Spec}(H_{{\rm DSE}_{p}^{(k)}})$ is a
representable covariant functor which sends a topological space to a
group. Set $G_{{p}^{(k)}}:= {\rm Spec}(H_{{\rm
DSE}_{p}^{(k)}})(\mathbb{C})$ as the complex Lie subgroup
corresponding to the Hopf subalgebra $H_{{\rm DSE}_{p}^{(k)}}$ such
that its group structure is given by the convolution product
generated by the coproduct $\Delta_{H_{{\rm DSE}_{p}^{(k)}}}$.
Thanks to this setting, we can build a new lattice
$(\mathcal{G}_{p},\succeq)$ of complex Lie groups such that
\begin{equation}
G \succeq K \Longleftrightarrow
\end{equation}
There exists a finite sequence of complex Lie subgroups
$G_{1},...,G_{r}\in \mathcal{G}_{p}$ together with surjective group
homomorphisms $G \rightarrow G_{1} \rightarrow G_{2} \rightarrow ...
\rightarrow G_{r} \rightarrow K$ which connect $G$ to $K$.

In
addition, for each $n \ge 1$, define $G(X_{1}^{j},...,X_{n}^{j})$ as
the finite dimensional complex Lie subgroup corresponding to the
free commutative graded Hopf subalgebra $H(X_{1}^{j},...,X_{n}^{j})$
of $H_{{\rm DSE}_{p}^{(j)}}$. The lattice $(\mathcal{G}_{p},\succeq)$ of Lie groups encodes the
information flow between $p$ and its related virtual particles.

At the second step, we plan to show that there exists a lattice of Lie subgroups which describes the information flow in an entangled system of elementary particles in a given interacting gauge field theory.
For this purpose, we need to build a lattice of Lie subgroups for the description of the
quantum entanglement process between space-time far distant
elementary particles which belong to an entangled system in the
physical theory $\Phi$.

Theorem \ref{entanglement-2} determines a lattice
$(\mathcal{C}_{c_{pq}}^{j^{*}_{1}j^{*}_{2}},\preceq)$ of Hopf
subalgebras which describes the entanglement process. Thanks to the first part of the Proof, it is possible to lift the increasing chains (\ref{1}),
(\ref{2}) onto the following decreasing chains of Lie subgroups
\begin{equation} \label{3}
G_{{p}^{(j^{*}_{1})}} \ge G_{{p}^{(j^{*}_{1}-1)}} \ge ... \ge
G_{{p}^{(1)}} \ge G_{{\rm DSE}_{p}}
\end{equation}
\begin{equation} \label{4}
G_{{q}^{(j^{*}_{2})}} \ge G_{{q}^{(j^{*}_{2}-1)}} \ge ... \ge
G_{{q}^{(1)}} \ge G_{{\rm DSE}_{q}}
\end{equation}
with the corresponding group homomorphisms
\begin{equation} \label{5}
G_{q^{(j^{*}_{2})}} \longrightarrow G_{q^{(\epsilon)}}
\longrightarrow G_{c^{(k)}} \longrightarrow G_{p^{(\epsilon)}}
\longrightarrow G_{p^{(j^{*}_{1})}}
\end{equation}
or
\begin{equation} \label{6}
G_{p^{(j^{*}_{1})}} \longrightarrow G_{p^{(\epsilon)}}
\longrightarrow G_{c^{(k)}} \longrightarrow G_{q^{(\epsilon)}}
\longrightarrow G_{q^{(j^{*}_{2})}}
\end{equation}
such that $G_{c^{(k)}}$ is the complex Lie subgroup corresponding to
the Hopf subalgebra $H_{{\rm DSE}_{c}^{(k)}}$ generated by the
equation ${\rm DSE}_{c}^{(k)}$ which lives in the topological region
$\overline{R^{c_{pq}}}$. The existence of the virtual particle $c$,
which contributes to interactions of the particles $p,q$, is the
consequence of the energy of the vacuum in interacting physical theory.
Now we can define a new lattice
\begin{equation}
(\mathcal{G}_{c_{pq}}^{j^{*}_{1}j^{*}_{2}},\succeq)
\end{equation}
of Lie subgroups and Lie group epimorphisms. This lattice encodes
the entanglement processes between $p$ and $q$.

As we have shown in the previous parts of this work, the renormalization Hopf algebra of Feynman graphons $\mathcal{S}^{\Phi}_{{\rm graphon}}$ is capable to recover the renormalization Hopf algebra $H_{{\rm FG}}(\Phi)$ and Hopf subalgebras generated by all Dyson--Schwinger equations. Therefore for each Dyson--Schwinger equation DSE with the corresponding Hopf subalgebra $H_{{\rm DSE}}$, we can cover the associated complex Lie subgroup $G_{{\rm DSE}}(\mathbb{C})$ via the complex Lie group $\mathbb{G}^{\Phi}_{{\rm graphon}}(\mathbb{C})$ by a natural surjection while in the quotient Hopf algebra level ${\rm Hom}(\frac{\mathcal{S}^{\Phi}_{{\rm graphon}}}{H_{{\rm DSE}}},\mathbb{C})$ is a Lie subgroup of $\mathbb{G}^{\Phi}_{{\rm graphon}}(\mathbb{C})$.
\end{proof}

\subsection{\textsl{Tannakian subcategories as intermediate algorithms}}

The study of Dyson--Schwinger equations had been developed to a
categorical setting where we associated a category of geometric
objects to each system of these equations. Then we have embedded this class of categories into the
universal Connes--Marcolli category $\mathcal{E}^{{\rm CM}}$ of flat
equi-singular vector bundles. Thanks to this machinery, some new
geometric and combinatorial tools for the computation of
non-perturbative parameters have already been obtained
\cite{shojaeifard-1, shojaeifard-3, shojaeifard-5, shojaeifard-8}.
Thanks to Theorem \ref{entanglement-3}, now it is possible to describe quantum
entanglement in the context of the representation theory of Lie
groups where we can address a new application of Tannakian
formalism to Quantum Field Theory.

\begin{thm} \label{entanglement-5}
There exists a lattice of Tannakian subcategories which describes
quantum entanglement in a given interacting Quantum Field Theory.
\end{thm}

\begin{proof}
At the first step, we show the existence of a lattice $({\rm
Cat}_{p},\succeq)$ of Tannakian
subcategories which encodes the quantum entanglement between an
elementary particle $p$ and all unobserved intermediate states (as
virtual particles).

Thanks to Theorem \ref{entanglement-1}, the entanglement of the
particle $p$ and its related virtual particles is encapsulated
by the lattice $(\mathcal{C}_{p},\preceq)$ of topological Hopf
subalgebras. Each pair of objects $H^{{\rm cut}}_{{\rm DSE}_{p}^{(k)}}
\preccurlyeq H^{{\rm cut}}_{{\rm DSE}_{p}^{(l)}}$ can determine the
natural injective Hopf algebra homomorphism
\begin{equation}
i_{kl}: H_{{\rm DSE}_{p}^{(k)}} \longrightarrow H_{{\rm
DSE}_{p}^{(l)}}
\end{equation}
which can be lifted onto the surjective group homomorphism
\begin{equation}
\widetilde{i}_{kl}:G_{{p}^{(l)}} \longrightarrow G_{{p}^{(k)}}.
\end{equation}
For each object $G_{{p}^{(l)}}$ of the lattice
$(\mathcal{G}_{p},\succeq)$, let
\begin{equation}
G^{*}_{{p}^{(l)}}:= G_{{p}^{(l)}} \rtimes \mathbb{G}_{m}
\end{equation}
such that $\mathbb{G}_{m}$ is the multiplicative group which acts on
the original group.

Define ${\rm Rep}_{G^{*}_{{p}^{(l)}}}$ as the
category of finite dimensional representations of the complex Lie
group $G^{*}_{{p}^{(l)}}$ which is a neutral Tannakian category.
Thanks to the representation theory of affine group schemes
(\cite{milne-1}), the surjective morphism $\widetilde{i}_{kl}$
allows us to send each representation $\sigma: G_{{p}^{(k)}}
\longrightarrow GL_{V}$ to a representation
\begin{equation}
\sigma \circ \widetilde{i}_{kl}: G_{{p}^{(l)}} \longrightarrow
GL_{V}
\end{equation}
which leads us to achieve an exact fully faithful functor
\begin{equation}
{\rm Rep}_{G^{*}_{{p}^{(k)}}} \longrightarrow {\rm
Rep}_{G^{*}_{{p}^{(l)}}}.
\end{equation}

This information is enough to build a new lattice $({\rm
Cat}_{p},\succeq)$ of subcategories such that
\begin{equation}
{\rm Rep}_{H^{*}} \succeq {\rm Rep}_{K^{*}} \Longleftrightarrow
\end{equation}
there exists a finite sequence of subcategories ${\rm Rep}_{H_{1}^{*}},
..., {\rm Rep}_{H_{t}^{*}} \in {\rm Cat}_{p}$ together with exact
fully faithful functors
\begin{equation}
{\rm Rep}_{K^{*}} \rightarrow {\rm Rep}_{H_{1}^{*}} \rightarrow  ...
\rightarrow {\rm Rep}_{H_{t}^{*}} \rightarrow {\rm Rep}_{H^{*}}
\end{equation}
derived from the epimorphisms $\widetilde{i}_{kl}$.

At the second step, we show the existence of a lattice of Tannakian
subcategories for the description of the entanglement between
space-time far distant elementary particles which belong to an
entangled system.

Thanks to Theorem \ref{entanglement-2}, the information flow between
$p$ and other distant particle $q$ is described by the lattice
$(\mathcal{C}_{c_{pq}}^{j^{*}_{1}j^{*}_{2}},\preceq)$ of topological
Hopf subalgebras which inherits a lattice
$(\mathcal{G}_{c_{pq}}^{j^{*}_{1}j^{*}_{2}},\succeq)$ of Lie
subgroups (i.e. Theorem \ref{entanglement-3}). The decreasing chains
(\ref{3}) and (\ref{4}) can be lifted onto the categorical setting
to achieve the following chains of categories and exact fully
faithful functors
\begin{equation}
{\rm Rep}_{G^{*}_{{\rm DSE_{p}}}} \ge {\rm Rep}_{G^{*}_{p^{(1)}}}
\ge ... \ge {\rm Rep}_{G^{*}_{p^{(j^{*}_{1}-1)}}} \ge {\rm
Rep}_{G^{*}_{p^{{(j^{*}_{1})}}}}
\end{equation}
\begin{equation}
{\rm Rep}_{G^{*}_{{\rm DSE_{q}}}} \ge {\rm Rep}_{G^{*}_{q^{(1)}}}
\ge ... \ge {\rm Rep}_{G^{*}_{q^{(j^{*}_{2}-1)}}} \ge {\rm
Rep}_{G^{*}_{q^{(j^{*}_{2})}}}.
\end{equation}
We can connect these two chains to each other in terms of one of the
following sequences of exact fully faithful functors
\begin{equation} \label{f7}
{\rm Rep}_{G^{*}_{q^{(j^{*}_{2})}}} \longrightarrow {\rm
Rep}_{G^{*}_{q^{(\epsilon)}}} \longrightarrow {\rm
Rep}_{G^{*}_{c^{(k)}}} \longrightarrow {\rm
Rep}_{G^{*}_{p^{(\epsilon)}}} \longrightarrow {\rm
Rep}_{G^{*}_{p^{(j^{*}_{1})}}}
\end{equation}
or
\begin{equation} \label{f8}
{\rm Rep}_{G^{*}_{p^{(j^{*}_{1})}}} \longrightarrow {\rm
Rep}_{G^{*}_{p^{(\epsilon)}}} \longrightarrow {\rm
Rep}_{G^{*}_{c^{(k)}}} \longrightarrow {\rm
Rep}_{G^{*}_{q^{(\epsilon)}}} \longrightarrow {\rm
Rep}_{G^{*}_{q^{(j^{*}_{2})}}}.
\end{equation}

This information is enough to define $({\rm Cat}_{c_{pq}}^{j^{*}_{1}j^{*}_{2}},\succeq)$
as a lattice of Tannakian subcategories which encodes the
entanglement process between $p$ and $q$.
\end{proof}

\begin{thm} \label{entanglement-6}
Flat equi-singular vector bundles provide a new geometric description for the information flow in interaction Quantum Field Theories on the basis of the Riemann--Hilbert correspondence.
\end{thm}

\begin{proof}
Flat equi-singular vector bundles have been applied for the construction of
the Connes--Marcolli category $\mathcal{E}^{{\rm CM}}$ which is a neutral Tannakian category.
It is isomorphic to the category ${\rm Rep}_{\mathbb{U}^{*}}$ of finite dimensional
representations of the universal affine group scheme
$\mathbb{U}^{*}$. \cite{connes-marcolli-1}

This universal category is rich enough to recover all categories ${\rm
Rep}_{G^{*}_{p^{(j)}}}$ as subcategories which enable us to define a
surjective functor
\begin{equation} \label{universal-1}
\pi_{j}^{*}: {\rm Rep}_{\mathbb{U}^{*}} \longrightarrow {\rm
Rep}_{G^{*}_{p^{(j)}}}
\end{equation}
of categories. Now if we apply the chain (\ref{f7}) or
(\ref{f8}), then we can determine one of the functors
\begin{equation} \label{entan-1}
\rho^{pq}_{j^{*}_{1}j^{*}_{2}}: {\rm Rep}_{G^{*}_{p^{(j^{*}_{1})}}}
\longrightarrow {\rm Rep}_{G^{*}_{q^{(j^{*}_{2})}}}
\end{equation}
or
\begin{equation} \label{entan-2}
\rho^{pq}_{j^{*}_{2}j^{*}_{1}}: {\rm Rep}_{G^{*}_{q^{(j^{*}_{2})}}}
\longrightarrow {\rm Rep}_{G^{*}_{p^{(j^{*}_{1})}}}.
\end{equation}
These functors allow us to formulate one of the following equations
\begin{equation}
\pi_{j^{*}_{2}} = \rho^{pq}_{j^{*}_{1}j^{*}_{2}} \circ
\pi_{j^{*}_{1}}, \ \ {\rm or} \ \ \pi_{j^{*}_{1}} =
\rho^{pq}_{j^{*}_{2}j^{*}_{1}} \circ \pi_{j^{*}_{2}}
\end{equation}
at the level of functors. They are the key tools for us to determine some flat equi-singular vector bundles which contribute to the information flow in entangled systems.
\end{proof}

\begin{cor} \label{entangle-7}
There exists a category of mixed Tate motives which interprets the
information flow in an entangled system of particles in a given
interacting Quantum Field Theory.
\end{cor}

\begin{proof}
On the one hand, the category $\mathcal{E}^{{\rm CM}}$ is equivalent to the motivic
category
\begin{equation}
\mathcal{TM}_{{\rm mix}}({\rm Spec} \ \mathcal{O}[1/N])
\end{equation}
of mixed Tate motives (i.e. Proposition 1.110, Corollary 1.111 in
\cite{connes-marcolli-1}). Thanks to Theorem \ref{entanglement-6},
neutral Tannakian subcategories ${\rm
Rep}_{G^{*}_{p^{(j^{*}_{1})}}}$ and ${\rm
Rep}_{G^{*}_{q^{(j^{*}_{2})}}}$ can be embedded into this motivic
category. Therefore the information flow is equivalent to
determining subcategories of the category $\mathcal{TM}_{{\rm
mix}}({\rm Spec} \ \mathcal{O}[1/N])$ which contain those mixed Tate
motives identified by motivic Galois groups
$G^{*}_{p^{(j^{*}_{1})}}$ and $G^{*}_{q^{(j^{*}_{2})}}$.
We denote the resulting motivic subcategories with ${\rm
Mot}(G^{*}_{p^{(j^{*}_{1})}})$ and ${\rm
Mot}(G^{*}_{q^{(j^{*}_{2})}})$, respectively.

On the other hand,
thanks to Theorem \ref{entanglement-3}, we have already determined a
new class of Dyson--Schwinger equations ${\rm DSE}^{(k)}_{c}$ which
describes the information flow in the topological region
$\overline{R^{c_{pq}}}$. Now by applying Theorem
\ref{entanglement-5}, we have the category ${\rm
Rep}_{G^{*}_{c^{(k)}}}$ with respect to this class of
Dyson--Schwinger equations which can encode the information flow. This
category can be also embedded into $\mathcal{E}^{{\rm CM}}$ which
leads us to characterize another class of mixed Tate motives
identified with the motivic Galois group $G^{*}_{c^{(k)}}$. We
denote the resulting motivic subcategory with ${\rm
Mot}(G^{*}_{c^{(k)}})$.

The disjoint union of objects of these
motivic subcategories make a new subcategory
\begin{equation}
{\rm Mot}_{pq}:= {\rm Mot}(G^{*}_{p^{(j^{*}_{1})}}) \bigsqcup {\rm
Mot}(G^{*}_{c^{(k)}}) \bigsqcup {\rm Mot}(G^{*}_{q^{(j^{*}_{2})}})
\end{equation}
of $\mathcal{TM}_{{\rm mix}}({\rm Spec} \ \mathcal{O}[1/N])$. ${\rm
Mot}_{pq}$ is a category of mixed Tate motives which contribute to
the entanglement processes between space-time far distant particles
$p,q$ via virtual particles of the vacuum energy.
\end{proof}

The explained mathematical machinery for the description of quantum
entanglement is related to the strength of the bare or running coupling constants of
physical theories where we deal with Dyson--Schwinger equations under different running couplings. It means that
our machinery still works after changing the scale
of the momenta (i.e. running coupling constant). As we know that
theory of Renormalization Group is the key tool in Quantum Field
Theory to study the changes of the dynamics of a quantum system with
infinite degrees of freedom when the scales of some physical
parameters such as momentum, energy and mass have been changed. It
allows us to concern the possibility of exchanging information from
scale to scale in the appearance of uncertainty principle. Now we can apply
the Connes--Marcolli universal affine group scheme to define a
suitable Renormalization Group which encodes the information flow
under the rescaling of the momentum parameter.

\begin{cor} \label{entangle-8}
The information flow between an elementary particle $p$ and all
unobserved intermediate states exists independent of changing the
scale of the momenta of particles.
\end{cor}

\begin{proof}
The universal category $\mathcal{E}^{{\rm CM}}$ is isomorphic to the
category ${\rm Rep}_{\mathbb{U}^{*}}$ with respect to
the universal affine group scheme $\mathbb{U}^{*}$.  The
Connes--Marcolli universal shuffle type Hopf algebra of
renormalization $H_{\mathbb{U}}$ is the result of the graded dual of
the universal enveloping algebra of the free graded Lie algebra
$L_{\mathbb{U}}$ which is generated by elements $e_{-n}$ of degree
$-n$ for each $n>0$. The Milnor--Moore Theorem can determine the
corresponding affine group scheme $\mathbb{U}$. The sum $e:=
\sum_{n} e_{-n}$ is an element of the Lie algebra $L_{\mathbb{U}}$
where thanks to the pro-unipotent structure of $\mathbb{U}$, we can
lift it onto the morphism ${\rm rg}: \mathbb{G}_{a} \rightarrow
\mathbb{U}$. \cite{connes-marcolli-1}

Now we can apply Theorem 1.106 in \cite{connes-marcolli-1} to determine a
graded representation
\begin{equation}
\tau_{p}:\mathbb{U}^{*}(\mathbb{C}) \rightarrow G^{*}_{{\rm
DSE}_{p}}
\end{equation}
such that the resulting map $\tau_{p} \circ {\rm rg}$ provides the
Renormalization Group with respect to the equation ${\rm DSE}_{p}$.
By this method, we can build a Renormalization Group with respect to
each Dyson--Schwinger equation ${\rm DSE}_{p}^{(j)}$ for each $j \ge
1$. These Renormalization Groups allow us to study the
behavior of the information flow in the topological regions such as
$R^{p}$ under changing the scales of the momentum parameter.
\end{proof}

\section{\textsl{Quantum logic via non-perturbative propositional calculus}}

The fundamental purpose in this section is to address a new category
model for the study of logical propositions about solutions of (strongly coupled) Dyson--Schwinger equations. We use cut-distance topological regions of Feynman diagrams to build a new topos model which can encode the role of the strength of coupling constants in the logical evaluation about non-perturbative aspects of physical theories. This new topos model can lead us to clarify a new class of computable Heyting algebras for the logical study of gauge field theories.

\subsection{\textsl{Quantum Topos}}

Generally speaking, a quantum system is described by its von
Neumann algebra $\mathcal{B}(H)$ of observables which contains all
bounded operators on an infinite dimensional separable Hilbert space
$H$. Each physical quantity $A$ has a corresponding self-adjoint
operator $\hat{A}$ in $\mathcal{B}(H)$ and vice versa. Set
$\mathcal{V}(H)$ as the poset of all unital abelian von Neumann
subalgebras of $\mathcal{B}(H)$ which can be seen also as the
context category. For objects $V_{1} \subset V_{2}$ in the context
category, the subalgebra $V_{1}$ has less number of self-adjoint
operators and less number of projections than the subalgebra
$V_{2}$. The restriction process from the subalgebra $V_{2}$ to the
subalgebra $V_{1}$ or the lifting process from  $V_{1}$ onto $V_{2}$
are fundamental issues in propositional calculus of quantum systems.
These translation issues have been studied under coarse-graining
process. On the one hand, it enables us to map self-adjoint
operators and projections from $V_{2}$ to $V_{1}$. For a proposition
$A \in \Delta$ about a given physical quantity $A$, suppose its
corresponding projection $\hat{P}_{A}^{\Delta}$ belongs to $V_{2}$
but not belong to $V_{1}$. It means that this proposition can not be
evaluated from the perspective of $V_{1}$. The daseinisation process
has been designed to adapt the projection $\hat{P}_{A}^{\Delta}$ and
the proposition $A \in \Delta$ to $V_{1}$ by making them coarser. On
the other hand, every self-adjoint operator and every projection in
$V_{1}$ belong also to $V_{2}$ but the embedding of the smaller
subalgebra into the larger one requires some extra structures and
objects which live in $V_{2}$. To concern this issue has led people
to build a topos of contravariant functors from the context category
$\mathcal{V}(H)$ to the category ${\bf Set}$. This categorical
setting has been developed very fast for the reconstruction of
physical theories in the context of higher order logic.
\cite{adelman-corbett-1, birkhoff-vonneumann-1, doring-isham-1,
doring-isham-2, doring-isham-3, doring-isham-4, i1, ib1, ib2, lambek-scott-1,
mclarty-1, maclane-moerdijk-1, soare-1}

The original motivation for the construction of a new
topos model is to provide a new analogous of this propositional calculus
for the study of situations beyond perturbation theory in Quantum
Field Theories with strong couplings in the context of logical
conceptions. Our new topos model shows the importance of the strength of
the (bare) couplings in the construction of the category of context
(i.e. base category). The context category of our new topos is built in terms of cut-distance topological Hopf subalgebras of the enriched renormalization Hopf algebra $H_{{\rm FG}}^{{\rm cut}}(\Phi)$. Its complexity is more than the complexity of the standard context category $\mathcal{V}(H)$. The inclusion
$H^{{\rm cut}}_{{\rm DSE}_{1}} \subset H^{{\rm cut}}_{{\rm DSE}_{2}}$ in our context category does not mean in general
that the equation ${\rm DSE}_{1}$ should have less physical
information than the equation ${\rm DSE}_{2}$. We can remind the
calculus of ordinals in Set Theory, where we deal with different
types of infinities while sometimes a subset of a set and the set can have the
same cardinal. Therefore
coarse-graining process is not noticed in the foundations of
our topos model and we need to concern other parameters to
deal with propositions at the level of large Feynman diagrams.

Let us give a short overview about the concept of topos . The
fundamental motivation for the study of topos came from the concept
of abstraction in Mathematics. In fact, Category Theory, as a modern discipline, comes to the game
whenever we plan to study a general theory of structures. Categories enable us to
concern mathematical structures in terms of interrelationships among
objects (which are formally known as morphisms) while under a set
theoretic perspective, we choose to deal with properties of
mathematical structures on the basis of elements and membership
relations. Many basic concepts such as spaces and elements in Set
Theory can be replaced by objects and arrows in the categorical
setting, respectively. It is reasonable to think about Category
Theory as a generalization of Set Theory where we are capable to
study a mathematical structure in terms of its relations with other
structures. This approach leads us to a universal fundamental
language in dealing with mathematical structures where we have
general powerful tools such as functors between categories and
natural transformations between functors instead of the equality
relation between elements of sets in Set Theory. Actually, the
language of Category Theory provides a new understanding of the
notion of "element" of an object in a mathematical structure which
is more general than its set theoretic version. Each arrow is indeed a
generalized element of its own codomain which means that each object
$X$ can be described in terms of consisting of different collections of
arrows $Y \rightarrow X$. This interpretation is known as the
varying of elements of $X$ over the stages $Y$ which corresponds to
the notion of absolute element $x$ of a set $X$ in Set Theory. This
story is encapsulated in terms of a map   $x:\{*\}
\rightarrow X$ where we enable to address the terminal object
underlying the categorical setting. Questions about the existence of
a class of categories which could be regarded as a
categorical-theoretic replacement and generalization of the category
of sets and functions have led people to build elementary topos and
Grothendieck topos such that the second class is known as a
replacement for the notion of "space". The concept of topos has all
tools of the set-theoretical world which are necessary for the
construction of mathematical structures and their models under a
categorical configuration. It provides a generalized version of the
notions of space and logic where we enable to interpret it as a
categorical-theoretic generalization of the structure of a universe
of sets and functions that disappears certain logical and geometric
restrictions of the base mathematical structure. Generally speaking, for a given
mathematical theory, we can have a treatment to
evaluate and study the theory under different stages with respect to objects
of a fixed base category. So there is a
chance to consider possible relations among toposes in the context
of a special family of functors which are called geometric
morphisms. The category ${\rm {\bf Set}}$ of all sets and functions is an example of a topos which
is the basis for the construction of more complicated toposes such
as the Grothendieck topos of sheaves over a given base category.
Consider the category ${\rm
\bf{Set}}^{\mathcal{C}^{\rm op}}$ of contravariant functors from the
base category $\mathcal{C}$ to the standard category ${\rm
\bf{Set}}$ of sets and functions. Elements of this mathematical
theory corresponded to the base category $\mathcal{C}$, which have
been already modeled as objects in the mentioned topos, become
representable in terms of set-valued functors over the base category
$\mathcal{C}$. Roughly speaking, a topos is a Cartesian closed
category with equalisers and subobject classifier. In other words, a topos has terminal object, equalisers, pullbacks, all other limits,
exponential objects and subobject classifier.
\cite{b1, johnstone-1, lawvere-1, lawvere-rosebrugh-1, maclane-1, p1,
pare-1, street-1}

\subsection{\textsl{Non-perturbative Topos}}

In this part we build the context category of our new topos model. Then we show the existence of a new class of computable Heyting algebras which encode truth-values of propositions about
non-perturbative aspects of gauge field theories.

\begin{defn} \label{base-cat-1}
Topological Hopf algebras $H_{{\rm DSE(\lambda g)}}^{{\rm cut}}$ generated by
solutions of Dyson--Schwinger equations under different running couplings in a given gauge field theory $\Phi$ with the strong bare coupling constant $g$ and their closed Hopf
subalgebras can be organized into a poset structure. For
each pair $(H_{1},H_{2})$ of objects, we can define arrows pointing from
$H_{1}$ to $H_{2}$ (i.e. $H_{1} \le H_{2}$) if and only if there
exists a homomorphism $H_{1} \rightarrow H_{2}$ of Hopf algebras
which is continuous with respect to the cut-distance topology.
\end{defn}

This poset can be seen as a category denoted by $\mathcal{C}^{{\rm non},g}_{\Phi}$.

\begin{lem}
$\mathcal{C}^{{\rm non},g}_{\Phi}$ is a small category.
\end{lem}

\begin{proof}
The existence of the graduation parameters on the renormalization Hopf algebra $H_{{\rm FG}}(\Phi)$ allow us to represent it in terms of an infinite system of Dyson--Schwinger equations generated by fixed point equations of vertex-type and edge-type Green's functions of the physical theory under different running couplings. Therefore $H_{{\rm FG}}(\Phi)$ can be seen as an object in the category $\mathcal{C}^{{\rm non},g}_{\Phi}$ which makes this category as a small category.
\end{proof}

The small category $\mathcal{C}^{{\rm non},g}_{\Phi}$ is useful to determine cut-distance topological neighborhoods of Feynman diagrams and expansions around a fixed Feynman diagram. These topological regions are Hausdorff and therefore they allow us to separate not weakly isomorphic Feynman diagrams from each other. It gives us an advantage to study the unique solution of a given Dyson--Schwinger equation in terms of Feynman diagrams which contribute to arbitrary small neighborhoods around the corresponding large Feynman diagram.

\begin{lem}
There exists a topos structure on the small category
$\mathcal{C}^{\rm{non},g}_{\Phi}$.
\end{lem}

\begin{proof}
The natural choice is the topos of presheaves on the category
$\mathcal{C}^{\rm{non},g}_{\Phi}$ (as the category of context). We denote
this new category by ${\rm \textbf{T}}^{{\rm non},g}_{\Phi}$ and call it a non-perturbative topos.

An object in ${\rm \textbf{T}}^{{\rm non},g}_{\Phi}$ is a contravariant functor
from the category $\mathcal{C}^{{\rm non},g}_{\Phi}$ to the standard
category ${\rm {\bf Set}}$ of sets and functions.

A morphism between a pair $(F_{1},F_{2})$ of objects is a natural
transformation such as $\eta:F_{1} \rightarrow F_{2}$ which is
actually a family of morphisms $\{\eta_{H}: F_{1}(H) \rightarrow
F_{2}(H)\}_{H\in {\rm Obj}(\mathcal{C}^{{\rm non},g}_{\Phi})}$ with respect
to the contravariant property. It means that for each morphism
$f:H_{1} \rightarrow H_{2}$ in $\mathcal{C}^{{\rm non},g}_{\Phi}$, we have
$\eta_{H_{1}} \circ F_{1}(f) = F_{2}(f) \circ \eta_{H_{2}}$.

A sieve on an object $H \in \mathcal{C}^{{\rm non},g}_{\Phi}$ is defined as a
collection $S$ of morphisms $f:H \longrightarrow H'$ in
$\mathcal{C}^{{\rm non},g}_{\Phi}$ with the property that if $f$ belongs to
$S$ and if $g:H' \rightarrow H''$ is any other morphism in
$\mathcal{C}^{{\rm non},g}_{\Phi}$, then $g \circ f: H \rightarrow H''$
also belongs to $S$.

The terminal object $1:\mathcal{C}^{{\rm non},g}_{\Phi} \rightarrow {\rm
{\bf Set}}$ can be defined by $1(H):=\{*\}$ at all stages $H$ in
$\mathcal{C}^{{\rm non},g}_{\Phi}$. If $f:H \rightarrow H'$ is a morphism
in $\mathcal{C}^{{\rm non},g}_{\Phi}$, then $1(f): \{*\} \rightarrow
\{*\}$. It is the terminal object because for any other presheaf $F$
we can define a unique natural transformation $\eta: F
\rightarrow 1$ such that its components $\eta_{H}: F(H)
\rightarrow 1(H)=\{*\}$ are the constant maps $\Gamma
\mapsto *$ for all $\Gamma \in F(H)$.

The subobject classifier $\Omega^{{\rm non}}$ is the presheaf
$\Omega^{{\rm non}}: \mathcal{C}^{{\rm non},g}_{\Phi} \rightarrow {{\rm
{\bf Set}}}$ such that for any object $H \in \mathcal{C}^{{\rm non},g}_{\Phi}$,
$\Omega^{{\rm non}}(H)$ is identified by the set of all sieves on
$H$. If $f:H' \rightarrow H''$ is a morphism in
$\mathcal{C}^{{\rm non},g}_{\Phi}$, then $\Omega^{{\rm non}}(f): \Omega^{{\rm
non}}(H'') \rightarrow \Omega^{{\rm non}}(H')$ is given by
\begin{equation}
\Omega^{{\rm non}}(f)(S):= \{h:H'' \rightarrow H''', \ \ \ h
\circ f \in S\}
\end{equation}
for all sieves $S$ which lives in $\Omega^{{\rm non}}(H)$.

The outer presheaf is the contravariant functor $\underline{O}$ which maps each Hopf subalgebra $H_{{\rm DSE}}$ in the base category to the set ${\rm In}({\rm DSE})$ of all infinitesimal characters corresponding to Feynman diagrams in $H_{{\rm DSE}}$.

The spectral presheaf is the contravariant functor $\underline{\Sigma}$ which sends each Hopf subalgebra $H_{{\rm DSE}}$ in the base category to its corresponding complex Lie group of characters.
\end{proof}

More details about the basic structure of the non-perturbative topos has been studied by author in other separate research work.

Heyting algebras are practical models of the intuitionistic logic
where we do not have the law of excluded middle. It means that the
proposition $\phi \vee \neg \phi$ is not intuitionistically valid. The importance of this class of logical algebras in physics have been clarified when people started to describe quantum logics in terms of topos models.
\cite{b1, heyting-1, i1, ib1, ib2, lambek-scott-1}

\begin{defn}
A Heyting algebra $A$ is a bounded distributive lattice with the
largest element $1$ and the smallest element $0$ which obeys this
condition that for each couple $(a,b)$ of its elements there exists
a greatest element $x \in A$ such that $a \wedge x \le b$. This
particular element is called the relative pseudo-complement of $a$
with respect to $b$. $A$ is called a complete Heyting algebra, if it
is a complete lattice.
\end{defn}

\begin{thm} \label{topos-dse-6}
The topos ${\rm \textbf{T}}^{{\rm non},g}_{\Phi}$ can encode the logical evaluation of
propositions about topological regions of Feynman diagrams.
\end{thm}

\begin{proof}
Truth objects corresponding to cut-distance topological regions of
Feynman diagrams can be determined by the Heyting algebraic
structure defined naturally on the subobject classifier of the topos
${\rm \textbf{T}}^{{\rm non},g}_{\Phi}$.

For a given topological Hopf algebra  $H_{{\rm DSE}}^{{\rm cut}}$,
consider the space $\Omega^{{\rm non}}(H_{{\rm DSE}}^{{\rm cut}})$
which contains all sieves on $H_{{\rm DSE}}^{{\rm cut}}$. Now for
arbitrary sieves $S_{1}, S_{2}$ on $H_{{\rm
DSE}}^{{\rm cut}}$ which live in $\Omega^{{\rm non}}(H_{{\rm
DSE}}^{{\rm cut}})$, the partial order relation on $\Omega^{{\rm
non}}(H_{{\rm DSE}}^{{\rm cut}})$ can be determined naturally by
the relation
\begin{equation} \label{e1}
S_{1} \le S_{2} \Leftrightarrow S_{1} \subset S_{2}.
\end{equation}
It leads us to make the following elementary logical statements
$$S_{1} \wedge S_{2}:= S_{1} \cap S_{2}, \ \ \ \
S_{1} \vee S_{2}:= S_{1} \cup S_{2}, $$
$$S_{1} \Rightarrow S_{2}:=$$
\begin{equation}
\{f: H_{{\rm DSE},1}^{{\rm cut}} \rightarrow H_{{\rm
DSE},2}^{{\rm cut}} \ {\rm s.t.} \ \forall g: H_{{\rm DSE},2}^{{\rm cut}}
\rightarrow H_{{\rm DSE},3}^{{\rm cut}}, \ g \circ f \in S_{1}
\Rightarrow g \circ f \in S_{2}\}.
\end{equation}
The negation of an element $S$ is defined by the proposition
$\neg S:= S \Longrightarrow 0$ which means that
\begin{equation}
\neg S:= \{f: H_{{\rm DSE},1}^{{\rm cut}} \rightarrow H_{{\rm
DSE},2}^{{\rm cut}} \ {\rm s.t.} \ \forall g: H_{{\rm DSE},2}^{{\rm cut}}
\rightarrow H_{{\rm DSE},3}^{{\rm cut}}, \ \ g \circ f \notin
S\}.
\end{equation}

Thanks to the defined partial order (\ref{e1}), for $S_{1},
S_{2} \in \Omega^{\rm non}(H_{{\rm DSE}}^{{\rm cut}})$, there exists
a proposition $S_{1} \Rightarrow S_{2}$ of $\Omega^{{\rm
non}}(H_{{\rm DSE}}^{{\rm cut}})$ with the property that for all $S
\in \Omega^{{\rm non}}(H_{{\rm DSE}}^{{\rm cut}})$,
\begin{equation}
S \le (S_{1} \Rightarrow S_{2}) \Leftrightarrow S \wedge S_{1} \le
S_{2}.
\end{equation}
In addition, the unit element in $\Omega^{{\rm non}}(H_{{\rm
DSE}}^{{\rm cut}})$ is the principal sieve on $H_{{\rm DSE}}^{{\rm
cut}}$ and the null element is the empty sieve $\emptyset$.

The presheaf $\Omega^{{\rm non}}$ as the
subobject classifier shows that subobjects of any object $F$ in the
topos ${\rm \textbf{T}}^{{\rm non},g}_{\Phi}$ are in an one to one correspondence with
morphisms such as $\chi: F \rightarrow \Omega^{{\rm non}}$. In
other words, for a subobject $K$ of $F$, its associated
characteristic morphism $\chi^{K}$ is determined by its components
$\chi^{K}_{H_{{\rm DSE}}^{{\rm cut}}}: F(H_{{\rm DSE}}^{{\rm cut}})
\rightarrow \Omega^{{\rm non}}(H_{{\rm DSE}}^{{\rm cut}})$ where
\begin{equation}
\chi_{H_{{\rm DSE}}^{{\rm cut}}}^{K}(A):= \{f:H_{{\rm DSE}}^{{\rm
cut}} \rightarrow H_{{\rm DSE},1}^{{\rm cut}}: \ \ F(f)(A) \in
K(H_{{\rm DSE},1}^{{\rm cut}})\},
\end{equation}
for all $A \in F(H_{{\rm DSE}}^{{\rm cut}})$, is actually a sieve on
$H_{{\rm DSE}}^{{\rm cut}}$. Furthermore, each morphism $\chi: F
\rightarrow \Omega^{{\rm non}}$, which is a natural
transformation between presheaves, defines a subobject $K^{\chi}$ of
$F$ which is given by
\begin{equation}
K^{\chi}(H_{{\rm DSE}}^{{\rm cut}}):=\chi_{H_{{\rm DSE}}^{{\rm
cut}}}^{-1}\{1_{\Omega^{{\rm non}}(H_{{\rm DSE}}^{{\rm cut}})}\}.
\end{equation}

As the conclusion, for each equation DSE, $(\Omega^{{\rm non}}(H_{{\rm
DSE}}^{{\rm cut}}), \le, \wedge, \vee, \rightarrow)$ is our
promising Heyting algebra.
\end{proof}

A Heyting algebra is called finitely free, if it is generated by the
equivalence classes of formulas of finite number of propositional
variables under provable equivalence in the intuitionistic logic.

A subset $A$ of natural numbers is called computable if there exists
an algorithm to decide whether a natural number belongs to $A$ or
not. In other words, $A$ is computable if its corresponding
characteristic function is computable. An algebraic structure is
called computable if its domain can be identified with a computable
set of natural numbers where the (finitely many) operations and
relations on the structure are computable. If the structure is
infinite, people usually identify the cardinal of its domain with
the symbol $\omega$. The computable dimension of a computable
structure is the number of classically isomorphic computable copies
of the structure up to the computable isomorphism.  \cite{soare-1}

\begin{defn} \label{heyting-1}
A Heyting algebra $(H,\le,\wedge,\vee,\rightarrow)$ is called
computable, if $H$ and all its corresponding operations are
computable.
\end{defn}

For a given Heyting algebra with one generator, there exist
infinitely nonequivalent intuitionistic formulas of one
propositional variable. The connection between free Heyting algebras
and the intuitionistic logic leads us to the concept of
"computable dimension" for Heyting algebras in particular the ones which can
encode the logics of the non-perturbative toposes ${\rm \textbf{T}}^{{\rm non},g}_{\Phi}$.

\begin{thm}  \label{topos-non-3}
There exists a computable Heyting algebra which encodes truth
objects associated to topological regions of Feynman diagrams which
contribute to the unique solution of the Dyson--Schwinger
equation DSE in a given gauge field theory $\Phi$.
\end{thm}

\begin{proof}
We work on the non-perturbative topos ${\rm \textbf{T}}^{{\rm non},g}_{\Phi}$.
Thanks to Theorem \ref{topos-dse-6}, we can associate the Heyting
algebra $\Omega^{{\rm non}}(H_{{\rm DSE}}^{{\rm cut}})$ to each
combinatorial Dyson--Schwinger equation DSE. Therefore the subobject
classifier $\Omega^{{\rm non}}$ has the internal structure of our interesting
Heyting algebra (as the algebraic structure appropriate for
the intuitionistic logic). We want to lift this logical type of algebra
onto a enriched version $\hat{\Omega}^{{\rm non}}$ which is
computable at the level of dimension.

Consider the propositional intuitionistic logic over the given
language $(\Omega^{{\rm non}}(H_{{\rm DSE}}^{{\rm cut}}),\wedge,
\vee, \rightarrow, \bot, \top)$ such that $\Omega^{{\rm
non}}(H_{{\rm DSE}}^{{\rm cut}})$ can be seen as the collection of
propositional formulas in infinitely many variables modulo
equivalence under the intuitionistic logic where
$\wedge,\vee,\rightarrow$ are the logical connectives, $\bot$ is
false and $\top$ is truth. The codes for formulas such as $\phi
\wedge \psi, \ \ \phi \vee \psi$ or $\phi \longrightarrow \psi$ are
always greater than the codes for $\phi$ and $\psi$.

The intuitionistic propositional logic is decidable (\cite{coecke-1,
turlington-1}) which means that we need a finite constructive process to apply uniformly to every propositional formula to understand either it produces an intuitionistic proof of the formula or it shows no such proof can exist. Therefore we have a computable copy of the
free Heyting algebra on $\omega$ generators. Now we can consider
elements of $\Omega^{{\rm non}}(H_{{\rm DSE}}^{{\rm cut}})$ as the
equivalence classes $[\phi]$ under provable equivalence in the
intuitionistic logic which leads us to the following computational
operations.
$$[\phi] \le [\psi] \Leftrightarrow \phi \rightarrow \psi \ \ {\rm
is \ provable \ under \ the \ intuitionistic \ logic,}$$
\begin{equation}
[\phi] \wedge [\psi] = [\phi \wedge \psi], \ \ \ \ \  [\phi] \vee
[\psi] = [\phi \vee \psi].
\end{equation}

The plan is to build $\hat{\Omega}^{{\rm non}}$ as a computable copy
which is not computability isomorphic to $\Omega^{{\rm non}}$. Let
$\alpha_{s}(n)$ be a label at stage $s$ determined by the domain of
$\hat{\Omega}^{{\rm non}}$ in the construction process. It is a
propositional formula in the intuitionistic logic such that \\
- $\alpha(n) = {\rm lim}_{s} \alpha_{s}(n)$, \\
- For $n \neq m$, the propositional formulas $\alpha(n)$ and
$\alpha(m)$ are not intuitionistically equivalent, \\
- For each intuitionistic propositional formula $\phi$, there exists
such $n$ such that $\alpha(n)$ is intuitionistically equivalent to
$\phi$, \\
- Morphisms with the general form $\phi_{e}:
\hat{\Omega}^{{\rm non}} \rightarrow \Omega^{{\rm non}}$ can be applied to deal with the diagonalization against all possible computable isomorphisms (\cite{turlington-1}).

Once we define the join, meet or relative
pseudo-complement of elements, these relationships never change in
future stages. Therefore, the function $\alpha$, which indicates an
isomorphism map between $\hat{\Omega}^{{\rm non}}$ and $\Omega^{{\rm
non}}$, makes $\hat{\Omega}^{{\rm non}}$ computable.
\end{proof}


\chapter{\textsf{Conclusion}}

\newpage

In this part we summarize the original achievements of this monograph which focused on the mathematical foundations of non-perturbative gauge field theories in the context of combinatorial, geometric--analytic and categorical settings.

\textbf{(A)} We have studied solutions of Dyson--Schwinger equations in terms of an infinite combinatorial setting. We applied the theory of analytic graphs (for sparse graphs) to formulate non-trivial graphon models of Feynman diagrams and their formal expansions. Then we encapsulated the renormalization of these graphons in the context of a new topological Hopf algebra $\mathcal{S}^{\Phi}_{{\rm graphon}}$ of Feynman graphons with respect to Feynman diagrams of a given (strongly) coupled gauge field theory $\Phi$. The topology of this Hopf algebra is the topology of graphons, its coproduct is derived from the Kreimer's renormalization coproduct and its graduation parameter can be determined in terms of loop numbers of the corresponding Feynman diagrams. The compactness of the topology of graphons enables us to identify Feynman graph limits as the convergent limits of sequences of finite Feynman diagrams. In addition, solutions of fixed point equations of Green's functions have also been studied as the cut-distance convergent limits of the sequences of their corresponding partial sums. We then applied these Feynman graphon models of Dyson--Schwinger equations to formulate a new enrichment of the BPHZ renormalization theory for infinite formal expansions of Feynman diagrams. Our study completes the foundations of a differential Galois theory (in terms of the Riemann--Hilbert problem) for the study of non-perturbative parameters derived from the renormalization of Dyson--Schwinger equations under strong running or bare coupling constants.

We have also studied Dyson--Schwinger equations in the language of some combinatorial polynomials. We built a new parametric representation for solutions of these non-perturbative equations in terms of Tutte polynomials and Kirchhoff--Symanzik polynomials. We then formulated a new multi-scale Renormalization Group on the
collection $\mathcal{S}^{\Phi,g}$ of all Dyson--Schwinger equations in a given gauge
filed theory in terms of changing simultaneously the scales of
momenta (i.e. running) and bare coupling constants. This Renormalization Group is useful to study strongly coupled Dyson--Schwinger equations in terms of cut-distance convergent limits of sequences of large Feynman diagrams under weaker running couplings. Feynman graphon models and the topology of graphons are the essential tools for this fundamental result. In addition, this Renormalization Group has been applied to formulate a new concept of complexity for Dyson--Schwinger equations under different running coupling constants. We considered $\mathcal{S}^{\Phi,g}$ as a new constructive world to formulate a new generalization of the Kolmogorov complexity for Dyson--Schwinger equations. We showed that our generalized non-perturbative BPHZ renormalization program can encode the Halting problem of partial recursive functions on $\mathcal{S}^{\Phi,g}$.

\textbf{(B)} We have built a new Noncommutative Geometry model for the study of Dyson--Schwinger equations in terms of a new class of infinite dimensional spectral triples. These spectral triples are derived from graded Hopf subalgebras generated by solutions of Dyson--Schwinger equations. The strength of running couplings has a direct influence on the structures of these spectral triples. This is a fundamental achievement for the description of the geometry of non-perturbative quantum motions via noncommutative differential geometry and the theory of Spectral Geometry.

Thanks to Feynman graphon models, we have explained the foundations of a new Functional Analysis approach for the study of solutions of Dyson--Schwinger
equations in terms of the Haar integration theory on $\mathcal{S}^{\Phi,g}$ as a
modification of the classical Riemann--Lebesgue integration theory
with respect to the Borel $\sigma$-algebra on real numbers. As we have shown this new approach is useful for the analytic description of evolution of large Feynman diagrams in terms of sequences of their corresponding partial sums where we have formulated a new generalization of .
Johnson--Lapidus Dyson series for strongly coupled Dyson--Schwinger equations.

We have also worked on the Banach
algebra $L^{1}(\mathcal{S}^{\Phi,g},\mu_{{\rm Haar}})$ to formulate
a generalized version of the Fourier transformation on the basis of the Gelfand transform. It is useful to encode the evolution of Dyson--Schwinger equations in terms of $\mu_{{\rm Haar}}$-integrable functions on $\mathcal{S}^{\Phi,g}$.

In addition, we explained the basic elements of
the G\^{a}teaux differential calculus machinery on the separable Banach space
$\mathcal{S}^{\Phi,g}$ with respect to the cut-norm to study smooth
functions on $\mathcal{S}^{\Phi,g}$ in the language of Taylor series
of higher order G\^{a}teaux differentiations and homomorphism
densities.

\textbf{(C)} Thanks to Feynman graphon models of Dyson--Schwinger equations, we have obtained a new mathematical model for the description of quantum entanglement in interacting (strongly coupled) gauge field theories in the language of intermediate algorithms organized in some new lattice models. These intermediate algorithms have been encoded via lattices of topological Hopf subalgebras generated by solutions of Dyson--Schwinger equations. We then lifted these lattices on to the level of lattices of Lie subgroups and Tannakian subcategories. Our study provides a new bridge between information flow in Quantum Field Theory and the Theory of Computation and Complexity.

In addition, we have organized topological Hopf subalgebras derived from solutions of Dyson--Schwinger equations in a given gauge field theory into a new topos. This new topos can provide the logical foundations of gauge field theories under strongly coupled running couplings. This new topos model has been developed by Author in his recent research works. We have shown the existence of a new class of computable Heyting algebras which encode
the evaluation of logical propositions about topological regions of Feynman
diagrams.

\end{document}